\documentclass{amsart}
\usepackage [latin1]{inputenc}
\usepackage{fullpage,graphicx,subfigure,mathpazo,color}
\usepackage{amsmath,amscd,tikz,mathrsfs}
\usepackage[normalem]{ulem}
\usepackage{amsmath}
\usepackage{epstopdf}
\usepackage{tikz}
\usepackage{setspace}%ÃÂ¹ÃÃÅÃ¤Å¸Ã ÂºÃªÂ°Ã¼
\newcommand{\ii}{\mathrm{i}}
\newcommand{\ee}{\mathrm{e}}
\newcommand{\dd}{\mathrm{d}}
\newcommand{\T}{\mathrm{T}}
\newcommand{\modify}[1]{{\color{black} #1}}
\newcommand{\lm}[1]{{\color{black} #1}}
\newcommand{\ys}[1]{{\color{black} #1}}
\newcommand{\lmn}[1]{{\color{black} #1}}

\newcommand{\cn}{\mathrm{cn}}

\newtheorem{rhp}{Riemann-Hilbert Problem}
\newtheorem{theorem}{Theorem}
\newtheorem{lemma}{Lemma}
\newtheorem{prop}{Proposition}
\newtheorem{corollary}{Corollary}

\newtheorem{remark}{Remark}

\usepackage{graphicx}      % insert graphic
\usepackage{titlesec}

\titleformat{\section}{\centering\LARGE\bfseries}{\thesection}{1em}{}
\titleformat{\subsection}{\Large\bfseries}{\thesubsection}{1em}{}
\begin{document}

\title{Large and infinite order solitons of the coupled nonlinear Schr\"odinger equation}

\author{Liming Ling}
\address{School of Mathematics, South China University of Technology, Guangzhou, China 510641}
\email{linglm@scut.edu.cn}
\author{Xiaoen Zhang}
\address{School of Mathematics, South China University of Technology, Guangzhou, China 510641}
\email{zhangxiaoen@scut.edu.cn}

\begin{abstract}
We study the large order and infinite order soliton of the coupled nonlinear Schr\"{o}dinger equation with the Riemann-Hilbert method. By using the Riemann-Hilbert representation for the high order Darboux dressing matrix, the large order and infinite order solitons can be analyzed directly without using inverse scattering transform.  We firstly disclose the asymptotics for large order soliton, which is divided into four different regions---the elliptic function region, the non-oscillatory region, the exponential and algebraic decay region. We verify the consistence between asymptotic expression and exact solutions by the Darboux dressing method numerically. Moreover, we consider the property and dynamics for infinite order solitons---a special limitation for the larger order soliton.  It is shown that the elliptic function and exponential region will disappear
for the infinite order solitons.


\end{abstract}

\date{\today}

\maketitle

\section{Introduction}

%\lm{Lacking the statement on the significance of high order soliton.}

\ys{As one of the most important models in the nonlinear mathematical physics, the nonlinear Schr\"{o}dinger (NLS) equation governs both nonlinear optics fibers and the Bose-Einstein condensation (BEC). It can also be applied in other fields, involving the fluid mechanics, the plasma and even the finance\cite{Chiao-PRL-1964,Zakharov-JAM-1968,Yan-PLA-2011}. But when there appears coupling effect between more than two components, the scalar NLS equation can not describe it any more, instead, the coupled systems have been paid much more attentions. %The study of the coupled NLS(CNLS) equation provide a good platform to investigate the characteristics of atomic matter waves.
Compared to the scalar equation, the coupled models have more diversified dynamics properties, such as in 1997 \cite{Rad-JPA-1995}, the authors denominated that there exist the inelastic collision in the interaction of two bright soliton \cite{Rad-PRE-1997}, %and these two components behavior differently,
which was not observed in the scalar NLS equation.
It is well known that many excellent phenomena have been observed to have an intimate connection with the coupled NLS (CNLS) equation, such as the vortices, the bright-bright soliton, the dark-dark soliton, the bright-dark soliton. %With the theory of quantum wells, the bright soliton corresponds to the ground state, which has an essential difference to the dark soliton.
Therefore, it is of significance to investigate the multi-component system and discover its more unknown dynamics.

Similar to the scalar NLS equation, the CNLS equation is also integrable, both of them have the soliton solutions. As one of the most popular solutions in the integrable system, the solitons as well as its dynamic behavior have been studied for a long time, which involves the high-order solitons and the multi-solitons. In the terminology of inverse scattering transform \cite{GGKM-PRL-1967}, $N$-solitons are given in the reflectionless case under the transmission scattering data $a(\lambda)$ has $N$ distinct simple zeros, so the amplitude and the velocity to every soliton is different, while the $N$-th order soliton is given with $a(\lambda)$ has $N$ multiple zeros, which indicates the interaction between $N$ solitons of equal amplitude but having a particular chirp. The concept ``chirp" comes from the nonlinear optics and it can be generated from the self-phase modulation with the intensity-dependent refractive index \cite{Tom-JOSAB-1984}. From the physical viewpoints, the high order soliton can be used to estimate the compression factor by $F_c\approx N/1.6$, where $N$ is the soliton order \cite{Tom-JOSAB-1984}. Moreover, in \cite{Li-JOSAB-2010}, the authors studied the cascaded $N$-soliton to the non-adiabatic pulse compression theoretically and obtained the cascaded $N$-soliton compression numerically. From a mathematical viewpoint, it is important to consider how the soliton will behavior when $N$ is large or $N\to\infty$. This concept was first put forward by Zhou \cite{Zhou-CPAM-1990}, afterwards, Gesztesy et. al studied the infinite many soliton to the KdV equation comprehensively \cite{Gesztesy-Duke-1992}, which contains the convergence, the spectral properties and the solution behavior when $N\to\infty$. Furthermore, the long time asymptotics of the NLS equation with infinite order was analyzed via the Riemann-Hilbert factorization of the inverse scattering method \cite{Kamvissis-JMP-1995}. In \cite{Peter-CPAM-2007}, the authors gave a detailed analysis to the multi-soliton when $N$ is large. Recently, Bilman, one of the authors and Miller obtained a novel solution---rogue waves of infinite order via the robust inverse scattering method and Darboux transformation, as they studied the near field asymptotics for the high order rogue waves, in which the solutions were related with the Painlev\'{e}-III hierarchy \cite{Peter-Duke-2019}. In \cite{Deniz-JNS-2019,Deniz-arXiv-2019}, the authors analyzed the near field and far field asymptotics for high order soliton (or multi-pole soliton) with large order for the NLS equation, where the asymptotics of near field will converge to the rogue wave of infinite order rogue waves and the far field limit of high order soliton with large order contained four different asymptotic regions. Very recently, Deniz and Miller considered the analysis about the far-field limit for the large order rogue waves, in which the high order soliton and rogue waves were analyzed in a uniform way \cite{BilmanM-21}.

However, when it comes to the higher-order matrix spectral problem, the corresponding asymptotic analysis becomes more difficult, one should take some skills to decompose the jump matrix into the upper and lower triangles. Thus there are only few results about the long-time asymptotics to the higher-order matrix spectral problem, such as the CNLS equation \cite{Geng-JNS-2018}, the coupled mKdV equation \cite{Geng-JGP-2019}, the spin-one Gross-Pitaevskii equation \cite{Geng-CMP-2021}, the Sasa-Satuma equation \cite{Huang-JDE-2020}, the Degasperis-Procesi equation \cite{Anne-AIFG-2019} and ``good" Boussinesq equation \cite{Lenells-arxiv-2020}. While to the best of our knowledge, there is no result of large order solitons for the higher-order matrix spectral problem due to the its own difficulty. In this work, we intend to extend the asymptotic analysis of large order soliton with the aforementioned method to the CNLS equation associated with a $3\times 3$ matrix spectral problem.}
%Moreover, Eq.\eqref{eq:cnls} was studied with the inverse scattering method on both zero boundary condition \cite{Yang-book-2010} and nonzero boundary condition with a special constraint . In \cite{Pelinovsky-IJTP-2014}, Pelinovsky et.al. proved the generalized Eq.\eqref{eq:cnls} is globally well-posed in $H^{1}(\mathbb{R})$. In \cite{Liu-Nonli-2019}, Liu establishes the bijectivity on the weighted Sobolev space $H^{i}, (i=1,2)$ to the $3\times 3$ AKNS system, which involves Eq.\eqref{eq:cnls} as a special flow reduction. Under the soliton-free case, the long-time asymptotics with Schwartz initial conditions was studied in \cite{Geng-JNS-2018}, which becomes more difficult for the decomposition of the jump matrix than the $2\times 2$ Lax pair. Apart from the asymptotic analysis to CNLS equation, some other high order Lax pair equations also pay much attention, such as, Lenells et.al. studied the asymptotics for the Sasa-Satsuma equation\cite{Huang-JDE-2020}, ``good" Boussinesq equation\cite{Lenells-arxiv-2020} and Degasperis-Procesi equation\cite{Anne-AIFG-2019}, which all need the new skill to tackle with the jump matrix. \lm{Adding the recently results on Geng's.}
In general, the high-order soliton can be represented as a determinant form, which is hard to be analyzed as $N$ is large. As a result, it is impossible to study the asymptotics on the basis of calculation the Darboux matrix. Thus we are preparing to construct a Riemann-Hilbert problem to convert the large order problem onto the jump matrix, then the large order asymptotics can be studied with the nonlinear steepest descend method. Furthermore, if $N\to\infty$, how the asymptotics will become? This problem is different from the large order problem essentially, which reflects an infinite dimensional system \cite{Shabat-FAA-1993,Shabat-IP-1992}. In this paper, we are preparing to disclose these two asymptotics by the Riemann-Hilbert method. Thus we give some preliminaries about the CNLS equation. Before we progress, we give a notational convention in the following remark.
\begin{remark}
We use the superscript $\lambda^*$ to denote the conjugate of $\lambda$, which is also applied to the matrix matrix notation. We use the ``dagger" natation $\mathbf{q}^{\dagger}$ to the Hermitian conjugate of $\mathbf{q}$. As the Lax pair for the CNLS is $3\times 3$, thus we give some special natation matrix
\begin{equation}
\begin{split}
\pmb{\sigma}_3={\rm diag}\left(1, -\mathbb{I}_2\right), \quad \pmb{\hat{\sigma}}_3={\rm diag}\left(1, -1\right), \quad \pmb{\sigma}_2=\begin{bmatrix}0&-\ii\\\ii&0\end{bmatrix},
\end{split}
\end{equation}
where $\mathbb{I}$ indicates the identity. We use the boldface capital letters to denote the matrix and the bold lowercase letter to denote the column vectors.
\end{remark}
%As to the integrable system, it has the Lax equation
%\begin{equation}
%\mathbf{L}_{t}(\lambda)=\left[\mathbf{M}(\lambda), \mathbf{L}(\lambda)\right].
%\end{equation}
%If $\mathbf{L}(\lambda)$ has distinct eigenvalues at the neighbourhood of $\lambda_k$, then $\mathbf{L}(\lambda)$ can be decomposed as
%\begin{equation}
%\mathbf{L}(\lambda)=\mathbf{g}^{(k)}(\lambda)\mathbf{A}^{(k)}(\lambda)\mathbf{g}^{(k)-1}(\lambda),
%\end{equation}
%where $\mathbf{A}^{(k)}(\lambda)$ is diagonal matrix and has a pole of order $n_k$ and $\lambda_k$. Then the poisson brackets of the matrix elements of $\mathbf{L}(\lambda)$ can be written as
%\begin{equation}
%\{\mathbf{L}_{1}(\lambda), \mathbf{L}_{2}(\mu)\}=-\left[\frac{\mathbf{C}_{12}}{\lambda-\mu}, \mathbf{L}_{1}(\lambda)+\mathbf{L}_{2}(\mu) \right],
%\end{equation}
%with $\mathbf{L}_{1}(\lambda)=\mathbf{L}(\lambda)\bigotimes 1, \mathbf{L}_{2}(\lambda)=1\bigotimes\mathbf{L}(\lambda), \mathbf{C}_{12}=\sum_{i,j}\mathbf{E}_{ij}\bigotimes\mathbf{E}_{ji}$, and $\mathbf{E}_{ij}$ are the canonical basis matrices.

\subsection{Review on the CNLS equation}

The focusing CNLS equation reads in the vector form:
\begin{equation}\label{eq:cnls}
\ii \mathbf{q}_t+\frac{1}{2}\mathbf{q}_{xx}+\mathbf{q}\mathbf{q}^{\dagger}\mathbf{q}=0, \qquad \mathbf{q}=\left(q_1(x,t), q_2(x,t)\right),
\end{equation}
where $q_1$ and $q_2$ are the wave envelopes, %the subscripts $_x$ and $_t$ are the partial derivatives with respect to the space and time variables.
%Like the scalar nonlinear Schr\"odinger equation, Eq.\eqref{eq:cnls} also has the Lax pair and infinite number conservation laws. This focusing CNLS equation
which was first derived by Manakov to describe the electric propagation in 1974 \cite{Manakov-JETP-1974}. So it is also named the Manakov model. Subsequently, the focusing CNLS equation \eqref{eq:cnls} has attracting much attentions in different areas, ranging from the BEC \cite{Busch-PRL-2001} to the optical fibres \cite{Agrawal} and bio-physics \cite{Scott-PS-1984}. From the mathematical viewpoints, the CNLS equation \eqref{eq:cnls} is an integrable equation admitting the following Lax pair
\begin{equation}\label{eq:laxpair}
\begin{split}
\mathbf{\Phi}_{x}&=\mathbf{U}(\lambda; x, t)\mathbf{\Phi}, \qquad \mathbf{U}(\lambda; x, t)=\ii\left(-\lambda \pmb{\sigma}_3+\pmb{\mathcal{Q}}\right),\\
\mathbf{\Phi}_t&=\mathbf{V}(\lambda; x, t)\mathbf{\Phi}, \qquad \mathbf{V}(\lambda; x, t)=\ii\lambda\left(-\lambda\pmb{\sigma}_3+\pmb{\mathcal{Q}}\right)
+\frac{1}{2}\left(\pmb{\mathcal{Q}}_x+\ii\pmb{\mathcal{Q}}^2\right)\pmb{\sigma}_3,
\end{split}
\end{equation}
where $\lambda\in\mathbb{C}$ is the spectral parameter,
\begin{equation*}
\pmb{\mathcal{Q}}=\begin{bmatrix}
0&\mathbf{q}\\
\mathbf{q}^{\dagger}&0
\end{bmatrix}.
\end{equation*}
The compatibility condition of Lax pair \eqref{eq:laxpair}: $\pmb{\Phi}_{xt}=\pmb{\Phi}_{tx}$ gives the focusing CNLS equation \eqref{eq:cnls}.
Due to integrablity, the CNLS equation \eqref{eq:cnls} had been studied widely on both zero and non-zero boundary condition by the inverse scattering method \cite{Manakov-JETP-1974,KrausBK-15}.

We intend to study Eq.\eqref{eq:cnls} with the Riemann-Hilbert method, thus we present some brief reviews on the inverse scattering transform for CNLS in Appendix \ref{app:IST}. Under the framework of this method, the $n$-soliton solution is given. It is natural to ask how about the large limit $n\to\infty.$ As we know that, if there exists a sequence $\{\lambda_i,i=1,2,\cdots\}$ in a compact region, there is a sub-sequence $\{\lambda_{n_i}\}$ will converge into $\lambda_0$. In which, one of the special case is $\{\lambda_{n_i}=\lambda_0\}$ corresponding to the high order solitons. The high order solitons can be solved by the inverse scattering method. To analyze the dynamics for infinite order solitons for CNLS, we would like to combine the Darboux transformation with Riemann-Hilbert representation.% under the framework of recent proposed robust inverse scattering method.
Then we can use the Deift-Zhou method to analyze the asymptotics for large order and infinite order solitons. Before the detailed analysis, we discuss the asymptotics to Darboux matrix under special case.

Then we give a brief review on the soliton solutions for the focusing CNLS equation \eqref{eq:cnls}. In 1974, the bright soliton to Eq.\eqref{eq:cnls} has been given with the inverse scattering method. In general, the single bright soliton has the form
\begin{equation}\label{eq:soliton}
\begin{split}
q_1&=2b\cos(\alpha)\ii\,{\rm sech}(2b(x+2at))\ee^{2\ii((b^2-a^2)t-ax-\beta)}, \\ q_2&=2b\sin(\alpha)\ii\,{\rm sech}(2b(x+2at))\ee^{2\ii((b^2-a^2)t-ax-\gamma)},\,\,\,\,\, \alpha\in\mathbb{R},
\end{split}
\end{equation}
which shows that both two components have the same shape as the solutions of scalar NLS equation, where $a\in\mathbb{R}$ and $b>0$. Furthermore, there exist a general nondegenerate or multi-hump fundamental solitons in Eq.\eqref{eq:cnls}, which admit that this collisions has no energy redistribution \cite{Ram-PRL-2019,QinZL-PRE-19}. While in some cases, two component system has a similar characteristics with the scalar one. For example, Eq.\eqref{eq:cnls} can also be regarded as a good model for the rogue wave, the interaction between the rogue wave and the solitons \cite{Guo-CPL-2011} as well as the Akemediev breathers, Ma solitons and the general breathers \cite{Akemediev-PRE-2013}.

It is well known that the elementary Darboux transformation in the framework of loop group \cite{TerngU-00} for the system \eqref{eq:laxpair} with $\mathbf{q}=\mathbf{0}$ is
\begin{equation}
\mathbf{T}^{[1]}(\lambda; x, t)=\mathbb{I}-\frac{\lambda_1-\lambda_1^*}{\lambda-\lambda_1^*}\mathbf{P}^{[1]}(x,t),\,\,\,\,\, \mathbf{P}^{[1]}(x,t)=\frac{\mathbf{\Phi}_1(x,t)[\mathbf{\Phi}_1(x,t)]^{\dagger}}{[\mathbf{\Phi}_1(x,t)]^{\dagger}\mathbf{\Phi}_1(x,t)},
\end{equation}
where $\mathbf{\Phi}_1(x,t)=\ee^{-\ii\lambda_1(x+\lambda_1 t)\pmb{\sigma}_3}\mathbf{c}$, $\mathbf{c}=(1,c_1,c_2)^{\T}$, which converts the Lax pair \eqref{eq:laxpair} into a new one by replacing the potential function $\pmb{\mathcal{Q}}=\mathbf{0}$ with a new one
\begin{equation}\label{eq:single-bt}
\pmb{\mathcal{Q}}^{[1]}=\pmb{\mathcal{Q}}+(\lambda_1-\lambda_1^*)\left[\pmb{\sigma}_3, \mathbf{P}^{[1]}(x,t)\right],
\end{equation}
\lm{
i.e.
\begin{equation}\label{eq:bt-1}
\left(q_1^{[1]}, q_2^{[1]}\right)=(q_1, q_2)+\frac{4\ii\Im(\lambda_1)\phi_1^{(1)}\left(\phi_1^{(2)*},\phi_1^{(3)*}\right)}{\left|\phi_1^{(1)}\right|^2+\left|\phi_1^{(2)}\right|^2+\left|\phi_1^{(3)}\right|^2}.
\end{equation}
By the Minkowski inequality, we have
\begin{equation}
\left\|\left(q_1^{[1]}, q_2^{[1]}\right)\right\|_2\leq \|(q_1, q_2)\|_2+\frac{4\left|\Im(\lambda_1)\phi_1^{(1)}\right| \left\|\left(\phi_1^{(2)*},\phi_1^{(3)*}\right)\right\|_2}{\left|\phi_1^{(1)}\right|^2+\left|\phi_1^{(2)}\right|^2+\left|\phi_1^{(3)}\right|^2}
\end{equation}
where $\|(a,b)\|_2=[|a|^2+|b|^2]^{1/2}$ and the equality holds only if $(q_1, q_2)=k\left(\phi_1^{(2)*},\phi_1^{(3)*}\right)$. Furthermore, by the mean inequality, we have the estimate
\begin{equation}\label{eq:estimate-v}
\left\|\left(q_1^{[1]}, q_2^{[1]}\right)\right\|_2\leq \|(q_1, q_2)\|_2+2|\Im(\lambda_1)|
\end{equation}
where the equality holds only if $|\phi_1^{(1)}|=\left\|\left(\phi_1^{(2)*},\phi_1^{(3)*}\right)\right\|_2$.
Thus, to obtain the single soliton has the maximum norm at origin, we set $c_1=\cos(\alpha)\ee^{\ii\beta}$ and $c_2=\sin(\alpha)\ee^{\ii\gamma}$.} The above equation \eqref{eq:single-bt} will yield the single bright soliton \eqref{eq:soliton} by choosing $\lambda_1=a+\ii b$.
We can readily see that the Darboux matrix $\mathbf{T}^{[1]}(\lambda;x,t)$ is a meromorphic matrix in the whole complex plane $\mathbb{C}$, which implies that the new wave function $\Phi^{[1]}(\lambda;x,t)=\mathbf{T}^{[1]}(\lambda;x,t)\ee^{-\ii\lambda(x+\lambda t)\pmb{\sigma}_3}[\mathbf{T}^{[1]}(\lambda;0,0)]^{-1}$ is analytic for $\lambda\in\mathbb{C}$ with $(x,t)\in \mathbb{R}^2$. The high order Darboux matrix and multi-fold one can be iterated recursively.
To obtain the asymptotics when $x\to\pm\infty$, we give the Darboux transformation with $N$ multiple poles as
\begin{equation}\label{eq:darboux12}
\mathbf{T}_{N}(\lambda; x, t)=\mathbf{T}^{[N]}(\lambda;x,t)\mathbf{T}^{[N-1]}(\lambda;x,t)\cdots\mathbf{T}^{[1]}(\lambda;x,t),
\end{equation}
where
\begin{equation}\label{eq:high-dt}
\begin{split}
\pmb{\Phi}^{[j]}(\lambda;x,t)&=\mathbf{T}^{[j]}(\lambda;x,t)\pmb{\Phi}^{[j-1]}(\lambda;x,t)\left(\mathbf{T}^{[j]}(\lambda;0,0)\right)^{-1},\,\,\,\, \pmb{\Phi}^{[0]}(\lambda;x,t)=\pmb{\Phi}(\lambda;x,t)\equiv\ee^{-\ii\lambda(x+\lambda t)\pmb{\sigma}_3} \\
\mathbf{T}^{[j]}(\lambda;x,t)&=\mathbb{I}-\frac{\lambda_1-\lambda_1^*}{\lambda-\lambda_1^*}\mathbf{P}^{[j]}(x,t),\qquad\mathbf{P}^{[j]}(x,t)=\frac{\pmb{\Phi}_1^{[j-1]}(x,t)\left(\pmb{\Phi}_{1}^{[j-1]}(x,t)\right)^{\dagger}}{\left(\pmb{\Phi}_{1}^{[j-1]}(x,t)\right)^{\dagger}\pmb{\Phi}_1^{[j-1]}(x,t)}, \\
\mathbf{\Phi}_1^{[j-1]}(x,t)&=\lim\limits_{\lambda\to\lambda_1}\pmb{\Phi}^{[j-1]}(\lambda;x,t)\lm{\left(\mathbf{c}+\epsilon_{j-1}^{[1]}\mathbf{v}_1+\epsilon_{j-1}^{[2]}\mathbf{v}_2\right)},\qquad j=1,\cdots, N,
\end{split}
\end{equation}
\lm{where $\mathbf{v}_1=(-c_1^*,1,0)^{\T}$ and $\mathbf{v}_2=(-c_2^*,0,1)^{\T}$ such that $\mathbf{c}\mathbf{v}_1^{\T}=0$ and $\mathbf{c}\mathbf{v}_2^{\T}=0$. }
The $N$-th order soliton solution is given by
\begin{equation}
\pmb{\mathcal{Q}}^{[N]}=(\lambda_1-\lambda_1^*)\sum_{i=1}^{N}\left[\sigma_3, \mathbf{P}^{[i]}(x,t)\right],
\end{equation}
which also can be represented in a determinant formula. \lm{By the above construction \eqref{eq:high-dt}, we know that the fundamental solution satisfies $\pmb{\Phi}^{[j]}(\lambda;0,0)=\mathbb{I}$. Together with the estimate \eqref{eq:estimate-v}, we find the high order soliton with the maximum peak of $2$-norm $\|\cdot\|_2$ will attain by choosing the parameters $\epsilon_{j}^{[i]}=0$, which is associated with the scalar NLS equation by ${\rm SU}(2)$ symmetry. Actually, when the parameters $\epsilon_{j}^{[i]}$ have a small perturbation, the solutions will not satisfy the scalar NLS equation under the ${\rm SU}(2)$ symmetry. The second order solitons for non-vanishing $\epsilon_{j}^{[i]}$ is shown in the Appendix \ref{appendix:second-order}. In this work, we mainly consider the special case that all $\epsilon_j^{[i]}$ are vanishing.}

\begin{lemma}\label{lem:asym-m}
Suppose ${\rm Im}(\lambda_1)>0$, when $x\to+\infty$, the asymptotics of $\mathbf{T}_{N}(\lambda; x, t)$ is
\begin{equation}
\mathbf{T}_{N}(\lambda;x,t)=\begin{bmatrix}\left(\frac{\lambda-\lambda_1}{\lambda-\lambda_1^*}\right)^{N}&0&0\\
0&1&0\\
0&0&1
\end{bmatrix}+\mathcal{O}(x^{N-1}\ee^{-2{\rm Im}(\lambda_1)x}),
\end{equation}
conversely, when $x\to-\infty$, the asymptotics of $\mathbf{T}_{N}(\lambda; x, t)$ is
\begin{equation}
\mathbf{T}_{N}(\lambda;x,t)=\begin{bmatrix}1&0&0\\
0&\frac{1}{2}+\frac{1}{2}\left(\frac{\lambda-\lambda_1}{\lambda-\lambda_1^*}\right)^{N}&-\frac{1}{2}+\frac{1}{2}\left(\frac{\lambda-\lambda_1}{\lambda-\lambda_1^*}\right)^{N}\\
0&-\frac{1}{2}+\frac{1}{2}\left(\frac{\lambda-\lambda_1}{\lambda-\lambda_1^*}\right)^{N}&\frac{1}{2}+\frac{1}{2}\left(\frac{\lambda-\lambda_1}{\lambda-\lambda_1^*}\right)^{N}
\end{bmatrix}+\mathcal{O}(x^{N-1}\ee^{2{\rm Im}(\lambda_1)x}).
\end{equation}
\end{lemma}

Similar to the NLS equation, the CNLS Eq.\eqref{eq:cnls} also has infinite number of conservation laws, which is shown in the following properties:
\begin{prop}
The first three conservation laws for the CNLS equation are
\begin{equation}
\begin{split}
I_1&=\int_{-\infty}^{+\infty}\left(|q_1|^2+|q_2|^2\right)dx=4N{\rm Im}(\lambda_1),\\ I_2&=-\frac{\ii}{2}\int_{-\infty}^{+\infty}\left(q_1q_{1,x}^*+q_2q_{2,x}^*-q_1^*q_{1,x}-q_2^*q_{2,x}\right)dx=8N{\rm Re}(\lambda_1){\rm Im}(\lambda_1),\qquad
\\I_3&=\int_{-\infty}^{+\infty}\left[|q_{1,x}|^2+|q_{2,x}|^2-\left(|q_1|^2+|q_2|^2\right)^2\right]dx
=\frac{16}{3}N{\rm Im}(\lambda_1)\left(3{\rm Re}(\lambda_1)^2-{\rm Im}(\lambda_1)^2\right).\end{split}
\end{equation}
\end{prop}

Observing the asymptotics of $\mathbf{T}_{N}(\lambda; x, t)$ in lemma \ref{lem:asym-m}, we can set the $\pmb{\Phi}^{\pm}(\lambda; x, t)$ in Eq.\eqref{eq:scattering} as
\begin{equation}\label{eq:phipn}
\pmb{\Phi}^{\pm}(\lambda; x, t)=\mathbf{T}_{N}(\lambda; x, t)\ee^{-\ii\lambda\left(x+\lambda t\right)\pmb{\sigma}_3}\mathbf{T}_{\pm}^{-1}(\lambda),
\end{equation}
where $\mathbf{T}_{\pm}(\lambda)$ are defined in Eq.\eqref{eq:T-infinity} and Eq.\eqref{eq:T-infinity-1}.

From the above high order Darboux matrices \eqref{eq:darboux12}, we can construct the solutions for the following ordinary differential equations or finite dimensional integrable system \cite{BBT-03}:
\begin{equation}\label{eq:finite-dimension}
\frac{\partial}{\partial x}\mathbf{L}(\lambda;x,\cdot)=\left[\mathbf{U}^{[N]}(\lambda;x,\cdot), \mathbf{L}(\lambda;x,\cdot)\right],\,\,\,\,\,\, \frac{\partial}{\partial t}\mathbf{L}(\lambda;\cdot,t)=\left[\mathbf{V}^{[N]}(\lambda;\cdot,t), \mathbf{L}(\lambda;\cdot,t)\right]
\end{equation}
where $\mathbf{U}^{[N]}(\lambda;x,t)=\mathbf{U}(\pmb{\mathcal{Q}}\to\pmb{\mathcal{Q}}^{[N]})$, $\mathbf{V}^{[N]}(\lambda;x,t)=\mathbf{V}(\pmb{\mathcal{Q}}\to\pmb{\mathcal{Q}}^{[N]})$ and  $$\mathbf{L}(\lambda;x,t)=\mathbf{T}_N(\lambda;x,t)\pmb{\sigma}_3\mathbf{T}_N^{\dag}(\lambda^*;x,t)=\pmb{\sigma}_3+\sum_{i=1}^N\left(\frac{\mathbf{A}_i(x,t)}{(\lambda-\lambda_1)^i}+\frac{\mathbf{A}_i^{\dag}(x,t)}{(\lambda-\lambda_1^*)^i}\right).$$ As $N\to\infty$, the finite dimensional integrable system \eqref{eq:finite-dimension} will turn to the infinite dimensional system.
%From the scattering data and the definition of $\pmb{\Phi}^{\pm}(\lambda; x, t)$, we know the
% scattering data $\mathbf{S}(\lambda)$ is only related to the boundary condition of $\mathbf{T}_{\pm}(\lambda)$, which equals to
% \begin{equation}
%\mathbf{S}(\lambda)=\mathbf{T}_{+}(\lambda)\mathbf{T}_{-}^{-1}(\lambda)=\begin{bmatrix}\left(\frac{\lambda-\lambda_1}{\lambda-\lambda_1^*}\right)^{n_1}&0&0\\
%0&\frac{1}{2}+\frac{1}{2}\left(\frac{\lambda-\lambda_1}{\lambda-\lambda_1^*}\right)^{n_1}&-\frac{1}{2}+\frac{1}{2}\left(\frac{\lambda-\lambda_1}{\lambda-\lambda_1^*}\right)^{n_1}\\
%0&-\frac{1}{2}+\frac{1}{2}\left(\frac{\lambda-\lambda_1}{\lambda-\lambda_1^*}\right)^{n_1}&\frac{1}{2}+\frac{1}{2}\left(\frac{\lambda-\lambda_1}{\lambda-\lambda_1^*}\right)^{n_1}
%\end{bmatrix}.
%\end{equation}

Now, we proceed to analyze the asymptotics of $\mathbf{T}_{N}(\lambda; x, t)$ when $N\to\infty$. To consider the convergence of Darboux matrix, we establish the uniform estimate
\begin{equation}\label{eq:estimate}
\left\|\mathbf{T}^{[j]}\right\|\leq 1+\frac{2{\rm Im}(\lambda_1)}{|\lambda-\lambda_1^*|}, \,\,\,\,\,\, \|\mathbf{T}^{[j]}\|\equiv [{\rm Tr}(\mathbf{T}^{[j]}(\mathbf{T}^{[j]})^{\dagger})]^{1/2},\qquad j=1,2,\cdots, N,
\end{equation}
for $(x,t)\in\mathbb{R}^2$ and $\lambda\in\mathbb{C}\setminus \mathcal{O}(\lambda_1^*,\epsilon)$, where $\epsilon$ is a small real parameter. If ${\rm Im}(\lambda_1)=1/N$, we have $\left\|\mathbf{T}_{N}(\lambda; x, t)\right\|\leq\ee^{2/|\lambda-{\rm Re}(\lambda_1)|}$, which implies that $\mathbf{T}_{N}(\lambda; x, t)$ is uniform convergent for $(x,t)\in\mathbb{R}^2$ and $\lambda\in\mathbb{C}\setminus \mathcal{O}({\rm Re}(\lambda_1),\epsilon)$ by the M-test. Then the corresponding scattering matrix $\mathbf{S}(\lambda)$ is obtained by
\begin{equation}
\begin{split}
\mathbf{S}(\lambda)&=\lim_{x\to+\infty}\lim_{N\to\infty}\mathbf{T}_{N}(\lambda;x,t) \left(\lim_{x\to-\infty}\lim_{N\to\infty}\mathbf{T}_{N}(\lambda;x,t)\right)^{-1}\\
&=\lim_{N\to\infty}\lim_{x\to+\infty}\mathbf{T}_{N}(\lambda;x,t) \left(\lim_{N\to\infty}\lim_{x\to-\infty}\mathbf{T}_{N}(\lambda;x,t)\right)^{-1}\\
&=\lim\limits_{N\to\infty}\begin{bmatrix}\left(\frac{\lambda-\lambda_1}{\lambda-\lambda_1^*}\right)^{N}&0&0\\
0&\frac{1}{2}+\frac{1}{2}\left(\frac{\lambda-\lambda_1}{\lambda-\lambda_1^*}\right)^{-N}&-\frac{1}{2}+\frac{1}{2}\left(\frac{\lambda-\lambda_1}{\lambda-\lambda_1^*}\right)^{-N}\\
0&-\frac{1}{2}+\frac{1}{2}\left(\frac{\lambda-\lambda_1}{\lambda-\lambda_1^*}\right)^{-N}&\frac{1}{2}+\frac{1}{2}\left(\frac{\lambda-\lambda_1}{\lambda-\lambda_1^*}\right)^{-N}
\end{bmatrix}\\
&=\begin{bmatrix}
\ee^{-\frac{2\ii}{\lambda-{\rm Re}(\lambda_1)}}&0&0\\
0&\frac{1}{2}+\frac{1}{2}\ee^{\frac{2\ii}{\lambda-{\rm Re}(\lambda_1)}}&-\frac{1}{2}+\frac{1}{2}\ee^{\frac{2\ii}{\lambda-{\rm Re}(\lambda_1)}}\\
0&-\frac{1}{2}+\frac{1}{2}\ee^{\frac{2\ii}{\lambda-{\rm Re}(\lambda_1)}}&\frac{1}{2}+\frac{1}{2}\ee^{\frac{2\ii}{\lambda-{\rm Re}(\lambda_1)}}
\end{bmatrix}.
\end{split}
\end{equation}
From the definition of scattering data in Eq.\eqref{eq:scattering}, we know the conservation laws can be given by expanding the factor $\frac{2\ii}{\lambda-{\rm Re}(\lambda_1)}$ with respect to $-2\ii\lambda$, then the conservation laws under this special case can be given as
\begin{equation}
\begin{split}
I_1&=\int_{-\infty}^{+\infty}\left(|q_1|^2+|q_2|^2\right)dx=4,\\ I_2&=-\frac{\ii}{2}\int_{-\infty}^{+\infty}\left(q_1q_{1,x}^*+q_2q_{2,x}^*-q_1^*q_{1,x}-q_2^*q_{2,x}\right)dx=8{\rm Re}(\lambda_1),\qquad
\\I_3&=\int_{-\infty}^{+\infty}\left[|q_{1,x}|^2+|q_{2,x}|^2-\left(|q_1|^2+|q_2|^2\right)^2\right]dx
=16{\rm Re}(\lambda_1)^2.\end{split}
\end{equation}

Under the special choice ${\rm Im}(\lambda_1)=1/N,$ the Darboux matrix $\mathbf{T}_{N}(\lambda; x, t)$ and the scattering matrix $\mathbf{S}(\lambda)$ both are uniformly convergent, then we will give the two different asymptotics by choosing the arbitrary spectral parameter $\lambda_1$. One is the large order and the other one is the infinite order.
%Generally speaking, the Riemann-Hilbert problem is given via the inverse scattering method, when it is reflectionless, we can get the pure soliton solution, when it is soliton free case, we can get the radiation, otherwise, it is the soliton and radiation interaction.
In the reflectionless case, we have established the relation between the Darboux transformation and the Riemann-Hilbert problem in Appendix \ref{app:IST} Eq. \eqref{eq:ansatz}, whose jump matrix can be given clearly. Compared to the $2\times 2$ Lax pair, the asymptotics to $3\times 3$ case becomes much more difficult due to the structures of the higher-order matrix spectral problem.

%especially, we do not know how to decompose the jump matrix into a good upper and lower trigonometric matrix for further study. To overcome this problem, we construct a key transformation to convert the jump matrix to a block one, which can be decomposed and analyzed more easily.

\ys{The innovation of this paper contains the following four points, (i) We extend to analyze the large order and infinite order soliton asymptotics to the $3\times 3$ matrix spectral problem and give two different asymptotic behaviors, one is the large order asymptotics and the other is the infinite order case. For the first case, we give four different asymptotic region, the oscillatory region, the non-oscillatory region, the algebraic decay region and the exponential decay region. Especially, the leading order in the oscillatory region can be written as the Riemann-Theta function and its corresponding modulus can be simplified as the Jacobi elliptic function, which has never been reported to the best of our knowledge. And in the second case, we get two kinds of asymptotics, one in the large $\chi$ asymptotics and the other is the large $\tau$ asymptotics, where $\chi$ and $\tau$ are related to the original variables $x$ and $t$. It should be noted that the infinite order asymptotics is a reflection for the infinite dimensional system, which is a new research topic and is different from the finite dimensional system essentially. (ii) To get the asymptotics for the $3\times 3$ system, we construct a key transformation to convert the jump matrix to a block one, which can be decomposed into the upper and lower triangles successfully. Apart from the CNLS equation, this transformation can also be used to other higher-order matrix spectral problem with a minor revision, which can be regarded as an effective tool for analyzing the higher-order matrix spectral problem. (iii) Under the same coordinate frame, we give the comparison between the large order case and the infinite order case and verify these two asymptotic expression are consistent, the large $\chi$ asymptotics of infinite order is coincident with the algebraic decay region of large order and the large $\tau$ asymptotics of infinite order is coincident with the non-oscillatory region of large order. (iiii) Based on the theory of loop group, we directly establish the Riemann-Hilbert representation via the Darboux matrix without using inverse scattering method, which seems more simple and convenient. From the result of these asymptotics, we find that every component $q_i, (i=1,2)$ in the Eq.\eqref{eq:cnls} exhibit the similar characteristics with the scalar NLS equation on both these two types of asymptotics, which verifies the fact that the dynamics of large and infinite order of the coupled equations is consistent with the one of scalar equation. In other words, the large and infinite order solitonic solution with the maximum peak of $2$-norm admit the universal property in the scalar or vector NLS system. The results are claimed in the following theorems.}
%Afterwards, to study the large order asymptotics, we should give a scale transformation to convert the phase term to a new one so as to can be analyzed for large $N$. In this case, we give four different asymptotics region, the exponent decay region, the algebraic decay region, the non-oscillatory region and the oscillatory region. While on the other infinity order case, the scale transformation is different in essence, with the new Riemann-Hilbert, we can derive an ordinary equation according to the compatibility of the Lax pair.

\subsection{The main results and numeric verification}

\begin{theorem}
\label{theo:os}
(The oscillatory region and elliptic representation)  If $(X, T)$ is in the oscillatory region, then the leading order asymptotic solution of large order solitons for the CNLS Eq.\eqref{eq:cnls} can be given by the Riemann-theta function:
\small
\begin{equation}
q_i(X, T)=c_{i}^*{\rm Im}(b_o-a_o)\frac{\theta_2\left(\pi A-\frac{Nd\tau}{2\ii}\right)\theta_3(0)}{\theta_2\left(\pi A\right)\theta_3\left(\frac{Nd\tau}{2\ii}\right)}\ee^{-\frac{Nd\tau}{2}-F_1E-2F_0}\left(\mathbb{I}+O(N^{-1})\right)\qquad (i=1,2),
\end{equation}
and its modulus can be expressed with the Jacobi elliptic function:
\begin{equation}
|q_i(X, T)|^2=|c_i|^2\left(\left({\rm Im}(b_{o}-a_o)\right)^2-|a_{o}-b_o|^2\cn^2\left(u+K(m),m\right)\right)\ee^{-Nd\tau-|F_1E|^2}\left(\mathbb{I}+O(N^{-1})\right),
\end{equation}
\normalsize
where $(i=1,2)$, $m=\frac{\theta_2^4(0)}{\theta_3^4(0)}, u=\frac{Nd\tau}{2\ii}\theta_3^2(0), $ and $a_o, b_o, d, A, E,  F_0, F_1,  \tau$ all are functions with respect to $X, T$, which are defined in Eq.\eqref{RHP-G-O}, Eq.\eqref{eq:dGdvarphi}, Eq. \eqref{eq:F1}, Eq. \eqref{eq:F0}, Eq. \eqref{eq:tau}, Eq.\eqref{eq:Omega1} and Eq. \eqref{eq:A-NO}, $c_1$ and $c_2$ are constants in the vector $\mathbf{c}$.\end{theorem}
\begin{figure}[!h]
\centering{\includegraphics[width=0.3\textwidth]{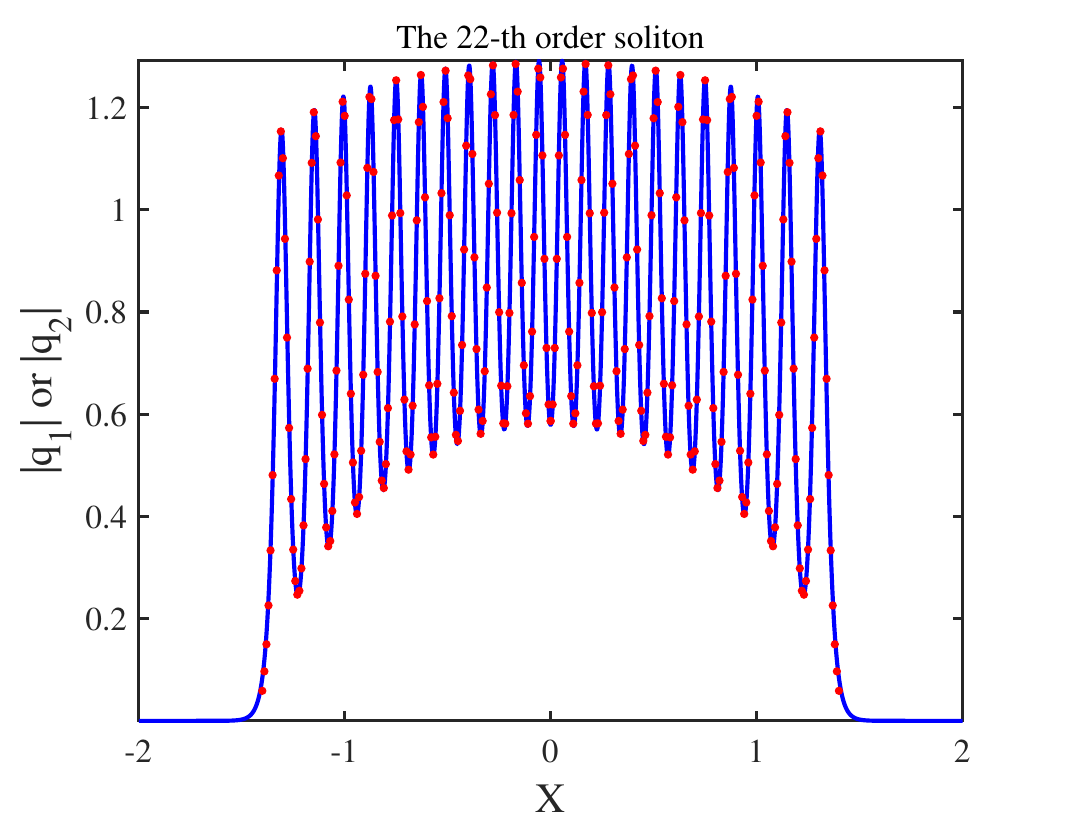}}
{\includegraphics[width=0.3\textwidth]{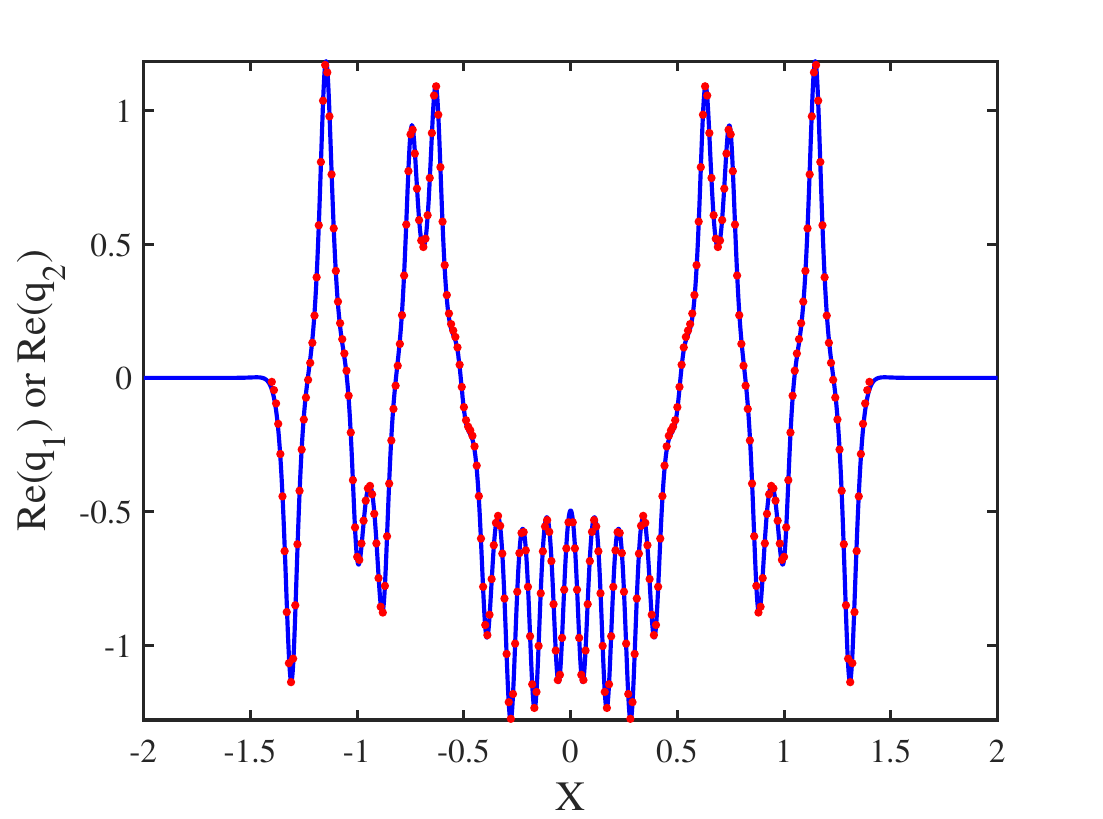}}
{\includegraphics[width=0.3\textwidth]{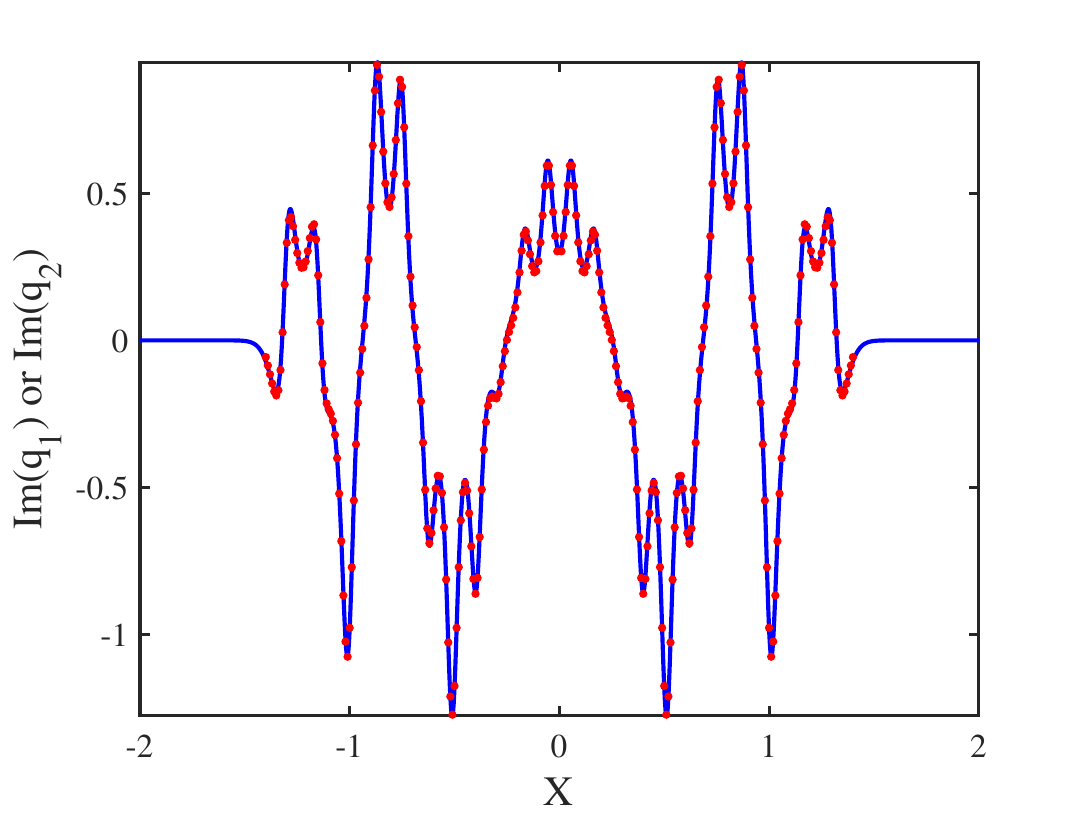}}
\caption{The comparison between the exact solution and its corresponding leading order in the oscillatory region, the parameters is $T=2/3, c_1=c_2=\frac{\sqrt{2}}{2}\ii, $ the exact solution is marked with the solid blue line and asymptotic soliton is marked with the red dotted line. The right two are the comparison of the corresponding real part and the imaginary part.}
\label{o-region}
\end{figure}

The comparison for modulus as well its real part and the imaginary part of $q_1$ or $q_2$ between leading order asympotics and $22$-th order soliton solution is given in Fig. \ref{o-region}. The $22$-th order soliton solution are plotted by the iterative algorithm for Darboux matrix \eqref{eq:darboux12}.
\begin{figure}[!h]
\centering{\includegraphics[width=0.3\textwidth]{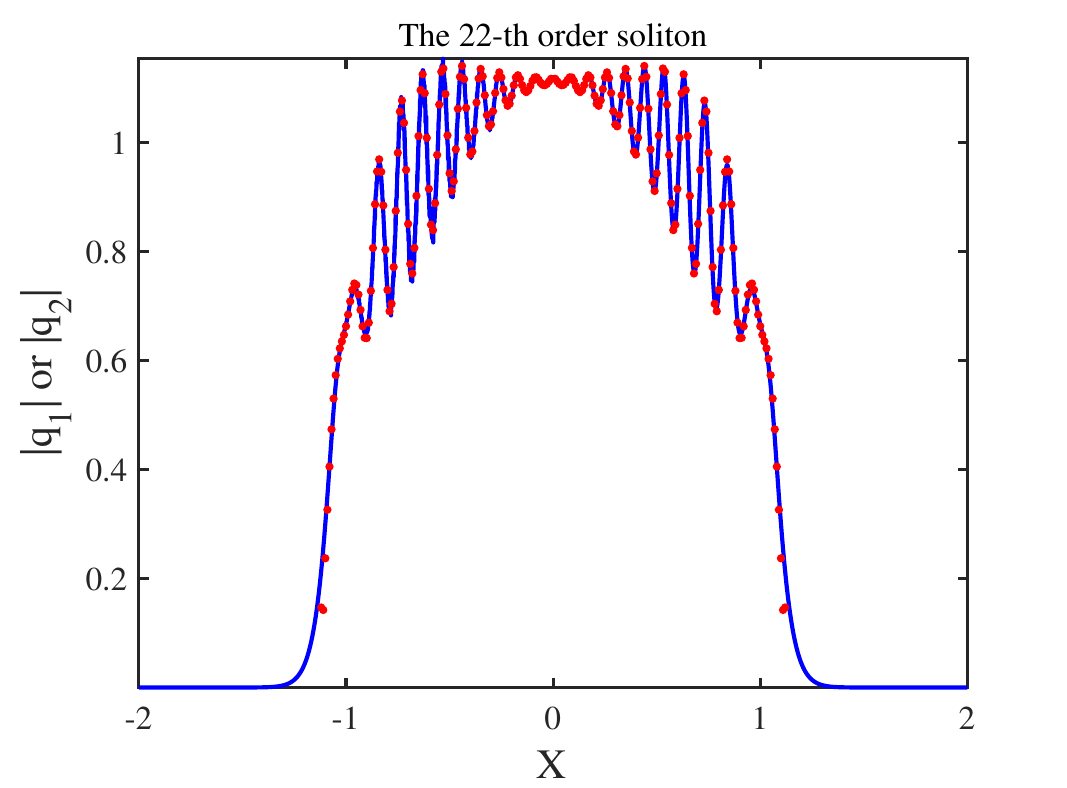}}{\includegraphics[width=0.3\textwidth]{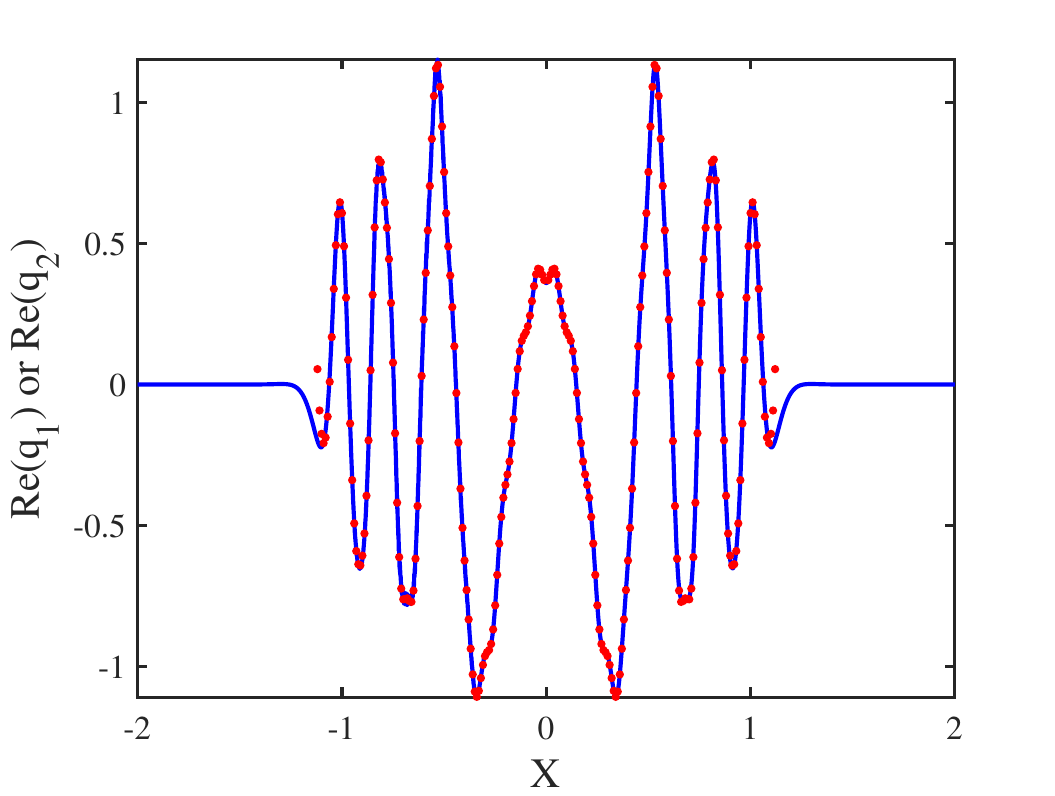}}
{\includegraphics[width=0.3\textwidth]{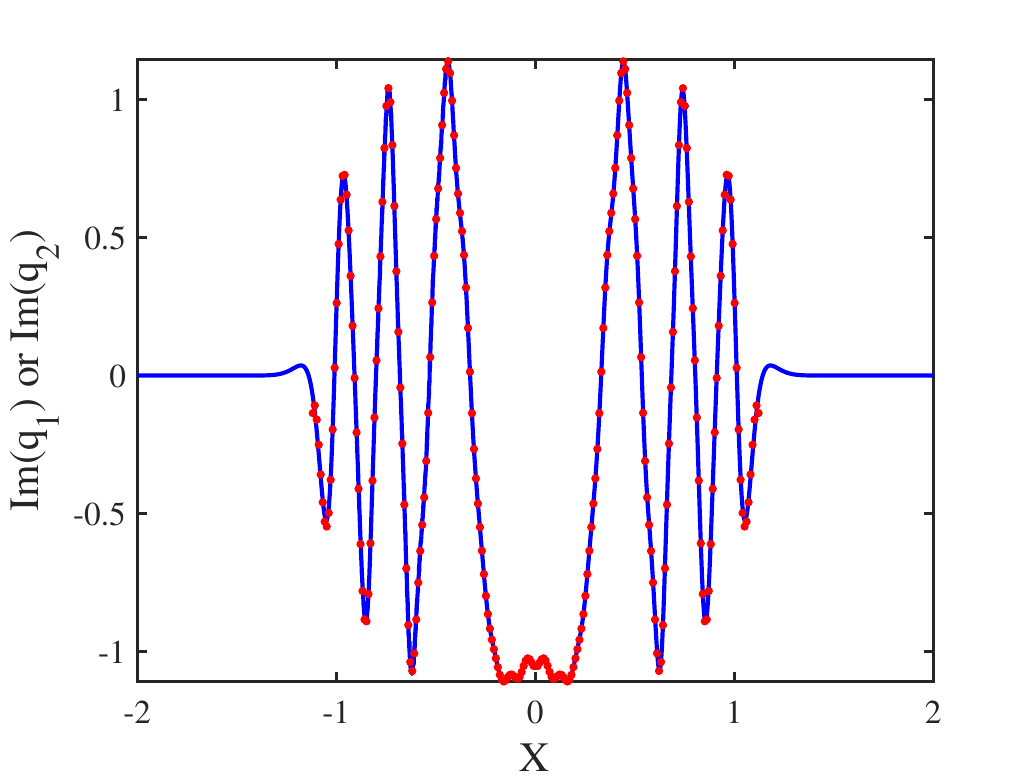}}
\caption{The comparison between the exact solution and its corresponding leading order in the non-oscillatory region, the parameters is $T=1/3, c_1=c_2=\frac{\sqrt{2}}{2}\ii, $ the exact solution is marked with the solid blue line and asymptotic soliton is marked with the red dotted line. The right two are the comparison of the corresponding real part and the imaginary part.}
\label{no-region}
\end{figure}

\ys{We give a detailed calculation about this leading order term in subsection \ref{sec:os}. This whole region is marked by $O$ in Fig.\ref{high-order-soliton-1}. Seeing this figure, we know the oscillatory region is adjacent to two regions, one is the non-oscillatory region and the other is the exponent decay region. In this region, the modulus of $q_1$ or $q_2$ behaves a good oscillation, whose maximum and minimum amplitude has a similar shape with the boundary line between the oscillatory region and the non-oscillatory region. As $X$ increases, the modulus of solution is slow decay. These two different dynamic behaviors confirm the fact there are two boundary lines between the oscillatory region and other regions.}
\begin{theorem}
\label{theo:nos}
(The non-oscillatory region) When $(X, T)$ is located in the non-oscillatory region, the asymptotic leading order term can be given as
\begin{equation}\label{eq:q-no}
\begin{split}
q_{i}(X, T)=&-\ii c_{i}^*\ee^{N\Omega_{n}-\ii\mu}\left(-\ii {\rm Im}(a_{n})+\frac{\sqrt{2p}}{N^{1/2}\sqrt{-\hat{h}''(\lambda_c)}}(m_{-}^{\lambda_c}\ee^{\ii\phi_{\lambda_c}}-m_{+}^{\lambda_c}\ee^{-\ii\phi_{\lambda_c}})\right)\\
&-\ii c_{i}^*\ee^{N\Omega_{n}-\ii\mu}\frac{\sqrt{2p}}{N^{1/2}\sqrt{\hat{h}''(\lambda_d)}}\left(m_{+}^{\lambda_d}\ee^{\ii\phi_{\lambda_d}}-m_{-}^{\lambda_d}\ee^{-\ii\phi_{\lambda_d}}\right)+\mathcal{O}(N^{-1}),
\end{split}
\end{equation}
where $a_{n}, \lambda_c, \lambda_d, \hat{h}, \Omega_{n}, \phi_{\lambda_c}, \phi_{\lambda_d}, m_{\pm}^{\lambda_c}, m_{\pm}^{\lambda_d}$ are functions with $X$ and $T$, which are defined in Eq.\eqref{eq:R-no}, Eq.\eqref{eq:g-no}, RHP \ref{RHP:NO} and Eq.\eqref{eq:phiaphib}. By choosing suitable parameters, we give the comparison between exact solutions and its asymptotic analysis in Fig.\ref{no-region}, which shows that they are fitting very well.

Furthermore, when $N$ is large, the modulus of $q_{i}(X, T)$ in the non-oscillatory region is
\begin{multline}
\left|q_{i}(X, T)\right|^2{=}\left|c_i\right|^2\left({\rm Im}(a_{n})^2{-}2\frac{{\rm Im}(a_n)\sqrt{2p}}{N^{1/2}}\left(\frac{\sin(\phi_{\lambda_c})}{\sqrt{{-}\hat{h}''(\lambda_c)}}{+}\frac{\sin(\phi_{\lambda_d})}{\sqrt{\hat{h}''(\lambda_d)}}\right)\right){+}\mathcal{O}(N^{-1})\end{multline}
\end{theorem}
Obviously, when $\sin\left(\phi_{\lambda_c}\right)=\sin\left(\phi_{\lambda_d}\right)=-1$, the modulus attains its maximum, which can be seen from the figures. These grids can be given by a union
\begin{equation}
\begin{split}
\mathcal{U}_{\lambda_c}=\left\{\left(X, T\right)\in NO, \sin(\phi_{\lambda_c})=1\right\}, \quad \mathcal{U}_{\lambda_d}=\left\{\left(X, T\right)\in NO, \sin(\phi_{\lambda_d})=1\right\}.
\end{split}
\end{equation}
To verify this fact, we give a figure to show this property:
\begin{figure}[!h]
\centering{\includegraphics[width=0.5\textwidth]{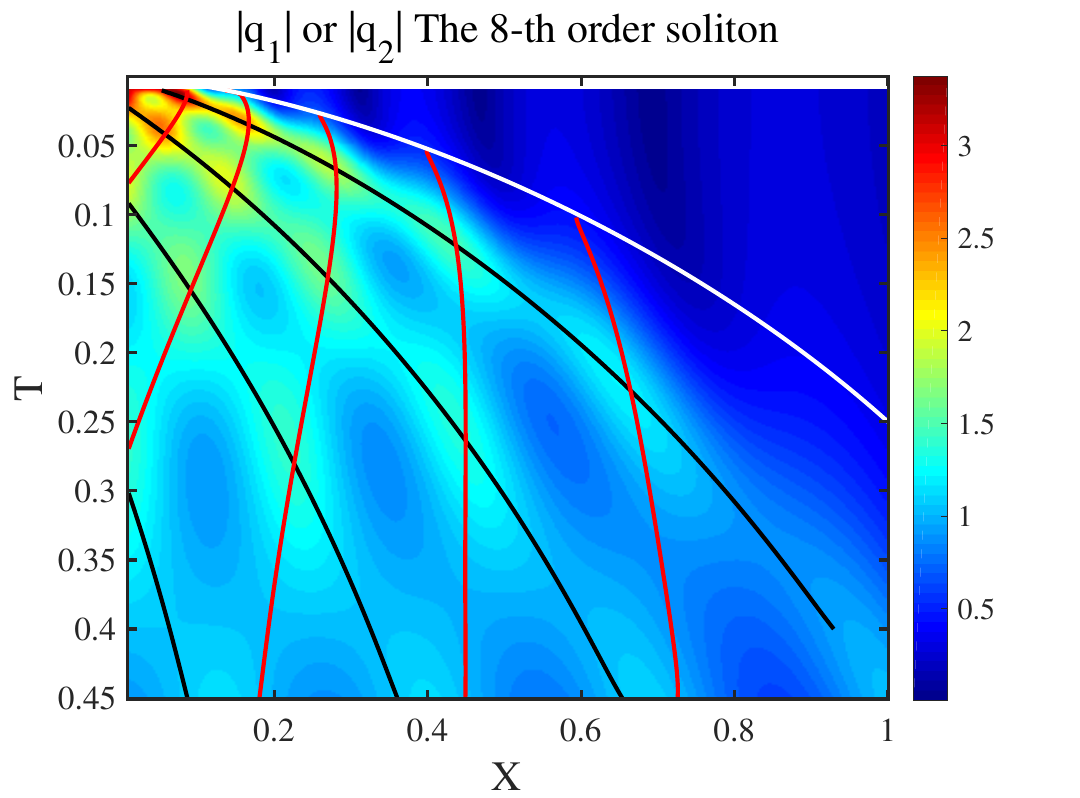}}
\caption{The density of $\mathbf{q}(X, T)$ in the non-oscillatory region, the white line is the boundary between the non-oscillatory region and the algebraic decay region, the black lines indicate the union of $\mathcal{U}_{\lambda_c}$ and the red lines represent the union of $\mathcal{U}_{\lambda_d}$. It is shown that the maximum value points lie in the intersection between red and black lines.}
\label{re-im}
\end{figure}

It can be seen that the amplitude in the cyan point is larger than the blue point, the yellow point is larger than the cyan point along the color bar direction. Following the detailed analysis in \cite{BilmanM-21}, we know this kind of plane is the modulational stability wave for the CNLS equation.

\ys{This leading order term is given in subsection \ref{sec:no} detailedly, and it is marked by $NO$ in Fig.\ref{high-order-soliton-1}. The calculation in this region is similar to the oscillatory region by replacing the $G$-function with $g$-function. But algebraic curves for  $G$-function and $g$-function have different genus, which results in two different types of leading order term.}
\begin{theorem}
\label{theo:al}
(Algebraic decay region)When $(X, T)$ is in the algebraic decay region, the soliton will be decay with $N^{-1/2}$, whose asymptotic leading order term is
\begin{equation}\label{eq:q-algebra}
\begin{split}
q_i(X, T)&{=}
-\frac{c_i^*\ii}{N^{1/2}}\sqrt{\frac{\ln(2)}{\pi}}\Bigg(\ee^{-2N\varphi(b_{A}; X, T)+\ii\phi_{f}}\left(\sqrt{-\ii \varphi''(b_A; X, T)}\right)^{2\ii p-1}\\&{+}\ee^{-2N\varphi(a_A; X, T)-\ii\phi_{f}}\left(\sqrt{\ii\varphi''(a_A; X, T)}\right)^{-2\ii p-1}\Bigg)+\mathcal{O}(N^{-1}), \qquad (i=1,2),
\end{split}
\end{equation}
where
$\phi_f(X, T)=\frac{\ln(2)^2}{2\pi}+\frac{\pi}{4}-\arg\left(\Gamma\left(\frac{\ln(2)}{2\pi}\ii\right)\right)+2p\ln\left(b_A-a_A\right)+p\ln(N), \quad p=\frac{\ln(2)}{2\pi}$, $a_A$ and $b_A$ are the critical points of $\varphi(\lambda; X, T)$, and $\varphi(\lambda; X, T)$ is defined in RHP \ref{rhp:far-field}. By choosing $T=0,$ then $\mathbf{q}(X, T)$ will be a real vector function, whose evolutional behavior is shown in Fig. \ref{algebraic-decay-region}
\begin{figure}[!h]
\centering{\includegraphics[width=0.4\textwidth]{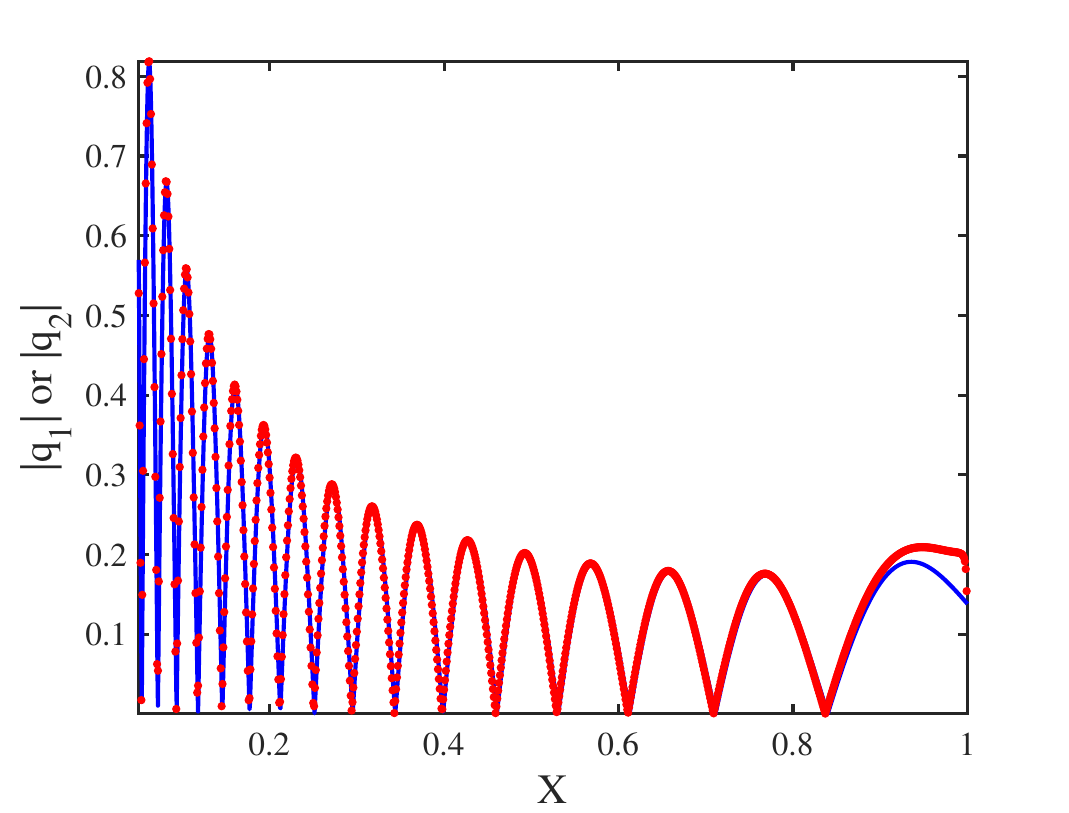}}
{\includegraphics[width=0.4\textwidth]{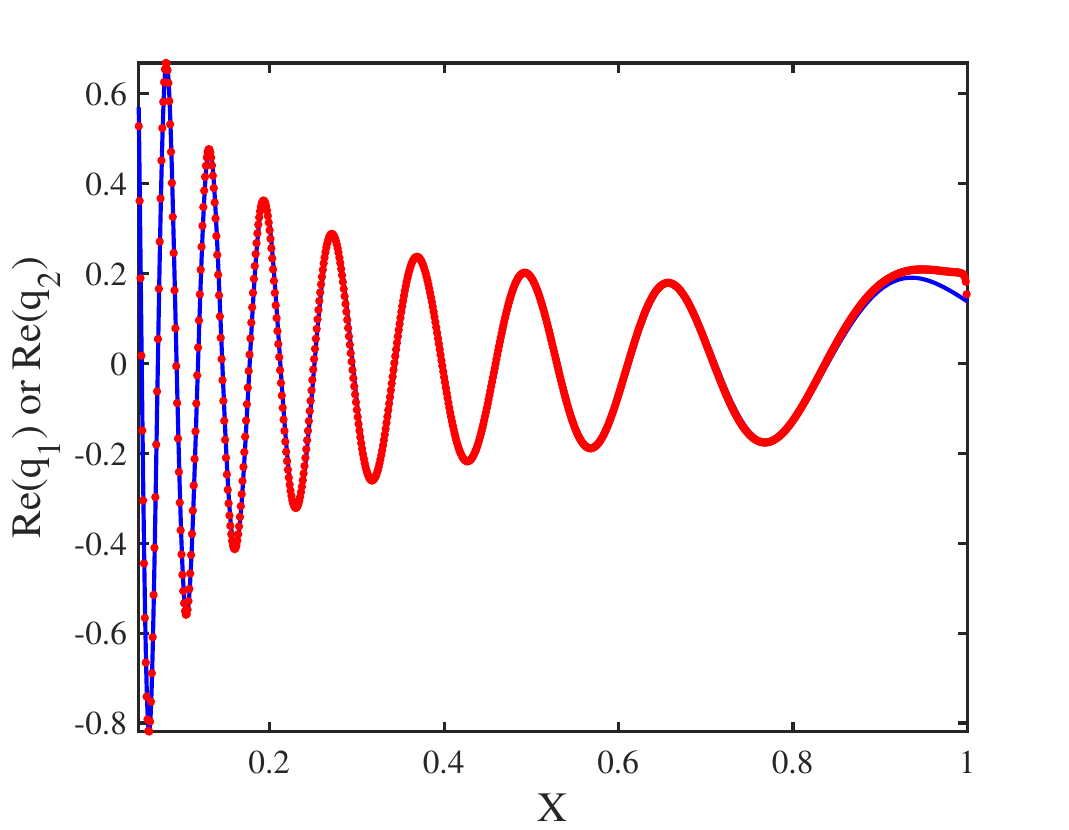}}
\caption{The comparison between the exact solution and its corresponding leading order in the algebraic decay region, the parameters is $T=0, c_1=c_2=\frac{\sqrt{2}}{2}\ii,$ the exact solution is marked with the solid blue line and asymptotic soliton is marked with the red dotted line.}
\label{algebraic-decay-region}
\end{figure}
\end{theorem}
\ys{This region is marked with $A$ in Fig.\ref{high-order-soliton-1}. It is clear that the leading order term in this region decays with $N^{-1/2}$. Compared to the oscillatory region and the non-oscillatory region, the original contour in this region can be deformed directly. We do not need construct ``g"-function any more. Additionally, the leading order term in this region is given by the error between the Riemann-Hilbert matrix and its parametrix matrix, which is shown in subsection \ref{sec:error-analysis} in the Appendix.}
\begin{theorem}\label{theo:expo}
(Exponential decay region) When $(X, T)$ is in this region, the leading order term is simpler, which is
\begin{equation}
\begin{split}
q_i(X, T)&=2c_{i}^*\lim\limits_{\lambda\to\infty}\lambda\mathbf{\widetilde{N}}_{f,e}(\lambda; X, T)_{12}\to O(\ee^{-dN}),\qquad (i=1,2),
\end{split}
\end{equation}
where $d$ is a positive constant.
\end{theorem}
\ys{This region is marked with $E$ in Fig.\ref{high-order-soliton-1}. In this region, with a basic deformation to this contour, the solution about this Riemann-Hilbert problem can be given directly, which is the simplest.}
\begin{remark}
Theorem \ref{theo:os} to Theorem \ref{theo:expo} are the leading order terms for the four different regions in the large order asymptotics. Next we will give the main results about the infinite order asymptotics in theorem \ref{theo:infinity-large-chi} and \ref{theo:large-tau}.
\end{remark}
\begin{theorem}
\label{theo:infinity-large-chi}(Large $\chi$ asymptotics)
\modify{For the soliton of infinite order}, with a Galilean transformation $x=\chi-2{\rm Re}(\lambda_1) t, \lambda=\Lambda+{\rm Re}(\lambda_1), t=\tau$, when $\chi$ is large, the asymptotic leading order term is
\begin{equation}
\begin{split}
q_i(\chi, \chi^{3/2}v)&=-c_i^*\frac{\ii}{\chi^{3/4}}\sqrt{\frac{\ln(2)}{\pi}}\Bigg(\ee^{-2 \chi^{1/2}\vartheta(b(v); v)+\ii\phi_{n}}\left(\sqrt{-\ii\vartheta''(b(v); v)}\right)^{2\ii p-1}\\&+\ee^{-2\chi^{1/2}\vartheta(a(v); v)-\ii \phi_{n}}\left(\sqrt{\ii\vartheta''(a(v); v)}\right)^{-2\ii p-1}\Bigg)+O(\chi^{-1}),\qquad (i=1,2),
\end{split}
\end{equation}
where $a(v)$ and $b(v)$ are the critical points of $\vartheta(z; v)$, which are defined in Eq.\eqref{eq:near-field-S}, and $\phi_{n}=\frac{\ln(2)^2}{2\pi}+\frac{\pi}{4}-\arg\left(\Gamma\left(\frac{\ln(2)}{2\pi}\ii\right)\right)+2p\ln\left(b(v)-a(v)\right)+\frac{p}{2}\ln(\chi)$.
\end{theorem}
\ys{It can be seen the expression on large $\chi$ asymptotics in the infinite order case is similar to the expression in the algebraic decay region. But to the large $\chi$ asymptotics, the critical points $a(v)$ and $b(v)$ is only related to the parameter $v$, while the $a_{A}$ and $b_A$ in theorem \ref{theo:al} are functions with respect to $X$ and $T$. By choosing a suit $v$, we give the comparison between these two asymptotics, which is shown in the left figure in Fig.\ref{large-t-and-no}.}
\begin{theorem}
\label{theo:large-tau}
(Large $\tau$ asymptotics) \modify{For the soliton of infinite order}, with the same Galilean transformation to Theorem \ref{theo:infinity-large-chi}, the leading order term in large $\tau$ asymptotics is
\begin{multline}
q_i(\tau^{2/3}w, \tau)=-\ii c_{i}^*\ee^{\tau^{1/3}\widetilde{\Omega}_{n}-\ii\tilde{\mu}}\left(\frac{\sqrt{2p}}{\tau^{1/2}\sqrt{-\hat{h}''(z_{c})}}m_{-}^{z_{a}}\ee^{\ii\phi_{z_{c}}}-\frac{\sqrt{2p}}{\tau^{1/2}\sqrt{-\hat{h}''(z_{c})}}m_{+}^{z_{a}}\ee^{-\ii\phi_{z_c}}\right)\\
-\ii c_{i}^*\ee^{\tau^{1/3}\widetilde{\Omega}_{n}-\ii\tilde{\mu}}\left(\frac{\sqrt{2p}}{\tau^{1/2}\sqrt{\hat{h}''(z_d)}}m_{+}^{z_d}\ee^{\ii\phi_{z_d}}-\frac{\sqrt{2p}}{\tau^{1/2}\sqrt{\hat{h}''(z_d)}}m_{-}^{z_b}\ee^{-\ii\phi_{z_b}}\right)\\
- c_{i}^*\ee^{\tau^{1/3}\widetilde{\Omega}_{n}-\ii\tilde{\mu}}\frac{\sqrt{9-w^2}}{3\tau^{1/3}}+\mathcal{O}(\tau^{2/3}),\quad i=1,2,
\end{multline}
where $0\leq w<3, \widetilde{\Omega}_n$ is a pure imaginary constant, defined by RHP \ref{RHP:large-tau}.
%Its corresponding modulus can be given by
%\begin{equation}
%q_i(X, T)=-\frac{|c_i|}{3}\tau^{-1/3}\sqrt{9-w^2}, \quad i=1,2.
%\end{equation}
\end{theorem}
\ys{The formula about this asymptotics is similar to the non-oscillatory region, both of them are calculated with a similar $g$-function. By choosing one special $w$, $|q_1|$ or $|q_2|$ is a elementary function with $\tau$, whose evolutional behavior can be seen clearly, which is shown in the right panel in Fig.\ref{large-t-and-no}. It is seen that each component of infinite order soliton for the CNLS equation is proportional to the rogue waves of infinite order for the scalar NLS equation \cite{Peter-Duke-2019}. The concept to rogue waves of infinite order was proposed in the literature \cite{Peter-Duke-2019}, which also appears in the vanishing background \cite{Deniz-JNS-2019} and can be used to model the self-focusing phenomena in nonlinear geometrical optics or unstable gas dynamics \cite{Suleimanov-17}. All of these results verify the universality for this special solution.}
\begin{remark}(The comparison between the large order and infinite order) In the large order asymptotics, we give four different regions, which is shown in Theorem \ref{theo:os} to Theorem \ref{theo:expo}. And in the infinite order case, we give two types of asymptotics about large $\chi$ and large $\tau$. From the scale transformation, we know the large $\chi$ asymptotics lies in the algebraic decay region and the large $\tau$ asymptotics lies in the non-oscillatory region, which can also be manifested from the above theorem. By choosing some parameters, we give the comparison between these asymptotics:
\begin{figure}[!h]
\centering{\includegraphics[width=0.4\textwidth]{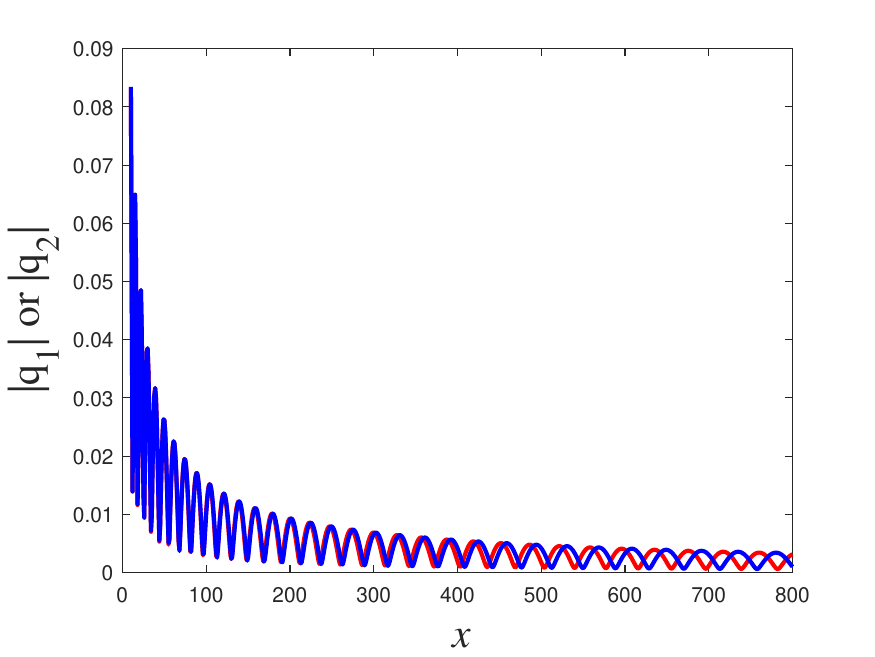}}
{\includegraphics[width=0.4\textwidth]{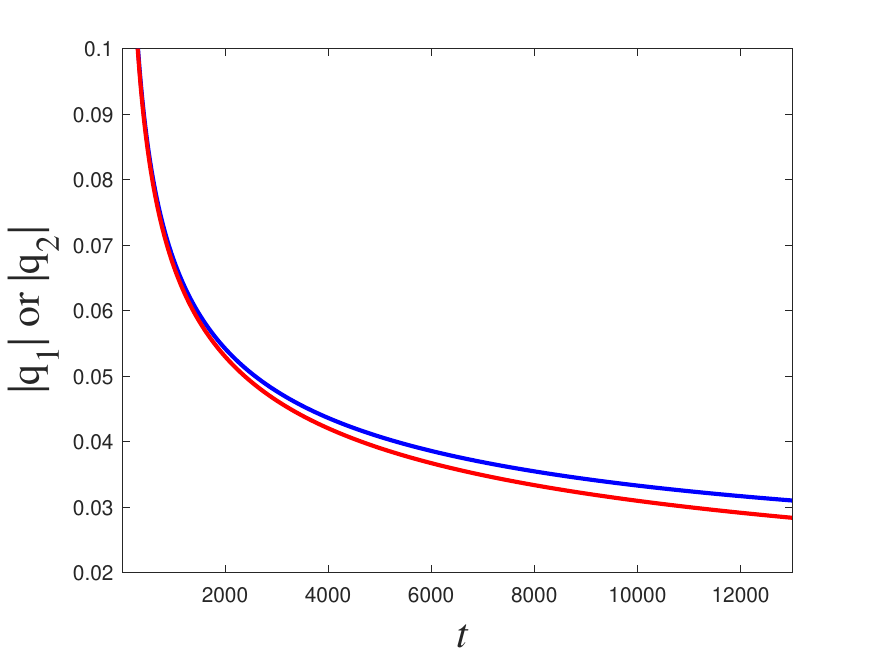}}
\caption{The comparison between the asymptotics of the large order and the infinite order case. The left panel is the large $\chi$ asymptotics and the algebraic decay region, the corresponding parameters is $v=\frac{1}{12}, N=100$. The right one is the main leading order in large $\tau$ asymptotics and the non-oscillatory region by choosing $w=1, N=30$. The leading order in Theorem \ref{theo:infinity-large-chi}, Theorem \ref{theo:large-tau} is marked with the solid blue line and the leading order in Theorem \ref{theo:al}, Theorem \ref{theo:nos} is marked with the red dotted line.}
\label{large-t-and-no}
\end{figure}
\end{remark}

\subsection{Outline of this work}

The outline of this paper is organized as follows: In section \ref{sec:RH}, we give the Darboux transformation of the high order soliton to CNLS Eq.\eqref{eq:cnls}, then the corresponding Riemann-Hilbert representation can be constructed, which is the foundation for the later asymptotic analysis. Afterwards, in section \ref{sec:large-order}, we give the large order asymptotics by using a suit scale transformation to the variables $x$ and $t$. According to the nature of the critical points, the large order contains four different asymptotic region, the oscillatory region, the non-oscillatory region, the algebraic decay region and the exponential decay region, which is shown in four different subsections. Furthermore, in section \ref{sec:infinity-order}, we give the infinite order asymptotics under the same framework of Riemann-Hilbert problem but with a different transformation about the variables $x$ and $t$. This section involves two subsections, one is about how to derive the ordinary differential equation with the compatibility condition of the new Lax pair and the other one is the detailed calculation to the asymptotics. For this section, we give two kinds of asymptotics in total, one is the large $\chi$ asymptotics and the other one is the large $\tau$ asymptotics. In the last part, there are five appendices as the complement materials. In Appendix \ref{app:IST}, we give some preliminaries about the inverse scattering method; in Appendix \ref{appendix:second-order} we give the asymptotic analysis for the general second order soliton solutions, which shows the high order soliton solution with vanishing parameters $\epsilon_j^{[i]}$, $i=1,2$ is not only an important type of limiting solution but also associated with the scalar NLS equation; in Appendix \ref{app:CLS}, we present the proofs to some lemmas and in Appendix \ref{App:large-chi}, we mainly provide the ways to construct the parametrix matrix and analyze the error for the large $\chi$ asymptotics or the algebraic decay region, which is crucial to the asymptotic analysis; in the last Appendix \ref{app:no}, we give the detailed calculation about the leading order in the non-oscillatory region or the large $\tau$ analysis.

\section{Riemann-Hilbert representation of high order solitons}
\label{sec:RH}
In the Appendix \ref{app:IST}, we have given the relation between the Darboux matrix and the Riemann-Hilbert representation in the non-reflection cases. To analyze the large order and infinite order asymptotics, we try to establish the RHP and convert the properties of large order soliton into the jump matrix. Then some complicated exact solutions  will be tackled with the Riemann-Hilbert approach to analyze the dynamics for the solutions.

Based on the idea in \cite{peter-CPAM-2018}, we can construct two sectional analytic matrix
\begin{equation}\label{eq:Mmatrix}
\mathbf{M}^{[N]}(\lambda; x, t)=\left\{\begin{split}\mathbf{M}_{+}^{[N]}(\lambda; x, t)&=\left(\frac{\lambda-\lambda_1}{\lambda-\lambda_1^*}\right)^{-N/3}
\mathbf{T}_{N}(\lambda; x, t),\\
\mathbf{M}_{-}^{[N]}(\lambda; x, t)&=
\mathbf{T}_{N}(\lambda; x, t)\ee^{-\ii\lambda(x+\lambda t)\pmb{\sigma}_3}\mathbf{T}_{N}^{-1}(\lambda; 0, 0)\ee^{\ii\lambda(x+\lambda t)\pmb{\sigma}_3},
\end{split}\right.
\end{equation}
where $\mathbf{T}_N(\lambda;x,t)$ can be rewritten in a compact form:
\begin{equation*}
\begin{split}
\mathbf{T}_{N}(\lambda; x, t)&=\mathbb{I}+\mathbf{Y}_N\mathbf{M}^{-1}\mathbf{D}\mathbf{Y}_{N}^{\dagger}, \qquad \mathbf{M}=\mathbf{X}^{\dagger}\mathbf{S}\mathbf{X}\\
\mathbf{q}^{[N]}&=-2\mathbf{X}_1\mathbf{M}^{-1}(\mathbf{X}_2^{\dag},\mathbf{X}_3^{\dag})
\end{split}
\end{equation*}
and
\begin{equation*}
\begin{split}
\mathbf{Y}_{N}&=\left[\mathbf{\Phi}_{1}^{[0]},\mathbf{\Phi}_{1}^{[1]}, \cdots, \mathbf{\Phi}_{1}^{[N-1]}\right],\\
\mathbf{D}&=\begin{bmatrix}\frac{1}{\lambda-\lambda_1^*}&0&\cdots&0\\
\frac{1}{(\lambda-\lambda_1^*)^2}&\frac{1}{\lambda-\lambda_1^*}&\cdots&0\\
\vdots&\vdots&\ddots&\vdots\\
\frac{1}{\left(\lambda-\lambda_1^*\right)^{N}}&\frac{1}{\left(\lambda-\lambda_1^*\right)^{N-1}}&\cdots&\frac{1}{\lambda-\lambda_1^*}
\end{bmatrix},\qquad
\mathbf{X}=\begin{bmatrix}\mathbf{\Phi}_{1}^{[0]}&\mathbf{\Phi}_{1}^{[1]}&\cdots&\mathbf{\Phi}_{1}^{[N-1]}\\
0&\mathbf{\Phi}_{1}^{[0]}&\cdots&\mathbf{\Phi}_{1}^{[N-2]}\\
\vdots&\vdots&\ddots&\vdots\\
0&0&\cdots&\mathbf{\Phi}_{1}^{[0]}
\end{bmatrix},\\
\mathbf{S}&=\begin{bmatrix}
\binom{0}{0}\frac{\mathbb{I}_3}{\lambda_1^*-\lambda_1}&\binom{1}{0}\frac{\mathbb{I}_3}{\left(\lambda_1^*-\lambda_1\right)^2}&\cdots&\binom{N-1}{0}\frac{\mathbb{I}_3}{\left(\lambda_1^*-\lambda_1\right)^{N}}\\
\binom{1}{1}\frac{(-1)\mathbb{I}_3}{\left(\lambda_1^*-\lambda_1\right)^2}&\binom{2}{1}\frac{(-1)\mathbb{I}_3}{\left(\lambda_1^*-\lambda_1\right)^3}&\cdots&\binom{N}{1}\frac{(-1)\mathbb{I}_3}{\left(\lambda_1^*-\lambda_1\right)^{N+1}}\\
\vdots&\vdots&\ddots&\vdots\\
\binom{N-1}{N-1}\frac{(-1)^{N-1}\mathbb{I}_3}{\left(\lambda_1^*-\lambda_1\right)^{N}}&\binom{N}{N-1}\frac{(-1)^{N-1}\mathbb{I}_3}{\left(\lambda_1^*-\lambda_1\right)^{N+1}}&\cdots&\binom{2N-2}{N-1}\frac{(-1)^{N-1}\mathbb{I}_3}{\left(\lambda_1^*-\lambda_1\right)^{2N-1}}
\end{bmatrix},
\end{split}
\end{equation*}
with $\pmb{\Phi}_1^{[k]}=\frac{1}{k!}\left(\frac{\rm d}{{\rm d}\lambda}\right)^{k}\pmb{\Phi}_1(\lambda;x,t)|_{\lambda=\lambda_1}$, $\pmb{\Phi}_1(\lambda_1;x,t)$ belongs to the spanning space: ~$$\mathrm{span}\{\ee^{-\ii\lambda_1(x+\lambda_1 t)\pmb{\sigma}_3}\mathbf{c}\}.$$
Especially, $\pmb{\Phi}_1(\lambda_1;x,t)$ can be choosen as $\left[1,c_1\ee^{2\ii\lambda_1(x+\lambda_1 t)},c_2\ee^{2\ii\lambda_1(x+\lambda_1 t)}\right]$, the $N$-th order soliton solution can be constructed directly. Actually, the parameters $c_1$ and $c_2$ can involve the spectral parameter $\lambda_1$, in which the dynamics of high order solitons could be regulated by these parameters. In this work, we merely consider the special case in which the parameters $c_1$ and $c_2$ are independent with $\lambda_1$.
In general, the formula of Darboux matrix $\mathbf{T}_N(\lambda;x,t)$ is hard to analyze since it involves the complicated determinant formula. However, the Darboux matrix $\mathbf{T}_N(\lambda;0,0)$ can be rewritten in a compact form:
\begin{prop}\label{prop:Tn}
When $(x, t)=(0, 0)$, the Darboux matrix $\mathbf{T}_{N}(\lambda;x,t)$ in Eq.\eqref{eq:Mmatrix} can be simplified to
\begin{equation}\label{eq:tn00}
\mathbf{T}_{N}(\lambda; 0, 0)=\mathbf{Q}_c{\rm diag}\left(\left(\frac{\lambda-\lambda_1}{\lambda-\lambda_1^*}\right)^N, 1, 1\right)
\mathbf{Q}_c^{-1},\qquad \mathbf{Q}_c=\begin{bmatrix}1&-c_1^*&-c_2^*\\ c_1&1&0\\
c_2&0&1
\end{bmatrix}.
\end{equation}
\end{prop}
\begin{proof}
Choosing the parameters $(x, t)=(0, 0)$, then the matrices $\mathbf{Y}_{N}$ and $\mathbf{X}$ will be reduced into a simple formula
\begin{equation}
\mathbf{Y}_{N}=\begin{bmatrix}\mathbf{c},&\mathbf{0},&\cdots, \mathbf{0}\end{bmatrix},\quad  \mathbf{c}=\begin{bmatrix}1\\c_1\\c_2
\end{bmatrix}, \quad \mathbf{X}=\begin{bmatrix}\mathbf{c}&\mathbf{0}&\cdots&\mathbf{0}\\
0&\mathbf{c}&\cdots&\mathbf{0}\\
\vdots&\vdots&\ddots&\vdots\\
0&0&\cdots&\mathbf{c}
\end{bmatrix}.
\end{equation}
With this simple decomposition, the matrix $\mathbf{M}:=\mathbf{X}^{\dagger}\mathbf{S}\mathbf{X}$ can be written as
\begin{equation}
\mathbf{M}=\frac{1+|\mathbf{c}|^2}{\lambda_1^*-\lambda_1}\mathbf{D}_1\mathbf{S}_{N}^{\dagger}\mathbf{S}_N\mathbf{D}_2,
\end{equation}
where $|\mathbf{c}|^2=|c_1|^2+|c_2|^2,$ $ \mathbf{S}_N=\left(\binom{j-1}{i-1}\right)_{1\leq i\leq j\leq N}$ and
\begin{equation}
\mathbf{D}_1={\rm diag}\left(1,\frac{-1}{\lambda_1^*-\lambda_1}, \cdots, \left(\frac{-1}{\lambda_1^*-\lambda_1}\right)^{N-1}\right),\,\,\,\,\, \mathbf{D}_2={\rm diag}\left(1,\frac{1}{\lambda_1^*-\lambda_1}, \cdots, \left(\frac{1}{\lambda_1^*-\lambda_1}\right)^{N-1}\right).
\end{equation}
Based on the previous formulas and the decomposition of Pascal matrix, the matrix $\mathbf{Y}_N\mathbf{M}^{-1}\mathbf{D}\mathbf{Y}_N^{\dagger}$ can be simplified into
\begin{equation}\label{eq:T-I}
\mathbf{Y}_N\mathbf{M}^{-1}\mathbf{D}\mathbf{Y}_N^{\dagger}=\frac{\lambda_1^*-\lambda_1}{1+|\mathbf{c}|^2}\mathbf{Y}_{N}\mathbf{D}_2^{-1}\mathbf{S}_N^{-1}\left(\mathbf{S}_{N}^{\dagger}\right)^{-1}\mathbf{D}_1^{-1}\mathbf{D}\mathbf{Y}_N^{\dagger}.
\end{equation}
It is clear that the rank of $\mathbf{Y}_N\mathbf{M}^{-1}\mathbf{D}\mathbf{Y}_N^{\dagger}$ equals to one, which motivates us to write it as a product of two matrices with the rank of one. Following this idea, we first calculate the product of the first three matrices in Eq.\eqref{eq:T-I} $\mathbf{Y}_{N}\mathbf{D}_{2}^{-1}\mathbf{S_{N}}^{-1}$, which equals to
\begin{equation}\label{eq:first-three}
\begin{aligned}
\mathbf{Y}_{N}\mathbf{D}_{2}^{-1}\mathbf{S}_{N}^{-1}&=\begin{bmatrix}\mathbf{c},&\mathbf{0},&\cdots,&\mathbf{0}
\end{bmatrix}
{\rm diag}\left(1, \lambda_1^*-\lambda_1, \cdots, \left(\lambda_1^*-\lambda_1\right)^{N-1}\right)\left((-1)^{i+j}\binom{j-1}{i-1}\right)_{1\leq i\leq j\leq N}\\
&=\begin{bmatrix}\mathbf{c},&\mathbf{0},&\cdots,&\mathbf{0}
\end{bmatrix}\left((-1)^{i+j}\binom{j-1}{i-1}\right)_{1\leq i\leq j\leq N}\\
&=\mathbf{c}\begin{bmatrix}1,&-1,&\cdots,&(-1)^{N-1}
\end{bmatrix}.
\end{aligned}
\end{equation}
Secondly, we give the products of last four matrices $\left(\mathbf{S}_{N}^{\dagger}\right)^{-1}\mathbf{D_1}^{-1}\mathbf{D}\mathbf{Y}_{N}^{\dagger}$, i.e.
\begin{equation}
\begin{aligned}
&\left((-1)^{i+j}\binom{j-1}{i-1}\right)_{1\leq i\leq j\leq N}^{\dagger}\begin{bmatrix}\frac{1}{\lambda-\lambda_1^*}&0&\cdots&0\\
\frac{\left(\lambda_1-\lambda_1^*\right)}{(\lambda-\lambda_1^*)^2}&\frac{\left(\lambda_1-\lambda_1^*\right)}{\lambda-\lambda_1^*}&\cdots&0\\
\vdots&\vdots&\ddots&\vdots\\
\frac{\left(\lambda_1-\lambda_1^*\right)^{N-1}}{\left(\lambda-\lambda_1^*\right)^{N}}&\frac{\left(\lambda_1-\lambda_1^*\right)^{N-1}}{\left(\lambda-\lambda_1^*\right)^{N-1}}&\cdots&\frac{\left(\lambda_1-\lambda_1^*\right)^{N-1}}{\lambda-\lambda_1^*}
\end{bmatrix}\begin{bmatrix}\mathbf{c}^{\dagger}\\
\mathbf{0}\\
\vdots\\
\mathbf{0}
\end{bmatrix}\\
=&\begin{bmatrix}\frac{1}{\lambda-\lambda_1^*}&0&\cdots&0\\
\frac{\lambda_1-\lambda}{\left(\lambda-\lambda_1^*\right)^2}&*&\cdots&0\\
\vdots&\vdots&\ddots&\vdots\\
\frac{\left(\lambda_1-\lambda\right)^{N-1}}{\left(\lambda-\lambda_1^*\right)^{N}}&*&\cdots&*
\end{bmatrix}\begin{bmatrix}\mathbf{c}^{\dagger}\\
\mathbf{0}\\
\vdots\\
\mathbf{0}
\end{bmatrix}
=\begin{bmatrix}\frac{1}{\lambda-\lambda_1^*}\\
\frac{\lambda_1-\lambda}{\left(\lambda-\lambda_1^*\right)^2}\\
\vdots\\
\frac{\left(\lambda_1-\lambda\right)^{N-1}}{\left(\lambda-\lambda_1^*\right)^{N}}
\end{bmatrix}\mathbf{c}^{\dagger}.
\end{aligned}
\end{equation}
Then the part of Darboux matrix $\mathbf{T}_{N}(\lambda; 0, 0)$ in Eq.\eqref{eq:T-I} can be simplified into
\begin{equation}
\begin{aligned}
\mathbf{Y}_N\mathbf{M}^{-1}\mathbf{D}\mathbf{Y}_N^{\dagger}&=\frac{\lambda_1^*-\lambda_1}{1+|\mathbf{c}|^2}\mathbf{c}\begin{bmatrix}1,&-1,&\cdots,&(-1)^{N-1}\end{bmatrix}\begin{bmatrix}\frac{1}{\lambda-\lambda_1^*}\\
(-1)\frac{\lambda-\lambda_1}{\left(\lambda-\lambda_1^*\right)^2}\\
\vdots\\
(-1)^{N-1}\frac{\left(\lambda-\lambda_1\right)^{N-1}}{\left(\lambda-\lambda_1^*\right)^{N}}
\end{bmatrix}\mathbf{c}^{\dagger}\\
&=\frac{\left(\frac{\lambda-\lambda_1}{\lambda-\lambda_1^*}\right)^N-1}{1+|\mathbf{c}|^2}\mathbf{c}\mathbf{c}^{\dagger}
\\&=\mathbf{Q}_c{\rm diag}\left(\left(\frac{\lambda-\lambda_1}{\lambda-\lambda_1^*}\right)^N-1, 0, 0\right)\mathbf{Q}_c^{-1},
\end{aligned}
\end{equation}
which implies that the Darboux matrix $\mathbf{T}_{N}(\lambda; 0, 0)$ can be given as the final result \eqref{eq:tn00}. This completes the proof.
\end{proof}

Due to the above proposition \ref{prop:Tn}, we have the following corollary immediately
\begin{corollary}
$|\mathbf{q}(0,0)|=2N|\Im(\lambda_1)|.$
\end{corollary}

\begin{lemma}\label{lem:qc-decom}
For the matrix $\mathbf{Q}_c$, we have the following matrix decomposition:
\begin{equation}
\mathbf{Q}_c=\mathbf{Q}_H \mathbf{Q}_d\mathbf{Q}_H^{\dagger},
\end{equation}
where
\begin{equation}
\mathbf{Q}_H=\begin{bmatrix}
1&0&0\\
0&\cos(\alpha)\ee^{\ii\beta} &-\sin(\alpha)\ee^{\ii\beta}  \\
0&\sin(\alpha)\ee^{\ii\gamma} & \cos(\alpha)\ee^{\ii\gamma}   \\
\end{bmatrix},\,\,\,\,\,\,\, \mathbf{Q}_d=\begin{bmatrix}
1&-|c|&0\\
|c|&1&0  \\
0&0 & 1  \\
\end{bmatrix},
\end{equation}
and $c_1=|c|\cos(\alpha)\ee^{\ii\beta}$, $c_2=|c|\sin(\alpha)\ee^{\ii\gamma}$, $|c|=\sqrt{|c_1|^2+|c_2|^2}$.
\end{lemma}
\begin{remark}
Based on the result in \cite{Peter-Duke-2019}, \cite{Deniz-JNS-2019} and \cite{Deniz-arXiv-2019}, we know the key point to study the NLS asymptotics is how to decompose the $\mathbf{Q}^{-1}$ into upper and lower matrix. But to a general $3\times 3$ matrix, this decomposition becomes more difficult. To deal with this problem, we give a decomposition to $\mathbf{Q}_c$ in lemma \ref{lem:qc-decom} for the purpose of changing the important information to $\mathbf{Q}_d$, which is a good block matrix and can be regarded as the $2\times 2$ matrix in essence. With the aid of this decomposition, the asymptotic analysis to Eq.\eqref{eq:cnls} becomes available.
\end{remark}
By the theory of classic Darboux transformation, the parameter $|c|$ just depends the location of solitons. For convenience, we let $|c|=1.$
From Eq.\eqref{eq:Mmatrix}, we know that the jump about $\mathbf{M}^{[N]}(\lambda; x, t)$ is only related to the matrix $\mathbf{T}_{N}(\lambda; 0,0)$ and $\pmb{\Phi}(\lambda;x,t)$, in {\bf proposition} \ref{prop:Tn}, the matrix $\mathbf{T}_{N}(\lambda; 0,0)$ can be diagonalized in a well-formed one. As a matter of course, we can give the Riemann-Hilbert problem about $\mathbf{M}^{[N]}(\lambda; x, t)$:
\begin{rhp}\label{RHP0}
Let $(x,t)\in\mathbb{R}^2$ be arbitrary parameters, and $N\in\mathbb{Z}_{>0}$. Find a $3\times 3$ matrix function $\mathbf{M}^{[N]}(\lambda; x, t)$ satisfying the following properties:
\begin{itemize}
\item \textbf{Analyticity}: $\mathbf{M}^{[N]}(\lambda; x, t)$ is analytic for $\lambda\in\mathbb{C}\setminus \partial D_0$, it takes the continuous boundary values from the interior and exterior of $\partial D_0$.
\item \textbf{Jump condition}: The boundary values on the jump contour $\partial D_0$ are related as
    \begin{multline}
    \mathbf{M}^{[N]}_{+}(\lambda; x, t){=}\mathbf{M}^{[N]}_{-}(\lambda; x, t)\ee^{{-}\ii\left(\lambda x{+}\lambda^2 t\right)\pmb{\sigma}_3}\mathbf{Q}_{c}\\\cdot {\rm diag}\left(\left(\frac{\lambda{-}\lambda_1}{\lambda{-}\lambda_1^*}\right)^{\frac{2}{3}N}, \left(\frac{\lambda{-}\lambda_1^*}{\lambda{-}\lambda_1}\right)^{\frac{1}{3}N}, \left(\frac{\lambda{-}\lambda_1^*}{\lambda{-}\lambda_1}\right)^{\frac{1}{3}N} \right)\mathbf{Q}_{c}^{-1}\ee^{\ii\left(\lambda x{+}\lambda^2t\right)\pmb{\sigma}_3},
    \end{multline}
\item \textbf{Normalization}: $\mathbf{M}^{[N]}(\lambda; x, t)=\mathbb{I}+O(\lambda^{-1}),$
\end{itemize}
where $D_0$ is a disk centered at the origin in which the spectral point $\lambda=\lambda_1$ is involved inside.
\end{rhp}
Given the above Riemann-Hilbert problem, the potential $\mathbf{q}(x,t)$ can be recovered by the formula:
\begin{equation}
\mathbf{q}(x, t)=2\lim\limits_{\lambda\to\infty}(\lambda\mathbf{M}_{1,2}(\lambda; x, t)),
\end{equation}
where the subscript $_{1,2}$ denotes the elements $(1,2)$ and $(1,3)$ of a matrix.

Even though the lower order soliton solutions can be obtained by the determinant formula, the complicated expressions for these high order solutions with large $N$ are hard to analyze. To avoid this problem, we would like to replace the determinant formulas with the Riemann-Hilbert representation.

Now we proceed to consider the case of infinite order, i.e. $N\to\infty$. As $N\to \infty$, by the estimate \eqref{eq:estimate}, the sectional analytic matrices $\mathbf{M}^{[N]}(\lambda; x, t)$ in Eq. \eqref{eq:Mmatrix} are uniformly convergence by choosing ${\rm Im}(\lambda_1)=1/N$ in the corresponding region. Then the limit of matrix $\mathbf{M}^{[\infty]}(\lambda; x, t)$ satisfies the following new RHP:
\begin{rhp}\label{RHP2-0}
	(Soliton of infinite order). Let $(x, t)\in \mathbb{R}^2.$ Find a $3\times 3$ matrix $\mathbf{M}^{[\infty]}(\lambda; x, t)$ with the following properties
	\begin{itemize}
		\item \textbf{Analyticity}: $\mathbf{M}^{[\infty]}(\lambda; x, t)$ is analytic in $\lambda\in \mathbb{C}\setminus \partial D_0,$ and it takes the continuous boundary condition from the interior and exterior of $D_0$.
		\item \textbf{Jump Condition}: The boundary condition on $D_0$ (clockwise orientation) are related with the following jump condition
		\begin{equation*}
		{\mathbf{M}^{[\infty]}}_{+}(\lambda; x, t)={\mathbf{M}^{[\infty]}}_{-}(\lambda; x, t)\ee^{-\ii\left[\left(\lambda x+\lambda^2t\right)\right]\pmb{\sigma}_3}\mathbf{Q}_{c}\mathbf{{J}}(\lambda)\mathbf{Q}_{c}^{-1}\ee^{\ii\left[\left(\lambda x+\lambda^2t\right)\right]\pmb{\sigma}_3},
		\end{equation*}
		where $\mathbf{{J}}(\lambda)={\rm diag}\left(\ee^{-\frac{4}{3}\ii (\lambda-{\rm Re}(\lambda_1))^{-1}}, \ee^{\frac{2}{3}\ii(\lambda-{\rm Re}(\lambda_1))^{-1}}, \ee^{\frac{2}{3}\ii(\lambda-{\rm Re}(\lambda_1))^{-1}} \right),$
		\item \textbf{Normalization}: ${\mathbf{M}^{[\infty]}}(\lambda; x, t)\to \mathbb{I}$ as $\lambda\to\infty.$
	\end{itemize}
\end{rhp}
Then the potential $\mathbf{q}(x,t)$ can be recovered by the formula:
\begin{equation}
\mathbf{q}(x, t)=2\lim\limits_{\lambda\to\infty}(\lambda\mathbf{M}^{[\infty]}_{1,2}(\lambda; x, t)),
\end{equation}
where the subscript $_{1,2}$ denotes the elements $(1,2)$ and $(1,3)$ of a matrix.

%From the RHP \ref{RHP0}, we know the jump about $\mathbf{M}^{[N]}(\lambda; x, t)$ contains the parameter $N$, if $N$ is large, what will this RHP will be? How the corresponding asymptotic behavior of the solutions will be. To study this problem, we first need to convert the phase term $\ee^{-\ii\left(\lambda x+\lambda^2 t\right)\pmb{\sigma}_3}$ to an analysable term. In the part of introduction, we have proved that under the choice ${\rm Im}(\lambda_1)=1/N$, the Darboux matrix is uniformly convergent. And there appear a phase term $\ee^{\frac{-2\ii}{\lambda-{\rm Re}(\lambda_1)}}$ when $N\to\infty$.

To reduce the above RHP in a simper form, we would like to use the Lie symmetry to reduce the above RHP. As we know that the solution of CNLS equation has a Galilean transformation $\mathbf{q}(x,t)\to\mathbf{q}(x-2\alpha t,t)\ee^{2\ii\alpha(x-\alpha t)}$, which can be used to absorb the parameter ${\rm Re}(\lambda_1)$ in the RHP. Under the Galilean transformation:
\begin{equation}\label{eq:xandX}
x=\chi-2{\rm Re}(\lambda_1) t,\quad \lambda=\Lambda+{\rm Re}(\lambda_1),\quad t=\tau,
\end{equation}
together with the gauge transformation $\mathbf{N}(\Lambda; \chi, \tau)=\ee^{{\rm ad_{\pmb{\sigma}_3}}\ii{\rm Re}(\lambda_1)(\chi-{\rm Re}(\lambda_1)\tau)}\mathbf{M}^{[\infty]}(\Lambda+{\rm Re}(\lambda_1); \chi-2{\rm Re}(\lambda_1)\tau, \tau)$, where $\ee^{{\rm ad_{\pmb{\sigma}_3}}}\cdot=\ee^{\pmb{\sigma}_3}\cdot\ee^{-\pmb{\sigma}_3}$, the RHP \ref{RHP2-0} will turn into following new RHP:
\begin{rhp}\label{RHP2}
(Soliton of infinite order). Let $(\chi, \tau)\in \mathbb{R}^2.$ Find a $3\times 3$ matrix $\mathbf{N}(\Lambda; \chi, \tau)$ with the following properties
\begin{itemize}
\item \textbf{Analyticity}: $\mathbf{N}(\Lambda; \chi, \tau)$ is analytic in $\Lambda\in \mathbb{C}\setminus \partial D_0,$ and it takes the continuous boundary condition from the interior and exterior of $D_0$.
\item \textbf{Jump Condition}: The boundary condition on $D_0$ (clockwise orientation) are related with the following jump condition
    \begin{equation*}
    \mathbf{N}_{+}(\Lambda; \chi, \tau)=\mathbf{N}_{-}(\Lambda; \chi, \tau)\ee^{-\ii\left[\left(\Lambda \chi+\Lambda^2\tau\right)\right]\pmb{\sigma}_3}\mathbf{Q}_{c}\mathbf{\widetilde{J}}(\Lambda)\mathbf{Q}_{c}^{-1}\ee^{\ii\left[\left(\Lambda \chi+\Lambda^2\tau\right)\right]\pmb{\sigma}_3},
    \end{equation*}
    where $$\mathbf{\widetilde{J}}(\Lambda)={\rm diag}\left(\ee^{-\frac{4}{3}\ii\Lambda^{-1}}, \ee^{\frac{2}{3}\ii\Lambda^{-1}}, \ee^{\frac{2}{3}\ii\Lambda^{-1}} \right),$$
\item \textbf{Normalization}: $\mathbf{N}(\Lambda; \chi, \tau)\to \mathbb{I}$ as $\Lambda\to\infty.$
\end{itemize}
\end{rhp}
Then the potential function defined by $\mathbf{N}(\lambda; \chi, \tau)$ with the form
\begin{equation}\label{eq:new-potential}
\mathbf{q}(\chi, \tau):=2\lim\limits_{\Lambda\to\infty}\Lambda\mathbf{N}_{1,2}(\Lambda; \chi, \tau)
\end{equation}
also satisfies the coupled nonlinear Schr\"{o}dinger equation \eqref{eq:cnls} with the replacement $x\to \chi, t\to \tau$.

In what follows, we would like to analyze dynamics for the high order soliton with large order and infinite order by the aid of above RHP \ref{RHP0} and RHP \ref{RHP2} respectively. The main tools to analyze them are the Deift-Zhou nonlinear steepest phase method \cite{Zhou-Annals-1993}.
%Given this new RHP \ref{RHP2}, we should prove it has the unique solution, which can be proved by the Zhou's vanishing Lemma \cite{Zhou-SIAM-1989}. Firstly we change the orientation of jump in the lower half plane to insure the Schwartz symmetry. Then the jump will be rewritten as
%\begin{equation*}
%\mathbf{J}(\Lambda; X, T):=\left\{\begin{split}&\ee^{-\ii\left(\Lambda X+2\Lambda^2T\right)\pmb{\sigma}_3}\mathbf{Q}\mathbf{\widetilde{J}}(\Lambda)\mathbf{Q}^{-1}\ee^{\ii\left(\Lambda X+2\Lambda^2T\right)\pmb{\sigma}_3},&|\Lambda|=1, \quad {\rm Im}(\Lambda)>0,\\
%&\ee^{-\ii\left(\Lambda X+2\Lambda^2T\right)\pmb{\sigma}_3}\mathbf{Q}\mathbf{\widetilde{J}}^{-1}(\Lambda)\mathbf{Q}^{-1}\ee^{\ii\left(\Lambda X+2\Lambda^2T\right)\pmb{\sigma}_3},&|\Lambda|=1,\quad {\rm Im}(\Lambda)<0.
%\end{split}\right.
%\end{equation*}
%It is clear that
%\begin{equation*}
%\begin{split}
%\mathbf{J}(\Lambda^*; X, T)&=\left[\ee^{\ii\left(\Lambda X+2\Lambda^2T\right)\pmb{\sigma}_3}\mathbf{Q}\mathbf{\widetilde{J}}(\Lambda)\mathbf{Q}^{\T}\ee^{-\ii\left(\Lambda X+2\Lambda^2T\right)\pmb{\sigma}_3}\right]^{*}\\
%&=\left[\ee^{-\ii\left(\Lambda X+2\Lambda^2T\right)\pmb{\sigma}_3}\mathbf{Q}\mathbf{\widetilde{J}}(\Lambda)\mathbf{Q}^{\T}\ee^{\ii\left(\Lambda X+2\Lambda^2T\right)\pmb{\sigma}_3}\right]^{\dagger}\\
%&=\mathbf{J}(\Lambda; X, T)^{\dagger},
%\end{split}
%\end{equation*}
%where the superscript $^{\dagger}$ denotes the conjugate transpose of the matrix. Thus the RHP \ref{RHP2} is uniquely solvable.

\section{Asymptotic behavior for the large order solitons}
\label{sec:large-order}

Last section, we have constructed the Riemann-Hilbert problem corresponding to the high order solitons with the aid of Darboux matrix. In this section, we utilize the RHP \ref{RHP0} to analyze the high order solitons with large order. Following the way \cite{Deniz-arXiv-2019}, we define%It can be seen the jump matrix contains the factor $\left(\frac{\lambda-\lambda_1}{\lambda-\lambda_1^*}\right)^{N}$.

\begin{equation}\label{eq:xandXII}
X:=\frac{x}{N}, \qquad T:=\frac{t}{N},
\end{equation}
then the jump matrix in RHP \ref{RHP0} changes into
\begin{equation}
%\begin{split}
%&\ee^{-\ii\lambda N\left(X+\lambda T\right)\pmb{\sigma}_3}\mathbf{Q}_{c}{\rm diag}\left(\left(\frac{\lambda-\lambda_1}{\lambda-\lambda_1^*}\right)^{\frac{2}{3}N},\left(\frac{\lambda-\lambda_1^*}{\lambda-\lambda_1}\right)^{\frac{1}{3}N},\left(\frac{\lambda-\lambda_1^*}{\lambda-\lambda_1}\right)^{\frac{1}{3}N}\right)\mathbf{Q}^{-1}_{c}\ee^{\ii\lambda N\left(x+\lambda T\right)\pmb{\sigma}_3}\\
%:=&
\ee^{-\ii\lambda N\left(X+\lambda T\right)\pmb{\sigma}_3}\mathbf{Q}_{c}\mathbf{J}_{f,d}(\lambda)
\mathbf{Q}^{-1}_{c}\ee^{\ii\lambda N\left(x+\lambda T\right)\pmb{\sigma}_3},
%\end{split}
\end{equation}
where $$\mathbf{J}_{f,d}(\lambda):={\rm diag}\left(\left(\frac{\lambda-\lambda_1}{\lambda-\lambda_1^*}\right)^{\frac{2}{3}N}, \left(\frac{\lambda-\lambda_1^*}{\lambda-\lambda_1}\right)^{\frac{1}{3}N}, \left(\frac{\lambda-\lambda_1^*}{\lambda-\lambda_1}\right)^{\frac{1}{3}N}\right).$$

%Compared to the $2\times 2$ jump matrix, the analysis to $3\times 3$ RHP becomes more difficult.
%One of the most effective ways to tackle with it is rewriting the $3\times 3$ jump matrix to a block matrix.
In {\bf lemma} \ref{lem:qc-decom}, we give a decomposition to $\mathbf{Q}_{c}$, in which the block matrix $\mathbf{Q}_{d}$ can carry out the upper-lower triangle decomposition readily. Moreover, we utilize the gauge transformation $\mathbf{Q}_H$ to redefine the sectional analytic matrix
\begin{equation}\label{eq:N_f}
\widetilde{\mathbf{N}}_{f}(\lambda; X, T):=\left\{\begin{aligned}
&\mathbf{Q}_{H}^{\dagger}\mathbf{M}^{[N]}(\lambda; NX, NT)\mathbf{Q}_{H}\ee^{-\ii\lambda N\left(X+\lambda T\right)\pmb{\sigma}_3}\mathbf{Q}_{d}\mathbf{C}_{d}\ee^{\ii\lambda N\left(X+\lambda T\right)\pmb{\sigma}_3}, &\quad \lambda\in D_0,\\
&\mathbf{Q}_{H}^{\dagger}\mathbf{M}^{[N]}(\lambda; NX, NT)\mathbf{Q}_{H}\mathbf{J}_{f, d}^{-1}(\lambda),&\quad \lambda\notin D_0,
\end{aligned}\right.
\end{equation}
where %$\mathbf{C}_{d}$ is designed to normalize $\mathbf{Q}_{d}$ as $\det\left(\mathbf{Q}_{d}\right)=1$,
$\mathbf{C}_{d}:={\rm diag}\left(\frac{1}{\sqrt{|c|^2+1}}, \frac{1}{\sqrt{|c|^2+1}}, 1\right)$, which satisfies the following Riemann-Hilbert problem by choosing $c=1$:
\begin{rhp}\label{rhp:far-field}
\begin{itemize}
Let $(X, T)\in \mathbb{R}^2$, find a $3\times 3$ matrix $\widetilde{\mathbf{N}}_{f}(\lambda; X, T)$ satisfying the following conditions
\item {\bf Analyticity}: $\widetilde{\mathbf{N}}_{f}$ is analytic for $\lambda\in\mathbb{C}\setminus\partial D_0$, and it takes continuous boundary condition from the interior and the exterior of $\partial D_0$.
\item {\bf Jump Condition}: The boundary condition on the jump contour $\partial D_0$ are related with the following jump
\begin{equation}\label{eq:n-jump}
\begin{split}
\widetilde{\mathbf{N}}_{f,+}(\lambda; X, T)&%=\widetilde{\mathbf{N}}_{f,-}(\lambda; X, T)\ee^{\left(-\ii N\left(\lambda X+\lambda^2T\right)+\frac{N}{2}\log\left(\frac{\lambda-\lambda_1}{\lambda-\lambda_1^*}\right)\right)\pmb{\sigma}_3}\widehat{\mathbf{Q}}_{d}^{-1}\ee^{\left(\ii N\left(\lambda X+\lambda^2T\right)-\frac{N}{2}\log\left(\frac{\lambda-\lambda_1}{\lambda-\lambda_1^*}\right)\right)\pmb{\sigma}_3}\\
=\widetilde{\mathbf{N}}_{f,-}(\lambda; X, T)\ee^{-N\varphi(\lambda; X, T)\pmb{\sigma}_3}\widehat{\mathbf{Q}}_{d}^{-1}\ee^{N\varphi(\lambda; X, T)\pmb{\sigma}_3},
\end{split}
\end{equation}
where $\varphi(\lambda; X, T)=\ii \left(\lambda X+\lambda^2T\right)-\frac{1}{2}\log\left(\frac{\lambda-\lambda_1}{\lambda-\lambda_1^*}\right),\,\,\,\, \widehat{\mathbf{Q}}_d=\begin{bmatrix}\frac{\sqrt{2}}{2}&-\frac{\sqrt{2}}{2}&0\\\frac{\sqrt{2}}{2}&\frac{\sqrt{2}}{2}&0\\0&0&1
\end{bmatrix}.$
\item {\bf Normalization}: When $\lambda\to\infty$, we have $\widetilde{\mathbf{N}}_{f}(\lambda; X, T)\to\mathbb{I}$.
\end{itemize}
\end{rhp}
Then the potential $\mathbf{q}(X, T)$ can be recovered by
\begin{equation}
\mathbf{q}(X,T)=2\lim\limits_{\lambda\to\infty}\lambda\left(\mathbf{Q}_{H}\mathbf{\widetilde{N}}_{f}(\lambda; X, T)\mathbf{Q}_{H}^{-1}\right)_{12}.
\end{equation}
Now we proceed to analyze the asymptotics of high order solitons by the above RHP \ref{rhp:far-field}. We firstly need to discuss the sign structure of ${\rm Re}[\varphi(\lambda; X, T)]$ in Eq.\eqref{eq:n-jump}. To achieve this aim,  we utilize the critical points of $\varphi(\lambda; X, T)$ satisfy the following cubic equation with respect to $\lambda$:
\begin{multline}
2\lambda^3T{+}\left(X{-}4{\rm Re}(\lambda_1) T\right)\lambda^2{+}\left(2{\rm Re}(\lambda_1)^2T{+}2{\rm Im}(\lambda_1)^2T{-}2{\rm Re}(\lambda_1) X\right)\lambda \\{+}X{\rm Im}(\lambda_1)^2{+}X{\rm Re}(\lambda_1)^2{-}{\rm Im}(\lambda_1){=}0,
\end{multline}
whose discriminant is
\begin{equation}\label{eq:discrim}
\begin{split}
-4\Bigg[&8{\rm Re}(\lambda_1){\rm Im}(\lambda_1)\left(4{\rm Re}(\lambda_1)^2{\rm Im}(\lambda_1) X{+}4{\rm Im}(\lambda_1)^3 X{-}{\rm Re}(\lambda_1)^2{-}9{\rm Im}(\lambda_1)^2\right)T^3\\&+{\rm Im}(\lambda_1) X^3 \left({\rm Im}(\lambda_1) X-1\right)\\&+{\rm Im}(\lambda_1)\left(24{\rm Re}(\lambda_1)^2{\rm Im}(\lambda_1) X^2{+}8{\rm Im}(\lambda_1)^3X^2{-}12{\rm Re}(\lambda_1)^2X{-}36{\rm Im}(\lambda_1)^2X{+}27{\rm Im}(\lambda_1)\right)T^2\\&+2{\rm Re}(\lambda_1){\rm Im}(\lambda_1) X^2\left(4{\rm Im}(\lambda_1) X-3\right)T+16{\rm Im}(\lambda_1)^2\left({\rm Re}(\lambda_1)^2{+}{\rm Im}(\lambda_1)^2\right)^2T^4\Bigg].
\end{split}
\end{equation}
Then we will show the boundary lines of different asymptotic regions for the high order soliton.
\begin{figure}[!h]
\centering
\subfigure[contour plot]{
\includegraphics[height=0.28\textheight,width=0.45\textwidth]{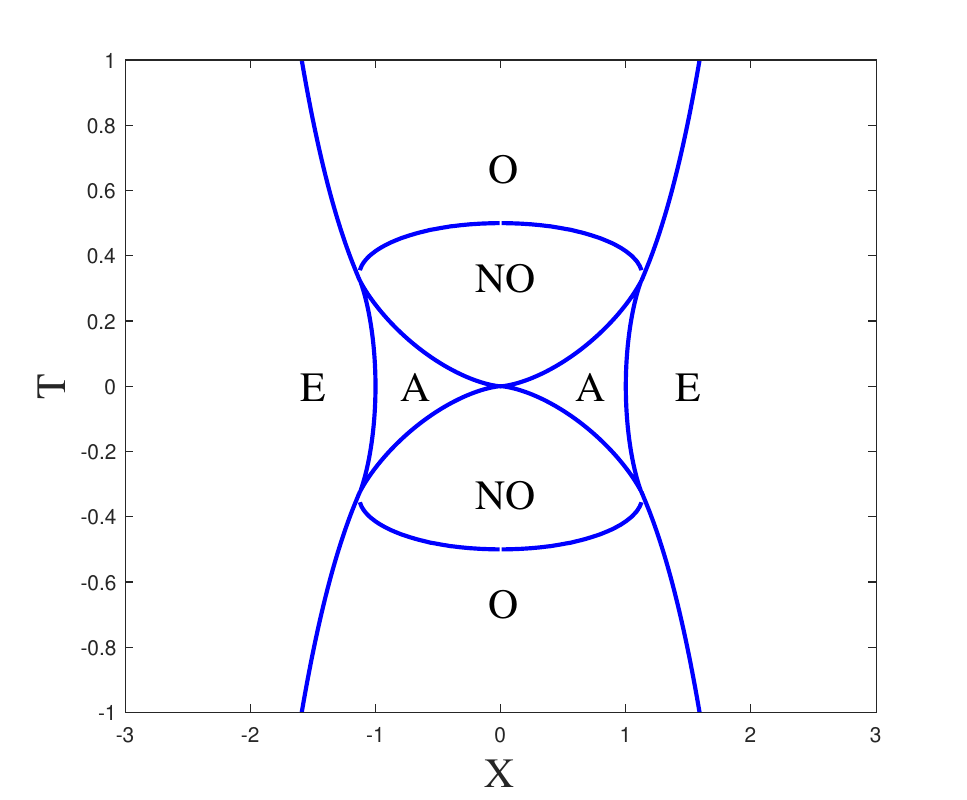}}
\subfigure[density plot]{\includegraphics[height=0.28\textheight,width=0.45\textwidth]{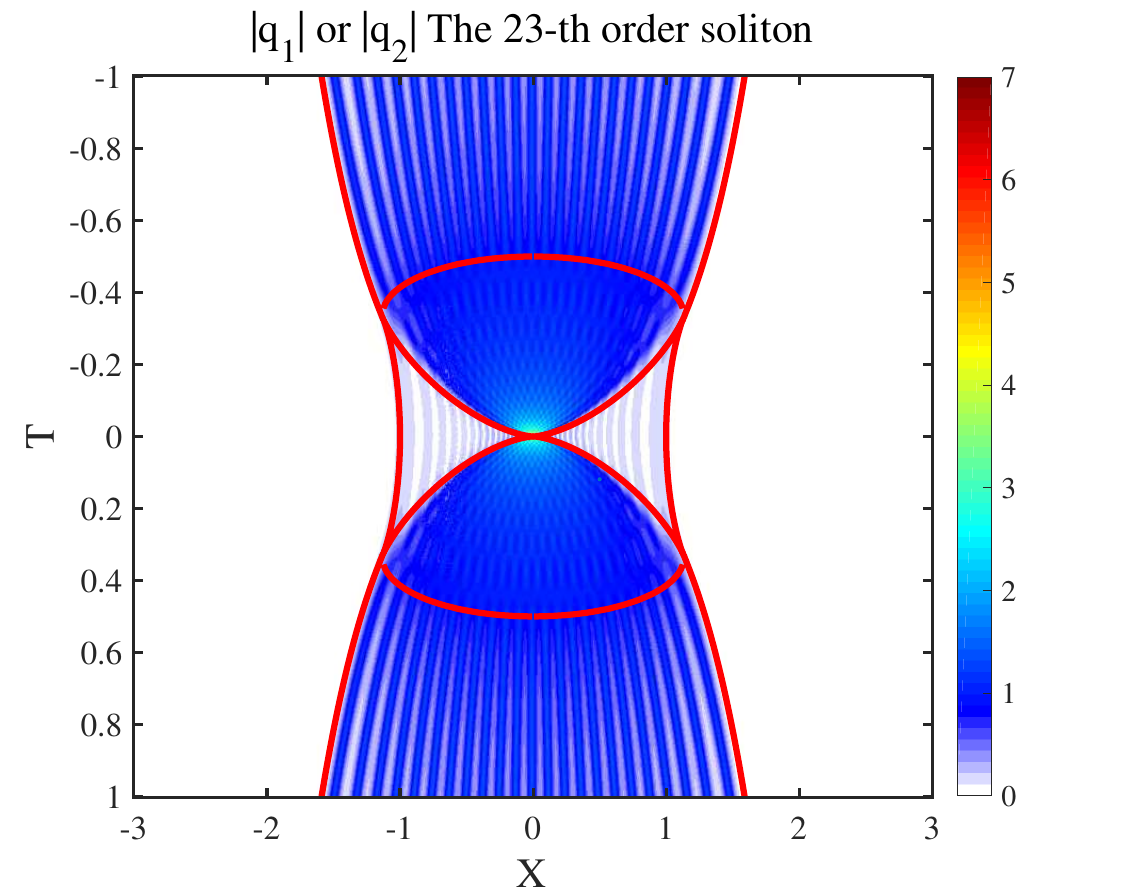}}
\caption{The 23-th order soliton solution and the asymptotic region.}
\label{high-order-soliton-1}
\end{figure}

From the Fig.\ref{high-order-soliton-1}, we know there exist four different asymptotic regions when $N$ is large, which is directly related by the discriminant \eqref{eq:discrim}. Thus we give a detailed analysis about Eq.\eqref{eq:discrim}.

Case I: If \lm{Eq.\eqref{eq:discrim}$>0$}, then the three critical points are all real roots, set $\lambda_1<\lambda_2<\lambda_3$. Under this condition, we choose some special $X, T$ and give contour of ${\rm \Re}[\varphi(\lambda; X, T)]=0$, which is shown in Fig.\ref{contour-algebraic}. And the third critical point locates on the left of \lm{$a$}. As the variables $X$ and $T$ changes continuously until Eq.\eqref{eq:discrim}$=0$, there will appear a double roots, which is the boundary between the algebraic decay region to other different regions.

Case II: After the constraint Eq.\eqref{eq:discrim}$=0$, the next stage is Eq.\eqref{eq:discrim}$<0$. With a slight change to the variable $X$ and $T$, the contour of ${\rm \Re}(\varphi(\lambda; X, t))$ will change essentially. There will appear two different types of contours, one is the closed contour in Case I will split into two closed contours, locating in the upper and lower half plane respectively, which represent the exponential decay region, as shown in Fig. \ref{contour-exponent}. The other one is the original closed contour will not be closed any more, we can not deform the contour with the original contour any more, to analyze the asymptotics in the this region, we need construct a suit $g$-function, whose contour is shown in Fig.\ref{contour-non-os}.

Case III: If the contour is no longer closed, the shape of the contour is variable. Fix $X$ and variation of $T$, if $T$ is large, the critical points approximate to $\lambda=\lambda_1, \lambda=\lambda_1^*$ and $\lambda=0$. These three critical points lie on the imaginary axis, then the $g$ function constructed in Case II is not available any more, we must reconstruct a new $g$-function to match this region, whose contour is shown in Fig.\ref{contour-os}.

Then we give the boundary lines for these four regions. The boundary lines of the algebra-decay region is the discriminant Eq.\eqref{eq:discrim} equals to zero, which contains three contours, two of which are the boundaries between the algebra-decay region and the oscillatory region, the left one is the boundary between the algebra-decay region and the exponential-decay one. The second boundary line is separated by the exponential-decay region and the oscillatory one. In \cite{Deniz-arXiv-2019}, the authors give a detailed description about the boundary contour between different asymptotic regions. Especially, the boundary between the non-oscillatory and the oscillatory region is determined by a $g$ function, based on the rule of constructing $g$ function, when $\lambda$ is located in the non-oscillatory region, we set
%Both of these two regions have one real critical point $\lambda=\lambda_0$ and a complex-conjugate pair $\lambda=\lambda_1, \lambda=\lambda_1^{*}$. As for the exponent-decay region, these critical points are located out of the closed contour, thus the real part of $\varphi(\lambda; X, T)$ is nonzero, while the oscillatory region is an open region, so the boundary between these two regions is the parameters $\lambda_1$ and $\lambda_1^{*}$ located onto the contour, ${\rm R}(\varphi(\lambda_1; X, T))={\rm R}(\varphi(\lambda_1^{*}; X, T))=0$. The last boundary is the oscillatory region and the non-oscillatory region, both of which need construct the $g$ function to analyze the asymptotic behavior, so the boundary curve must be related to the $g$ function.
\begin{equation}
\begin{split}
R(\lambda; X, T)&=\left(\left(\lambda-a_{n}\right)\left(\lambda-a_{n}^*\right)\right)^{1/2}, \\
g'(\lambda; X, T)&=\frac{R(\lambda; X, T)}{2\pi \ii}\int_{\Sigma}\frac{2\ii X+4\ii sT+\frac{1}{\lambda-\lambda_1^*}-\frac{1}{\lambda-\lambda_1}}{R_{+}(s; X, T)\left(s-\lambda\right)}ds,
\end{split}
\end{equation}
where the subscript $_+$ indicates the left side of the positive direction, the subscript $_n$ stands the non-oscillatory region. Thus we have
\begin{equation}\label{eq:g-function-1}
\varphi'(\lambda; X, T)-g'(\lambda; X, T)=R(\lambda; X, T)\left(2\ii T-\frac{1}{2R(\lambda_1^*; X, T)\left(\lambda_1^*-\lambda\right)}+\frac{1}{2R(\lambda_1; X, T)\left(\lambda_1-\lambda\right)}\right).
\end{equation}
Obviously, $\varphi'(\lambda)-g'(\lambda)$ has four roots.

When $\lambda$ is in the oscillatory region, we set
\begin{equation}\label{eq:G-function-1}
\begin{split}
\mathscr{R}(\lambda; X, T):&=\left(\left(\lambda-a_{o}\right)\left(\lambda-a_{o}^*\right)\left(\lambda-b_{o}\right)\left(\lambda-b_{o}^*\right)\right)^{1/2},\\
\varphi'(\lambda; X, T)-G'(\lambda; X, T)&=\mathscr{R}(\lambda; X, T)\left(\frac{1}{2\mathscr{R}(\lambda_1; X, T)
\left(\lambda_1-\lambda\right)}-\frac{1}{2\mathscr{R}(\lambda_1^*; X, T)\left(\lambda_1^*-\lambda\right)}\right),
\end{split}
\end{equation}
where the subscript $_o$ stands the oscillatory region, then $\varphi'(\lambda; X, T)-G'(\lambda; X, T)$ has five roots. Compared to the $g$-function and $G$-function in Eq.\eqref{eq:g-function-1} and Eq.\eqref{eq:G-function-1}, we find that Eq.\eqref{eq:G-function-1} is the generalization of Eq.\eqref{eq:g-function-1}, which implies that Eq.\eqref{eq:g-function-1} should have a double real root with respect to $\lambda$ on the boundary line. Thus the discriminant of the factor
$$2\ii T-\frac{1}{2R(\lambda_1^*)\left(\lambda_1^*-\lambda\right)}+\frac{1}{2R(\lambda_1)\left(\lambda_1-\lambda\right)}$$ with respect to $\lambda$ should equal to zero, namely
\begin{equation}
\begin{split}
&16R(\lambda_1^*; X, T)^2R(\lambda_1; X, T)^2\left(\lambda_1-\lambda_1^*\right)^2T^2+\left(R(\lambda_1; X, T)-R(\lambda_1^*; X, T)\right)^2\\&-8\ii\left(R(\lambda_1^*; X, T)+R(\lambda_1; X, T)\right)R(\lambda_1; X, T)R(\lambda_1^*; X, T)\left(\lambda_1-\lambda_1^*\right)T=0.
\end{split}
\end{equation}
Given all the boundary lines about these four asymptotic regions, we proceed to analyze the asymptotic behaviors in order.

\subsection{The Oscillatory region}
\label{sec:os}
We first consider the asymptotics in the oscillatory region, which has the most complicated dynamics among these four regions. For the further study, we give a RHP about the $G(\lambda)$ function for this region.
\begin{rhp}
(The G-function in the oscillatory region). In the oscillatory region, there must exist three unique contours $\Sigma_{\rm u}, \Sigma_{\rm d}$ and $\Sigma_{\rm mid}$, which are determined by the unique $G(\lambda)$-function completely, and $G(\lambda)$ satisfies the following condition
\begin{itemize}
\item {\bf Analyticity}: $G(\lambda; X, T)$ is analytic except on $\Sigma_{\rm u}\cup\Sigma_{\rm d}\cup\Sigma_{\rm mid}$, where it can arrive the boundary values.
\item{\bf Jump condition}: The boundary value given by $G(\lambda; X, T)$ are related by
\begin{equation}\label{RHP-G-O}
\begin{aligned}
G_{+}(\lambda; X, T)+G_-(\lambda; X, T)&=2\varphi(\lambda; X, T)+\Omega_{o}, \qquad &\lambda\in\Sigma_{\rm u},\\
G_{+}(\lambda; X, T)+G_-(\lambda; X, T)&=2\varphi(\lambda; X, T)-\Omega^*_{o}, \qquad &\lambda\in\Sigma_{\rm d},\\
G_{+}(\lambda; X, T)-G_-(\lambda; X, T)&=d, \qquad &\lambda\in\Sigma_{\rm mid}.\\
\end{aligned}
\end{equation}
\item {\bf Normalizations}: When $\lambda\to \infty$, $G(\lambda; X, T)\to O(\lambda^{-1})$.
\item {\bf Symmetry}: $G(\lambda; X, T)$ satisfies the symmetry condition
\begin{equation*}
G(\lambda; X, T)=-G(\lambda^*; X, T)^*.
\end{equation*}
\end{itemize}
where the subscript $_u$ and $_d$ indicates the up and down in the whole complex plane.
\end{rhp}
Generally speaking, to solve $G(\lambda; X, T)$-function, we must look for a suitable function with a branch cut in the $\Sigma_{\rm u}$ and $\Sigma_{\rm d}$. In the previous analysis, we construct a genus one algebraic curve $\mathscr{R}(\lambda; X, T)$ in Eq.\eqref{eq:G-function-1}, if the unknown branch cut $\Sigma_{\rm u}$ and $\Sigma_{\rm d}$ is determined by $a_o, a_o^*$ and $b_o, b_{o}^*$, then the new function $\frac{G'(\lambda; X, T)}{\mathscr{R}(\lambda; X, T)}$ will satisfy
\begin{equation}
\begin{aligned}
\frac{G'_+(\lambda; X, T)}{\mathscr{R}_+(\lambda; X, T)}-\frac{G'_-(\lambda; X, T)}{\mathscr{R}_-(\lambda; X, T)}&=2\frac{\varphi'(\lambda; X, T)}{\mathscr{R}(\lambda; X, T)},\qquad\lambda\in \Sigma_{\rm u}\cup\Sigma_{\rm d}.
\end{aligned}
\end{equation}
Therefore, $G'(\lambda; X, T)$ can be given by the Plemelj formula
\begin{equation}
G'(\lambda; X, T)=\frac{\mathscr{R}(\lambda; X, T)}{2\pi\ii}\int_{\Sigma_{\rm u}\cup \Sigma_{\rm d}}\frac{2\ii X+4\ii s T+\frac{1}{s-\lambda_1^*}-\frac{1}{s-\lambda_1}}{\mathscr{R}_{+}(s; X, T)(s-\lambda)}ds.
\end{equation}
Furthermore, with the generalized residue theorem, $G'(\lambda; X, T)$ can be given as
\begin{equation}
G'(\lambda; X, T)=\mathscr{R}(\lambda; X, T)\left(\mathop{\rm Res}\limits_{s=\lambda}+\mathop{\rm Res}\limits_{s=\lambda_1}+\mathop{\rm Res}\limits_{s=\lambda_1^*}+\mathop{\rm Res}\limits_{s=\infty}\right)\left(\frac{\ii X+2\ii sT+\frac{1}{2(s-\lambda_1^*)}-\frac{1}{2(s-\lambda_1)}}{\mathscr{R}(s; X, T)\left(s-\lambda\right)}\right).
\end{equation}
With a simple calculation, $G'(\lambda; X, T)-\varphi'(\lambda; X, T)$ equals to
\begin{equation}\label{eq:dGdvarphi}
G'(\lambda; X, T)-\varphi'(\lambda; X, T)=\frac{1}{2}\frac{\mathscr{R}(\lambda; X, T)}{\mathscr{R}(\lambda_1^*)(\lambda_1^*-\lambda)}-\frac{1}{2}\frac{\mathscr{R}(\lambda; X, T)}{\mathscr{R}(\lambda_1)(\lambda_1-\lambda)}.
\end{equation}
Obviously, there are five roots to $G'(\lambda; X, T)-\varphi'(\lambda; X, T)$, i.e. $$ \lambda=\frac{\mathscr{R}(\lambda_1^*; X, T)\lambda_1^*-\mathscr{R}(\lambda_1; X, T)\lambda_1}{\mathscr{R}(\lambda_1^*; X, T)-\mathscr{R}(\lambda_1; X, T)}:=c_o, \,\lambda=b_o,\, \lambda=b_o^*, \,\lambda=a_o,\, \lambda=a_o^*.$$  Based on the standard existence theory of the ordinary differential equations, the point $\lambda$ that is not the pole or the zero of $G'(\lambda; X, T)-\varphi'(\lambda; X, T)$ will lie on a unique trajectory. With the aid of local analysis for \eqref{eq:dGdvarphi}, there are three trajectories emanating from $a_o, a_o^*, b_o, b_o^*$, and there are four trajectories emanating from the point $c_o$.

%From the definition of $G(\lambda; X, T)$, which is analytic when $\lambda\in\mathbb{C}\setminus\left(\Sigma_{\rm u}\cup\Sigma_{\rm d}\cup\Sigma_{\rm mid}\right)$, then the integral
%\begin{equation}
%\begin{split}
%&\int_{a^*(X, T)}^{b^*(X, T)}G'_{+}(\lambda; X, T)d\lambda+\int_{b^*(X, T)}^{b(X, T)}G'(\lambda; X, T)d\lambda+\int_{b(X, T)}^{a(X, T)}G'_{+}(\lambda; X, T)\\+&\int_{a(X, T)}^{b(X, T)}G'_{-}(\lambda; X, T)d\lambda+\int_{b(X, T)}^{b^*(X, T)}G'(\lambda; X, T)d\lambda+\int_{b^*(X, T)}^{a^*(X, T)}G'_{-}(\lambda; X, T)d\lambda=0,
%\end{split}
%\end{equation}
%which indicates $\int_{a^*(X, T)}^{b^*(X, T)}G'_{+}(\lambda; X, T)d\lambda+\int_{b(X, T)}^{a(X, T)}G'_{+}(\lambda; X, T)d\lambda=0$. With the symmetry condition $G(\lambda^*; X, T)=-G^*(\lambda; X, T)$. Hence, we can give the basic contour to the real part of $G(\lambda; X, T)-\varphi(\lambda; X, T)$ in a lemma:
%\begin{lemma}
%${\rm Re}\left(G(a; X, T)-G(b; X, T)\right)=0$.
%\end{lemma}
%Take the point $a(X, T)$ for example, two trajectories are $a_o(X, T)$ to $b_o(X, T)$ and construct a closed loop, and the third trajectory is $a_o(X, T)$ to $\infty$. While to the point $c_o(X, T)$, two trajectories are $\lambda_0(X, T)$ to $\infty$, and the left two trajectories are $c_o(X, T)$ to $b_o(X, T)$ and $b_o^*(X, T)$ respectively. By this analysis, it can be seen $\Omega(X, T)$ and $d(X, T)$ in the  RHP\ref{rhp-G-function} are all pure imaginary constants. The contour of ${\rm Re}(\varphi(\lambda; X, T)-G(\lambda; X, T))$ is as follows:
\begin{figure}[!h]
\centering
\includegraphics[width=0.45\textwidth]{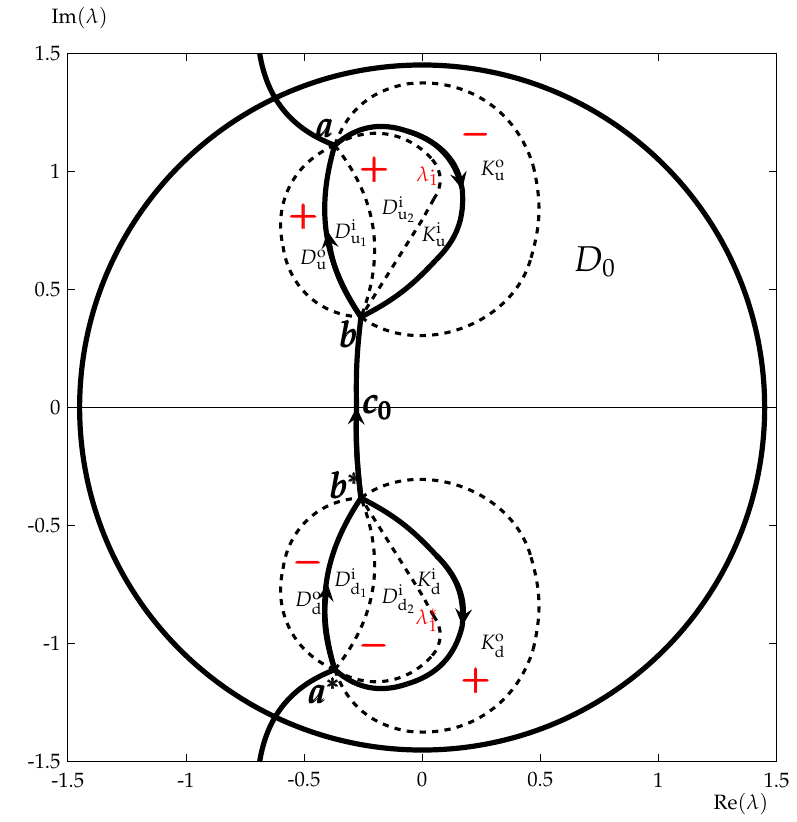}
\centering
\includegraphics[width=0.45\textwidth]{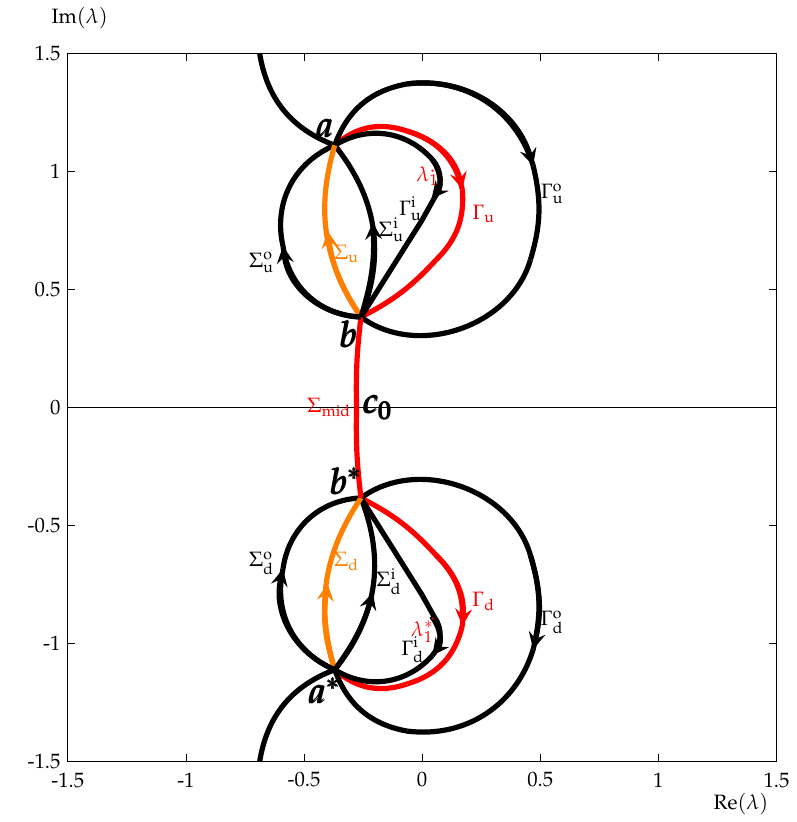}
\caption{The contour of ${\rm Re}(\varphi(\lambda; X, T)-G(\lambda; X, T))$ in the oscillatory region with the parameters $X=1, T=\frac{1}{2}.$ The left one gives the original contour about the $\widetilde{\mathbf{N}}_{f}(\lambda; X, T)$ and the sign of ${\rm Re}(\varphi(\lambda; X, T)-G(\lambda; X, T))$. The right one is the corresponding contour deformation.}
\label{contour-os}
\end{figure}

Next, we would calculate $\mathscr{R}(\lambda; X, T)$ with the properties of $G(\lambda; X, T)$ and trajectories.
For convenience, we set $\lambda_1=\ii$ and rewrite $\mathscr{R}(\lambda; X, T)$ to
\begin{equation}
\mathscr{R}(\lambda; X, T)=\sqrt{\lambda^4-l_1\lambda^3+l_2\lambda^2-l_3\lambda+l_4},
\end{equation}
where
\begin{equation}\label{eq:sub-l-a-b}
\begin{split}
l_1&=a_o+a_o^*+b_o+b_o^*, \quad l_2=a_ob_o+a_oa_o^*+a_ob_o^*+a_o^*b+b_ob_o^*+a_o^*b_o^*,\\
l_3&=a_oa_o^*\left(b_o+b_o^*\right)+b_ob_o^*\left(a_o+a_o^*\right), \qquad l_4=a_oa_o^*b_ob_o^*,
\end{split}
\end{equation}
are all functions with respect to $X$ and $T$. And the Jacobi determinant between the variables $(l_1,l_2,l_3,l_4)$ and $(a_o,a_o^*,b_o,b_o^*)$ is given by
\begin{equation}
J(a_o, a_o^*, b_o, b_o^*)=\frac{\partial (l_1, l_2, l_3, l_4)}{\partial (a_o, a_o^*, b_o, b_o^*)}=\left(a_o-a_o^*\right)\left(b_o-b_o^*\right)\left|a_o-b_o\right|^2\left|a_o-b_o^*\right|^2\neq0,
\end{equation}
which guarantee the above substitution \eqref{eq:sub-l-a-b} is valid by the inverse function theorem. When $\lambda\to\infty$, we have $G'(\lambda; X, T)=O(\lambda^{-2})$, which indicates the following three equations:
\begin{equation}
\begin{aligned}
&4\ii T+\frac{1}{\mathscr{R}(-\ii; X, T)}-\frac{1}{\mathscr{R}(\ii; X, T)}=0,\\
&\ii X+\frac{1}{4}\frac{l_1-2\ii}{\mathscr{R}(\ii; X, T)}-\frac{1}{4}\frac{l_1+2\ii}{\mathscr{R}(-\ii; X, T)}=0,\\
&\left(l_1^2-4l_2+4\ii l_1+8\right)\frac{1}{\mathscr{R}(\ii; X, T)}-\left(l_1^2-4\ii l_1-4l_2+8\right)\frac{1}{\mathscr{R}(-\ii; X, T)}=0.
\end{aligned}
\end{equation}
By solving the above equations, the parameters $l_2, l_3, l_4$ can be given as:
\begin{equation}
\begin{aligned}
l_2&=\frac{3}{4}l_1^2+\frac{X}{2T} l_1+2,\\
l_3&=l_1+\frac{4\left(l_1 T^2+XT\right)}{\left(4T^2+l_1^2T^2+2l_1TX+X^2\right)^2},\\
l_4&=\frac{3l_1^2}{4}+\frac{l_1 X}{2T}+1+\frac{-4T^2+l_1^2T^2+2l_1TX+X^2}{\left(4T^2+l_1^2T^2+2l_1TX+X^2\right)^2}.
\end{aligned}
\end{equation}
If $l_1$ is a real number, then $l_2, l_3, l_4$ are all real. Next, we should determine the last parameter $l_1$. It is clear that the real part of $G(\lambda; X, T)-\varphi(\lambda; X, T)$ is a function with respect to $l_1$. By choosing one special $X$ and $T$ in this region, we imposes that ${\rm \Re}\left(G(c_0)-\varphi(c_0)\right)=0$, which can be used to determine $l_1$ uniquely (the sign of the real part $G(\lambda)-\varphi(\lambda)$ in Fig.\ref{contour-os}). Actually, for every fixed $X$ and $T$, the function $G(\lambda)-\varphi(\lambda)$ is an elliptic integral, that is
\begin{equation}
\begin{aligned}
&G(\lambda; X, T){-}\varphi(\lambda; X, T)\\
&=\frac{1}{2}\frac{\mathscr{R}(-\ii; X, T)-\mathscr{R}(\ii; X, T)}{\mathscr{R}(\ii; X, T)\mathscr{R}(-\ii; X, T)}\int_{0}^{\lambda}\frac{s\mathscr{R}(s; X, T)} {s^2+1}ds\\
&+\frac{\ii}{2}\frac{\mathscr{R}(-\ii; X, T)+\mathscr{R}(\ii; X, T)}{\mathscr{R}(\ii; X, T)\mathscr{R}(-\ii; X, T)}\int_{0}^{\lambda}\frac{\mathscr{R}(s; X, T)} {s^2+1}ds-C_0
\end{aligned}
\end{equation}
where $C_0=G(0)-\varphi(0)$. With a complicated calculation, $G(\lambda){-}\varphi(\lambda)$ can be converted into the standard three types of elliptic integral. The explicit expression about this integral is very complicated, so we do not give it anymore. According to the contour in Fig.\ref{contour-os}, we can continue to deform it with the nonlinear steepest descent method. By the generalized residue theorem, we know that $\Re(G(a)-\varphi(a))=\Re(G(b)-\varphi(b))=0$ and $\Re(G(a^*)-\varphi(a^*))=\Re(G(b^*)-\varphi(b^*))=0$. Together with the local analysis, the signature chart can be constructed in the left panel of Fig.\ref{contour-os}. Before the deformation, we give a lemma to insure this definition about $G(\lambda; X, T)$ is reasonable.
%Before the deformation to the jump matrix in Eq.\eqref{eq:n-jump}, we first give a basic composition to $\mathbf{Q}^{-1}$.
%\begin{lemma}
%The matrix $\mathbf{Q}^{-1}$ has four different decompositions,
%\begin{equation}
%\begin{aligned}
%\mathbf{Q}^{-1}&=\begin{bmatrix}
%\frac{\sqrt{2}}{2}&-\frac{\sqrt{2}}{2}\ii&0\\
%-\frac{\sqrt{2}}{2}\ii&\frac{\sqrt{2}}{2}&0\\0&0&1
%\end{bmatrix}=\begin{bmatrix}
%1&\ii&0\\0&1&0\\0&0&1
%\end{bmatrix}\begin{bmatrix}0&-\sqrt{2}\ii&0\\
%-\frac{\sqrt{2}}{2}\ii&0&0\\
%0&0&1
%\end{bmatrix}
%\begin{bmatrix}
%1&\ii&0\\0&1&0\\
%0&0&1
%\end{bmatrix},
%\\\mathbf{Q}^{-1}&=\begin{bmatrix}
%1&0&0\\\ii&1&0\\0&0&1
%\end{bmatrix}\begin{bmatrix}0&-\frac{\sqrt{2}}{2}\ii&0\\
%-\sqrt{2}\ii&0&0\\
%0&0&1
%\end{bmatrix}
%\begin{bmatrix}
%1&0&0\\\ii&1&0\\
%0&0&1
%\end{bmatrix},
%\\
%\mathbf{Q}^{-1}&=\begin{bmatrix}
%1&0&0\\
%-\ii&1&0\\
%0&0&1
%\end{bmatrix}\begin{bmatrix}
%\frac{\sqrt{2}}{2}&0&0\\
%0&\sqrt{2}&0\\
%0&0&1
%\end{bmatrix}\begin{bmatrix}1&-\ii&0\\0&1&0\\0&0&1
%\end{bmatrix},\\
%\mathbf{Q}^{-1}&=\begin{bmatrix}1&-\ii&0\\
%0&1&0\\
%0&0&1
%\end{bmatrix}\begin{bmatrix}\sqrt{2}&0&0\\
%0&\frac{\sqrt{2}}{2}&0\\
%0&0&1
%\end{bmatrix}\begin{bmatrix}1&0&0\\
%-\ii&1&0\\
%0&0&1
%\end{bmatrix}.
%\end{aligned}
%\end{equation}
%which is useful to the deformation of the jump matrix in Eq.\eqref{eq:n-jump}.
%\end{lemma}
\begin{lemma}
The two points $\lambda=\lambda_1$ and $\lambda=\lambda_1^*$ in the denominator of Eq.\eqref{eq:dGdvarphi} are in a closed contour $\Re(G(\lambda; X, T)-\varphi(\lambda; X, T))=0$  connecting the point $a$ and $b$ in the upper plane or $a^*$ and $b^*$ in lower half plane respectively, as shown in Fig.\ref{contour-os}.
\end{lemma}
\begin{proof}
From the definition of $G(\lambda; X, T)$, we can see $G(\lambda; X, T)$ is analytic when $\lambda\in\mathbb{C}\setminus\left(\Sigma_{\rm u}\cap\Sigma_{\rm d}\cap\Sigma_{\rm mid}\right)$, and $\varphi(\lambda; X, T)$ has the logarithm singularity at $\lambda=\lambda_1$ and $\lambda=\lambda_1^*$. %$\varphi(\lambda; X, T)$ is analytic when $\lambda\in\mathbb{C}\setminus\Sigma_{\lambda_1,\lambda_1^*}$, where $\Sigma_{\lambda_1,\lambda_1^*}$ is the branch cut of $\log\left(\frac{\lambda-\lambda_1}{\lambda-\lambda^*}\right)$.What we need is $G(\lambda; X, T)-\varphi(\lambda; X, T)$ is analytic when $\lambda$ is located outside the contour. Thus if $\lambda=\lambda_1$ is in the outside of the contour, from the expression of Eq.\eqref{eq:dGdvarphi}, it is clear that $\lambda=\lambda_1$ and $\lambda=\lambda_1^*$ is not analytic, which is against with the definition of $G(\lambda; X, T)$.
We know that the real part of $G(\lambda;X,T)-\varphi(\lambda;X,T)$ is harmonic function.
So these two logarithm points $\lambda=\lambda_1$ and $\lambda=\lambda_1^*$ should be in the internal of closed contour, otherwise the harmonic function is zero by the extreme principle, which is impossible.
\end{proof}

Now, we begin to deform the contour in RHP \ref{rhp:far-field}. Define
\begin{equation}\label{eq:O-far-field}
\mathbf{O}_{f,o}(\lambda; X, T):=\left\{\begin{aligned}&\widetilde{\mathbf{N}}_{f}(\lambda; X, T)\ee^{-N\varphi(\lambda; X, T)\pmb{\sigma}_3}\widehat{\mathbf{Q}}_{d}^{-1}\ee^{N\varphi(\lambda; X, T)\pmb{\sigma}_3},\qquad &\lambda\in D_0\cap \left(D_{\rm u}\cup D_{\rm d}\right)^{c},\\
&\widetilde{\mathbf{N}}_{f}(\lambda; X, T)\qquad &{\rm otherwise}.\end{aligned}\right.
\end{equation}
where $D_{\rm u}=D_{\rm u_1}^{\rm i}\cup D_{\rm u_{2}}^{\rm i}\cup K_{\rm u}^{\rm i}$, $D_{\rm d}=D_{\rm d_1}^{\rm i}\cup D_{\rm d_{2}}^{\rm i}\cup K_{\rm d}^{\rm i}$.
Then the jump about $\mathbf{O}_{f,o}(\lambda; X, T)$ lies onto $\partial D_{\rm u}$ and $\partial D_{\rm d}$. With the traditional decomposition to the jump matrix, the right hand of $\partial D_{\rm u}$ and $\partial D_{\rm d}$ is good, but the left one is bad due to the sign of both sides of $\sum_{\rm u}$ and $\sum_{\rm d}$ is same. To deal with this problem, we should redefine a new matrix $\mathbf{P}_{f,o}(\lambda; X, T)$ together with the $G(\lambda; X, T)$-function,
\begin{equation}\label{P-matrix}
\mathbf{P}_{f,o}(\lambda; X, T):=\mathbf{O}_{f,o}(\lambda; X, T){\rm diag}\left(\ee^{-NG(\lambda; X, T)}, \ee^{NG(\lambda; X, T)}, 1 \right).
\end{equation}
Then the jump condition about the $\mathbf{P}_{f,o}(\lambda; X, T)$ is
\begin{equation}\label{eq:P-o-jump}
\begin{aligned}
\mathbf{P}_{f,o,+}&=\mathbf{P}_{f,o,-}\begin{bmatrix}\frac{\sqrt{2}}{2}\ee^{N\left(G_-(\lambda; X, T)-G_+(\lambda; X, T)\right)}&\frac{\sqrt{2}}{2}\ee^{N\left(-2\varphi(\lambda; X, T)+G_{+}(\lambda; X, T)+G_{-}(\lambda; X, T)\right)}&0\\
-\frac{\sqrt{2}}{2}\ee^{N\left(2\varphi(\lambda; X, T)-G_{+}(\lambda; X, T)-G_{-}(\lambda; X, T)\right)}&\frac{\sqrt{2}}{2}\ee^{N\left(G_+(\lambda; X, T)-G_-(\lambda; X, T)\right)}&0\\0&0&1
\end{bmatrix},\\&\lambda\in\Sigma_{\rm u}\cup\Sigma_{\rm d}\\
\mathbf{P}_{f,o,+}&=\mathbf{P}_{f,o,-}\begin{bmatrix}\frac{\sqrt{2}}{2}&\frac{\sqrt{2}}{2}\ee^{N\left(-2\varphi(\lambda; X, T)+2G_{-}(\lambda; X, T)\right)}&0\\
-\frac{\sqrt{2}}{2}\ee^{N\left(2\varphi(\lambda; X, T)-2G_{-}(\lambda; X, T)\right)}&\frac{\sqrt{2}}{2}&0\\0&0&1
\end{bmatrix},\lambda\in\Gamma_{\rm u}\cup\Gamma_{\rm d}\\
\mathbf{P}_{f,o,+}&=\mathbf{P}_{f,o,-}{\rm diag}\left(\ee^{N\left(G_-(\lambda; X, T)-G_{+}(\lambda; X, T)\right)}, \ee^{N\left(G_+(\lambda; X, T)-G_{-}(\lambda; X, T)\right)}, 1\right),\lambda\in\Sigma_{\rm mid}
\end{aligned}
\end{equation}
With the sign in Fig.\ref{contour-os}, we plan to study the asymptotic behavior with the nonlinear steepest descent method, before analyzing it, we give the decomposition about $\widehat{\mathbf{Q}}_{d}^{-1}$ in the following lemma:
\begin{lemma}
\begin{equation}
\begin{split}
\widehat{\mathbf{Q}}_d^{-1}&=\begin{bmatrix}\frac{\sqrt{2}}{2}&0&0\\
0&\sqrt{2}&0\\
0&0&1
\end{bmatrix}\begin{bmatrix}1&0&0\\
-\frac{1}{2}&1&0\\
0&0&1
\end{bmatrix}
\begin{bmatrix}
1&1&0\\0&1&0\\
0&0&1
\end{bmatrix}:=\mathbf{Q}_{L}^{[1]}\mathbf{Q}_{C}^{[1]}\mathbf{Q}_{R}^{[1]}\\
\widehat{\mathbf{Q}}_d^{-1}&=\begin{bmatrix}\sqrt{2}&0&0\\
0&\frac{\sqrt{2}}{2}&0\\
0&0&1
\end{bmatrix}\begin{bmatrix}1&\frac{1}{2}&0\\
0&1&0\\
0&0&1
\end{bmatrix}
\begin{bmatrix}
1&0&0\\-1&1&0\\
0&0&1
\end{bmatrix}:=\mathbf{Q}_{L}^{[2]}\mathbf{Q}_{C}^{[2]}\mathbf{Q}_{R}^{[2]}\\
\widehat{\mathbf{Q}}_{d}^{-1}&=\begin{bmatrix}1&-1&0\\
0&1&0\\
0&0&1
\end{bmatrix}\begin{bmatrix}0&\sqrt{2}&0\\
-\frac{\sqrt{2}}{2}&0&0\\
0&0&1
\end{bmatrix}\begin{bmatrix}
1&-1&0\\
0&1&0\\
0&0&1
\end{bmatrix}:=\mathbf{Q}_{L}^{[3]}\mathbf{Q}_{C}^{[3]}\mathbf{Q}_{R}^{[3]},\\
\widehat{\mathbf{Q}}_{d}^{-1}&=\begin{bmatrix}1&0&0\\
1&1&0\\
0&0&1
\end{bmatrix}\begin{bmatrix}0&\frac{\sqrt{2}}{2}&0\\
-\sqrt{2}&0&0\\
0&0&1
\end{bmatrix}\begin{bmatrix}
1&0&0\\
1&1&0\\
0&0&1
\end{bmatrix}:=\mathbf{Q}_{L}^{[4]}\mathbf{Q}_{C}^{[4]}\mathbf{Q}_{R}^{[4]},\\
\widehat{\mathbf{Q}}_{d}^{-1}&=\begin{bmatrix}1&1&0\\
0&1&0\\
0&0&1
\end{bmatrix}\begin{bmatrix}\sqrt{2}&0&0\\
0&\frac{\sqrt{2}}{2}&0\\
0&0&1
\end{bmatrix}\begin{bmatrix}
1&0&0\\
-1&1&0\\
0&0&1
\end{bmatrix}:=\mathbf{Q}_{L}^{[5]}\mathbf{Q}_{C}^{[5]}\mathbf{Q}_{R}^{[5]},\\
\widehat{\mathbf{Q}}_{d}^{-1}&=\begin{bmatrix}1&0&0\\
-1&1&0\\
0&0&1
\end{bmatrix}\begin{bmatrix}\frac{\sqrt{2}}{2}&0&0\\
0&\sqrt{2}&0\\
0&0&1
\end{bmatrix}\begin{bmatrix}
1&1&0\\
0&1&0\\
0&0&1
\end{bmatrix}:=\mathbf{Q}_{L}^{[6]}\mathbf{Q}_{C}^{[6]}\mathbf{Q}_{R}^{[6]},
\end{split}
\end{equation}
\end{lemma}
By the above lemma, define a new matrix $\mathbf{Q}_{f,o}$ by
\begin{equation}\label{eq:Q-definition}
\begin{aligned}
\mathbf{Q}_{f,o}:&=\mathbf{P}_{f,o}\begin{bmatrix}1&-\ee^{NH(\lambda; X, T)}&0\\
0&1&0\\0&0&1
\end{bmatrix},\,\,\lambda\in D^{\rm i}_{\rm u_1},\,\,
\mathbf{Q}_{f,o}:=\mathbf{P}_{f,o}\begin{bmatrix}1&\ee^{NH(\lambda; X, T)}&0\\
0&1&0\\0&0&1
\end{bmatrix},\,\,\lambda\in D^{\rm o}_{\rm u},\\
\mathbf{Q}_{f,o}:&=\mathbf{P}_{f,o}\begin{bmatrix}1&0&0\\
\ee^{-NH(\lambda; X, T)}&1&0\\0&0&1
\end{bmatrix},\,\,\lambda\in D^{\rm i}_{\rm d_1},\,\,
\mathbf{Q}_{f,o}:=\mathbf{P}_{f,o}\begin{bmatrix}1&0&0\\
-\ee^{-NH(\lambda; X, T)}&1&0\\0&0&1
\end{bmatrix},\,\,\lambda\in D^{\rm o}_{\rm d},\\
\mathbf{Q}_{f,o}:&=\mathbf{P}_{f,o}\begin{bmatrix}1&\ee^{NH(\lambda; X, T)}&0\\
0&1&0\\0&0&1
\end{bmatrix},\,\,\lambda\in K^{\rm i}_{\rm u},\,\,
\mathbf{Q}_{f,o}:=\mathbf{P}_{f,o}\begin{bmatrix}1&0&0\\
\ee^{-NH(\lambda; X, T)}&1&0\\0&0&1
\end{bmatrix},\,\,\lambda\in K^{\rm o}_{\rm u},\\
\mathbf{Q}_{f,o}:&=\mathbf{P}_{f,o}\begin{bmatrix}1&0&0\\
-\ee^{-NH(\lambda; X, T)}&1&0\\0&0&1
\end{bmatrix},\,\,\lambda\in K^{\rm i}_{\rm d},\,\,
\mathbf{Q}_{f,o}:=\mathbf{P}_{f,o}\begin{bmatrix}1&-\ee^{NH(\lambda; X, T)}&0\\
0&1&0\\0&0&1
\end{bmatrix},\,\,\lambda\in K^{\rm o}_{\rm d},\\
\mathbf{Q}_{f,o}:&=\mathbf{P}_{f,o},\quad{\rm otherwise}
\end{aligned}
\end{equation}
where $H(\lambda; X, T)=2G(\lambda; X, T)-2\varphi(\lambda; X, T)$.
Then the jump matrices about the $\mathbf{Q}_{f,o}(\lambda; X, T)$ change into
\begin{equation}\label{eq:jump-Qf}
\mathbf{Q}_{f,o,+}{=}\mathbf{Q}_{f,o,-}\begin{bmatrix}1&-\ee^{NH(\lambda; X, T)}&0\\
0&1&0\\0&0&1
\end{bmatrix},\,\,\lambda\in\Sigma^{\rm o}_{\rm u},\,\,
\mathbf{Q}_{f,o,+}{=}\mathbf{Q}_{f,o,-}\begin{bmatrix}0&\sqrt{2}\ee^{N\Omega_{o}}&0\\
-\frac{\sqrt{2}}{2}\ee^{-N\Omega_{o}}&0&0\\
0&0&1
\end{bmatrix},\,\lambda\in\Sigma_{\rm u},
\end{equation}
and
\begin{equation}\label{eq:jump-Qf1}
\begin{aligned}
\mathbf{Q}_{f,o,+}&{=}\mathbf{Q}_{f,o,-}\begin{bmatrix}1&-\ee^{NH(\lambda; X, T)}&0\\
0&1&0\\0&0&1
\end{bmatrix},\,\,\lambda\in\Sigma_{\rm u}^{\rm i},\quad
\mathbf{Q}_{f,o,+}{=}\mathbf{Q}_{f,o,-}\begin{bmatrix}1&0&0\\
\ee^{-NH(\lambda; X, T)}&1&0\\0&0&1
\end{bmatrix},\,\,\lambda\in\Sigma_{\rm d}^{\rm o},\\
\mathbf{Q}_{f,o,+}&{=}\mathbf{Q}_{f,o,-}\begin{bmatrix}0&\frac{\sqrt{2}}{2}\ee^{N\Omega_o}&0\\
-\sqrt{2}\ee^{-N\Omega_o}&0&0\\
0&0&1
\end{bmatrix},\,\,\lambda\in\Sigma_{\rm d},\quad
\mathbf{Q}_{f,o,+}{=}\mathbf{Q}_{f,o,-}\begin{bmatrix}1&0&0\\
\ee^{-NH(\lambda; X, T)}&1&0\\0&0&1
\end{bmatrix},\,\,\lambda\in\Sigma_{\rm d}^{\rm i},\\
\mathbf{Q}_{f,o,+}&{=}\mathbf{Q}_{f,o,-}\begin{bmatrix}1&0&0\\
-\ee^{-NH(\lambda; X, T)}&1&0\\0&0&1
\end{bmatrix},\,\,\lambda\in \Gamma^{\rm o}_{\rm u},\quad
\mathbf{Q}_{f,o,+}{=}\mathbf{Q}_{f,o,-}{\rm diag}\left(\sqrt{2}, \frac{\sqrt{2}}{2}, 1\right),\quad\lambda\in \Gamma_{\rm u},\\
\mathbf{Q}_{f,o,+}&{=}\mathbf{Q}_{f,o,-}\begin{bmatrix}1&\ee^{NH(\lambda; X, T)}&0\\
0&1&0\\0&0&1
\end{bmatrix},\,\,\lambda\in \Gamma_{\rm u}^{\rm i},\quad
\mathbf{Q}_{f,o,+}{=}\mathbf{Q}_{f,o,-}\begin{bmatrix}1&\ee^{NH(\lambda; X, T)}&0\\
0&1&0\\0&0&1
\end{bmatrix},\,\,\lambda\in \Gamma_{\rm d}^{\rm o},\\
\mathbf{Q}_{f,o,+}&=\mathbf{Q}_{f,o,-}{\rm diag}\left(\frac{\sqrt{2}}{2}, \sqrt{2}, 1\right),\quad\lambda\in \Gamma_{\rm d},\quad
\mathbf{Q}_{f,o,+}{=}\mathbf{Q}_{f,o,-}\begin{bmatrix}1&0&0\\
-\ee^{-NH(\lambda; X, T)}&1&0\\0&0&1
\end{bmatrix},\,\,\lambda\in \Gamma_{\rm d}^{\rm i},\\
\mathbf{Q}_{f,o,+}&{=}\mathbf{Q}_{f,o,-}{\rm diag}\left(\ee^{-Nd}, \ee^{Nd}, 1\right),\,\,\lambda\in\Sigma_{\rm mid}.
\end{aligned}
\end{equation}
Observing the jump in Eq.\eqref{eq:jump-Qf} and \eqref{eq:jump-Qf1}, only when $\lambda\in\Sigma_{\rm u}, \Sigma_{\rm d}, \Gamma_{\rm u}, \Gamma_{\rm d}$ and $\Sigma_{\rm mid}$, the jump matrices are not approaching to identity as $N\to\infty$. Thus we can define the outer model problem:
\begin{rhp}\label{RHP-out-o}
(The outer problem in the oscillatory region) Find a $3\times 3$ matrix $\mathbf{Q}^{\rm out}_{f,o}(\lambda; X, T)$ satisfying the following condition
\begin{itemize}
\item
\textbf{Analyticity}: The function $\mathbf{Q}^{\rm out}_{f,o}$ is analytic in $\lambda\in\mathbb{C}\setminus\left(\Sigma_{u}\cup\Sigma_{d}\cup\Gamma_{u}\cup\Gamma_{d}\cup\Sigma_{\rm mid}\right)$, and it can achieve the continuous boundary condition values from the left and right side of every arc.
\item\textbf{Jump Condition}: The boundary values taken by $\mathbf{Q}_{f,o}^{\rm out}(\lambda; X, T)$ satisfy the jump conditions $\mathbf{Q}^{\rm out}_{f,o,+}=\mathbf{Q}^{\rm out}_{f,o,-}\mathbf{V}_{f,o}^{\rm out}$, where
    \begin{equation}
    \mathbf{V}_{f,o}^{\rm out}:=\left\{\begin{split}&\begin{bmatrix}0&\sqrt{2}\ee^{N\Omega_o}&0\\
-\frac{\sqrt{2}}{2}\ee^{-N\Omega_o}&0&0\\
0&0&1
\end{bmatrix},\qquad&\lambda\in\Sigma_{\rm u},\\
&\begin{bmatrix}0&\frac{\sqrt{2}}{2}\ee^{N\Omega_o}&0\\
-\sqrt{2}\ee^{-N\Omega_o}&0&0\\
0&0&1
\end{bmatrix},\qquad&\lambda\in\Sigma_{\rm d},\\
&{\rm diag}\left(\sqrt{2}, \frac{\sqrt{2}}{2}, 1\right),\qquad&\lambda\in \Gamma_{\rm u},\\
&{\rm diag}\left(\frac{\sqrt{2}}{2}, \sqrt{2}, 1\right),\qquad&\lambda\in \Gamma_{\rm d},\\
&{\rm diag}\left(\ee^{-Nd}, \ee^{Nd}, 1\right),\qquad&\lambda\in\Sigma_{\rm mid}.
\end{split}\right.
    \end{equation}
\item
\textbf{Normalization}: As $\lambda\to\infty$, $\mathbf{Q}_{f,o}^{\rm out}(\lambda; X, T)\to\mathbb{I}$.
\end{itemize}
\end{rhp}
To solve the RHP \ref{RHP-out-o}, we define a new function $F(\lambda)$ by
\begin{equation}
\begin{aligned}
F(\lambda):&=\frac{\mathscr{R}(\lambda; X, T)}{2\pi\ii}\Bigg(\int_{\Sigma_{\rm u}}\frac{-N\Omega_o-\log\left(\sqrt{2}\right)}{\mathscr{R}_{+}(s; X, T)\left(s-\lambda\right)}ds+\int_{\Sigma_{\rm d}}\frac{-N\Omega_o-\log\left(\frac{\sqrt{2}}{2}\right)}{\mathscr{R}_{+}(s; X, T)(s-\lambda)}ds\\+&\int_{\Gamma_{\rm u}}\frac{\log\left(\sqrt{2}\right)}{\mathscr{R}(s; X, T)(s-\lambda)}ds+\int_{\Gamma_{\rm d}}\frac{\log\left(\frac{\sqrt{2}}{2}\right)}{\mathscr{R}(s; X, T)(s-\lambda)}ds+\int_{\Sigma_{\rm mid}}\frac{-Nd}{\mathscr{R}(s; X, T)(s-\lambda)}ds
\Bigg)
\end{aligned}
\end{equation}
which satisfy the jump condition
\begin{equation}
\begin{aligned}
F_{+}(\lambda)+F_{-}(\lambda)&=-N\Omega_o-\log\left(\sqrt{2}\right),\qquad &\lambda\in\Sigma_{\rm u},\\
F_{+}(\lambda)+F_{-}(\lambda)&=-N\Omega_o-\log\left(\frac{\sqrt{2}}{2}\right),\qquad &\lambda\in\Sigma_{\rm d},\\
F_{+}(\lambda)-F_{-}(\lambda)&=\log\left(\sqrt{2}\right),\qquad &\lambda\in \Gamma_{\rm u},\\
F_{+}(\lambda)-F_{-}(\lambda)&=\log\left(\frac{\sqrt{2}}{2}\right),\qquad &\lambda\in \Gamma_{\rm d},\\
F_{+}(\lambda)-F_{-}(\lambda)&=-Nd,\qquad &\lambda\in\Sigma_{\rm mid}.
\end{aligned}
\end{equation}
It should be noted that the contour $\Gamma_{\rm u}$ and $\Gamma_{\rm d}$ are located on the right of branch cut $\Sigma_{\rm u}, \Sigma_{\rm d}$, so $\mathscr{R}(s; X, T)$ in the denominator should be chosen the negative sign.

Since the potential $\mathbf{q}(X,T)$ is related to the RHP when $\lambda\to\infty$, then we expand $F(\lambda)$ at $\lambda\to\infty$,
\begin{equation}
F(\lambda)=F_1\lambda+F_0+\mathcal{O}(\lambda^{-1}),
\end{equation}
where
\begin{equation}
\begin{aligned}
\label{eq:F1}
F_1:=-\frac{1}{2\pi\ii}\Bigg(&\int_{\Sigma_{\rm u}}\frac{-N\Omega_o-\log\left(\sqrt{2}\right)}{\mathscr{R}_{+}(s; X, T)}ds+\int_{\Sigma_{\rm d}}\frac{-N\Omega_o-\log\left(\frac{\sqrt{2}}{2}\right)}{\mathscr{R}_{+}(s; X, T)}ds\\+&\int_{\Gamma_{\rm u}}\frac{\log\left(\sqrt{2}\right)}{\mathscr{R}(s; X, T)}ds+\int_{\Gamma_{\rm d}}\frac{\log\left(\frac{\sqrt{2}}{2}\right)}{\mathscr{R}(s; X, T)}ds+\int_{\Sigma_{\rm mid}}\frac{-Nd}{\mathscr{R}(s; X, T)}ds
\Bigg),
\end{aligned}
\end{equation}
and
\begin{equation}
\begin{aligned}
\label{eq:F0}
F_0:=-\frac{l_1}{2}F_1-\frac{1}{2\pi\ii}\Bigg(&\int_{\Sigma_{\rm u}}\frac{-N\Omega_o-\log\left(\sqrt{2}\right)}{\mathscr{R}_{+}(s; X, T)}sds+\int_{\Sigma_{\rm d}}\frac{-N\Omega_o-\log\left(\frac{\sqrt{2}}{2}\right)}{\mathscr{R}_{+}(s; X, T)}sds\\+&\int_{\Gamma_{\rm u}}\frac{\log\left(\sqrt{2}\right)}{\mathscr{R}(s; X, T)}sds+\int_{\Gamma_{\rm d}}\frac{\log\left(\frac{\sqrt{2}}{2}\right)}{\mathscr{R}(s; X, T)}sds+\int_{\Sigma_{\rm mid}}\frac{-Nd}{\mathscr{R}(s; X, T)}sds
\Bigg).
\end{aligned}
\end{equation}
Based on the definition of $F(\lambda)$, we can redefine a new matrix
\begin{equation}\label{eq:Qout-o}
\mathbf{S}_{f,o}(\lambda; X, T):={\rm diag}\left(\ee^{F_0}, \ee^{-F_0}, 1\right)\mathbf{Q}_{f,o}^{\rm out}(\lambda; X, T){\rm diag}\left(\ee^{-F(\lambda)}, \ee^{F(\lambda)}, 1\right),
\end{equation}
whose jump matrix converts into a constant:
\begin{equation}\label{eq:jumpS}
\mathbf{S}_{f,o,+}(\lambda; X, T)=\mathbf{S}_{f,o,-}(\lambda; X, T)\begin{bmatrix}0&1&0\\-1&0&0\\
0&0&1
\end{bmatrix},\qquad\lambda\in\Sigma_{\rm u}\cup\Sigma_{\rm d},
\end{equation}
and when $\lambda\to\infty$, we have
\begin{equation}\label{eq:normal-S-o}
\mathbf{S}_{f, o}(\lambda; X, T){\rm diag}\left(\ee^{F_1\lambda}, \ee^{-F_1\lambda}, 1\right)=\mathbb{I}+\mathcal{O}(\lambda^{-1}).
\end{equation}
Obviously, $\mathscr{R}(\lambda)$ gives a genus one Riemann-surface $\Sigma$ with two sheets $\Sigma_1$, $ \pmb{\sigma}_2$ and a basis $\{\alpha, \beta\}$ cycles, which is shown in Fig. \ref{cycles}. The contour $\alpha$ is closed, anticlockwise in the first sheet. And the $\beta$ cycle starts from the right side of branch cut $[a,b]$ and arrives at the branch cut $[a^*, b^*]$, then enters into the second sheet and return to the right of branch cut $[a, b]$.
\begin{figure}[!h]
\centering
\includegraphics[width=0.45\textwidth]{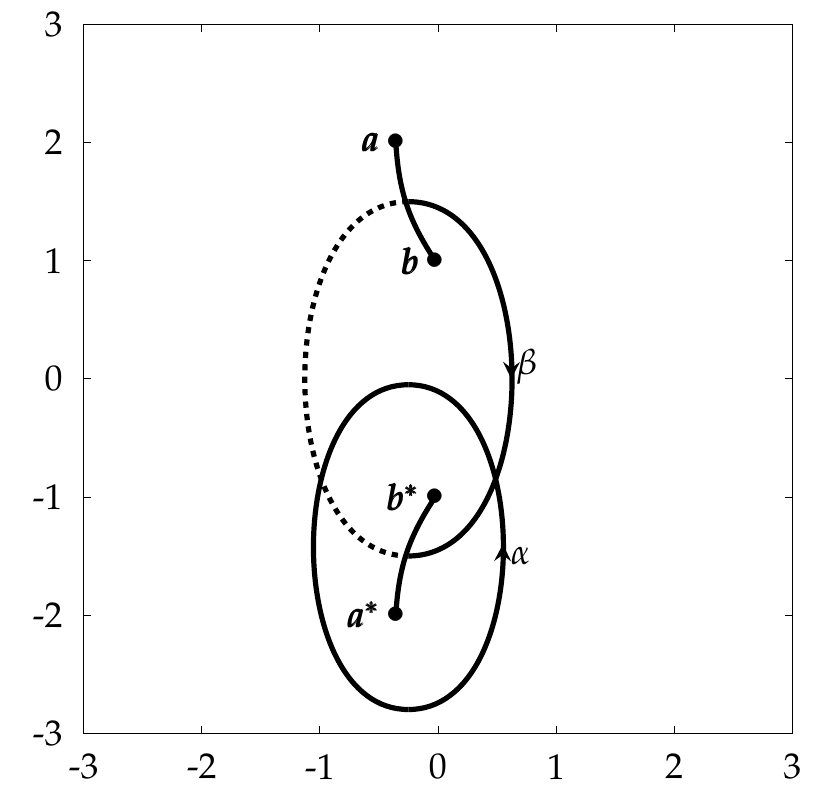}
\caption{The homology cycles $\alpha$ and $\beta$ for the branch cut $\mathscr{R}(\lambda)$, where solid line is the first sheet and the dotted line is the second sheet.}
\label{cycles}
\end{figure}

In order to solve Eq.\eqref{eq:jumpS} with the boundary condition Eq.\eqref{eq:normal-S-o}, we introduce the Abel map $A(\lambda)$ and the quantity $B$ as
\begin{equation}\label{eq:AandB}
A(\lambda):=\frac{2\pi\ii}{\oint_{\alpha}\frac{ds}{\mathscr{R}(s; X, T)}}\int_{a_o^*}^{\lambda}\frac{ds}{\mathscr{R}(s; X, T)},\qquad B:=\frac{2\pi\ii}{\oint_{\alpha}\frac{ds}{\mathscr{R}(s; X, T)}}\oint_{\beta}\frac{ds}{\mathscr{R}(s; X, T)},
\end{equation}
both of which have the normalization
\begin{equation}
A_+(b^*_o)=-\pi\ii, \qquad B=2\pi\ii\tau,
\end{equation}
where $\tau$ can be calculated as
\begin{equation}\label{eq:tau}
\tau=\frac{\ii K\left(1-m\right)}{K\left(m\right)},
\end{equation}
where $K(m)$ is the first kind of elliptic integral
\begin{equation}
K(m)=\int_{0}^{\frac{\pi}{2}}\frac{ds}{\sqrt{1-m\sin^2(s)}}.
\end{equation}
The Able mape $A(\lambda)$ has the following properties:
\begin{equation}\label{eq:Alambda}
\begin{aligned}
A_{+}(\lambda)+A_{-}(\lambda)&=-B \mod\Lambda_{o},\qquad \lambda\in\Sigma_{\rm u}\\
A_{+}(\lambda)-A_{-}(\lambda)&=-2\pi\ii\, \mod\Lambda_{o},\qquad \lambda\in\Sigma_{\rm mid},\\
A_{+}(\lambda)+A_{-}(\lambda)&=0\, \mod\Lambda_{o},\qquad\lambda\in\Sigma_{\rm {d}},
\end{aligned}
\end{equation}
where $\Lambda_{o}:=2\pi\ii j+B k,$ $ j,k\in\mathbb{Z}$.

To deal with the asymptotic behavior as $\lambda\to+\infty$, we define a new function $\Omega_1(s)$ as
\begin{equation}
d\Omega_1=\frac{s^2-\frac{1}{2}l_1s-l_0}{\mathscr{R}(s; X, T)}ds,
\end{equation}
where $l_0=\frac{\oint_{\alpha}\frac{s^2-\frac{1}{2}l_1s}{\mathscr{R}(s; X, T)}ds}{\oint_{\alpha}\frac{1}{\mathscr{R}(s; X, T)}ds}$, which can be normalized $\Omega_1(X, T)$ to
\begin{equation}
\oint_{\alpha}d\Omega_1=0.
\end{equation}

Following the theory in \cite{Calini-JNS-2005}, when $s\to\infty$, we set
\begin{equation}\label{eq:Omega1}
E=2\lim\limits_{s\to\infty}\left(s-\int_{a_o^*}^{s}\Omega_1\right),
\end{equation}
and define a new variable $V$ as
\begin{equation}\label{eq:V}
V:=\oint_{\beta}\Omega_1,
\end{equation}
which is useful to construct the Theta function solution. Spontaneously, we give some properties about the Theta function.
\begin{prop}\label{prop:theta}
\begin{equation}
\begin{aligned}
\Theta(\lambda)&\equiv\Theta(\lambda; B):=\sum\limits_{k\in\mathbb{Z}}\ee^{k\lambda+\frac{1}{2}Bk^2}=\theta_{3}\left(\frac{\lambda}{2\ii}, \ee^{\frac{B}{2}}\right),\\
\Theta(-\lambda)&=\Theta(\lambda),\qquad \Theta(\lambda+2\pi\ii)=\Theta(\lambda),\qquad \Theta(\lambda+B)=\ee^{-\frac{1}{2}B}\ee^{-\lambda}\Theta(\lambda)
\end{aligned}
\end{equation}
where $\theta_3(\lambda; \wp)$ is the third Jacobi theta function, defined as
\begin{equation}
\theta_3(\lambda; \wp)=\sum\limits_{k\in\mathbb{Z}}\ee^{2\ii k\lambda}\wp^{k^2}.
\end{equation}
The Jacobi Theta functions have the following shift formulas:
\begin{equation}
\begin{aligned}
\theta_{1}\left(u+\frac{1}{2}\right)&=\theta_2(u), \qquad \theta_{1}\left(u+\frac{\tau}{2}\right)=\ii\ee^{-\pi\ii\left(u+\tau/4\right)}\theta_4(u),\\
\theta_{2}\left(u+\frac{1}{2}\right)&=-\theta_1(u), \qquad \theta_{2}\left(u+\frac{\tau}{2}\right)=\ee^{-\pi\ii\left(u+\tau/4\right)}\theta_3(u),\\
\theta_{3}\left(u+\frac{1}{2}\right)&=\theta_4(u), \qquad \theta_{3}\left(u+\frac{\tau}{2}\right)=\ee^{-\pi\ii\left(u+\tau/4\right)}\theta_2(u),\\
\theta_{4}\left(u+\frac{1}{2}\right)&=\theta_3(u), \qquad \theta_{4}\left(u+\frac{\tau}{2}\right)=\ii\ee^{-\pi\ii\left(u+\tau/4\right)}\theta_1(u).
\end{aligned}
\end{equation}
The zeros of these theta functions are:
\begin{equation}
\begin{aligned}
\theta_1(u)&=0:\qquad u=n+m\tau,\\
\theta_2(u)&=0:\qquad u=n+\frac{1}{2}+m\tau,\\
\theta_3(u)&=0:\qquad u=n+\frac{1}{2}+\left(m+\frac{1}{2}\right)\tau,\\
\theta_4(u)&=0:\qquad u=n+\left(m+\frac{1}{2}\right)\tau,
\end{aligned}
\end{equation}
where $m, n\in\mathbb{Z}$.
\end{prop}

Next we redefine a new matrix $\mathscr{T}(\lambda; X, T)$ as
\begin{equation}
\mathscr{T}(\lambda; X, T)=\begin{bmatrix}\frac{\Theta\left(A(\lambda)+A(Q)+\ii\pi+\frac{B}{2}-F_1V\right)}{\Theta\left(A(\lambda)+A(Q)+\ii\pi+\frac{B}{2}\right)}\ee^{-F_1\int_{a^*}^{\lambda}\Omega_1ds}&\frac{\Theta\left(A(\lambda)-A(Q)-\ii\pi-\frac{B}{2}+F_1V\right)}{\Theta\left(A(\lambda)-A(Q)-\ii\pi-\frac{B}{2}\right)}\ee^{F_1\int_{a^*}^{\lambda}\Omega_1ds}&0\\
\frac{\Theta\left(A(\lambda)-A(Q)-\ii\pi-\frac{B}{2}-F_1V\right)}{\Theta\left(A(\lambda)-A(Q)-\ii\pi-\frac{B}{2}\right)}\ee^{-F_1\int_{a^*}^{\lambda}\Omega_1ds}&\frac{\Theta\left(A(\lambda)+A(Q)+\ii\pi+\frac{B}{2}+F_1V\right)}{\Theta\left(A(\lambda)+A(Q)+\ii\pi+\frac{B}{2}\right)}\ee^{F_1\int_{a^*}^{\lambda}\Omega_1ds}&0\\
0&0&1
\end{bmatrix},
\end{equation}
where $Q$ will be determined in the following. With the properties of $\Theta$ function and $A(\lambda; X, T)$, the function $\mathscr{T}(\lambda; X, T)$ satisfies
\begin{equation}
\mathscr{T}_{+}(\lambda; X, T)=\mathscr{T}_{-}(\lambda; X, T)\begin{bmatrix}0&1&0\\
1&0&0\\
0&0&1
\end{bmatrix}.
\end{equation}
Moreover, $\mathbf{S}_{f, o}(\lambda; X, T)$ can be given with the aid of the elements of $\mathscr{T}(\lambda; X, T)$. Suppose
\begin{equation}
\mathbf{S}_{f,o}(\lambda; X, T)={\rm diag}\left(s_{11}(\lambda; X, T), s_{22}(\lambda; X, T), 1\right)\begin{bmatrix}\frac{p(\lambda)+p^{-1}(\lambda)}{2}\left[\mathscr{T}(\lambda)\right]_{11}&\frac{p(\lambda)-p^{-1}(\lambda)}{2}\ii
\left[\mathscr{T}(\lambda)\right]_{12}&0\\
\frac{p^{-1}(\lambda)-p(\lambda)}{2}\ii
\left[\mathscr{T}(\lambda)\right]_{21}&\frac{p(\lambda)+p^{-1}(\lambda)}{2}
\left[\mathscr{T}(\lambda)\right]_{22}&0\\0&0&1
\end{bmatrix},
\end{equation}
where $p(\lambda):=\left(\frac{\left(\lambda-b_o\right)\left(\lambda-a_o^*\right)}{\left(\lambda-a_o\right)\left(\lambda-b_o^*\right)}\right)^{1/4}$. For $\lambda\in\Sigma_{\rm u}\cup\Sigma_{\rm d}$, we have $p_{+}(\lambda)=-\ii p_{-}(\lambda)$. With a simple calculation, we know the function $p^{-1}(\lambda)-p(\lambda)$ has a simple zero at
$$\lambda_0=\frac{a_ob_o^*-a_o^*b_o}{a_o-a_o^*-b_o+b^*_o},$$ and $[\mathscr{T}(\lambda)]_{12}$ has the singularity at $\lambda=Q$.  Then $Q$ can be unique determined by the zero of $p(\lambda)-p^{-1}(\lambda)$, i.e. $Q=\lambda_0$. Furthermore, the unknown parameters $s_{11}(\lambda; X, T), s_{22}(\lambda; X, T)$ can be determined from the normalization condition \eqref{eq:normal-S-o}:
\begin{equation}
\begin{aligned}
s_{11}(\lambda; X, T)&=\frac{\Theta\left(A(\infty)+A(Q)+\ii\pi+\frac{B}{2}\right)}{\Theta\left(A(\infty)+A(Q)+\ii\pi+\frac{B}{2}-F_1V\right)}\ee^{-\frac{F_1E}{2}},\\
s_{22}(\lambda; X, T)&=\frac{\Theta\left(A(\infty)+A(Q)+\ii\pi+\frac{B}{2}\right)}{\Theta\left(A(\infty)+A(Q)+\ii\pi+\frac{B}{2}+F_1V\right)}\ee^{\frac{F_1E}{2}}.
\end{aligned}
\end{equation}
Based on this solution, we get the outer parametrix matrix $\mathbf{Q}^{\rm out}_{f, o}$ through Eq.\eqref{eq:Qout-o}. Due to the matrix has singularities at $\lambda=a_o, b_o, a_o^*, b_o^*$, we should define the local parametrix matrices $\mathbf{Q}^{a_o}_{f,o}(\lambda; X, T), \mathbf{Q}^{b_o}_{f,o}(\lambda; X, T)$ and $ \mathbf{Q}^{a_o^*}_{f,o}(\lambda; X, T), \mathbf{Q}^{b_o^*}_{f,o}(\lambda; X, T)$ at the the disk of $a_o, b_o, a_o^*$ and $b_o^*$ respectively. Each of these local parametrix can be constructed by the Airy function. Then the global parametrices can be written as
\begin{equation}
\dot{\mathbf{Q}}_{f, o}(\lambda; X, T):=\left\{\begin{split}&\mathbf{Q}^{a_o}_{f,o}(\lambda; X, T), \quad \lambda\in\mathbb{D}^{a_o},\\
&\mathbf{Q}^{b_o}_{f,o}(\lambda; X, T), \quad \lambda\in\mathbb{D}^{b_o},\\
&\mathbf{Q}^{a_o^*}_{f,o}(\lambda; X, T), \quad \lambda\in\mathbb{D}^{a_o^*},\\
&\mathbf{Q}^{b_o^*}_{f,o}(\lambda; X, T), \quad \lambda\in\mathbb{D}^{b_o^*},\\
&\mathbf{Q}^{\rm out}_{f,o}(\lambda; X, T), \quad {\rm otherwise},\\
\end{split}\right.
\end{equation}
where the disks $\mathbb{D}^{a_o}$, $\mathbb{D}^{b_o}$, $\mathbb{D}^{a_o^*}$ and $\mathbb{D}^{b_o^*}$ are centered at $a_0$, $b_o$, $a_o^*$ and $b_o^*$ respectively.
\ys{Then the error between $\dot{\mathbf{Q}}_{f, o}(\lambda; X, T)$ and $\mathbf{Q}_{f, o}(\lambda; X, T)$ can be written as
\begin{equation}
\mathbf{Q}_{f, o}(\lambda; X, T)=\mathbf{Q}_{f,o}^{\rm err}(\lambda; X, T)\dot{\mathbf{Q}}_{f, o}(\lambda; X, T)
\end{equation}
Similar to the error analysis in the Sec. \ref{sec:error-analysis}, we can estimate the error when $N\to\infty$. Based on the reference \cite{Robert-CPAM-2007} in section 4.7 to construct the parametrices, when $N\to\infty$, we have $|\mathbf{Q}^{\rm err}_{f,o}(\lambda; X, T)|=\mathcal{O}(N^{-1})$.}

Thus the potential $\mathbf{q}(X, T)$ can be recovered as
\begin{multline}\label{eq:q-last}
q_i(X, T)=\lim\limits_{\lambda\to\infty}c_{i}^*\lambda\ii s_{11}\left(p(\lambda)-p^{-1}(\lambda)\right)\mathscr{T}(\lambda)_{12}\ee^{-2F(\lambda)-NG}\left(\mathbb{I}+O(N^{-1})\right)\\
=c_{i}^*{\rm Im}(b_o-a_o)\frac{\Theta\left(A(\infty)+A(Q)+\ii\pi+\frac{B}{2}\right)}{\Theta\left(A(\infty)+A(Q)+\ii\pi+\frac{B}{2}-F_1V\right)}
\frac{\Theta\left(A(\infty)-A(Q)-\ii\pi-\frac{B}{2}+F_1V\right)}{\Theta\left(A(\infty)-A(Q)-\ii\pi-\frac{B}{2}\right)}\ee^{-F_1E-2F_0}\\
+O(N^{-1}), \qquad (i=1,2).
\end{multline}
Following the method in\cite{Biondini-CPAM-2017}, we can simplify the modulus of $\mathbf{q}(X, T)$ further. Consider a new function $f(k)$
\begin{equation}
f(k):=1-\frac{\left(k-b_o\right)\left(k-a_o^*\right)}{\mathscr{R}(k)}=p(k)\left(p^{-1}(k)-p(k)\right)
\end{equation}
as a function on the Riemann surface sheet one, such that $\mathscr{R}(\lambda)\to\lambda^2$ as $\lambda\to\infty$. Obviously, $f(k)$ has the singularities at $k=a_o$ and $k=b_o^*$, and has the zeros at $k=\lambda_0$ and $k=\infty$. Thus $f(k)$ is a meromorphic function on the Riemann surface, and its divisor $(f)$ is
\begin{equation}
(f)=\lambda_0+\infty-a_o-b_o^*.
\end{equation}
By Abel theorem \cite{Its-book-1994}, we have $A((f))=0$, which is equivalent to
\begin{equation}\label{eq:AQ}
A(Q)=A(a_o)+A(b_o^*)-A(\infty).
\end{equation}
Substituting $A(Q)$ into the solution $\mathbf{q}(X, T)$ in Eq.\eqref{eq:q-last}, by a simple calculation, we have
\begin{equation}
\begin{aligned}
q_i(X, T)&=c_{i}^*{\rm Im}(b_o-a_o)\frac{\Theta\left(2A(\infty)-Nd\tau\right)\Theta(0)}{\Theta(2A(\infty)\Theta(Nd\tau))}\ee^{-F_1E-2F_0}\left(\mathbb{I}+\mathcal{O}(N^{-1})\right), \quad (i=1,2).
\end{aligned}
\end{equation}
With the aid of the Proposition \ref{prop:theta}, $\mathbf{q}(X, T)$ can be reduced into
\begin{equation}
\begin{split}
q_i(X, T)&=c_{i}^*{\rm Im}(b_o-a_o)\frac{\theta_2\left(\pi A-\frac{Nd\tau}{2\ii}\right)\theta_3(0)}{\theta_2\left(\pi A\right)\theta_3\left(\frac{Nd\tau}{2\ii}\right)}\ee^{-\frac{Nd\tau}{2}-F_1E-2F_0}\left(\mathbb{I}+\mathcal{O}(N^{-1})\right), (i=1,2),
\end{split}
\end{equation}
where
\begin{equation}\label{eq:A-NO}
A:=\frac{1}{\oint_{\alpha}\frac{1}{\mathscr{R}(s)}ds}\left(\int_{a^*}^{\infty}\frac{1}{\mathscr{R}(s)}ds+\int_{a}^{\infty}\frac{1}{\mathscr{R}(s)}ds\right)
\end{equation}
is a real constant.

Furthermore, $|q_i(X, T)|^2$ can be rewritten as
\begin{equation}\label{eq:q-last-1}
\begin{aligned}
|q_i(X, T)|^2&=|c_i|^2{\rm Im}(b_o-a_o)^2\frac{\theta_2\left(\pi A-\frac{Nd\tau}{2\ii}\right)\theta_2\left(\pi A+\frac{Nd\tau}{2\ii}\right)\theta_3^2(0)}{\theta_2^2\left(\pi A\right)\theta_3^2\left(\frac{Nd\tau}{2\ii}\right)}\ee^{-Nd\tau-|F_1E|^2}\left(\mathbb{I}+\mathcal{O}(N^{-1})\right)\\
&=|c_i|^2{\rm Im}(b_o-a_o)^2\left(1-\frac{
\theta_4^2(\pi A)}{\theta_2^2(\pi A)}\frac{\theta_1^2\left(\frac{Nd\tau}{2\ii}\right)}{\theta_3^2\left(\frac{Nd\tau}{2\ii}\right)}\right)\ee^{-Nd\tau-|F_1E|^2}\left(\mathbb{I}+\mathcal{O}(N^{-1})\right)\\
&=|c_i|^2{\rm Im}(b_o-a_o)^2\left(1-\frac{
\theta_4^2(\pi A)}{\theta_2^2(\pi A)}\frac{\theta_2^2(0)\theta_4^2(0)}{\theta_3^4(0)}{\rm sd}^2(u,m)\right)\ee^{-Nd\tau-|F_1E|^2}\left(\mathbb{I}+\mathcal{O}(N^{-1})\right),
\end{aligned}
\end{equation}
where $u:=\frac{Nd\tau}{2\ii}\theta_3^2(0)$, $(i=1,2)$.

Next, we should reduce the formula $\frac{
\theta_4^2(\pi A)}{\theta_2^2(\pi A)}$ such that the leading order term of $|q_1(X, T)|^2, |q_{2}(X, T)|^2$ can be represented into the Jacobi $\cn$ function.
\begin{lemma}
\begin{equation}
\frac{
\theta_4^2(\pi A)}{\theta_2^2(\pi A)}=-4\frac{\theta_2^4(0)}{\theta_2^2(0)}\frac{\frac{a_o^*-b_o^*}{a_o-b_o}}{\left(\frac{a_o^*-b_o^*}{a_o-b_o}-1\right)^2}.
\end{equation}
\begin{proof}
With the Proposition \ref{prop:theta} and the relation between $A$ and $A(\infty)$: $\pi A=-\ii A(\infty)+\frac{\pi\tau}{2}$, we have
\begin{equation}\label{eq:Ainfty}
\frac{\theta_4^2(\pi A)}{\theta_{2}^2(\pi A)}=-\frac{\theta_1^2\left(-\ii A(\infty)\right)}{\theta_3^2\left(-\ii A(\infty)\right)}=\frac{\theta_1^2\left(2\pi \tilde{A}\right)}{\theta_3^2\left(2\pi\tilde{A}\right)},
\end{equation}
where $\tilde{A}:=\frac{1}{\oint_{\alpha}\frac{1}{\mathscr{R}(s)}ds}\int_{a^*_o}^{\infty}\frac{1}{\mathscr{R}(s; X, T)}ds$.
Based on the formulas between $\theta(2A)$ and $\theta(A)$ in \cite{Armitage-book-2006}, Eq.\eqref{eq:Ainfty} turns into
\begin{equation}
\frac{\theta_4^2(\pi A)}{\theta_{2}^2(\pi A)}=-4\frac{\theta_4^2(0)}{\theta_2^2(0)}\left(\frac{\varrho}{\varrho^2-1}\right)^2,\,\,\,\, \varrho:=\frac{\theta_3(\pi\tilde{A})\theta_4(\pi\tilde{A})}{\theta_1(\pi\tilde{A})\theta_2(\pi\tilde{A})}.
\end{equation}
Then we just need to determine the constant $\varrho.$ To get it, we construct a new function $\varrho(k)$ on Riemann surface $\mathcal{S}_1$ by
\begin{equation}
\varrho(k):=\frac{\theta_3\left(\frac{\pi}{\oint_{\alpha}\frac{1}{\mathscr{R}(s)}ds}\int_{a_o^*}^{k}\frac{1}{\mathscr{R}(s)}ds\right)\theta_4\left(\frac{\pi}{\oint_{\alpha}\frac{1}{\mathscr{R}(s)}ds}\int_{a_o^*}^{k}\frac{1}{\mathscr{R}(s)}ds\right)}
{\theta_1\left(\frac{\pi}{\oint_{\alpha}\frac{1}{\mathscr{R}(s)}ds}\int_{a_o^*}^{k}\frac{1}{\mathscr{R}(s)}ds\right)\theta_2\left(\frac{\pi}{\oint_{\alpha}\frac{1}{\mathscr{R}(s)}ds}\int_{a_o^*}^{k}\frac{1}{\mathscr{R}(s)}ds\right)},
\end{equation}
which infers that $\varrho=\varrho(\infty)$. From the properties of zeros of Theta functions in Proposition \ref{prop:theta}, we know $k=a_o$ and $k=b_o$ are the zeros of $\varrho^2(k)$, $k=a_o^*$ and $k=b_o^*$ are the poles of $\varrho^2(k)$. Note that $\varrho^2(k)$ is meromorphic on $\mathcal{S}_1$. Hence $\varrho^2(k)$ can be expressed in the form
\begin{equation}
\varrho^2(k)=p^2\frac{\left(k-a_o\right)\left(k-b_o\right)}{\left(k-a_o^*\right)\left(k-b_o^*\right)},
\end{equation}
where $p^2$ is a constant to be determined by the residue at $k=a_o^*$, i.e.
\begin{equation}
p^2=\frac{(a_o^*-b_o^*)}{(a_o^*-a_o)(a_o^*-b_o)}\underset{k=a_o^*}{\rm Res}\left(\varrho^2(k)\right)=\frac{\theta_3^2(0)\theta_4^2(0)\left(\oint_{\alpha}\frac{1}{\mathscr{R}(s)}ds\right)^2\left(a_o^*-b_o^*\right)^2}{4\pi^2\left(\theta_1'(0)\right)^2\theta_2^2(0)}.
\end{equation}
Moreover, since $\theta_1'(0)=\theta_2(0)\theta_3(0)\theta_4(0)$ and the integral $\oint_{\alpha}\frac{1}{\mathscr{R}(s)}ds$ can be converted into the first type of complete elliptic integral. By a direct calculation, $p^2$ can be reduced into
\begin{equation}
p^2=\frac{a_o^*-b_o^*}{a_o-b_o}.
\end{equation}
Thus $\varrho=\varrho(\infty)=\left(\frac{a_o^*-b_o^*}{a_o-b_o}\right)^{1/2}$, which completes the proof.
\end{proof}
\end{lemma}
Finally, Eq.\eqref{eq:q-last-1} can be simplified into
\begin{equation}
\begin{split}
|q_i(X, T)|^2&=|c_i|^2\left(\left({\rm Im}(b_o-a_o)\right)^2-|a_o-b_o|^2\cn^2\left(u+K(m),m\right)\right)\ee^{-Nd\tau-|F_1E|^2}\left(\mathbb{I}+\mathcal{O}(N^{-1})\right),
\end{split}
\end{equation}
where $m=\frac{\theta_2^4(0)}{\theta_3^4(0)}$, $i=1,2$ and $u$ is given in equation \eqref{eq:q-last-1}. \ys{Compared to the leading order term in \cite{Deniz-arXiv-2019}, our result is reduced to the Jacobi elliptic function, which is a new formula and has never been reported before. The periodicity of this formula can be seen through the properties of the $cn$ function, which seems more simple and clear.}
\subsection{Non-Oscillatory Region}
\label{sec:no}
In the last subsection, we have obtained the leading order term in the oscillatory region, whose modulus is related to a Jacobi elliptic function $\cn$. In this subsection, we want to discuss the leading order term in the non-oscillatory region. Similar to the study in the oscillatory region, the original contour is not closed either, we still need to construct the $g$-function. However, this $g$-function is different from the $G$-function. To the capital $G$-function, it is related to the algebraic curve with genus one. But to the $g$-function, the corresponding genus of algebraic curve is zero. Following the method in the last subsection, we give the Riemann-Hilbert problem about the $g$-function.
\begin{rhp}\label{RHP:NO}
For the variables $X, T$, there exists a unique $g$-function, satisfying the following conditions:
\begin{itemize}
\item
{\bf Analyticity}: $g(\lambda; X, T)$ is analytic except on $\Sigma$ (to be determined), and it takes the continuous boundary condition from the left and right side of $\Sigma$.
\item
{\bf Jump Condition}: For $\lambda\in\Sigma$, $g(\lambda; X, T)$ satisfies the following jump condition:
\begin{equation}
g_{+}(\lambda; X, T)+g_{-}(\lambda; X, T)-2\varphi(\lambda; X, T)=\Omega_{n},\qquad\lambda\in\Sigma,
\end{equation}
where the subscript $_n$ indicates the non-oscillatory region,
\item
{\bf Normalization}: As $\lambda\to\infty$, $g(\lambda; X, T)$ satisfies
\begin{equation}
g(\lambda; X, T)\to \mathcal{O}(\lambda^{-1}),
\end{equation}
\item
{\bf Symmetry}: $g(\lambda;X,T)$ satisfies the Schwartz symmetric condition:
\begin{equation}
g(\lambda; X, T)=-g(\lambda^*; X, T)^*.
\end{equation}
\end{itemize}
\end{rhp}
According to the above Riemann-Hilbert problem, we can determine the $g(\lambda)$ and the contour $\Sigma$. With the standard method, we usual analyze the derivative of $g$-function to describe the properties of $g$-function. For $\lambda\in \Sigma$, we have
\begin{equation}\label{eq:g-function}
g'_{+}(\lambda)+g'_{-}(\lambda)=2\ii X+4\ii\lambda T+\frac{1}{\lambda-\lambda_1^*}-\frac{1}{\lambda-\lambda_1},\qquad \lambda\in\Sigma.
\end{equation}
To solve Eq.\eqref{eq:g-function}, we introduce a new function $R(\lambda)$:
\begin{equation}\label{eq:R-no}
R(\lambda):=\left((\lambda-a_n)(\lambda-a_n^*)\right)^{1/2},
\end{equation}
then $\Sigma$ can be a determined curve with the endpoints of $a_n$ and $a_n^*$.  Therefore, the function $\frac{g'(\lambda)}{R(\lambda)}$ can be solved with the Plemelj formula,
\begin{equation}
g'(\lambda)=\frac{R(\lambda)}{2\pi\ii}\int_{\Sigma}\frac{2\ii X+4\ii sT+\frac{1}{s-\lambda_1^*}-\frac{1}{s-\lambda_1}}{R(s)\left(s-\lambda\right)}ds.
\end{equation}
Still, by the generalized residue theorem, $g'(\lambda)$ can be given as
\begin{equation}
\begin{split}
g'(\lambda)&=R(\lambda)\left(\mathop{\rm Res}\limits_{s=\lambda}+\mathop{\rm Res}\limits_{s=\lambda_1}+\mathop{\rm Res}\limits_{s=\lambda_1^*}\ys{+}\mathop{\rm Res}\limits_{s=\infty}\right)\left(\frac{\ii X+2\ii sT+\frac{1}{2(s-\lambda_1^*)}-\frac{1}{2(s-\lambda_1)}}{R(s)\left(s-\lambda\right)}\right)\\
&=\frac{R(\lambda)}{2R(\lambda_1^*)\left(\lambda_1^*-\lambda\right)}-\frac{R(\lambda)}{2R(\lambda_1)(\lambda_1-\lambda)}-2\ii TR(\lambda)+\ii X+2\ii T\lambda+\frac{1}{2\left(\lambda-\lambda_1^*\right)}-\frac{1}{2\left(\lambda-\lambda_1\right)}.
\end{split}
\end{equation}
Then we have
\begin{equation}\label{eq:g-no}
g'(\lambda)-\varphi'(\lambda; X, T)=R(\lambda)\left(\frac{1}{2R(\lambda_1^*)\left(\lambda_1^*-\lambda\right)}-\frac{1}{2R(\lambda_1)(\lambda_1-\lambda)}-2\ii T\right):=\frac{h'(\lambda; X, T)}{2},
\end{equation}
whose roots are $\lambda_c, \lambda_d, a_n, a_n^*$. Similar to the oscillatory region, we set the parameter $\lambda_1=\ii,$ and introduce two parameters $q=a_n+a_n^*, p=a_na_n^*$. Through the definition of $g(\lambda), $ as $\lambda\to\infty, g'(\lambda)\to O(\lambda^{-2})$, we can give two equations about the parameters $p$ and $q$,
\begin{equation}\label{eq:pands}
\begin{aligned}
&\mathcal{O}(1): \frac{1}{2\sqrt{-1+p-\ii q}}-\frac{1}{2\sqrt{-1+p+\ii q}}+\ii X+\ii qT=0,\\
&\mathcal{O}(\lambda^{-1}): \frac{2\ii-q}{4\sqrt{-1+p-\ii q}}+\frac{2\ii+q}{4\sqrt{-1+p+\ii q}}-\ii pT+\frac{1}{4}\ii q^2T=0.
\end{aligned}
\end{equation}
From the first equation in Eq.\eqref{eq:pands}, we have
\begin{equation}
\frac{1}{\sqrt{-1+p+\ii q}}=\frac{1}{\sqrt{-1+p-\ii q}}+2\ii qT+2\ii X.
\end{equation}
Plugging the above equation into the second equation of Eq.\eqref{eq:pands} and eliminating the factor $\frac{1}{\sqrt{-1+p-\ii q}}$, we have
\begin{equation}\label{eq:pands-1}
\frac{1}{\sqrt{-1+p-\ii q}}=-\ii X-\frac{Xq}{2}+pT-\ii qT-\frac{3q^2T}{4}.
\end{equation}
It can be seen for the fixed $X$ and $T$, Eq.\eqref{eq:pands-1} is a complex equation with respect to $q$ and $p$, thus we need to discuss both the real and imaginary part.
From the imaginary part of Eq.\eqref{eq:pands-1}, we get a relation between $q$ and $p$, i.e.
\begin{equation}
p=\frac{q^2}{4}, \qquad p=\frac{4X^2q+8XT+12XTq^2+8qT^2+9q^3T^2}{4T\left(2X+3qT\right)},
\end{equation}
if $p=\frac{q^2}{4}$, then the zero of $R(\lambda)$ will be $a_n=a^*_n=\frac{q}{2}$, which is unreasonable. Otherwise, substituting the second solution into the real part of equation \eqref{eq:pands-1}, we get a septic equation with respect to $q$, which is rigorous proved to have one real root by the Sturm theorem in the appendix of \cite{BilmanM-21}. Define
\begin{equation}
g(\lambda):=\int_{\infty}^{\lambda}g'(s)ds
\end{equation}
for each given $X$ and $T$. By choosing one special $X$ and $T$, we give the contour for ${\rm Re}(g(\lambda)-\varphi(\lambda; X, T))$, which is shown in Fig.\ref{contour-non-os}. Moreover, we can use the steepest descent method to deform this contour and obtain the leading order term as $N\to\infty$.

Similar to the analysis in the oscillatory region in the last subsection, we define a new matrix $\mathbf{O}_{f,n}(\lambda; X, T)$ and $\mathbf{P}_{f,n}(\lambda; X, T)$ as
\begin{equation}\label{eq:O-far-field-n}
\begin{aligned}
&\mathbf{O}_{f,n}(\lambda; X, T):=\left\{\begin{aligned}&\widetilde{\mathbf{N}}_{f}\ee^{-N\varphi(\lambda; X, T)\pmb{\sigma}_3}\widehat{\mathbf{Q}}_{d}^{-1}\ee^{N\varphi(\lambda; X, T)\sigma_3},\qquad &\lambda\in D_0\cap \left(D_{\rm u}\cup D_{\rm d}\right)^{c},\\
&\widetilde{\mathbf{N}}_{f}\qquad &{\rm otherwise}.\end{aligned}\right.\\
&\mathbf{P}_{f,n}(\lambda; X, T):=\mathbf{O}_{f,n}(\lambda; X, T){\rm diag}\left(\ee^{-Ng(\lambda)}, \ee^{Ng(\lambda)}, 1\right)
\end{aligned}
\end{equation}
where $D_{\rm u}=D_{\rm u_1}^{\rm i}\cup D_{\rm u_2}^{\rm i}\cup K_{\rm u}^{\rm i}$ and $D_{\rm d}=D_{\rm d_1}^{\rm i}\cup D_{\rm d_2}^{\rm i}\cup K_{\rm d}^{\rm i}$ are shown in Fig.\ref{contour-non-os}.
\begin{figure}[!h]
\centering
\includegraphics[width=0.45\textwidth]{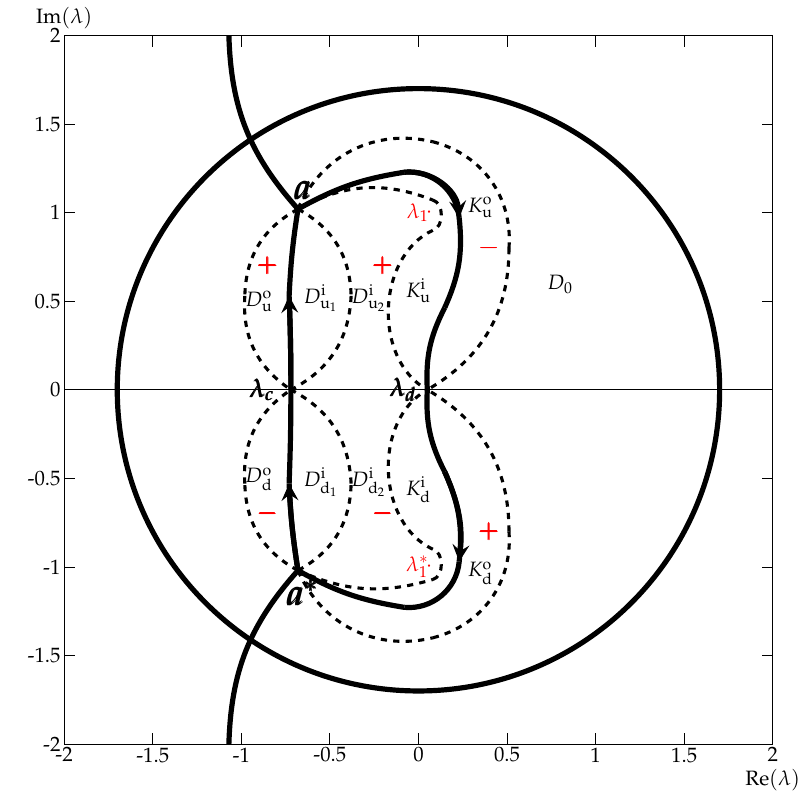}
\centering
\includegraphics[width=0.45\textwidth]{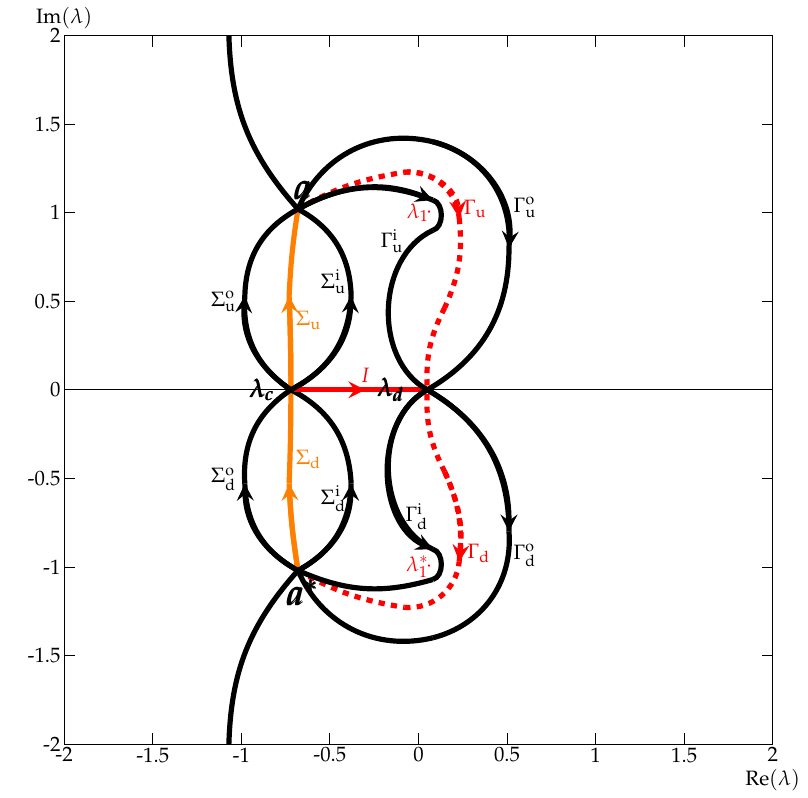}
\caption{The contour of ${\rm Re}(\varphi(\lambda; X, T)-g(\lambda))$ in the non-oscillatory region with the parameters $X=\frac{9}{10}, T=\frac{1}{3}.$ The left one gives the original contour about the $\widetilde{\mathbf{N}}_{f}(\lambda)$ and the sign of ${\rm Re}(\varphi(\lambda; X, T)-g(\lambda))$. The right one is the corresponding contour deformation.}
\label{contour-non-os}
\end{figure}
For $\lambda\in\Sigma_{\rm u}\cup\Sigma_{\rm d}\cup\Gamma_{\rm u}\cup\Gamma_{\rm d}$, the jump matrices about the new matrix $\mathbf{P}_{f,n}(\lambda; X, T)$ is the same as Eq.\eqref{eq:P-o-jump} with the capital $G(\lambda)$ replaced by $g(\lambda)$. With a similar definition of $\mathbf{Q}_{f,n}(\lambda; X, T)$ in Eq.\eqref{eq:Q-definition}, as $N\to\infty$, there only left four jump matrices that do not tend to identity matrix, that is
\begin{equation}\label{eq:Qfn-0}
\begin{aligned}
\mathbf{Q}_{f,n,+}&=\mathbf{Q}_{f,n,-}\begin{bmatrix}0&\sqrt{2} \ee^{N\Omega_{n}}&0\\
-\frac{\sqrt{2}}{2} \ee^{-N\Omega_{n}}&0&0\\
0&0&1
\end{bmatrix},\qquad &\lambda\in \Sigma_{\rm u},\\
\mathbf{Q}_{f,n,+}&=\mathbf{Q}_{f,n,-}\begin{bmatrix}0&\frac{\sqrt{2}}{2} \ee^{N\Omega_{n}}&0\\
-\sqrt{2}\ee^{-N\Omega_{n}}&0&0\\
0&0&1
\end{bmatrix},\qquad &\lambda\in \Sigma_{\rm d},\\
\mathbf{Q}_{f,n,+}&=\mathbf{Q}_{f,n,-}{\rm diag}\left(\sqrt{2}, \frac{\sqrt{2}}{2}, 1\right), \qquad &\lambda\in \Gamma_{\rm u},\\
\mathbf{Q}_{f,n,+}&=\mathbf{Q}_{f,n,-}{\rm diag}\left(\frac{\sqrt{2}}{2}, \sqrt{2}, 1\right), \qquad &\lambda\in \Gamma_{\rm d}.
\end{aligned}
\end{equation}
Moreover, we set a new sectionally analytic matrix to eliminate the jump on $\Gamma_{\rm u}$ and $\Gamma_{\rm d}$, which is
\begin{equation}\label{eq:hatQ-no}
\widehat{\mathbf{Q}}_{f, n}(\lambda; X, T)=\left\{\begin{split}&\mathbf{Q}_{f,n}{\rm diag}\left(\sqrt{2},\frac{\sqrt{2}}{2}, 1\right),\quad &\lambda\in K_{u}^{\ii}\cup D_{u_2}^{\ii}\cup D_{u_1}^{\ii},\\
&\mathbf{Q}_{f,n}{\rm diag}\left(\frac{\sqrt{2}}{2}, \sqrt{2}, 1\right),\quad &\lambda\in K_{d}^{\ii}\cup D_{d_2}^{\ii}\cup D_{d_1}^{\ii},\\
&\mathbf{Q}_{f,n}, \quad &{\rm otherwise}.
\end{split}\right.
\end{equation}
Then the nonzero jump about the $\widehat{\mathbf{Q}}_{f,n}(\lambda; X, T)$ when $N$ is large changes into
\begin{equation}\label{eq:jump-hatQ}
\begin{aligned}
\widehat{\mathbf{Q}}_{f, n, +}&=\widehat{\mathbf{Q}}_{f, n, -}\begin{bmatrix}0&\ee^{N\Omega_{n}}&0\\-\ee^{-N\Omega_{n}}&0&0\\0&0&1
\end{bmatrix},\quad&\lambda\in \Sigma_{\rm u}\cup \Sigma_{\rm d},\\
\widehat{\mathbf{Q}}_{f, n, +}&=\widehat{\mathbf{Q}}_{f, n, -}{\rm diag}\left(2, \frac{1}{2}, 1\right), \quad &\lambda\in I.
\end{aligned}
\end{equation}
By the corresponding detailed calculation shown in Appendix \ref{app:no}, the leading order term is given by equation \eqref{eq:q-no},
%\begin{equation}\label{eq:q-no}
%\begin{split}
%q_{i}(X, T)=&-\ii c_{i}^*\ee^{N\Omega_{n}-\ii\mu}\left(-\ii {\rm Im}(a_{n})+\frac{\sqrt{2p}}{N^{1/2}\sqrt{-\hat{h}''(\lambda_c)}}m_{-}^{\lambda_c}\ee^{\ii\phi_{\lambda_c}}-\frac{\sqrt{2p}}{N^{1/2}\sqrt{-\hat{h}''(\lambda_c)}}m_{+}^{\lambda_c}\ee^{-\ii\phi_{\lambda_c}}\right)\\
%&-\ii c_{i}^*\ee^{N\Omega_{n}-\ii\mu}\left(\frac{\sqrt{2p}}{N^{1/2}\sqrt{\hat{h}''(\lambda_d)}}m_{+}^{\lambda_d}\ee^{\ii\phi_{\lambda_d}}-\frac{\sqrt{2p}}{N^{1/2}\sqrt{\hat{h}''(\lambda_d)}}m_{-}^{\lambda_d}\ee^{-\ii\phi_{\lambda_d}}\right)+\mathcal{O}(N^{-1}),
%\end{split}
%\end{equation}
where
\begin{equation}\label{eq:phiaphib}
\begin{split}
\phi_{\lambda_c}&{=}\frac{\pi}{4}{+}\frac{\log(2)^2}{2\pi}{-}\arg\left(\Gamma\left(\ii\frac{\log(2)}{2\pi}\right)\right){-}2\ii k_{-}(\lambda_c){+}2N\hat{h}_{{-},\lambda_c}{+}p\log({-}Nh''_{\lambda_c}(\lambda_d{-}\lambda_c)^2){-}\ii N\Omega_{n}{-}\mu\\
\phi_{\lambda_d}&{=}\frac{\pi}{4}{+}\frac{\log(2)^2}{2\pi}{-}\arg\left(\Gamma\left(\ii\frac{\log(2)}{2\pi}\right)\right){+}2\ii k(\lambda_d){-}2N\hat{h}_{\lambda_d}{+}p\log(Nh''_{\lambda_d}(\lambda_d{-}\lambda_c)^2){+}\ii N\Omega_{n}{+}\mu\\
m_{+}^{\lambda_c}&{=}\frac{1}{2}{+}\frac{1}{4}\left(\sqrt{\frac{\lambda_c{-}a_n}{\lambda_c{-}a_n^*}}{+}\left(\sqrt{\frac{\lambda_c{-}a_n}{\lambda_c{-}a_n^*}}\right)^{-1}\right),\quad
m_{-}^{\lambda_c}{=}\frac{1}{2}{-}\frac{1}{4}\left(\sqrt{\frac{\lambda_c{-}a_n}{\lambda_c{-}a_n^*}}{+}\left(\sqrt{\frac{\lambda_c{-}a_n}{\lambda_c{-}a_n^*}}\right)^{-1}\right),\\
m_{+}^{\lambda_d}&{=}\frac{1}{2}{+}\frac{1}{4}\left(\sqrt{\frac{\lambda_d{-}a_n}{\lambda_d{-}a_n^*}}{+}\left(\sqrt{\frac{\lambda_d{-}a_n}{\lambda_d{-}a_n^*}}\right)^{-1}\right),\quad
m_{-}^{\lambda_d}{=}\frac{1}{2}{-}\frac{1}{4}\left(\sqrt{\frac{\lambda_d{-}a_n}{\lambda_d{-}a_n^*}}{+}\left(\sqrt{\frac{\lambda_d{-}a_n}{\lambda_d{-}a_n^*}}\right)^{-1}\right).
\end{split}
\end{equation}
\subsection{The algebraic decay region}
\label{sec:al}
In the previous subsection, we discuss the asymptotics as $(X,T)$ is in the oscillatory and non-oscillatory region. During the study for the asymptotics, we construct two different types of ``g"-function. The reason for the ``g"-function is that the original contour of ${\rm \Re}(\varphi(\lambda; X, T))$ is not efficient for further deformation. While in the algebraic region, there exist a closed contour connecting two critical points $a_A$ and $b_A$, which is shown in Fig.\ref{contour-algebraic}.
\begin{figure}[!h]
\centering
\includegraphics[width=0.45\textwidth]{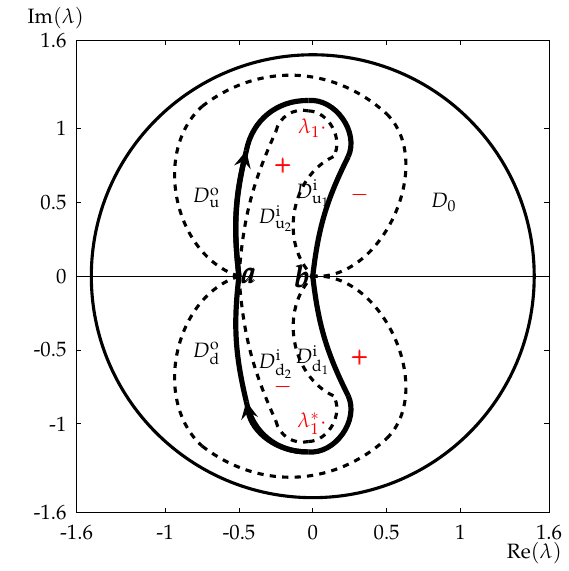}
\centering
\includegraphics[width=0.45\textwidth]{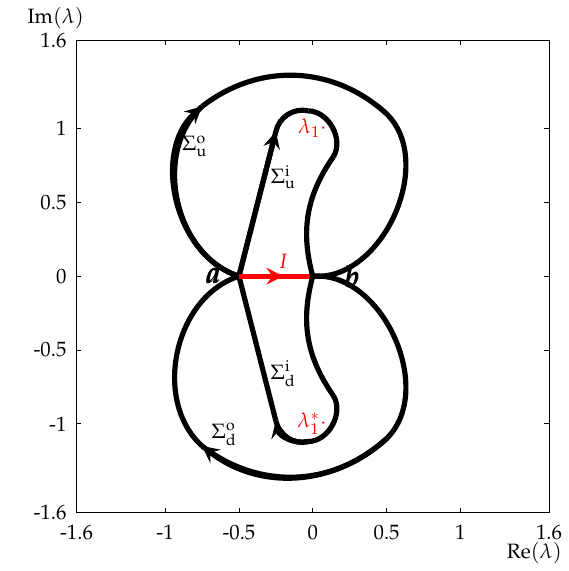}
\caption{The contour of ${\rm Re}(\varphi(\lambda; X, T))$ in the algebraic decay region with the parameters $X=1, T=\frac{1}{6}.$ The left one gives the original contour about the $\widetilde{\mathbf{N}}_{f}(\lambda; X, T)$ and the sign of ${\rm Re}(\varphi(\lambda; X, T))$. The right one is the corresponding contour deformation.}
\label{contour-algebraic}
\end{figure}
Similar to the construction of $\mathbf{O}_{f}$ matrix in the non-oscillatory region \eqref{eq:O-far-field} and the oscillatory region\eqref{eq:O-far-field}, we give the corresponding $\mathbf{O}_{f,a}$ matrix in the algebraic decay region
\begin{equation}\label{eq:O-far-field-a}
\mathbf{O}_{f,a}(\lambda; X, T):=\left\{\begin{aligned}&\widetilde{\mathbf{N}}_{f}\ee^{-N\varphi(\lambda; X, T)\pmb{\sigma}_3}\widehat{\mathbf{Q}}_{d}^{-1}\ee^{N\varphi(\lambda; X, T)\pmb{\sigma}_3},\qquad &\lambda\in D_0\cap \left(D_{\rm u}\cup D_{\rm d}\right)^{c},\\
&\widetilde{\mathbf{N}}_{f}\qquad &{\rm otherwise},\end{aligned}\right.
\end{equation}
where $D_{\rm u}=D_{\rm u_1}^{\rm i}\cup D_{\rm u_2}^{\rm i}, D_{\rm d}=D_{\rm d_1}^{\rm i}\cup D_{\rm d_2}^{\rm i}$. It is not difficult for find the contour shape in this case is similar to the large $\chi$ asymptotics in the infinite order. In the appendix \ref{App:large-chi}, we give a detailed calculation about this leading order term. Set a similar deformation with the large $\chi$ asymptotics:
\begin{equation}
\mathbf{Q}_{f,a}(\lambda; X, T):=\left\{\begin{aligned}
&\mathbf{O}_{f,a}(\lambda; X, T)\ee^{-{\rm ad}_{\pmb{\sigma}_{3}}N\varphi(\lambda; X, T)}\left(\mathbf{Q}_{R}^{[2]}\right)^{-1},\qquad&\lambda\in D_{\rm u}^{\rm o},\\
&\mathbf{O}_{f,a}(\lambda; X, T)\mathbf{Q}_{L}^{[2]}\ee^{-{\rm ad}_{\pmb{\sigma}_{3}}N\varphi(\lambda; X, T)}\mathbf{Q}_{C}^{[2]},\qquad&\lambda\in D_{\rm u_1}^{\rm i},\\
&\mathbf{O}_{f,a}(\lambda; X, T)\mathbf{Q}_{L}^{[2]},\qquad&\lambda\in D_{\rm u_2}^{\rm i},\\
&\mathbf{O}_{f,a}(\lambda; X, T)\mathbf{Q}_{L}^{[1]},\qquad&\lambda\in D_{\rm d_2}^{\rm i},\\
&\mathbf{O}_{f,a}(\lambda; X, T)\mathbf{Q}_{L}^{[1]}\ee^{-{\rm ad}_{\pmb{\sigma}_{3}}N\varphi(\lambda; X, T)}\mathbf{Q}_{C}^{[1]},\qquad&\lambda\in D_{\rm d_1}^{\rm i},\\
&\mathbf{O}_{f,a}(\lambda; X, T)\ee^{-{\rm ad}_{\pmb{\sigma}_{3}}N\varphi(\lambda; X, T)}\left(\mathbf{Q}_{R}^{[1]}\right)^{-1},\qquad&\lambda\in D_{\rm d}^{\rm o},\\
&\mathbf{O}_{f,a}(\lambda; X, T),\qquad&{\rm otherwise.}
\end{aligned}\right.
\end{equation}
As $N\to\infty$, there only left one jump which does not tend the identity, which is the same as with Eq.\eqref{eq:R-jump}, that is
\begin{equation}\label{eq:Q-jump-al}
\mathbf{Q}_{f,a,+}=\mathbf{Q}_{f,a,-}{\rm diag}\left(2, \frac{1}{2}, 1\right).
\end{equation}
Then the leading order term in the algebraic decay region can be given by replacing the variable $\Lambda \chi^{1/2}$ to $\lambda$, and $\chi^{1/4}$ to $N^{1/2}$ to Eq.\eqref{eq:q-near}, which becomes
\begin{equation}\label{eq:q-algebra}
\begin{split}
q_i(X, T)&{=}
-\frac{c_1^*\ii}{N^{1/2}}\sqrt{\frac{\ln(2)}{\pi}}\Bigg(\ee^{-2N\varphi(b_A; X, T)}\left(\sqrt{-\ii \varphi''(b_A; X, T)}\right)^{2\ii p-1}\ee^{\ii\phi_{f}}\\&{+}\ee^{-2N\varphi(a_A; X, T)}\left(\sqrt{\ii \varphi''(a_A; X, T)}\right)^{-2\ii p-1}\Bigg)+\mathcal{O}(N^{-1}),\quad (i=1,2).
\end{split}
\end{equation}
\ys{This leading order formula is similar to the asymptotics in \cite{Peter-Duke-2019} and \cite{Deniz-arXiv-2019}. We only give a minor adjustment on it.}
\subsection{Exponential decay region}
\label{sec:ex}
The last asymptotic region for the large order solitons is the exponential decay region. Compared with the asymptotic analysis in other regions, the analysis in this region is slightly easier. The corresponding signature chart for ${\rm{Re}}(\varphi(\lambda; X, T))=0$ is shown in the Fig.\ref{contour-exponent}.
\begin{figure}[!h]
\centering
\includegraphics[width=0.45\textwidth]{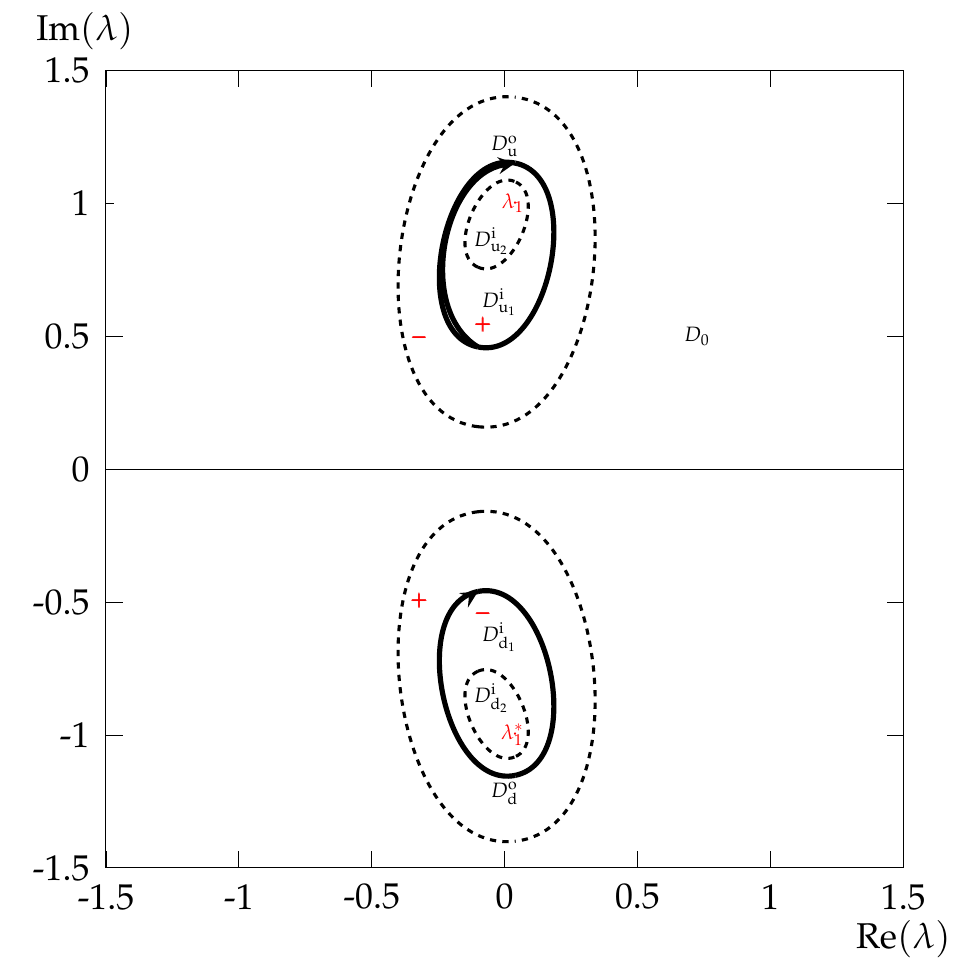}
\centering
\includegraphics[width=0.45\textwidth]{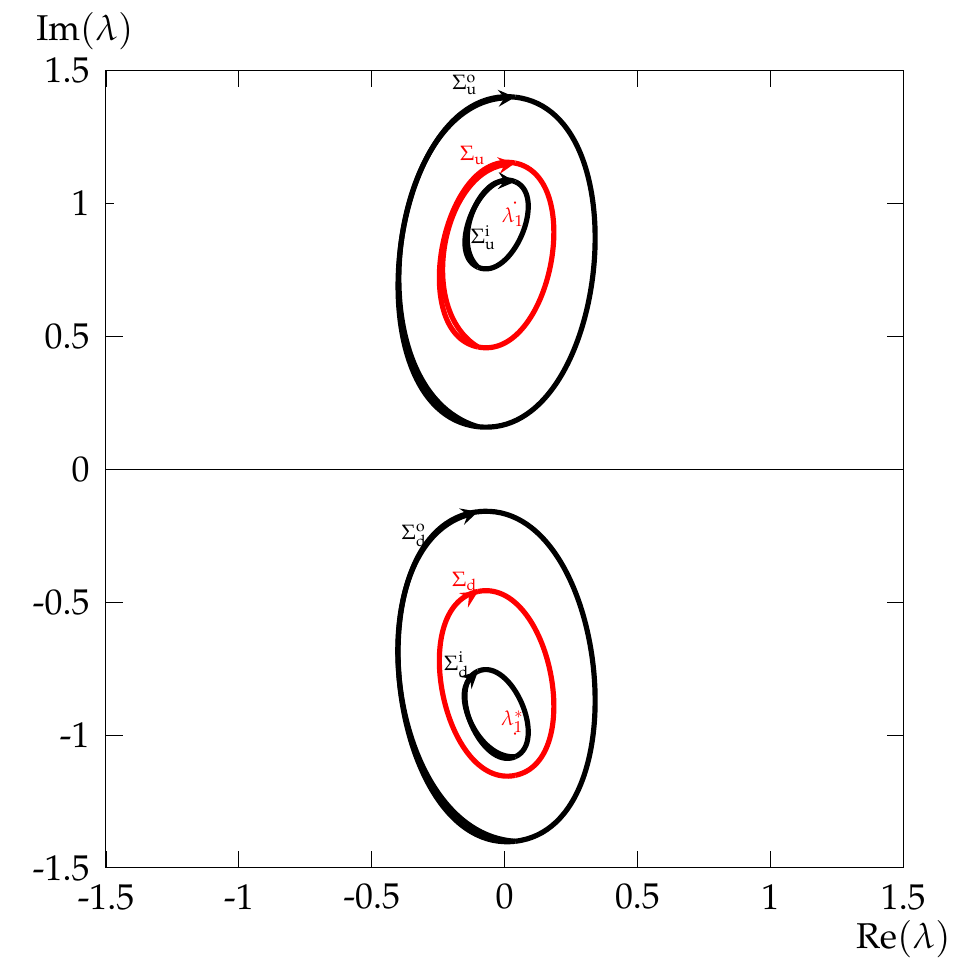}
\caption{The contour of ${\rm Re}(\varphi(\lambda; X, T))$ in the exponential decay region with the parameters $X=\frac{11}{10}, T=\frac{1}{6}.$ The left one gives the original contour about the $\widetilde{\mathbf{N}}_{f}(\lambda; X, T)$ and the sign of ${\rm Re}(\varphi(\lambda; X, T))$. The right one is the corresponding contour deformation.}
\label{contour-exponent}
\end{figure}

Firstly, define the matrix $\mathbf{O}_{f,e}$ as
\begin{equation}\label{eq:O-far-field-e}
\mathbf{O}_{f,e}(\lambda; X, T):=\left\{\begin{aligned}&\widetilde{\mathbf{N}}_{f}\ee^{-N\varphi(\lambda; X, T)\pmb{\sigma}_3}\widehat{\mathbf{Q}}_{d}^{-1}\ee^{N\varphi(\lambda; X, T)\pmb{\sigma}_3},\qquad &\lambda\in D_0\cap \left(D_{\rm u}\cup D_{\rm d}\right)^{c},\\
&\widetilde{\mathbf{N}}_{f}\qquad &{\rm otherwise}.\end{aligned}\right.
\end{equation}
where $D_{\rm u}=D_{\rm u_1}^{\rm i}\cup D_{\rm u_2}^{\rm i}, D_{d}=D_{\rm d_1}^{\rm i}\cup D_{\rm d_2}^{\rm i}$. Then the function $\mathbf{O}_{f,e}(\lambda; X, T)$ satisfies
\begin{equation}
\mathbf{O}_{f,e,+}=\mathbf{O}_{f,e,-}\ee^{-N\varphi(\lambda; X, T)\pmb{\sigma}_3}\widehat{\mathbf{Q}}_{d}^{-1}\ee^{N\varphi(\lambda; X, T)\pmb{\sigma}_3},\qquad \lambda\in\partial D_{\rm u}\cup\partial D_{\rm d}.
\end{equation}
Similarly, define $\mathbf{Q}_{f,e}$ as
\begin{equation}
\begin{aligned}
\mathbf{Q}_{f,e}:&=\mathbf{O}_{f,e}\ee^{-{\rm ad}\pmb{\sigma}_{3}N\varphi(\lambda; X, T)}\mathbf{Q}_{L}^{[5]},\qquad&\lambda\in D^{\rm i}_{\rm u_1},\\
\mathbf{Q}_{f,e}:&=\mathbf{O}_{f,e}\ee^{-{\rm ad}\pmb{\sigma}_{3}N\varphi(\lambda; X, T)}\left(\mathbf{Q}_{R}^{[5]}\right)^{-1},\qquad&\lambda\in D^{\rm o}_{\rm u},\\
\mathbf{Q}_{f,e}:&=\mathbf{O}_{f,e}\ee^{-{\rm ad}\pmb{\sigma}_{3}N\varphi(\lambda; X, T)}\mathbf{Q}_{L}^{[6]},\qquad&\lambda\in D^{\rm i}_{\rm d_1},\\
\mathbf{Q}_{f,e}:&=\mathbf{O}_{f,e}\ee^{-{\rm ad}\pmb{\sigma}_{3}N\varphi(\lambda; X, T)}\left(\mathbf{Q}_{R}^{[6]}\right)^{-1},\qquad&\lambda\in D^{\rm o}_{\rm d},\\
\mathbf{Q}_{f,e}:&=\mathbf{O}_{f,e},\qquad&{\rm otherwise.}
\end{aligned}
\end{equation}
Then the jump about $\mathbf{Q}_{f,e}$ changes into
\begin{equation}
\begin{aligned}
\mathbf{Q}_{f,e,+}&=\mathbf{Q}_{f,e,-}\ee^{-{\rm ad}\pmb{\sigma}_{3}N\varphi(\lambda; X, T)}\mathbf{Q}_{R}^{[5]},\qquad&\lambda\in \Sigma_{\rm u}^{\rm o},\\
\mathbf{Q}_{f,e,+}&=\mathbf{Q}_{f,e,-}{\rm diag}\left(\sqrt{2}, \frac{\sqrt{2}}{2}, 1\right),\qquad&\lambda\in\Sigma_{\rm u},\\
\mathbf{Q}_{f,e,+}&=\mathbf{Q}_{f,e,-}\ee^{-{\rm ad}\pmb{\sigma}_{3}N\varphi(\lambda; X, T)}\mathbf{Q}_{L}^{[5]},\qquad&\lambda\in \Sigma_{\rm u}^{\rm i},\\
\mathbf{Q}_{f,e,+}&=\mathbf{Q}_{f,e,-}\ee^{-{\rm ad}\pmb{\sigma}_{3}N\varphi(\lambda; X, T)}\mathbf{Q}_{L}^{[6]},\qquad&\lambda\in \Sigma_{\rm d}^{\rm i},\\
\mathbf{Q}_{f,e,+}&=\mathbf{Q}_{f,e,-}{\rm diag}\left(\frac{\sqrt{2}}{2}, \sqrt{2}, 1\right),\qquad &\lambda\in \Sigma_{\rm d},\\
\mathbf{Q}_{f,e,+}&=\mathbf{Q}_{f,e,-}\ee^{-{\rm ad}\pmb{\sigma}_{3}N\varphi(\lambda; X, T)}\mathbf{Q}_{R}^{[6]},\qquad&\lambda\in \Sigma_{\rm d}^{\rm o}.
\end{aligned}
\end{equation}
With this jump, we can define a model Riemann-Hilbert problem as follows:
\begin{rhp}
(The model problem in the exponential decay region) Find a unique $3\times 3$ matrix $\mathbf{P}_{f,e}(\lambda; X, T)$ satisfying the following conditions.
\item
{\bf Analyticity}: $\mathbf{P}_{f,e}(\lambda; X, T)$ is analytic for $\lambda\in\mathbb{C}\setminus \left(\Sigma_{\rm u}\cup\Sigma_{\rm d}\right)$,
\item {\bf Jump condition}: When $\lambda\in\Sigma_{\rm u}\cup\Sigma_{\rm d}$, $\mathbf{P}_{f,e}(\lambda; X, T)$ satisfy
    \begin{equation}
    \begin{split}
    \mathbf{P}_{f,e,+}(\lambda; X, T)&=\mathbf{P}_{f,e,-}(\lambda; X, T){\rm diag}\left(\sqrt{2},\frac{\sqrt{2}}{2}, 1 \right),\qquad \lambda\in \Sigma_{\rm u},\\
    \mathbf{P}_{f,e,+}(\lambda; X, T)&=\mathbf{P}_{f,e,-}(\lambda; X, T){\rm diag}\left(\frac{\sqrt{2}}{2},\sqrt{2}, 1 \right),\qquad \lambda\in \Sigma_{\rm d}.
    \end{split}
    \end{equation}
    \item
    {\bf Normalization}: When $\lambda\to\infty$, $\mathbf{P}_{f,e}\to\mathbb{I}$.
\end{rhp}
This Riemann-Hilbert problem can be solved by the Plemelj formula, whose form is
\begin{equation}
\begin{aligned}
\mathbf{P}_{f,e}&=\mathbb{I}, \qquad \lambda\in \left(D_{\rm u}\right)^{c}, \qquad
\mathbf{P}_{f,e}={\rm diag}\left(\frac{\sqrt{2}}{2}, \sqrt{2}, 1\right), \qquad \lambda\in D_{\rm u},\\
\mathbf{P}_{f,e}&=\mathbb{I}, \qquad \lambda\in \left(D_{\rm d}\right)^{c},\qquad
\mathbf{P}_{f,e}={\rm diag}\left(\sqrt{2}, \frac{\sqrt{2}}{2}, 1\right), \qquad \lambda\in D_{\rm d}.\\
\end{aligned}
\end{equation}
Then the error matrix can be constructed as
\begin{equation}
\begin{aligned}
\mathbf{R}_{f,e}(\lambda; X, T):=\mathbf{Q}_{f,e}(\lambda; X, T)\mathbf{P}_{f,e}(\lambda; X, T)^{-1}.
\end{aligned}
\end{equation}
Observing the jump matrix satisfied by $\mathbf{Q}_{f,e}(\lambda; X, T)$ and $\mathbf{P}_{f,e}(\lambda; X, T)$, we know $\mathbf{R}_{f,e}(\lambda; X, T)$ are exponentially to zero, thus we have $\mathbf{R}_{f,e}\to\mathbb{I}+O(\ee^{-dN})$, where $d$ is a constant. The potential function $\mathbf{q}(X, T)$ can be recovered as
\begin{equation}
q_i(X, T)=2c_{i}^*\lim\limits_{\lambda\to\infty}\lambda\mathbf{\widetilde{N}}_{f,e}(\lambda; X, T)_{12}\to \mathcal{O}(\ee^{-dN}),\qquad i=1,2.
\end{equation}

\section{The dynamics for the infinite order solitons}
\label{sec:infinity-order}
In the last section, we give the leading order asymptotics in four different regions, the oscillatory region, the non-oscillatory region, the algebraic decay region and the exponential decay region. All of which are studied as $N$ is large enough. However, as $N\to\infty$, the above asymptotic analysis will be error because the maximum amplitude about the high order soliton equals to $2N{\rm Im}(\lambda_1)$, which is unbounded when $N\to\infty$. To analyze the dynamics of infinite order solitons, we should tackle with this problem so as to the amplitude becomes bounded. One simple way to achieve it is setting ${\rm Im}(\lambda_1)=\frac{1}{N}$, under this condition, we can construct a RHP for the infinite order soliton, in which the factor ${\rm Re}(\lambda_1)$ can be also removed with the aid of Galilean transformation Eq.\eqref{eq:xandX}. Afterwards, we can give some new properties under the framework of the infinite order RHP. %Firstly, we will derive the Lax pair from the RHP \ref{RHP2}.

\subsection{Derivation the Lax pair of infinite order soliton}
To derive the Lax pair from RHP \ref{RHP2}, we still use the classical dressing technique. It is clear that $\mathbf{N}(\Lambda; \chi , \tau)$ is analytic for $\Lambda$ exterior to $D_0$, so it has a Laurent expansion as $\Lambda\to\infty$,
\begin{equation}\label{eq:Laurent-series}
\mathbf{N}(\Lambda; \chi , \tau)=\mathbb{I}+\sum\limits_{l=1}^{\infty}\mathbf{N}^{[l]}(\chi , \tau)\Lambda^{-l}, \qquad |\Lambda|\to\infty.
\end{equation}
The analytic Fredholm theory infers that the coefficients $\mathbf{N}^{[l]}$s are real analytic on $\mathbb{R}^2$ and the series \eqref{eq:Laurent-series} are differentiable with respect to $\chi $ and $\tau$.

We will prove the potential function defined by Eq.\eqref{eq:new-potential} also satisfies the coupled NLS equation with $x\to\chi, t\to\tau, \lambda\to \Lambda$. To get it, we set
\begin{equation}\label{eq:P-expression}
\mathbf{P}(\Lambda; \chi , \tau)=\mathbf{N}(\Lambda; \chi , \tau)\ee^{-\ii\left(\Lambda \chi +\Lambda^2\tau\right)\pmb{\sigma}_3}.
\end{equation}
Since $\mathbf{P}_\chi (\Lambda; \chi , \tau)$ and $\mathbf{P}_\tau(\Lambda; \chi , \tau)$ satisfy the same jump condition as $\mathbf{P}(\Lambda; \chi , \tau)$,  then the new defined matrices $$\mathbf{U}(\Lambda; \chi , \tau)=\mathbf{P}_{\chi }(\Lambda; \chi , \tau)\mathbf{P}(\Lambda; \chi , \tau)^{-1} \quad \text{and} \quad \mathbf{V}(\Lambda; \chi , \tau)=\mathbf{P}_{\tau}(\Lambda; \chi , \tau)\mathbf{P}(\Lambda; \chi , \tau)^{-1},$$
are the entire functions in the whole complex plane $\mathbb{C}$. By the Laurent series in Eq. \eqref{eq:Laurent-series}, and the Liouville theorem, we have
\begin{equation}\label{eq:new-laxpair}
\begin{split}
\mathbf{U}(\Lambda; \chi , \tau)&=-\ii\Lambda\pmb{\sigma}_3+\ii[\pmb{\sigma}_3, \mathbf{N}^{[1]}], \\
\mathbf{V}(\Lambda; \chi , \tau)&=-\ii\Lambda^2\pmb{\sigma}_3+\ii\Lambda[\pmb{\sigma}_3, \mathbf{N}^{[1]}]-\mathbf{N}_{\chi }^{[1]}. \\
\end{split}
\end{equation}

Then we consider the symmetric properties for $\mathbf{U}(\Lambda; \chi , \tau)$ and $\mathbf{V}(\Lambda; \chi , \tau)$. It is readily to verify that $[\mathbf{P}^{\dag}(\lambda^*;\chi ,\tau)]^{-1}$ satisfies the same jump condition on the contour $\partial D_0$ and the normalization when $\Lambda\to\infty$ as $\mathbf{P}(\lambda;\chi ,\tau)$. Thus by the uniqueness, we have $\mathbf{P}(\lambda;\chi ,\tau)=[\mathbf{P}^{\dag}(\lambda^*;\chi ,\tau)]^{-1}$, which implies that $\mathbf{U}(\Lambda; \chi , \tau)$ satisfy the symmetry $\mathbf{U}^{\dagger}(\Lambda^*; \chi , \tau)=-\mathbf{U}(\Lambda; \chi , \tau),$ and $\mathbf{V}^{\dagger}(\Lambda^*; \chi , \tau)=-\mathbf{V}(\Lambda; \chi , \tau)$. Finally,  the zero curvature condition $\mathbf{U}_\tau-\mathbf{V}_\chi +[\mathbf{U},\mathbf{V}]=0$ infers that $q_1(\chi ,\tau)$ and $q_2(\chi ,\tau)$ satisfy the coupled NLS equation.  %Then the matrix $\widetilde{\mathbf{U}}(\Lambda; X, T)$ and $\widetilde{\mathbf{V}}(\Lambda; X, T)$ equals to the $\mathbf{U}(\lambda; x, t)$ and $\mathbf{V}(\lambda; x, t)$ in Eq. \eqref{eq:laxpair} with $\mathbf{q}(x,t)\to \mathbf{q}(X, T), \lambda\to\Lambda$.

%\lm{Add the proposition on the error between large order solitons and infinite order solitons.}

\subsection{Differential equations}
To simplify the RHP \ref{RHP2} further, it is necessary to reformulate the RHP \ref{RHP2} to a new one. %union the factor $\ee^{\pm \ii(\Lambda \chi +2\Lambda^2T)\pmb{\sigma}_3}$ and $\widetilde{\mathbf{J}}(\Lambda)$.
Introduce the new matrix $\widetilde{\mathbf{N}}(\Lambda; \chi , \tau)$ as
\begin{equation}
\widetilde{\mathbf{N}}(\Lambda; \chi , \tau):=\left\{\begin{split}&\mathbf{Q}_{H}^{\dagger}\mathbf{N}(\Lambda; \chi , \tau)\mathbf{Q}_{H}\ee^{-\ii\left(\Lambda \chi +\Lambda^2\tau\right)\pmb{\sigma}_3}\mathbf{Q}_d\mathbf{C}_{d}\ee^{\ii\left(\Lambda \chi +\Lambda^2\tau\right)\pmb{\sigma}_3}, \qquad& \Lambda \, \text{in the interior of $D_0$},\\
&\mathbf{Q}_{H}^{\dagger}\mathbf{N}(\Lambda; \chi , \tau)\mathbf{Q}_{H}\widetilde{\mathbf{J}}^{-1}(\Lambda),\qquad & \Lambda \, \text{exterior of $D_0$}.
\end{split}\right.
\end{equation}

Obviously, the new matrix $\widetilde{\mathbf{N}}(\Lambda; \chi , \tau)$ is analytic in the interior and exterior of $D_0$, %$\mathbf{\widetilde{J}}(\Lambda)\to\mathbb{I}$,
which satisfies the following RHP:
\begin{rhp}\label{RHP3}
	Let $(\chi , \tau)\in \mathbb{R}^2,$ %$\lambda_1=\ii,$
	find a $3\times 3$ matrix $\widetilde{\mathbf{N}}(\Lambda; \chi , \tau)$ with the following properties.
	\begin{itemize}
		\item \textbf{Analyticity}: $\widetilde{\mathbf{N}}(\Lambda; \chi , \tau)$ is analytic when $\Lambda\in \mathbb{C}\setminus \partial D_0$, it takes the continuous boundary condition from the interior and the exterior of $D_0$.
		\item \textbf{Jump condition}: The boundary condition on $\partial D_0$ (clockwise orientation) are related with the following jump condition
		\begin{equation}
		\begin{split}
		\widetilde{\mathbf{N}}_{+}(\Lambda; \chi , \tau)=\widetilde{\mathbf{N}}_{-}(\Lambda; \chi , \tau)\ee^{-\ii\left(\Lambda \chi +\Lambda^2\tau+\Lambda^{-1}\right)\pmb{\sigma}_3}\widehat{\mathbf{Q}}_d^{-1}\ee^{\ii\left(\Lambda \chi +\Lambda^2\tau+\Lambda^{-1}\right)\pmb{\sigma}_3},
		\end{split}
		\end{equation}
		\item \textbf{Normalization}: $\widetilde{\mathbf{N}}(\Lambda; \chi , \tau)\to \mathbb{I}$ as $\Lambda\to \infty.$
	\end{itemize}
\end{rhp}
Combing with the fact $\widetilde{\mathbf{J}}(\lambda)\to \mathbb{I}$ as $\Lambda\to \infty$, the potential function can be recovered in a similar formula as Eq.\eqref{eq:new-potential}:
\begin{equation}\label{eq:new-potential-1}
\mathbf{q}(\chi , \tau):=2\lim\limits_{\Lambda\to\infty}\Lambda\left(\mathbf{Q}_{H}\widetilde{\mathbf{N}}\mathbf{Q}_{H}^{-1}\right)_{1,2}(\Lambda; \chi , \tau).
\end{equation}

Now we proceed to consider the differential equations derived from the RHP \ref{RHP3}. Similar to the formula in Eq.\eqref{eq:P-expression}, we set
\begin{equation}
\widetilde{\mathbf{P}}(\Lambda; \chi , \tau):=\widetilde{\mathbf{N}}(\Lambda; \chi , \tau)\ee^{-\ii\left(\Lambda \chi +\Lambda ^2\tau+\Lambda^{-1}\right)\pmb{\sigma}_3},
\end{equation}
then the jump matrix satisfied by $\widetilde{\mathbf{P}}(\Lambda; \chi , \tau)$ is independent on the variable $\chi $ and $\tau$,
\begin{equation*}
\widetilde{\mathbf{P}}_+(\Lambda; \chi , \tau)=\widetilde{\mathbf{P}}_-(\Lambda; \chi , \tau)\widehat{\mathbf{Q}}^{-1}_d.
\end{equation*}
Take $\widetilde{\mathbf{N}}(\Lambda; \chi , \tau)$ has the Laurent series expansion as
\begin{equation}\label{eq:Laurent-series-p}
\widetilde{\mathbf{N}}(\Lambda; \chi , \tau)=\mathbb{I}+\sum\limits_{l=1}^{\infty}\widetilde{\mathbf{N}}^{[l]}\Lambda^{-l}, \qquad \Lambda\to\infty.
\end{equation}
Based on the idea to derive the new Lax pair in Eq.\eqref{eq:new-laxpair}, the new matrix $\widetilde{\mathbf{P}}(\Lambda; \chi , \tau)$ satisfy
\begin{equation}\label{eq:new-laxpair-1}
\begin{split}
\frac{\partial \widetilde{\mathbf{P}}}{\partial \chi }(\Lambda; \chi , \tau)
&=\widetilde{\mathbf{U}}(\Lambda; \chi , \tau)\widetilde{\mathbf{P}}(\Lambda; \chi , \tau),\\
\frac{\partial \widetilde{\mathbf{P}}}{\partial \tau}(\Lambda; \chi , \tau)
&=\widetilde{\mathbf{V}}(\Lambda; \chi , \tau)\widetilde{\mathbf{P}}(\Lambda; \chi , \tau),\\
\end{split}
\end{equation}
where
\begin{equation}
\begin{split}
\widetilde{\mathbf{U}}(\Lambda; \chi , \tau)&=\widetilde{\mathbf{U}}^{[1]}\Lambda+\widetilde{\mathbf{U}}^{[0]}, \quad \widetilde{\mathbf{U}}^{[1]}=-\ii\pmb{\sigma}_3, \quad \widetilde{\mathbf{U}}^{[0]}=\begin{bmatrix}
0&\ii q_1(\chi,\tau)&0\\
\ii q_{1}^{*}(\chi,\tau)&0&0\\
0&0&0
\end{bmatrix}\\
\widetilde{\mathbf{V}}(\Lambda; \chi , \tau)&=\widetilde{\mathbf{V}}^{[2]}\Lambda^2+\widetilde{\mathbf{V}}^{[1]}\Lambda
+\widetilde{\mathbf{V}}^{[0]}, \quad \widetilde{\mathbf{V}}^{[2]}=-\ii\pmb{\sigma}_3, \\
\widetilde{\mathbf{V}}^{[1]}&=\widetilde{\mathbf{U}}^{[0]}, \quad \widetilde{\mathbf{V}}^{[0]}=\begin{bmatrix}
\ii\frac{|q_1(\chi,\tau)|^2}{2}&-\frac{p_{1}(\chi,\tau)}{2}&0\\
\frac{p_{1}^*(\chi,\tau)}{2}&-\ii\frac{|q_1(\chi,\tau)|^2}{2}&0\\
0&0&0
\end{bmatrix}
\end{split}
\end{equation}
and $\mathbf{q}(\chi , t)=2\widetilde{\mathbf{N}}^{[1]}_{1,2}$ are the functions with respect to $\chi $ and $\tau$.

It can be seen the jump for $\widetilde{\mathbf{N}}(\Lambda; \chi , \tau)$ is independent on the spectral parameter $\Lambda$, thus the matrix defined by
\begin{equation}\label{eq:eq:L-laxpair}
\mathbf{L}(\Lambda; \chi , \tau):=\frac{\partial\widetilde{\mathbf{P}}(\Lambda; \chi , \tau)}{\partial \Lambda}\widetilde{\mathbf{P}}^{-1}(\Lambda; \chi , \tau)
\end{equation}
is also analytic when $\Lambda\in\partial D_0$, while it has an isolated singular at $\Lambda=0$ because of the factor $\ee^{\pm\ii\Lambda^{-1}\pmb{\sigma}_3}$ appeared the phase term. So we can expand $\mathbf{L}(\Lambda; \chi, \tau)$ at $\Lambda\to\infty$ and $\Lambda=0$ respectively.

By the Laurent series Eq.\eqref{eq:Laurent-series-p}, we know the $\mathbf{L}(\Lambda; \chi , \tau)$ can be written as
\begin{equation}\label{eq:L-expression}
\mathbf{L}(\Lambda; \chi , \tau)=\mathbf{L}^{[1]}\Lambda+\mathbf{L}^{[0]}+\mathbf{L}^{[-1]}\Lambda^{-1}+\mathbf{L}^{[-2]}\Lambda^{-2}+\mathcal{O}(\Lambda^{-3}), \qquad \Lambda\to \infty,
\end{equation}
where \begin{equation*}
\mathbf{L}^{[1]}=-2\ii \tau\pmb{\sigma}_3, \quad \mathbf{L}^{[0]}=-\ii \chi \pmb{\sigma}_3+2\tau\widetilde{\mathbf{U}}^{[0]}, \quad \mathbf{L}^{[-1]}=2\tau\widetilde{\mathbf{V}}^{[0]}+\chi \widetilde{\mathbf{U}}^{[0]}.
\end{equation*}
Similarly, $\mathbf{L}(\Lambda; \chi , \tau)$ can also be expanded at $\Lambda=0$:
\begin{equation}\label{eq:L-expression-0}
\mathbf{L}(\Lambda; \chi , \tau)=\ii \widetilde{\mathbf{P}}(0; \chi , T)\pmb{\sigma}_3\widetilde{\mathbf{P}}(0; \chi , T)^{-1}\Lambda^{-2}+O(\Lambda^{-1}):=\mathbf{L}^{[-2]}\Lambda^{-2}+O(\Lambda^{-1}),
\end{equation}
together with the symmetry $\widetilde{\mathbf{P}}^{\dagger}(\Lambda^*; \chi , T)=\widetilde{\mathbf{P}}^{-1}(\Lambda; \chi , \tau)$, so $\mathbf{L}^{[-2]}$ can be rewritten in the form:
\begin{equation}
\mathbf{L}^{[-2]}(\chi , T)=\begin{bmatrix}\ii a(\chi,\tau)-\ii&\ii b(\chi,\tau)&0\\\ii b^*(\chi,\tau)&-\ii a(\chi,\tau)-\ii&0\\
0&0&-\ii
\end{bmatrix}
\end{equation}
and $\det(\mathbf{L}^{[-2]}(\chi , \tau))=-\ii, {\rm tr}(\mathbf{L}^{[-2]}(\chi , \tau))=-\ii, $ which means $|a(\chi,\tau)|^2+|b(\chi,\tau)|^2=2.$

Combining \eqref{eq:L-expression} with \eqref{eq:L-expression-0} shows that $\mathbf{L}(\Lambda;\chi,\tau)$ can be represented into the Laurent polynomial
\begin{equation}
\mathbf{L}(\Lambda;\chi,\tau)=\mathbf{L}^{[1]}\Lambda+\mathbf{L}^{[0]}+\mathbf{L}^{[-1]}\Lambda^{-1}+\mathbf{L}^{[-2]}\Lambda^{-2}
\end{equation}
and the equation \eqref{eq:eq:L-laxpair} can be reinterpreted as the Lax system:
\begin{equation}\label{eq:L-Lmabda}
\frac{\partial\widetilde{\mathbf{P}}(\Lambda; \chi , \tau)}{\partial \Lambda}=\mathbf{L}(\Lambda;\chi,\tau)\widetilde{\mathbf{P}}(\Lambda; \chi , \tau).
\end{equation}

\subsection{Ordinary differential equation in $\chi$ and $\tau$ and Painlev\'e-III hierarchy}
Since $\widetilde{\mathbf{P}}(\Lambda; \chi , \tau)$ satisfies the Lax pair Eq.\eqref{eq:new-laxpair-1} and Eq.\eqref{eq:L-Lmabda} simultaneously, so the coefficient $\widetilde{\mathbf{U}}(\Lambda; \chi , \tau)$ and $\mathbf{L}(\Lambda; \chi , \tau)$ will give a zero-curvature condition
$$\mathbf{L}_{\chi }-\widetilde{\mathbf{U}}_{\Lambda}
+\left[\mathbf{L},\widetilde{\mathbf{U}}\right]=0.$$ The coefficients
of $\Lambda^{j}$ $(j=1,2,0,-1,-2)$ in the left-hand side should be zero to match the one in the right hand side. By calculation, we know the first nonzero term appears in the off-diagonal of the coefficient $\Lambda^{0}$, which equals to
\begin{equation}
\mathbf{L}^{[0]}_{\chi }+\left[\mathbf{L}^{[0]},
\widetilde{\mathbf{U}}^{[0]}\right]+\left[\mathbf{L}^{[-1]}, -\ii\pmb{\sigma}_3\right]+\ii\pmb{\sigma}_3=0,
\end{equation}
the diagonal of which gives no information, and the off-diagonal part gives
\begin{equation}
2\ii \tau q_{1,\chi}-2\ii\tau p_1=0,\qquad 2\ii \tau q_{1, \chi}^*-2\ii\tau p_1^*=0.
\end{equation}
The coefficient of $\Lambda^{-1}$ gives the equation
\begin{equation}\label{eq:AL0}
\mathbf{L}^{[-1]}_{\chi }+\left[\mathbf{L}^{[-1]},
\widetilde{\mathbf{U}}^{[0]}\right]+\left[\mathbf{L}^{[-2]}, -\ii\pmb{\sigma}_3\right]=0.
\end{equation}
The diagonal elements produce the same identity as Eq.\eqref{eq:AL0}, and the off-diagonal elements give the following equation:
\begin{equation}\label{eq:AL1}
\begin{split}
&\ii \chi q_{1, \chi }-\tau p_{1,\chi }-2\tau |q_1|^2q_1+\ii q_1-2b=0,\\
&\ii \chi q^*_{1, \chi }+\tau p^{*}_{1, \chi}+2\tau |q_1|^2q_1^{*}+\ii q_1^*+2b^*=0.\\
\end{split}
\end{equation}
Similarly, the coefficient of $\Lambda^{-2}$ gives the following three equations:
\begin{equation}\label{eq:AL2}
\begin{split}
&a_{\chi }+\ii q_1^* b-\ii q_1b^*=0,\qquad b_{\chi }+2\ii aq_1=0, \qquad b^*_{\chi }-2\ii aq_1^*=0.
\end{split}
\end{equation}
If $\tau=0,$ Eq.\eqref{eq:AL1} reduces to
\begin{equation*}
\begin{split}
\chi  q_{1, \chi }&=-2\ii b-q_1
\end{split}
\end{equation*}
as well as its conjugate equation. Moreover, the phase or exponent term in the RHP \ref{RHP3} is $\ii\left(\Lambda \chi +\Lambda^{-1}\right),$ which is related to the Painlev\'{e} equation by a rescaled transformation. Based on method in \cite{Joshi-JMP-2007} and \cite{Peter-Duke-2019}, we introduce the rescaled transformation
\begin{equation*}
\begin{split}
\chi &=-\frac{x^2}{4}, \quad \Lambda=\frac{2}{x}\lambda, \quad q_1(\chi )=\frac{-\ii Q_1(x)}{x}, \quad
a(\chi )=A(x), \quad B(\chi )=B(x),
\end{split}
\end{equation*}
it follows that the new Lax pair depending on $x, \lambda, Q_1$ is given by
\begin{equation}\label{eq:laxpairhatAL}
\frac{\partial \widetilde{\mathbf{P}}}{\partial\lambda}
=\widehat{\mathbf{L}}(\lambda; x)\widetilde{\mathbf{P}}, \qquad \text{and} \qquad \frac{\partial \widetilde{\mathbf{P}}}{\partial x}
=\widehat{\mathbf{A}}(\lambda; x)\widetilde{\mathbf{P}},
\end{equation}
where
\begin{equation*}
\begin{split}
\widehat{L}(\lambda; x)&=\frac{\ii x}{2}\pmb{\sigma}_3+\frac{1}{\lambda}\begin{bmatrix}0&Q_1&0\\
-Q_1^*&0&0\\
0&0&0
\end{bmatrix}+\frac{x}{\lambda^2}\begin{bmatrix}\frac{\ii \left(A-1\right)}{2}&\frac{\ii B}{2}&0\\
\frac{\ii B^*}{2}&\frac{-\ii \left(A+1\right)}{2}&0\\0&0&-\frac{\ii}{2}
\end{bmatrix}\\
\widehat{A}(\lambda; x)&=\frac{\ii\lambda}{2}\pmb{\sigma}_3+\frac{1}{x}
\begin{bmatrix}
0&Q_1&0\\
-Q_1^*&0&0
\\0&0&0
\end{bmatrix}-\frac{1}{\lambda}
\begin{bmatrix}\frac{\ii \left(A-1\right)}{2}&\frac{\ii B}{2}&0\\
\frac{\ii B^*}{2}&\frac{-\ii \left(A+1\right)}{2}&0\\0&0&-\frac{\ii}{2}
\end{bmatrix}.
\end{split}
\end{equation*}
Then the compatibility condition of Eq.\eqref{eq:laxpairhatAL} gives the following equations
\begin{equation}
\begin{split}
Q_{1, x}&=-B x,\qquad x A_{x}=2\left(Q_1B^*+Q_1^*B\right),\qquad x B_{x}=-4A Q_1,
\end{split}
\end{equation}
which induce two integral equations,
\begin{equation}
\begin{split}
|A|^2+|B|^2=C_1=2,\qquad BQ_1^*-B^*Q_1=C_2,
\end{split}
\end{equation}
where the constants $C_i$ $(i=1,2)$ are the integrals of motion. Introducing the new variable $y_1$ as $y_1=Q_1/(xB)$,
then the function $y_1$ will satisfy the Painlev\'{e}-III equation
\begin{equation}
y_{1, xx}-\frac{y_{1,x}^2}{y_1}+\frac{1}{x}y_{1, x}-\frac{1}{x}\left(\alpha y_1^2+\beta\right)-\gamma y_1^3-\frac{\delta}{y_1}=0,
\end{equation}
where$$\delta=-1, \beta=-2, \alpha=8C_2, \gamma=8C_1^2.$$
\ys{This kind of Painlev\'{e} is similar to the one in \cite{Peter-Duke-2019}, but the parameters $\alpha, \beta, \gamma$ and $\delta$ have a different choice.}
Furthermore, we can derive an ordinary equation with respect to $\tau$. The corresponding compatibility condition for the second equation of Eq.\eqref{eq:new-laxpair-1} and Eq.\eqref{eq:L-Lmabda} gives $\mathbf{L}_{\tau}-\widetilde{\mathbf{V}}_{\Lambda}+[\mathbf{L},\widetilde{\mathbf{V}}]=0.$ It can be seen the first nonzero term appears in the coefficient of $\Lambda^{0}$, which gives the following equations:
\begin{equation}\label{eq:lambda0}
\begin{split}
&2 \tau q_{1, \tau}+2\ii b+\chi p_1+q_1=0,\\
&2 \tau q_{1, \tau}^*-2\ii b^*+\chi p_1^*+q_1^*=0.
\end{split}
\end{equation}
Moreover, the coefficient of $\Lambda^{-1}$ gives the following three equations:
\begin{equation}\label{eq:lambdan1}
\begin{split}
&\frac{1}{2}\ii \chi \left(q_1p_1^*{+}q_1^*p_1\right){+}\ii \tau|q_1|^2_{\tau}{+}\ii|q_1|^2
{+}q_1b^*{-}q_1^*b=0,\\
&\ii \chi  q_{1, \tau}-\tau p_{1, \tau}-p_{1}-2aq_1+\chi q_1|q_1|^2=0,\\
&\ii \chi  q_{1, \tau}^*+\tau p_{1, \tau}^*+p_{1}^*+2aq_{1}^*-\chi q_{1}^*|q_1|^2=0.
\end{split}
\end{equation}
The terms for the coefficient of $\Lambda^{-2}$ give the equations
\begin{equation}\label{eq:lambdan2}
\begin{split}
&p_{1}b^*+p^*_{1}b+2a_{\tau}=0,\\
&-\ii a p_{1}+\ii b_{\tau}+b|q_1|^2=0,\\
&-\ii a p_{1}^*+\ii b^*_{\tau}-b^*|q_1|^2=0.
\end{split}
\end{equation}
From the Eq.\eqref{eq:lambda0}, we know
\begin{equation}\label{eq:w1w2w3}
b=\ii \tau q_{1, \tau}+\frac{\ii \chi }{2}p_{1}+\frac{\ii}{2}q_1, \quad b^*=-\ii\tau q^*_{1,\tau}-\frac{1}{2}\ii\chi p_1^*-\frac{1}{2}\ii q_1^*.
\end{equation}
Substituting Eq.\eqref{eq:w1w2w3} into Eq.\eqref{eq:lambdan2}, then the ordinary equations \eqref{eq:lambdan2} change into
\begin{equation}\label{eq:lambdan1-new}
\begin{split}
&a_{\tau}+\frac{1}{2}\ii\left(p_1^*q_{1,\tau}-p_1q_{1,\tau}^*\right)\tau
+\frac{1}{4}\ii\left(q_1p_1^*-q_1^*p_1\right)=0,\\
&2\ii \tau q_{1,\tau\tau}+3\ii q_{1,\tau}+2\tau q_{1,\tau}|q_1|^2+\chi\left(\ii p_{1,\tau}+p_1|q_1|^2\right)+q_1|q_1|^2-2ap_1=0,\\
&2\ii \tau q_{1,\tau\tau}^*+3\ii q_{1,\tau}^*-2\tau q_{1,\tau}^*|q_1|^2+\chi\left(\ii p_{1,\tau}^*-p_1^*|q_1|^2\right)-q_1^*|q_1|^2+2ap_1^*=0,\\
\end{split}
\end{equation}
which is consistent with the identity $a^2+|b|^2=2$. \ys{This ordinary equation can be reduced to the Eq.(119) in \cite{Peter-Duke-2019} with $q_1\to\ii \Psi, p_1\to 2\Phi, \tau\to T$. }

\subsection{Asymptotic behavior of infinite order soliton for large $\chi$}
\label{sec:large-chi}
In the previous subsection, we have derived the RHP \ref{RHP3} for the infinite order soliton, which can be used to analyze the asymptotic behavior for the solutions as $\chi$ and $\tau$ is large. We firstly begin to study the asymptotics for large $\chi$. Taking the rescaled transformation between $\chi, \Lambda$ and $\chi$:
\begin{equation}\label{eq:scaleX}
\chi=\sigma |\chi|, \tau=v|\chi|^{3/2}, \Lambda=\chi^{-1/2}z,
\end{equation}
then the phase term  changes into
\begin{equation}\label{eq:phase-term}
\Lambda \chi+\Lambda^2\tau+\Lambda^{-1}=|\chi|^{1/2}\left(\sigma z+vz^2+z^{-1}\right),
\end{equation}
where $\sigma=\pm1$, the positive sign indicates $\chi>0$ and the negative one is $\chi<0$. In view of the symmetric property for the high order soliton, we only consider the case of $\sigma=1.$  The case of $\sigma=-1$ can be analyzed similarly. Defining $\mathbf{S}(z; \chi, v)=\widetilde{\mathbf{N}}(\chi^{-1/2}z; \chi, \chi^{3/2}v)$, then the jump condition in the RHP \ref{RHP3} changes into
\begin{equation}\label{eq:near-field-S}
\mathbf{S}_{+}(z; \chi, v)=\mathbf{S} _-(z; \chi, v)\ee^{-\chi^{1/2}\vartheta(z; v)\pmb{\sigma}_3}\widehat{\mathbf{Q}}_d^{-1}\ee^{\chi^{1/2}\vartheta(z; v)\pmb{\sigma}_3}
\end{equation}
where $\vartheta(z; v)=\ii\left(z+vz^2+z^{-1}\right)$,
and the potential function $\mathbf{q}(\chi, \chi^{3/2}v)$ can be recovered as
\begin{equation}
\mathbf{q}(\chi, \chi^{3/2}v)=2\lim\limits_{z\to \infty}z\chi^{-1/2}\left(\mathbf{Q}_{H}\mathbf{S}(z; \chi, v)\mathbf{Q}_{H}^{\dagger}\right) _{1, 2}(z; \chi, v).
\end{equation}
With the aid of a simple calculation, we find if $|v|<\sqrt{3}/{9}$, the critical points of $\vartheta(z; v)$ have three distinct real numbers; and if $|v|>\sqrt{3}/{9}$, the critical points of $\vartheta(z; v)$ involve one real number and a pair of conjugate roots. Specially, when $|v|=\sqrt{3}/{9},$ it has a double root. Now we begin to discuss the asymptotic behavior under $|v|<\sqrt{3}/{9}$.
\begin{figure}[!h]
\centering
\includegraphics[width=0.45\textwidth]{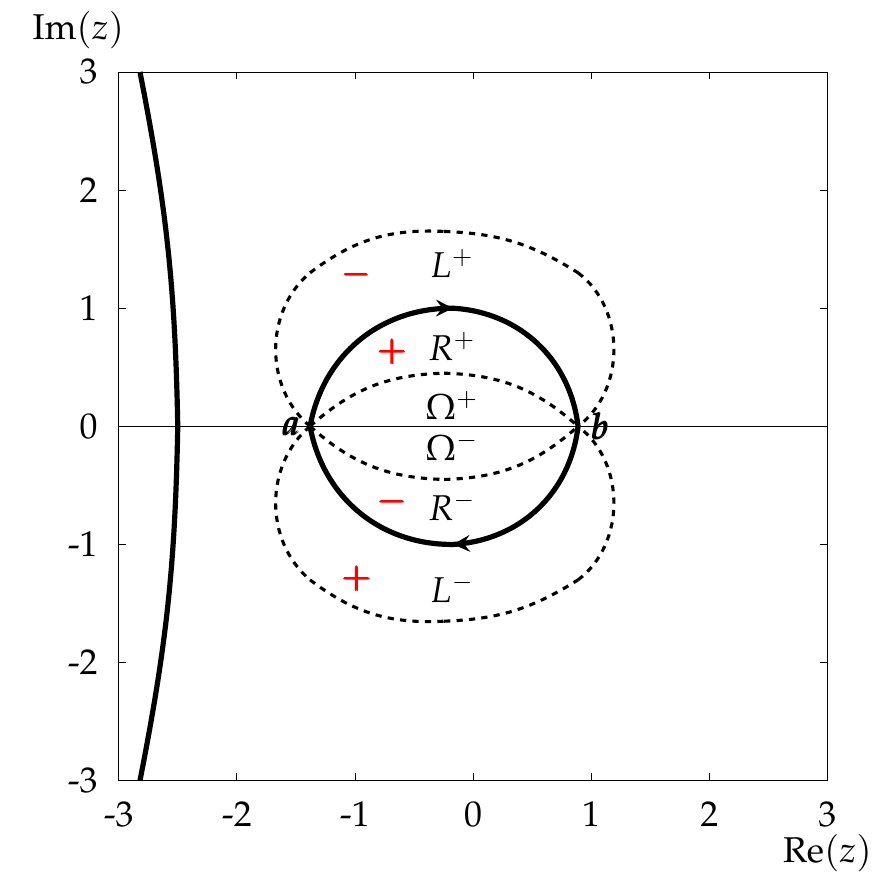}
\centering
\includegraphics[width=0.45\textwidth]{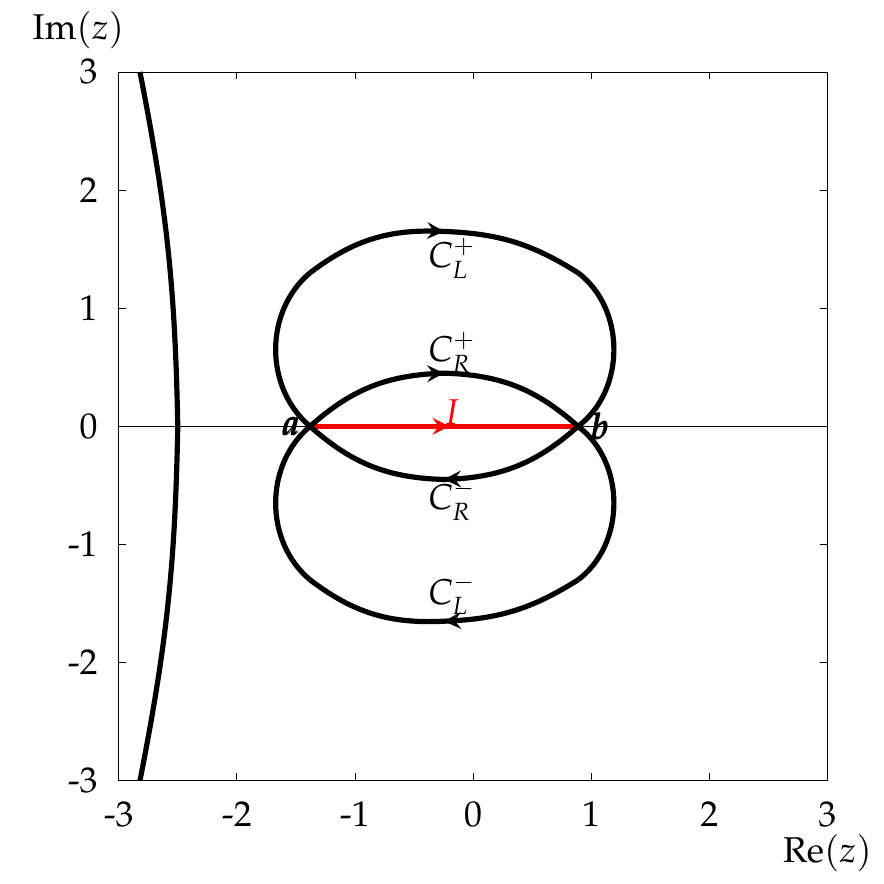}
\caption{The sign of ${\rm Re}(\vartheta(z, v))$ with the choice of $v=\frac{1}{6}$. The left is the original contour about the matrix $\mathbf{S}(z, v)$. The right one is the deformation contour, which is the jump for $\mathbf{R}(z; v)$.}
\label{contour-near}
\end{figure}
From the contour plot for signature chart of $\Re(\vartheta(z; v))$, we can choose the original contour with the circle involved the origin. To study the asymptotics when $\chi$ is large, we still use the classical steepest descent method. One of the key step of this method is to decompose the jump matrix into upper and lower triangular matrix. In lemma\ref{lem:qc-decom}, the matrix $\mathbf{Q}_c$ has been decomposed into three block matrices, which converts $\mathbf{Q}_c$ into the diagonal block $\mathbf{Q}_{d}$ and further $\widehat{\mathbf{Q}}_{d}$. In this way, the $3\times 3$ diagonal block matrix $\widehat{\mathbf{Q}}_d$ can be readily performed the LU decomposition, which can be summarized in the following lemma:
Based on the sign of ${\rm Re}(\vartheta(z; v))=0$ in Fig.\ref{contour-near}, we can give the first contour deformation.
\begin{equation}\label{eq:R-near-field}
\begin{aligned}
\mathbf{R}(z; \chi, v):&=\mathbf{S} (z; \chi, v)\ee^{-\text{ad}_{\pmb{\sigma}_3}\vartheta(z; v)}\left(\mathbf{Q}_{R}^{[2]}\right)^{-1} ,\quad &z\in L^+,\\
\mathbf{R}(z; \chi, v):&=\mathbf{S} (z; \chi, v)\mathbf{Q}_{L}^{[2]}\ee^{- \text{ad}_{\pmb{\sigma}_3}\vartheta(z; v)}\mathbf{Q}_{C}^{[2]}, \quad &z\in R^+,\\
\mathbf{R}(z; \chi, v):&=\mathbf{S}(z; \chi, v)\mathbf{Q}_{L}^{[2]}\quad &z\in \Omega^+,\\
\mathbf{R}(z; \chi, v):&=\mathbf{S} (z; \chi, v)\mathbf{Q}_{L}^{[1]}, \quad &z\in \Omega^-,\\
\mathbf{R}(z; \chi, v):&=\mathbf{S} (z; \chi, v)\mathbf{Q}_{L}^{[1]}\ee^{- \text{ad}_{\pmb{\sigma}_3}\vartheta(z; v)}\mathbf{Q}_{C}^{[1]},\quad &z\in R^-,\\
\mathbf{R}(z; \chi, v):&=\mathbf{S} (z; \chi, v)\ee^{- \text{ad}_{\pmb{\sigma}_3}\vartheta(z; v)}\left(\mathbf{Q}_{R}^{[1]}\right)^{-1},\quad &z\in L^-,
\end{aligned}
\end{equation}
in other domain, we set $\mathbf{R}(z; \chi, v):=\mathbf{S}(z; \chi, v).$ Then the jump matrix about $\mathbf{R}(z; \chi, v)$ changes into
\begin{equation}\label{eq:R-jump}
\begin{split}
\mathbf{R}_+(z; \chi, v)&=\mathbf{R}_{-}(z; \chi, v)\begin{bmatrix}1&0&0\\-\ee^{2 \chi^{1/2}\vartheta(z; v)}&1&0\\
0&0&1
\end{bmatrix}, z\in C_{L}^+,\\
\mathbf{R}_+(z; \chi, v)&=\mathbf{R}_{-}(z; \chi, v)\begin{bmatrix}1&\frac{1}{2}\ee^{-2 \chi^{1/2}\vartheta(z; v)}&0\\0&1&0\\0&0&1
\end{bmatrix}, z\in C_{R}^+,\\
\mathbf{R}_+(z; \chi, v)&=\mathbf{R}_-(z; \chi, v){\rm diag}\left(2, \frac{1}{2}, 1\right), z\in I,\\
\mathbf{R}_+(z; \chi, v)&=\mathbf{R}_{-}(z; \chi, v)\begin{bmatrix}1&0&0\\-\frac{1}{2}\ee^{2 \chi^{1/2}\vartheta(z; v)}&1&0\\0&0&1
\end{bmatrix}, z\in C_{R}^-,\\
\mathbf{R}_+(z; \chi, v)&=\mathbf{R}_{-}(z; \chi, v)\begin{bmatrix}1&\ee^{-2 \chi^{1/2}\vartheta(z; v)}&0\\0&1&0\\0&0&1
\end{bmatrix}, z\in C_{L}^-.
\end{split}
\end{equation}
When $\chi\to\infty$, the jump matrices all approach to the identity except the jump matrix in the contour $I$, which is marked in red in Fig.\ref{contour-near}. Then we can construct the parametrix matrix to match $\mathbf{R}(z; \chi, v)$, which is shown in Appendix \ref{App:large-chi}.

Then the asymptotics about $\mathbf{q}(\chi; \chi^{3/2}v)$ can be given as
%\begin{equation}
%\begin{split}
%q_1=&-c_i^*\frac{\ii}{\chi^{3/4}}r\ee^{-2\ii \chi^{1/2}\vartheta(b(v); v)}\chi^{\frac{\ii}{2}p}(b(v){-}a(v))^{2\ii p}\left(\sqrt{-\ii\vartheta''(b(v); v)}\right)^{2\ii p}\left(\sqrt{-\ii\vartheta''(b(v);v)}\right)^{-1}\\
%&-c_i^*\frac{\ii}{\chi^{3/4}}r^*\ee^{-2\ii \chi^{1/2}\vartheta(a(v); v)}\chi^{-\frac{\ii}{2}p}(b(v){-}a(v))^{-2\ii p}\left(\sqrt{\ii\vartheta''(a(v); v)}\right)^{-2\ii p}\left(\sqrt{\ii \vartheta''(a(v); v)}\right)^{-1}+O(\chi^{-1})\\
%\end{split}
%\end{equation}
\begin{equation}\label{eq:q-near}
\begin{split}
q_i(\chi, \chi^{3/2}v){=}&\frac{{-}\ii c_i^*}{\chi^{3/4}}\sqrt{\frac{\ln(2)}{\pi}}\Bigg(\ee^{{-}2\ii \chi^{1/2}\vartheta(b(v); v){+}\ii\phi_{n}}\left(\sqrt{{-}\ii \vartheta''(b(v); v)}\right)^{2\ii p{-}1}\\&{+}\ee^{{-}2\ii \chi^{1/2}\vartheta(a(v); v){-}\ii\phi_{n}}\left(\sqrt{\ii\vartheta''(a(v); v)}\right)^{{-}2\ii p{-}1}\Bigg)+\mathcal{O}(\chi^{-1}), \,\,\,\,\,(i=1,2).
\end{split}
\end{equation}
\subsection{Asymptotic behavior as $\tau$ is large}
\label{sec:large-tau}
In the last subsection, we give the asymptotics when $\chi$ is large under the transformation Eq.\eqref{eq:scaleX}. If $|v|<{\sqrt{3}}/{9}$, then $\vartheta(z; v)$ has three real critical points, whose contour plot about ${\rm Re}(\vartheta(z; v))=0$ is shown in Fig.\ref{contour-near}. In this subsection, we continue this analysis on the region $|v|>{\sqrt{3}}/{9}$, namely, the asymptotics for large $\tau$. In this case, we still suppose $\chi>0$ and consider large $\tau>0$. For convenience, we introduce rescaled transformation:
\begin{equation}
\chi=w\tau^{2/3}, \qquad \Lambda=\tau^{-1/3}z
\end{equation}
where $w=v^{-2/3}$, then the constraint about $w$ should be $0\leq w<3$. Under this transformation, the phase term changes into
\begin{equation}
\Lambda \chi+\Lambda^2\tau+\Lambda^{-1}=\tau^{1/3}\left(wz+z^2+z^{-1}\right).
\end{equation}
Similarly, defining $\mathbf{S}(z; \tau, w)=\widetilde{\mathbf{N}}(\tau^{-1/3}z; \tau^{2/3}w, \tau)$, then the jump condition in the RHP \ref{RHP3} changes into
\begin{equation}
\mathbf{S}_{+}(z; \tau, w)=\mathbf{S} _-(z; \tau, w)\ee^{-\tau^{1/3}\varpi(z; w)\pmb{\sigma}_3}\widehat{\mathbf{Q}}_d^{-1}\ee^{\tau^{1/3}\varpi(z; w)\pmb{\sigma}_3}
\end{equation}
where $\varpi(z; w)=\ii\left(wz+z^2+z^{-1}\right)$,
and the potential function $\mathbf{q}(\tau^{2/3}w, \tau)$ can be recovered as
\begin{equation}
\mathbf{q}(\tau^{2/3}w, \tau)=2\lim\limits_{z\to \infty}z\tau^{-1/3}\left(\mathbf{Q}_{H}\mathbf{S}(z; \tau, w)\mathbf{Q}_{H}^{\dagger}\right) _{1, 2}(z; \tau, w).
\end{equation}

Compared to the asymptotics for large $\chi$, the asymptotics for large $\tau$ is different, we can not deform the original contour of $\Re(\varpi)=0$ any more. To study further, we shall give a new $\widetilde{g}(z)$ function satisfying the following RHP.
\begin{rhp}\label{RHP:large-tau}
For a fixed $0\leq w<3$, there exists a unique $\widetilde{g}(z)$ function satisfying the following conditions:
\begin{itemize}
\item
{\bf Analyticity}: $\widetilde{g}(z)$-function is analytic except on $\Sigma_{f}$(to be determined), and it takes the continuous boundary condition from the left and right hand side of $\Sigma_f$.
\item
{\bf Jump Condition}: For $\lambda\in\Sigma_f$, the $\widetilde{g}(z)$-function satisfies the following jump condition:
\begin{equation}\label{eq:jump-g-tilde}
\widetilde{g}_{+}(z)+\widetilde{g}_{-}(z)-2\varpi(z)=\widetilde{\Omega}_{n},\qquad z\in\Sigma_{f}.
\end{equation}
\item
{\bf Normalization}: As $z\to\infty$, $\widetilde{g}(z)$ satisfies
\begin{equation}
\widetilde{g}(z)\to \mathcal{O}(z^{-1}).
\end{equation}
{\bf Symmetry}: The $\widetilde{g}(z)$ has the symmetry
\begin{equation*}
\widetilde{g}(z)=-\widetilde{g}(z^*)^*.
\end{equation*}
\end{itemize}
\end{rhp}
For the above Riemann-Hilbert problem, we can determine the $\widetilde{g}(z)$-function and the contour $\Sigma_{f}$. Similar to the formula in Eq.\eqref{eq:g-function}, for $z\in\Sigma_{f}$, by taking the derivative of the jump condition \eqref{eq:jump-g-tilde}, we have
\begin{equation}\label{eq:wg-function}
\widetilde{g}'_{+}(z)+\widetilde{g}'_{-}(z)=2\ii w+4\ii z-\frac{2\ii}{z^2},\qquad z\in\Sigma_{f}.
\end{equation}
Then the function $\widetilde{g}'(z)$ can be solved by the Plemelj formula:
\begin{equation}
\widetilde{g}'(z)=\frac{\widetilde{R}(z)}{2\pi\ii}\int_{\Sigma_{f}}\frac{2\ii w+4\ii s- \frac{2\ii}{s^2}}{\widetilde{R}(s)(s-z)}ds,
\end{equation}
where $\widetilde{R}(z):=\left[\left(z-\widetilde{a}\right)\left(z-\widetilde{a}^*\right)\right]^{1/2}=\left[z^2-\widetilde{p}z+\modify{\widetilde{q}}\right]^{1/2}$, with $\widetilde{p}=\widetilde{a}+\widetilde{a}^*, \widetilde{q}=|\widetilde{a}|^2$. With the generalized residue theorem, $\widetilde{g}'(z)$ can be obtained as
\begin{equation}
\begin{split}
\widetilde{g}'(z)&=\modify{\widetilde{R}(z)}\left(\mathop{\rm Res}\limits_{s=z}+\mathop{\rm Res}\limits_{s=0}\ys{+}\mathop{\rm Res}\limits_{s=\infty}\right)\left(\frac{\ii w+2\ii s-\ii \frac{1}{s^2}}{\widetilde{R}(s)\left(s-z\right)}\right)\\
&=\left(\frac{\ii}{\widetilde{q}^{1/2}z^2}+\frac{\ii \widetilde{p}}{2\widetilde{q}^{3/2}z}\right)\widetilde{R}(z)-2\ii\widetilde{R}(z)+\ii w+2\ii z-\ii \frac{1}{z^2},
\end{split}
\end{equation}
then we have
\begin{equation}
\widetilde{g}'(z)-\varpi'(z)=\ii \left[\frac{1}{\widetilde{q}^{1/2}z^2}+\frac{\widetilde{p}}{2\widetilde{q}^{3/2}z}-2\right]\widetilde{R}(z).
\end{equation}
From the definition of $\widetilde{g}(z)$, we have $g'(z)=\mathcal{O}(z^{-2})$ when $z\to\infty$, which derive two equations about $\widetilde{p}$ and $\widetilde{q}$:
\begin{equation}\label{eq:wpands}
\begin{split}
&\mathcal{O}(1): \ii\left(\frac{\widetilde{p}}{2\widetilde{q}^{3/2}}+\widetilde{p}+w\right)\\
&\mathcal{O}(z^{-1}): \ii \frac{\left(\sqrt{\widetilde{q}}-1\right)\left(\widetilde{q}+\sqrt{\widetilde{q}}+1\right)\left(\widetilde{p}^2-4\widetilde{q}\right)}{4\widetilde{q}^{3/2}},
\end{split}
\end{equation}
from the second equation in Eq.\eqref{eq:wpands}, we can easily get $\widetilde{q}=1$ or $\widetilde{p}^2=4\widetilde{q}$. If $\widetilde{p}^2=4\widetilde{q}$, then $\widetilde{a}=\widetilde{a}^*$, which is not we wanted, thus $\widetilde{q}=1$. Substituting it into the first equation in Eq.\eqref{eq:wpands}, we have $\widetilde{p}=-\frac{2}{3}w$, then $\widetilde{a}$ and $\widetilde{a}^*$ can be given as
\begin{equation}
\widetilde{a}=-\frac{1}{3}w+\frac{\ii}{3}\sqrt{9-w^2}, \qquad \widetilde{a}^*=-\frac{1}{3}w-\frac{\ii}{3}\sqrt{9-w^2}.
\end{equation}
With this result, we can also calculate the other two roots of $g'(z)-\varpi'(z)$ as $z_{c}=-\frac{1}{12}w-\frac{1}{12}\sqrt{w^2+72}, \,\, z_{d}=-\frac{1}{12}w+\frac{1}{12}\sqrt{w^2+72}.$
Under this condition, we have
\begin{equation}
-\ii \frac{d}{dz}\frac{(3+2wz+3z^2)^{3/2}}{3\sqrt{3}z}=g'(z)-\varpi'(z).
\end{equation}
By choosing one special $w$, we give the contour plot of ${\rm Re}(\varpi(z)-g(z))=0$ in Fig.\ref{contour-large-t}.
\begin{figure}[!h]
\centering
\includegraphics[width=0.45\textwidth]{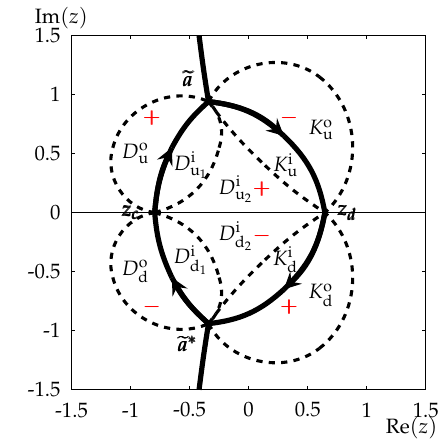}
\centering
\includegraphics[width=0.45\textwidth]{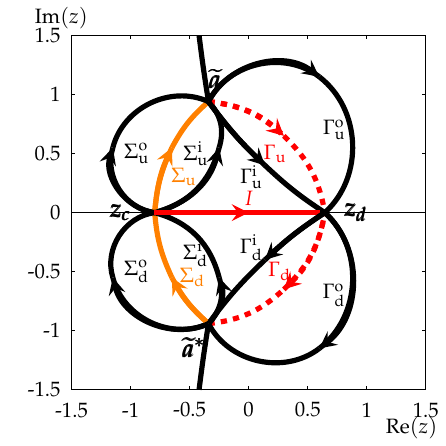}
\caption{The contour of ${\rm Re}(\varpi(z)-\widetilde{g}(z))$ by choosing $w=1$. The left one gives the original contour, and the right one the deformation contour.}
\label{contour-large-t}
\end{figure}

With the aid of $\widetilde{g}(z)$ function, we can construct a new matrix $$\mathbf{P}_{n}(z; \tau, w):=\mathbf{S}(z; \tau; w){\rm diag}\left(\ee^{-\tau^{1/3}\widetilde{g}(z)}, \ee^{ \tau^{1/3}\widetilde{g}(z)}, 1\right).$$ Furthermore, with a similar definition $\mathbf{Q}_{n}(z; \tau, w)$ in Eq.\eqref{eq:Q-definition}, then $\mathbf{Q}_{n}(z; \tau, w)$ satisfies
 \begin{equation}\label{eq:Qfn}
\begin{aligned}
\mathbf{Q}_{n,+}(z; \tau, w)&=\mathbf{Q}_{n,-}(z; \tau, w)\begin{bmatrix}0&\sqrt{2} \ee^{ \tau^{1/3}\widetilde{\Omega}_{n}}&0\\
-\frac{\sqrt{2}}{2} \ee^{-\tau^{1/3}\widetilde{\Omega}_{n}}&0&0\\
0&0&1
\end{bmatrix},\qquad &\lambda\in \Sigma_{\rm u}\\
\mathbf{Q}_{n,+}(z; \tau, w)&=\mathbf{Q}_{n,-}(z; \tau, w)\begin{bmatrix}0&\frac{\sqrt{2}}{2} \ee^{ \tau^{1/3}\widetilde{\Omega}_{n}}&0\\
-\sqrt{2}\ee^{-\tau^{1/3}\widetilde{\Omega}_{n}}&0&0\\
0&0&1
\end{bmatrix},\qquad &\lambda\in \Sigma_{\rm d}\\
\mathbf{Q}_{n,+}(z; \tau, w)&=\mathbf{Q}_{n,-}(z; \tau, w){\rm diag}\left(\sqrt{2}, \frac{\sqrt{2}}{2}, 1\right), \qquad &\lambda\in \Gamma_{\rm u}\\
\mathbf{Q}_{n,+}(z; \tau, w)&=\mathbf{Q}_{n,-}(z; \tau, w){\rm diag}\left(\frac{\sqrt{2}}{2}, \sqrt{2}, 1\right), \qquad &\lambda\in \Gamma_{\rm d}.
\end{aligned}
\end{equation}
Furthermore, we can set a similar matrix $\widehat{\pmb{Q}}_{n}$ with Eq.\eqref{eq:hatQ-no}, then the last non-identity converts into the $\Sigma_{\rm u}, \Sigma_{\rm d}$ and $I$. Then the leading order term in this case is similar to the asymptotics of non-oscillatory region for the large order soliton, we only need to replace the variable
\begin{equation*}
\lambda\to \Lambda\tau^{-1/3}, \quad N\to \tau^{1/3}, \quad a_{n}(X, T)\to \widetilde{a}, \quad \lambda_c\to z_c, \quad \lambda_d\to z_d, \quad \Omega_n\to \widetilde{\Omega}_n,\quad R(\lambda; X, T) \to \widetilde{R}(z)
\end{equation*}
in Eq.\eqref{eq:q-no}, thus the final result is
\begin{multline}
q_i(\tau^{2/3}w, \tau)=-\ii c_{i}^*\ee^{\tau^{1/3}\widetilde{\Omega}_{n}-\ii\tilde{\mu}}\left(-\ii \frac{\sqrt{9-w^2}}{3\tau^{1/3}}+\frac{\sqrt{2p}}{\tau^{1/2}\sqrt{-\widetilde{h}''(z_{c})}}m_{-}^{z_{a}}\ee^{\ii\phi_{z_{c}}}-\frac{\sqrt{2p}}{\tau^{1/2}\sqrt{-\widetilde{h}''(z_{c})}}m_{+}^{z_{a}}\ee^{-\ii\phi_{z_c}}\right)\\
-\ii c_{i}^*\ee^{\tau^{1/3}\widetilde{\Omega}_{n}-\ii\tilde{\mu}}\left(\frac{\sqrt{2p}}{\tau^{1/2}\sqrt{\widetilde{h}''(z_d)}}m_{+}^{z_d}\ee^{\ii\phi_{z_d}}-\frac{\sqrt{2p}}{\tau^{1/2}\sqrt{\widetilde{h}''(z_d)}}m_{-}^{z_b}\ee^{-\ii\phi_{z_b}}\right)+\mathcal{O}(\tau^{2/3}),\quad i=1,2,
\end{multline}
where $\widetilde{\mu}, \widetilde{h}$ is defined as
\begin{equation}
\widetilde{\mu}=2p\int_{z_c}^{z_d}\frac{1}{\widetilde{R}(s)}ds,\quad  \widetilde{h}(z):=\frac{\ii}{2}\left(\widetilde{g}(z)-\varpi(z)\right).
\end{equation}

\ys{In this case, based on the idea to construct the $g$ function in the non-oscillatory region, we use a different way to obtain the leading order term when $\tau$ is large, which seems simpler than the complicated polynomial calculation using in \cite{Peter-Duke-2019}.}

\appendix

\renewcommand{\appendixname}{Appendix \ \Alph{section} }

\setcounter{equation}{0}
\renewcommand\theequation{\Alph{section}.\arabic{equation}}
\setcounter{definition}{0}
\renewcommand\thedefinition{\Alph{section}.\arabic{definition}}

\renewcommand\thelemma{\Alph{section}.\arabic{lemma}}

\section{\appendixname: Some preliminaries about the inverse scattering method}
\label{app:IST}
We give a brief review on the inverse scattering transform for the CNLS equation. The boundary conditions are given by
\begin{equation}
\mathbf{q}(x,t)\to 0,\qquad \mathbf{q}_{x}(x,t)\to 0,\,\,\,\,\text{as  }x\to\pm\infty.
\end{equation}
Consider the Jost solutions for the spectral problem:
\begin{equation}
\pmb{\Phi}^{\pm}(\lambda;x,t)\to\ee^{-\ii\lambda(x+\lambda t)\pmb{\sigma}_3},\,\,\,\, x\to\pm\infty.
\end{equation}
%\lm{add the basic introduction on the Volterra equation}

If the function $\mathbf{q}(x,t)$ satisfies the focusing CNLS equation, which is equivalent to the compatibility condition $\pmb{\Phi}_{xt}(\lambda;x,t)=\pmb{\Phi}_{tx}(\lambda;x,t)$, then the matrix function $\mathbf{\Phi}(\lambda;x,t)$ is existence and can be determined uniquely by the curve integral.
Inserting the ansatz
\begin{equation}
\pmb{\Phi}^{\pm}(\lambda;x,t)=\pmb{\mu}^{\pm}(\lambda;x,t)\ee^{-\ii\lambda (x+\lambda t)\pmb{\sigma}_3}
\end{equation}
into the spectral problem, we will obtain the Jost solution
\begin{equation}\label{eq:volterra}
\begin{split}
\pmb{\mu}^{\pm}_x&=-\ii\lambda[\pmb{\sigma}_3,\pmb{\mu}^{\pm}]+\ii\pmb{\mathcal{Q}}\pmb{\mu}^{\pm},\\
\pmb{\mu}^{\pm}_t&=-\ii\lambda^2[\pmb{\sigma}_3,\pmb{\mu}^{\pm}]+\left[\ii\lambda\pmb{\mathcal{Q}}+\frac{1}{2}\left(\pmb{\mathcal{Q}}_x+\ii\pmb{\mathcal{Q}}^2\right)\pmb{\sigma}_3\right]\pmb{\mu}^{\pm},
\end{split}
\end{equation}
with the normalization conditions
\begin{equation}\label{eq:bound}
\lim_{x\to\pm\infty}\pmb{\mu}^{\pm}(\lambda;x,t)=\mathbb{I}.
\end{equation}
These solutions can be expressed as the Volterra type integrals
\begin{multline}
\pmb{\mu}^{\pm}(\lambda;x,t)=\mathbb{I}+\ii\int_{(\pm\infty,t)}^{(x,t)}\ee^{-\ii\lambda(x-y)\pmb{\sigma}_3}\pmb{\mathcal{Q}}
(y,s)\pmb{\mu}^{\pm}(\lambda;y,s)\ee^{\ii\lambda(x-y)\pmb{\sigma}_3}\dd y\\+\int_{(\pm\infty,0)}^{(\pm\infty,t)}\ee^{-\ii\lambda^2(t-s)\pmb{\sigma}_3}\left[\ii\lambda\pmb{\mathcal{Q}}(y,s)
+\frac{1}{2}\left(\pmb{\mathcal{Q}}_y(y,s)+\ii\pmb{\mathcal{Q}}^2(y,s)\right)\pmb{\sigma}_3\right]\pmb{\mu}^{\pm}(\lambda;y,s)\ee^{\ii\lambda^2(t-s)\pmb{\sigma}_3}\dd s.
\end{multline}
By the conditions $\mathbf{q}(x,t)\to0$ and $\mathbf{q}_{x}(x,t)\to0$ as $x\to\pm\infty$ for arbitrary time $t$, the above curve integral can be reduced to
\begin{equation}
\pmb{\mu}^{\pm}(\lambda;x,t)=\mathbb{I}+\ii\int_{\pm\infty}^{x}\ee^{-\ii\lambda(x-y)\pmb{\sigma}_3}\pmb{\mathcal{Q}}
(y,t)\pmb{\mu}^{\pm}(\lambda;y,t)\ee^{\ii\lambda(x-y)\pmb{\sigma}_3}\dd y.
\end{equation}

Define the following analytic matrix
\begin{equation}
\pmb{\mu}_+\equiv \left(\pmb{\mu}^{-}_1,\pmb{\mu}_2^{+}\right),\,\,\,\,\,\, \pmb{\mu}_-\equiv \left(\pmb{\mu}^{+}_1,\pmb{\mu}_2^{-}\right).
\end{equation}
Through the analytic property of Volterra integral equations, we obtain that the matrix $\pmb{\mu}_+$ is analytic in the upper half plane and the matrix $\pmb{\mu}_-$ is analytic in the lower half plane under the condition $q_1(x,t),q_2(x,t)\in L^1(\mathbb{R})$ for the fixed time $t$.

The symmetry for the Jost solutions is
\begin{equation}
\pmb{\Phi}^{\pm}(\lambda;x,t)\left[\pmb{\Phi}^{\pm}(\lambda^*;x,t)\right]^{\dag}=\mathbb{I}.
\end{equation}
Define the scattering matrix
\begin{equation}\label{eq:scattering}
\pmb{\Phi}^{-}(\lambda;x,t)=\pmb{\Phi}^{+}(\lambda;x,t)\mathbf{S}(\lambda),\,\,\,\,\lambda\in\mathbb{R}
\end{equation}
where
\begin{equation}
\mathbf{S}(\lambda)=\begin{bmatrix}
\bar{a}(\lambda)& \mathbf{b}(\lambda) \\
\bar{\mathbf{b}}(\lambda) & \mathbf{a}(\lambda) \\
\end{bmatrix},\,\,\,\,\, \det(\mathbf{S}(\lambda))=1.
\end{equation}
Then we obtain the symmetric relation for the scattering matrix
\begin{equation}
\mathbf{S}(\lambda)\mathbf{S}^{\dag}(\lambda^*)=\mathbb{I},
\end{equation}
which implies that
\begin{equation}
\bar{a}(\lambda)=\det(\mathbf{a}^{\dag}(\lambda^*)),\,\,\,\,\, \bar{\mathbf{b}}(\lambda)=-{\rm adj}(\mathbf{a}^{\dag}(\lambda^*))\mathbf{b}^{\dag}(\lambda^*).
\end{equation}
From the integral representation, we know that $\mathbf{a}(\lambda)$ is analytic in the lower half plane $\mathbb{C}_-.$ Moreover, we have $\det(\pmb{\mu_-})=\det(\mathbf{a}(\lambda))$. Assume the form $\det(\mathbf{a}(\lambda))=\prod_{i=1}^n\frac{\lambda-\lambda_i^*}{\lambda-\lambda_i}a_0(\lambda)$, where $a_0(\lambda)$ has no zero. On the other hand, since $\det(\pmb{\mu^{\pm}})=1$, then the matrix $\pmb{\mu_+}$ has the one dimensional kernel at $\lambda=\lambda_i$.
By the symmetric relation, we know that
\begin{equation}\label{eq:sym}
\begin{split}
&\mathbf{m}_{\pm}(\lambda;x,t)\left[\mathbf{m}_{\mp}(\lambda^*;x,t)\right]^{\dag}=\mathbb{I},\\ &\mathbf{m}_{+}=\pmb{\mu}_{+}{\rm diag}\left(\frac{1}{\det(\mathbf{a}^{\dag}(\lambda^*))},\mathbb{I}_2\right),\,\,\,\,\, \mathbf{m}_{-}=\pmb{\mu}_{-}{\rm diag}\left(1,\mathbf{a}^{-1}(\lambda)\right).
\end{split}
\end{equation}
Introducing the notation $\mathbf{R}(\lambda)=\mathbf{b}(\lambda)\mathbf{a}^{-1}(\lambda)$, we obtain that
\begin{equation}
\begin{split}
\mathbf{m}_+&=\pmb{\mu}^+\ee^{-\ii\lambda (x+\lambda t)\pmb{\sigma}_3}\begin{bmatrix}
1& 0 \\
-\mathbf{R}^{\dag}(\lambda^*) & \mathbb{I}_2\\
\end{bmatrix}\ee^{\ii\lambda (x+\lambda t)\pmb{\sigma}_3},\\
\mathbf{m}_-&=\pmb{\mu}^+\ee^{-\ii\lambda (x+\lambda t)\pmb{\sigma}_3}\begin{bmatrix}
1& \mathbf{R}(\lambda) \\
0 & \mathbb{I}_2\\
\end{bmatrix}\ee^{\ii\lambda (x+\lambda t)\pmb{\sigma}_3},
\end{split}
\end{equation}
which implies the jump condition
\begin{equation}
\mathbf{m}_+=\mathbf{m}_-\ee^{-\ii\lambda (x+\lambda t)\pmb{\sigma}_3}\begin{bmatrix}
1+\mathbf{R}(\lambda)\mathbf{R}^{\dag}(\lambda)& -\mathbf{R}(\lambda) \\
-\mathbf{R}^{\dag}(\lambda) & \mathbb{I}_2\\
\end{bmatrix}\ee^{\ii\lambda (x+\lambda t)\pmb{\sigma}_3},\,\,\,\,\,\, \lambda\in\mathbb{R}.
\end{equation}
For the simple discrete spectrum $\lambda=\lambda_i\in\mathbb{C}_+$ and $\lambda=\lambda_i^*\in\mathbb{C}_-$,  we have the following residue conditions for $\mathbf{m}_{\pm}$:
\begin{equation}\label{eq:residue+}
\underset{\lambda=\lambda_i}{{\rm Res}}\mathbf{m}_+(\lambda;x,t)=\lim_{\lambda\to\lambda_i}\mathbf{m}_+(\lambda;x,t)\begin{bmatrix}
0&0_{1\times 2} \\
-\mathbf{c}_i\ee^{\tilde{\theta}_i}&0_{2\times 2} \\
\end{bmatrix},
\end{equation}
and $\tilde{\theta}_i=2\ii\lambda_i(x+\lambda_i t),$
\begin{equation}\label{eq:residue-}
\underset{\lambda=\lambda_i^*}{{\rm Res}}\mathbf{m}_-(\lambda;x,t)=\lim_{\lambda\to\lambda_i^*}\mathbf{m}_-(\lambda;x,t)\begin{bmatrix}
0&\mathbf{c}_i^{\dag}\ee^{-\tilde{\theta}_i^*} \\
0_{2\times 1} &0_{2\times 2} \\
\end{bmatrix}.
\end{equation}
By solving the above Riemann-Hilbert problem, the solutions for CNLS equation can be constructed by:
\begin{equation}
\mathbf{q}(x,t)=2\lim_{\lambda\to\infty}\lambda\mathbf{m}_{12}(\lambda;x,t),
\end{equation}
where the subscript $_{12}$ represents the $(1,2)$ and $(1,3)$ elements of the corresponding matrix.
As for the non-reflection coefficient potential $\mathbf{q}(x,t)$ with $\mathbf{R}(\lambda)=0$ or $\bar{a}(\lambda)=\prod_{i=1}^n\frac{\lambda-\lambda_i}{\lambda-\lambda_i^*}$, the multi-soliton solutions can be constructed by
\begin{equation}\label{eq:ansatz}
\mathbf{m}(\lambda;x,t)=\mathbb{I}+\sum_{i=1}^n\left(\frac{\mathbf{m}_{1,i}(x,t)}{\lambda-\lambda_i},\frac{\mathbf{m}_{2,i}(x,t)}{\lambda-\lambda_i^*}\right).
\end{equation}
By the residue conditions, we have the linear system
\begin{equation}
\begin{split}
\mathbf{m}_{1,j}(x,t)=&\begin{bmatrix}
0 \\
-\mathbf{c}_j\ee^{\tilde{\theta}_j} \\
\end{bmatrix}-\sum_{i=1}^n\frac{\mathbf{m}_{2,i}(x,t) \mathbf{c}_j\ee^{\tilde{\theta}_j}}{\lambda_j-\lambda_i^*},\,\,\,\, j=1,2,\cdots,n \\
\mathbf{m}_{2,j}(x,t)=&\begin{bmatrix}
\mathbf{c}_j^{\dag}\ee^{-\tilde{\theta}_j^*} \\
0_{2\times 2} \\
\end{bmatrix}+\sum_{i=1}^n\frac{\mathbf{m}_{1,i}(x,t)\mathbf{c}_j^{\dag}\ee^{-\tilde{\theta}_j^*}}{\lambda_j^*-\lambda_i},
\end{split}
\end{equation}
which will determine the elements $\mathbf{m}_{1,j}(x,t)$ and $\mathbf{m}_{2,j}(x,t)$ uniquely.
Then the $n$-soliton solution can be constructed by
\begin{equation}
\mathbf{q}(x,t)=2\sum_{i=1}^n\mathbf{m}_{2,i}^{(1)}(x,t),
\end{equation}
where the notation $\mathbf{m}_{2,i}^{(1)}(x,t)$ represents the first row of the matrix $\mathbf{m}_{2,i}(x,t)$.

To obtain a better formula for the soliton solutions, we define the following new matrix function:
\begin{equation}\label{eq:tildem}
\widetilde{\mathbf{m}}(\lambda;x,t)=\mathbf{m}(\lambda;x,t){\rm diag}\left(\prod_{i=1}^n\frac{\lambda-\lambda_i}{\lambda-\lambda_i^*},1,1\right).
\end{equation}
Then the residue conditions \eqref{eq:residue+} and \eqref{eq:residue-} can be converted into the following kernel and residue condition respectively:
\begin{equation}
\widetilde{\mathbf{m}}(\lambda;x,t)|_{\lambda=\lambda_i}\begin{bmatrix}
1 \\
\mathbf{c}_i\gamma_i\ee^{\tilde{\theta}_i} \\
\end{bmatrix}=0,\,\,\,\,\underset{\lambda=\lambda_i^*}{\rm Res}\widetilde{\mathbf{m}}(\lambda;x,t)\begin{bmatrix}
-\mathbf{c}_i^{\dag}\gamma_i^*\ee^{-\tilde{\theta}_i^*} \\
\mathbb{I}_2 \\
\end{bmatrix}=0,\,\,\,\,\,  \gamma_i=\frac{1}{\lambda_i-\lambda_i^*}\prod_{j=1,j\neq i}^n\frac{\lambda_i-\lambda_j}{\lambda_i-\lambda_j^*}.
\end{equation}
Finally, the matrix function $\widetilde{\mathbf{m}}(\lambda;x,t)$ can be written in the form:
\begin{equation}
\widetilde{\mathbf{m}}(\lambda;x,t)=\mathbb{I}-\mathbf{X}_n\mathbf{M}_n^{-1}\left(\lambda \mathbb{I}-\mathbf{D}_n\right)^{-1}\mathbf{X}_n^{\dag}
\end{equation}
where $\mathbf{D}_n={\rm diag}\left(\lambda_1^*,\lambda_2^*,\cdots,\lambda_n^*\right)$ and
\begin{equation}
\mathbf{X}_n=\begin{bmatrix}
1 &1 &\cdots &1 \\
\mathbf{c}_1\gamma_1\ee^{\tilde{\theta}_1} &\mathbf{c}_2\gamma_2\ee^{\tilde{\theta}_2} &\cdots &\mathbf{c}_n\gamma_n\ee^{\tilde{\theta}_n} \\
\end{bmatrix},\,\,\,\,\mathbf{M}_n=\left(\frac{1+\mathbf{c}_i^{\dag}\mathbf{c}_j\gamma_i^*\gamma_j\ee^{\tilde{\theta}_i^*+\tilde{\theta}_j}}{\lambda_j-\lambda_i^*}\right)_{1\leq i,j\leq n}.
\end{equation}
It is readily to see that the matrix function $\widetilde{\mathbf{m}}(\lambda;x,t)$ is a special Darboux matrix.
Then the reconstruction of potential function can be performed by
\begin{equation}
\mathbf{q}(x,t)=2\lim_{\lambda\to\infty}\lambda\widetilde{\mathbf{m}}_{12}(\lambda;x,t),=-2\mathbf{X}_n^{(1)}\mathbf{M}_n^{-1}\left(\mathbf{X}_n^{(2)}\right)^{\dag},\,\,\,\,|\mathbf{q}(x,t)|^2=\ln_{xx}(\det(\mathbf{M}_n))
\end{equation}
where $\mathbf{X}_n^{(1)}$ represents the first row of $\mathbf{X}_n$ and $\mathbf{X}_n^{(2)}$ represents the second and third row of $\mathbf{X}_n$.

\section{\appendixname : Dynamic behavior to the general second order solitons $q_1^{[2]}$ and $q_2^{[2]}$}\label{appendix:second-order}
\ys{If the perturbation parameters $\epsilon_{j}^{[1]}, \epsilon_{j}^{[2]}, j=1,2,\cdots, N$ are zero, then the two components $q_1$ and $q_2$ have a similar behavior on both the large order and infinite order asymptotics, which can be seen from the theorem \ref{theo:os} to \ref{theo:large-tau}. While $\epsilon_{j}^{[1]}$ and $\epsilon_{j}^{[2]}$ are not zero, then they have a different asymptotics. Up to now, we do not know how to study the asymptotics by the Riemann-Hilbert problems. In the following, we merely give a simple analysis to the second order soliton when $\epsilon_{1}^{[1]}$ and $\epsilon_{1}^{[2]}$ are a small perturbation and $\lambda_1=\ii$, then the second order $q_1^{[2]}$ and $q_2^{[2]}$ can be given as
\begin{equation}
\begin{split}
q_1^{[2]}&=\frac{-8\cos(\theta)\left[\left(4t+2\ii x-\ii\right)\ee^{6x+2\ii t}+\left(4t-2\ii x-\ii\right)\ee^{2x+2\ii t}\right]+\delta^{[p_1]}_{1}\ee^{6x+2\ii t}+\delta^{[p_1]}_{2}\ee^{2x+2\ii t}}{\left[\left(64t^2+16x^2+2\right)\ee^{4x}+\ee^{8x}+1\right]+\delta_{3}\ee^{4x}+\delta_4\ee^{8x}}\\
q_2^{[2]}&=\frac{-8\sin(\theta)\left[\left(4t+2\ii x-\ii\right)\ee^{6x+2\ii t}+\left(4t-2\ii x-\ii\right)\ee^{2x+2\ii t}\right]+\delta^{[p_2]}_{1}\ee^{6x+2\ii t}+\delta^{[p_2]}_{2}\ee^{2x+2\ii t}}{\left[\left(64t^2+16x^2+2\right)\ee^{4x}+\ee^{8x}+1\right]+\delta_{3}\ee^{4x}+\delta_4\ee^{8x}}
\end{split}
\end{equation}
where $\delta^{[p_i]}_{j}, (i=1,2, j=1,2,3,4)$ are small parameters with respect to $\epsilon^{[1]}_{1}$ and $\epsilon^{[2]}_1$, which equal to
\begin{equation}
\begin{split}
\delta^{[p_1]}_1&=4\ii\left[(\epsilon^{[1]}_1)^2-(\epsilon^{[2]}_1)^2\right]\cos(\theta)^3+4\ii\epsilon^{[1]}_1\left(2\sin(\theta)\epsilon^{[2]}_1-4\ii x-8t-\ii\right)\cos(\theta)^2\\
&-2\ii\epsilon^{[2]}_1\left(4\ii x+8t+\ii\right)\sin(2\theta)+4\ii\left(2(\epsilon^{[2]}_1)^2-(\epsilon^{[1]}_1)^2\right)\cos(\theta)\\&-4\ii\epsilon^{[1]}_1\left(2\sin(\theta)\epsilon^{[2]}_1-4\ii x-8t+\ii\right),\\
\delta^{[p_1]}_{2}&=2\epsilon^{[2]}_1\sin(2\theta)+4\epsilon^{[1]}_1\cos(\theta)^2+4\epsilon^{[2]}_1,\\
\delta^{[p_2]}_1&=-8\ii\cos(\theta)^3\epsilon^{[1]}_1\epsilon^{[2]}_1+\left(4\ii\left((\epsilon^{[1]}_1)^2-(\epsilon^{[2]}_1)^2\right)\sin(\theta)+4\ii\epsilon^{[2]}_1\left(4\ii x+8t+\ii\right)\right)\cos(\theta)^2\\
&-2\ii\epsilon^{[1]}_1\left(4\ii x+8t+\ii\right)\sin(2\theta)+4\ii\sin(\theta)(\epsilon^{[1]}_1)^2+8\epsilon^{[2]}_1,\\
\delta^{[p_2]}_{2}&=2\epsilon^{[1]}_1\sin(2\theta)-4\epsilon^{[2]}_1\cos(\theta)^2+8\epsilon^{[2]}_1,\\
\delta_3&=3\left((\epsilon^{[1]}_1)^2-(\epsilon^{[2]}_1)^2\right)\cos(\theta)^2+2\epsilon^{[1]}_1\left(3\epsilon^{[2]}_1\sin(\theta)-16t\right)\cos(\theta)-32\epsilon^{[2]}_1\sin(\theta)t\\&
+(\epsilon^{[1]}_1)^2+4(\epsilon^{[2]}_1)^2,\\
\delta_4&=\left((\epsilon^{[2]}_1)^2-(\epsilon^{[1]}_1)^2\right)\cos(\theta)^2-\epsilon^{[1]}_1\epsilon^{[2]}_1\sin(2\theta)+(\epsilon^{[1]}_1)^2.
\end{split}
\end{equation}
Especially, when $t\to\pm\infty$, these two components have the following asymptotic expression
\begin{equation}
\begin{split}
q_1^{[2]}&\to\Bigg[{-}\frac{\ii\epsilon^{[2]}_1\sin(2\theta){-}2\ii\epsilon^{[1]}_1\sin(\theta)^2{+}2\cos(\theta)}{1{+}\sqrt{\sin(\theta)^2(\epsilon^{[1]}_1)^2{-}\cos(\theta)^2(\epsilon^{[2]}_1)^2{-}\epsilon^{[1]}_1\epsilon^{[2]}_1\sin(2\theta)}}
\\&\cdot{\rm sech}\left(2x{-}\log(t){+}\log\left(\frac{\sqrt{1{+}\cos(\theta)^2(\epsilon^{[2]}_1)^2{+}\sin(\theta)^2(\epsilon^{[1]}_1)^2{-}\epsilon^{[1]}_1\epsilon^{[2]}_1\sin(2\theta)}}{8}\right)\right)\\&-2\cos(\theta){\rm sech}\left(2x+\log(t)+3\log(2)\right)\Bigg]\ee^{2\ii t},\\
\end{split}
\end{equation}
\begin{equation}
\begin{split}
q_2^{[2]}&\to\Bigg[{-}\frac{\ii\epsilon^{[1]}_1\sin(2\theta){-}2\ii\epsilon^{[2]}_1\cos(\theta)^2{+}2\sin(\theta)}{1{+}\sqrt{\sin(\theta)^2(\epsilon^{[1]}_1)^2{-}\cos(\theta)^2(\epsilon^{[2]}_1)^2{-}\epsilon^{[1]}_1\epsilon^{[2]}_1\sin(2\theta)}}
\\&\cdot{\rm sech}\left(2x{-}\log(t){+}\log\left(\frac{\sqrt{1{+}\cos(\theta)^2(\epsilon^{[1]}_1)^2{+}\sin(\theta)^2(\epsilon^{[2]}_1)^2{-}\epsilon^{[1]}_1\epsilon^{[2]}_1\sin(2\theta)}}{8}\right)\right)\\&-2\sin(\theta){\rm sech}\left(2x+\log(t)+3\log(2)\right)\Bigg]\ee^{2\ii t},\\
\end{split}
\end{equation}
%\begin{equation}
%\begin{split}
%q_1&=-\frac{4\ee^{2\ii t+2x}}{5}\frac{\left(400\ii x-200\ii -205 k+800 t+180 kx+36\ii k^2-360\ii kt\right)\ee^{4x}-400\ii x-200\ii-205 k+800t}{\left(1600t^2+400x^2+50+73 k^2-640 kt\right)+\left(9k^2+25\right)\ee^{8x}+25},\\
%q_2&=-\frac{3\ee^{2\ii t+2x}}{5}\frac{\left(400\ii x-200\ii -80 k+800 t-320 kx-164\ii k^2+640\ii kt\right)\ee^{4x}-400\ii x-200\ii-80 k+800t}{\left(1600t^2+400x^2+50+73 k^2-640 kt\right)+\left(9k^2+25\right)\ee^{8x}+25}.
%\end{split}
%\end{equation}
%If $k=0,$ then we have $\frac{5}{4}q_1=\frac{5}{3}q_2$. Especially, if $t\to\infty$, the asymptotic expression to $q_1$ and $q_2$ is different, which depends on this perturbation parameter $k$, that is
%\begin{equation}
%\begin{split}
%q_1&=-\frac{8}{5}{\rm sech}\left(2x+\log(t)+3\log(2)\right)+\frac{2}{5}\frac{9\ii k-20}{\sqrt{9k^2+25}}{\rm sech}\left(2x-\log(t)+\log\left(\frac{\sqrt{9k^2+25}}{40}\right)\right),\\
%q_2&=-\frac{6}{5}{\rm sech}\left(2x+\log(t)+3\log(2)\right)-\frac{6}{5}\frac{4\ii k+5}{\sqrt{9k^2+25}}{\rm sech}\left(2x-\log(t)+\log\left(\frac{\sqrt{9k^2+25}}{40}\right)\right).
%\end{split}
%\end{equation}
Under this perturbation, the parameter $\epsilon^{[1]}_{1}$ and $\epsilon^{[2]}_{1}$ only appears in one of the asymptotics soliton, which can been seen from the Fig. \ref{fig:two-diff-soliton}.
\begin{figure}[!h]
\centering
\includegraphics[width=0.45\textwidth]{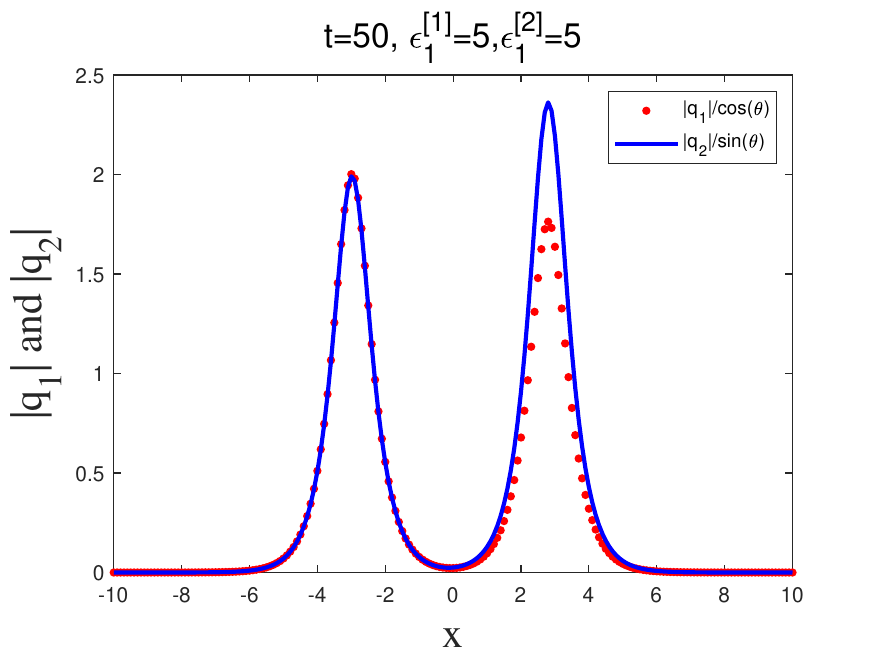}
\centering
\includegraphics[width=0.45\textwidth]{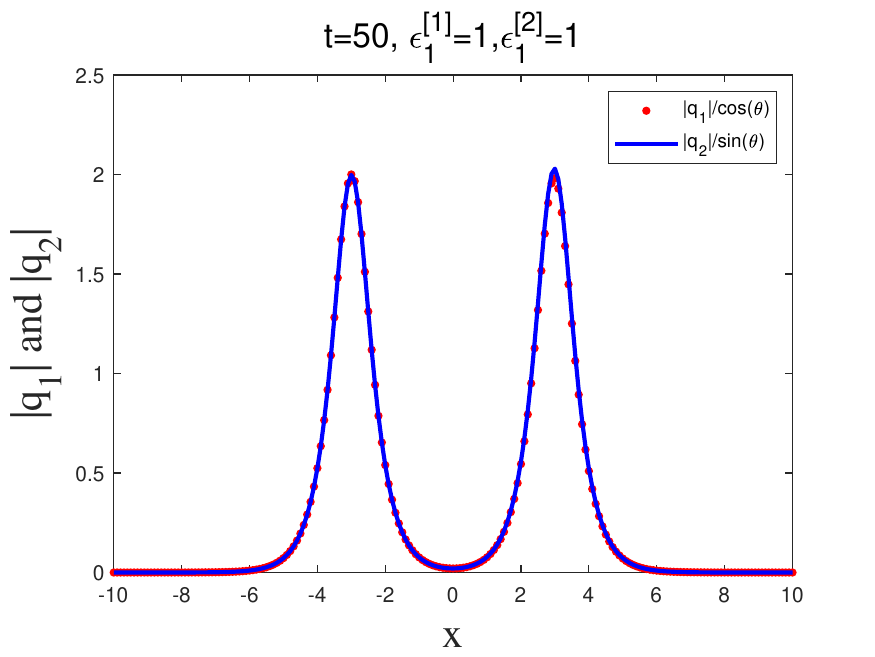}
\caption{Two different dynamical behavior of $q_1$ and $q_2$ by choosing different $\epsilon^{[1]}_{1}$ and $\epsilon^{[2]}_{1}$.}
\label{fig:two-diff-soliton}
\end{figure}
It can be seen when $\epsilon^{[1]}_1, \epsilon^{[2]}_{1}$ are small, the amplitudes to $\frac{|q_1|}{\cos(\theta)}$ and $\frac{|q_2|}{\sin(\theta)}$ have a slight difference. As the term $-4\ii\frac{\sin(\theta)\epsilon^{[1]}_{1}-\cos(\theta)\epsilon^{[2]}_{1}}{\sin(2\theta)}$ increases,  $\frac{|q_1|}{\cos(\theta)}$ and $\frac{|q_2|}{\sin(\theta)}$ will exhibit the distinct behaviors. In this similar manner, we could consider the other finite order cases. As the parameters $\epsilon_j^{[1]}$ and $\epsilon_j^{[2]}$ tends to zero, these solutions will converge to the solutions with the maximum $2$-norm $\|\cdot\|_2$ that we mainly consider in this work.}

\section{\appendixname : The proof of Lemma \ref{lem:asym-m}}
\label{app:CLS}
In this appendix, we will give the proof to lemma \ref{lem:asym-m}.
\begin{proof}
Based on the idea \cite{Faddeev-Book}, when $x\to+\infty$, we have
\begin{equation}\label{eq:phi20}
\begin{split}
\mathbf{\Phi}_1^{[0]}&=\begin{bmatrix}1\\\ee^{2\tilde{\theta}_1}\\\ee^{2\tilde{\theta}_1}
\end{bmatrix}=\begin{bmatrix}1\\0\\0
\end{bmatrix}+\mathcal{O}(\ee^{-2{\rm Im}(\lambda_1)x}),\quad x\to+\infty,\\
\end{split}
\end{equation}
Correspondingly, the asymptotics for $\mathbf{P}^{[1]}$ and $\mathbf{T}^{[1]}$ can be represented as
\begin{equation}\label{eq:p20t20}
\begin{split}
\mathbf{P}^{[1]}&=\frac{\Phi_{1}^{[0]}\left(\Phi_{1}^{[0]}\right)^{\dagger}}
{\left(\Phi_{1}^{[0]}\right)^{\dagger}\Phi_{1}^{[0]}}=\begin{bmatrix}1&0&0\\0&0&0\\
0&0&0
\end{bmatrix}+\mathcal{O}(\ee^{-2{\rm Im}(\lambda_1)x}),\quad x\to+\infty,\\ \mathbf{T}^{[1]}&=\left(\mathbb{I}-\frac{\lambda_1-\lambda_1^*}{\lambda-\lambda_1^*}\mathbf{P}^{[1]}\right)=\begin{bmatrix}\left(\frac{\lambda-\lambda_1}{\lambda-\lambda_1^*}\right)&0&0\\0&1&0\\0&0&1
\end{bmatrix}+\mathcal{O}(\ee^{-2{\rm Im}(\lambda_1)x}), \quad x\to+\infty.\\
\end{split}
\end{equation}
Furthermore, with Eq.\eqref{eq:high-dt}, Eq.\eqref{eq:phi20} and Eq.\eqref{eq:p20t20}, we have
\begin{equation}
\begin{split}
\mathbf{\Phi}_1^{[1]}&=\left(\mathbb{I}-\mathbf{P}^{[1]}\right)\begin{bmatrix}0\\ \frac{d}{d\lambda}\ee^{2\tilde{\theta}}\Big|_{\lambda=\lambda_1}\\\frac{d}{d\lambda}\ee^{2\tilde{\theta}}\Big|_{\lambda=\lambda_1}\end{bmatrix}
+\frac{\mathbf{P}^{[1]}}{\lambda_1-\lambda_1^*}\begin{bmatrix}1\\\ee^{2\tilde{\theta}_1}\\\ee^{2\tilde{\theta}_1}\end{bmatrix}
=\frac{1}{\lambda_1-\lambda_1^*}\begin{bmatrix}1\\0\\0
\end{bmatrix}+\mathcal{O}\left(x\ee^{-2{\rm Im}(\lambda_1)x}\right),\, x\to+\infty,\\
\mathbf{P}^{[2]}&=\frac{\mathbf{\Phi}_{1}^{[1]}\left(\mathbf{\Phi}^{[1]}_1\right)^{\dagger}}
{\left(\mathbf{\Phi}_{1}^{[1]}\right)^{\dagger}\mathbf{\Phi}^{[1]}_1}=\begin{bmatrix}1&0&0\\0&0&0\\0&0&0
\end{bmatrix}+\mathcal{O}\left[x\ee^{-2{\rm Im}(\lambda_1)x}\right],\quad x\to+\infty,\\ \mathbf{T}^{[2]}&=\left(\mathbb{I}-\frac{\lambda_1-\lambda_1^*}{\lambda-\lambda_1^*}
\mathbf{P}^{[2]}\right)=\begin{bmatrix}\left(\frac{\lambda-\lambda_1}{\lambda-\lambda_1^*}\right)&0&0\\0&1&0\\0&0&1
\end{bmatrix}+\mathcal{O}\left[x\ee^{-2{\rm Im}(\lambda_1)x}\right], \quad x\to+\infty.\\
\end{split}
\end{equation}
where $\tilde{\theta}=\ii\lambda\left(x+\lambda t\right)$.
In succession, after $N$-fold iteration, we have
\begin{equation}\label{eq:T-infinity}
\begin{split}
\mathbf{T}_{N}&=\begin{bmatrix}\left(\frac{\lambda-\lambda_1}{\lambda-\lambda_1^*}\right)^{N}&0&0\\0&1&0\\0&0&1
\end{bmatrix}+\mathcal{O}\left(x^{N-1}\ee^{-2{\rm Im}(\lambda_1)x}\right), \quad x\to+\infty,\\
:&=\mathbf{T}_{+}(\lambda)+\mathcal{O}\left(x^{N-1}\ee^{-2{\rm Im}(\lambda_1)x}\right), \quad x\to+\infty.
\end{split}
\end{equation}
Conversely, when $x\to-\infty$, we have
\begin{equation}\label{eq:phi20}
\begin{split}
\mathbf{\Phi}_1^{[0]}&=\begin{bmatrix}0\\1\\1
\end{bmatrix}+\mathcal{O}(\ee^{2{\rm Im}(\lambda_1)x}),\quad x\to-\infty,\\
\mathbf{P}^{[1]}&=\begin{bmatrix}0&0&0\\0&\frac{1}{2}&\frac{1}{2}\\
0&\frac{1}{2}&\frac{1}{2}
\end{bmatrix}+\mathcal{O}(\ee^{2{\rm Im}(\lambda_1)x}),\quad x\to-\infty,\\ \mathbf{T}^{[1]}&=\begin{bmatrix}1&0&0\\0&1-\frac{\lambda_1-\lambda_1^*}{2\left(\lambda-\lambda_1^*\right)}&-\frac{\lambda_1-\lambda_1^*}{2\left(\lambda-\lambda_1^*\right)}\\
0&-\frac{\lambda_1-\lambda_1^*}{2\left(\lambda-\lambda_1^*\right)}&1-\frac{\lambda_1-\lambda_1^*}{2\left(\lambda-\lambda_1^*\right)}
\end{bmatrix}+\mathcal{O}(\ee^{2{\rm Im}(\lambda_1)x}), \quad x\to-\infty.\\
\end{split}
\end{equation}

With $N$-fold iteration, we have
\begin{equation}\label{eq:T-infinity-1}
\begin{split}
\mathbf{T}_{N}(\lambda; x, t)&=\begin{bmatrix}1&0&0\\0&1-\frac{\lambda_1-\lambda_1^*}{2\left(\lambda-\lambda_1^*\right)}&-\frac{\lambda_1-\lambda_1^*}{2\left(\lambda-\lambda_1^*\right)}\\
0&-\frac{\lambda_1-\lambda_1^*}{2\left(\lambda-\lambda_1^*\right)}&1-\frac{\lambda_1-\lambda_1^*}{2\left(\lambda-\lambda_1^*\right)}
\end{bmatrix}^{N}+\mathcal{O}(x^{N-1}\ee^{2{\rm Im}(\lambda_1)x})\\
&=\begin{bmatrix}1&0&0\\
0&\frac{1}{2}+\frac{1}{2}\left(\frac{\lambda-\lambda_1}{\lambda-\lambda_1^*}\right)^{N}&-\frac{1}{2}+\frac{1}{2}\left(\frac{\lambda-\lambda_1}{\lambda-\lambda_1^*}\right)^{N}\\
0&-\frac{1}{2}+\frac{1}{2}\left(\frac{\lambda-\lambda_1}{\lambda-\lambda_1^*}\right)^{N}&\frac{1}{2}+\frac{1}{2}\left(\frac{\lambda-\lambda_1}{\lambda-\lambda_1^*}\right)^{N}
\end{bmatrix}+\mathcal{O}(x^{N-1}\ee^{2{\rm Im}(\lambda_1)x}), \quad x\to-\infty,\\
:&=\mathbf{T}_{-}(\lambda)+\mathcal{O}(x^{N-1}\ee^{2{\rm Im}(\lambda_1)x}), \quad x\to-\infty,
\end{split}
\end{equation}
it completes the proof.
\end{proof}
Furthermore, based on the proof of the lemma \ref{lem:asym-m} and the theory in \cite{Yang-book-2010}, we can give the conservations for CNLS equation.
\begin{proof}
Set $\pmb{\Phi}(\lambda; x, t)=\left(\phi_1(\lambda; x, t), \phi_2(\lambda; x, t), \phi_3(\lambda; x, t)\right)^T$ be the solution of the Lax pair. Define two new variables $\mu^{(1)}(\lambda; x, t)$ and $\mu^{(2)}(\lambda; x, t)$ as
\begin{equation}
\mu^{(1)}(\lambda; x, t)=\phi_2(\lambda; x, t)/\phi_1(\lambda; x, t),\qquad \mu^{(2)}(\lambda; x, t)=\phi_3(\lambda; x, t)/\phi_1(\lambda; x, t),
\end{equation}
then we have
\begin{equation}\label{eq:laws}
\left(\ln \phi_1(\lambda; x, t)\right)_x=-\ii\lambda+\ii q_1\mu^{(1)}(\lambda; x, t)+\ii q_2\mu^{(2)}(\lambda; x, t).
\end{equation}
%integrate both sides of Eq.\eqref{eq:laws}, we have
%\begin{equation}
%\left(\ln \phi_1(\lambda; x, t)\right)\Big|_{-\infty}^{+\infty}=\int_{-\infty}^{+\infty}\left(-\ii\lambda+\ii q_1\mu^{(1)}(\lambda; x, t)+\ii q_2\mu^{(2)}(\lambda; x, t)\right)dx.
%\end{equation}
From the Eq.\eqref{eq:tildem}, we know $\phi_1(\lambda; x, t)$ equals to $[\widetilde{\mathbf{m}}(\lambda;x,t)\ee^{-\ii\lambda\left(x+2\lambda t\right)\pmb{\sigma}_3}\left(1, 0, 0\right)^{T}]_{1,1}.$
To get the conservation laws, set the $\mu^{(1)}(\lambda; x, t)$, $\mu^{(2)}(\lambda; x, t)$  and $N\log\left(\frac{\lambda-\lambda_1}{\lambda-\lambda_1^*}\right)$ as an expansion about $\lambda$, that is
\begin{equation}
\begin{split}
\mu^{(i)}(\lambda; x, t):=-\sum_{n=1}^{\infty}\frac{\mu_{n}^{(i)}(x, t)}{\left(-2\ii\lambda\right)^n}, \quad i=1,2,\quad N\log\left(\frac{\lambda-\lambda_1}{\lambda-\lambda_1^*}\right)=-N\sum_{n=1}^{\infty}\frac{(-2\ii)^n\left(\lambda_1^n-\lambda_1^{*,n}\right)}{n\left(-2\ii\lambda\right)^n}
\end{split}
\end{equation}
Then $\mu^{(1)}(\lambda; x, t)$ and $\mu^{(2)}(\lambda; x, t)$ satisfy the coupled Riccati equations:
\begin{equation}
\begin{split}
\mu^{(1)}_{x}&=2\ii\lambda\mu^{(1)}+\ii q_1^*-\ii\mu^{(1)}\left( q_1\mu^{(1)}+q_2\mu^{(2)}\right),\\
\mu^{(2)}_{x}&=2\ii\lambda\mu^{(2)}+\ii q_2^*-\ii\mu^{(2)}\left( q_1\mu^{(1)}+q_2\mu^{(2)}\right).
\end{split}
\end{equation}
Comparing the coefficients of $\lambda$, we can get the following recurrence relation
\begin{equation}
\begin{split}
\mu^{(1)}_{n+1}&=-\mu^{(1)}_{n,x}+\ii\sum_{k=1}^{n-1}\mu^{(1)}_{k}\left(q_1\mu^{(1)}_{n-k}+q_{2}\mu^{(2)}_{n-k}\right),\\
\mu^{(2)}_{n+1}&=-\mu^{(2)}_{n,x}+\ii\sum_{k=1}^{n-1}\mu^{(2)}_{k}\left(q_1\mu^{(1)}_{n-k}+q_{2}\mu^{(2)}_{n-k}\right).
\end{split}
\end{equation}
With a simple calculation, the first three items can be given as
\begin{equation}
\begin{split}
\mu_{1}^{(1)}&=-\ii q_1^*,\qquad \mu^{(1)}_2=\ii q_{1,x}^*,\qquad \mu^{(1)}_{3}=-\ii\left[q_{1,xx}^{*}+q^*_{1}\left(|q_1|^2+|q_2|^2\right)\right], \\
\mu_{1}^{(2)}&=-\ii q_2^*,\qquad \mu^{(1)}_2=\ii q_{2,x}^*, \qquad \mu^{(2)}_{3}=-\ii\left[q_{2,xx}^{*}+q^*_{2}\left(|q_1|^2+|q_2|^2\right)\right].
\end{split}
\end{equation}
Then the conservation laws are
\begin{equation}
I_n=\int_{-\infty}^{+\infty}\left(\ii q_{1}\mu^{(1)}_{n}+\ii q_{2}\mu^{(2)}_{n}\right) dx,
\end{equation}
With the asymptotics of $\widetilde{\mathbf{m}}(\lambda; x, t)$ when $x\to\pm\infty$ in lemma \ref{lem:asym-m}, we know the conservation laws $I_{n}$ equal to
\begin{equation}
I_{n}=N\frac{\left(-2\ii\right)^n\left(\lambda_1^n-\lambda_1^{*,n}\right)}{n}.
\end{equation}
\end{proof}

\section{\appendixname : The asymptotic analysis to $\mathbf{Q}_{f,a}(\lambda; X, T)$ and $\mathbf{R}(\Lambda; \chi; \tau)$}
\label{App:large-chi}
From the jump matrix about the $\mathbf{Q}_{f,a}(\lambda; X, T)$ in Eq.\eqref{eq:Q-jump-al} and $\mathbf{R}(z; \chi, v)$ in Eq.\eqref{eq:R-jump}, we know
\begin{equation}
\begin{split}
\mathbf{Q}_{f, a, +}(\lambda; X, T)&=\mathbf{Q}_{f, a, -}(\lambda; X, T){\rm diag}\left(2, \frac{1}{2}, 1\right), \lambda\in I,\\
\mathbf{R}_+(z; \chi, v)&=\mathbf{R}_-(z; \chi, v){\rm diag}\left(2, \frac{1}{2}, 1\right), z\in I,
\end{split}
\end{equation}
and the phase term in the asymptotics for the large order soliton is $\ee^{\pm N\phi(\lambda; X, T)}$; while for the case of infinite order soliton, the phase term is $\ee^{\pm\chi^{1/2}\vartheta(z; v)}$. Thus the asymptotics in both two cases is similar, we just give a detailed analysis to $\mathbf{R}(z; \chi, v)$. And to the large order case, we only need a transformation between $\chi$ and $N$, $z$ and $\lambda$.
\subsection{Parametrix construction}
To deal with the jump condition on contour $I$ for the infinite order soliton, set the outer parametrix matrix $\dot{\mathbf{R}}^{\rm out}(z; \chi, v)$ as
\begin{equation}
\dot{\mathbf{R}}^{\rm out}(z; \chi, v)={\rm diag}\left(\left(\frac{z-a(v)}{z-b(v)}\right)^{\ii p}, \left(\frac{z-a(v)}{z-b(v)}\right)^{-\ii p},1 \right),
\end{equation}
where $p=\ln 2/(2\pi)$.
We might hope $\dot{\mathbf{R}}^{\rm out}(z; \chi, v)$ is the global solution of the function $\mathbf{R}(z; \chi, v)$, but unfortunately, $\dot{\mathbf{R}}^{\rm out}(z; \chi, v)$ has two critical points at $z=a(v)$ and $z=b(v)$. To tackle with this problem, we should reconstruct two better matrices in the neighbourhood of $a(v)$ and $b(v)$ that not only satisfy the jump condition in the local neighbourhood but also match well onto the outer parametrix when $\zeta\to\infty$. Before constructing the local matrix, we first give a proposition about some local properties at $a(v)$ and $b(v)$.
\begin{prop}\label{prop1}
Suppose $a(v)<b(v)$, then we have $\widetilde{\vartheta}''(a(v); v)<0, \widetilde{\vartheta}''(b(v); v)>0$, where $\widetilde{\vartheta}(z; v):=-\ii \vartheta(z; v)$.
\end{prop}
\begin{proof}
Under the condition $0<v<\frac{\sqrt{3}}{9}, $ we set $$f(z):=\frac{z^2\widetilde{\vartheta}'(z; v)}{2v}=\left(z-z_1\right)\left(z-z_2\right)\left(z-z_3\right):=\left(z-z_1\right)\left(z-a(v)\right)\left(z-b(v)\right),$$ where $z_1<z_2<z_3$. With the  aid of properties of cubic equation, we have
\begin{equation}
z_1+z_2+z_3=-\frac{1}{2v}, \qquad z_1z_2z_3=\frac{1}{2v},
\end{equation}
which indicates $z_1<z_2<0<z_3$. The second derivative about $\widetilde{\vartheta}(z; v)$ is $\widetilde{\vartheta}''(z; v)=2v+2z^{-3}$, it is obvious that $\widetilde{\vartheta}''(b(v); v)>0$. If $z\in\left\{z|-v^{-1/3}<z<0\right\}, $ then $\widetilde{\vartheta}''(z; v)<0$. Substituting $z=-v^{-1/3}$ into $f(z)$, we have $f(-v^{-1/3})>0$, and $f(z_2)=0$. Based on the monotonicity of $f(z)$, we get $(-v^{-1/3})<z_2<0$, thus $\widetilde{\vartheta}''(a(v); v)<0$.
\end{proof}
To study the solutions in the neighbourhood of $z=a(v)$ and $z=b(v)$, we set a conformal mapping $f_a(z; v)$ and $f_b(z; v)$ in the near $z=a(v)$ and $z=b(v)$ respectively,
\begin{equation}\label{eq:con-map}
f_a(z; v)^2=2\left(\widetilde{\vartheta}(a(v); v)-\widetilde{\vartheta}(z; v)\right), \qquad f_b(z; v)^2=2\left(\widetilde{\vartheta}(z; v)-\widetilde{\vartheta}(b(v); v)\right).
\end{equation}
By the proposition \ref{prop1}, we have $f'_a(a(v); v)<0, f'_{b}(b(v); v)>0,$ so the mapping is reasonable. Set $\zeta_a:=\chi^{1/4}f_a, \zeta_b=\chi^{1/4} f_b$, then we define two new matrices as
\begin{multline}\label{eq:uaub}
\mathbf{U}_a(\zeta)=\begin{bmatrix}\ii \pmb{\sigma}_2&0\\
\mathbf{0}&1
\end{bmatrix}^{-1}{\rm diag}\left(\ee^{\ii \chi^{1/2}\widetilde{\vartheta}(a; v)\hat{\pmb{\sigma}}_3}, 1\right)\mathbf{R}(z; \chi, v){\rm diag}\left(\ee^{-\ii \chi^{1/2}\widetilde{\vartheta}(a; v)\hat{\pmb{\sigma}}_3}, 1\right)\begin{bmatrix}\ii\pmb{\sigma}_2&0\\
\mathbf{0}&1
\end{bmatrix}, \\~\text{near} ~z=a,\\
\mathbf{U}_b(\zeta)={\rm diag}\left(\ee^{\ii \chi^{1/2}\widetilde{\vartheta}(b; v)\hat{\pmb{\sigma}}_3}, 1\right)\mathbf{R}(z; \chi, v){\rm diag}\left(\ee^{-\ii \chi^{1/2}\widetilde{\vartheta}(b; v)\hat{\pmb{\sigma}}_3}, 1\right), ~ \text{near}~z=b.\\
\end{multline}
Then the jump contour satisfied by the $\mathbf{U}_a$ and $\mathbf{U}_b$ has five rays, see Fig. \ref{contour-1}.
\begin{figure}[!h]
\centering
\includegraphics[width=0.6\textwidth]{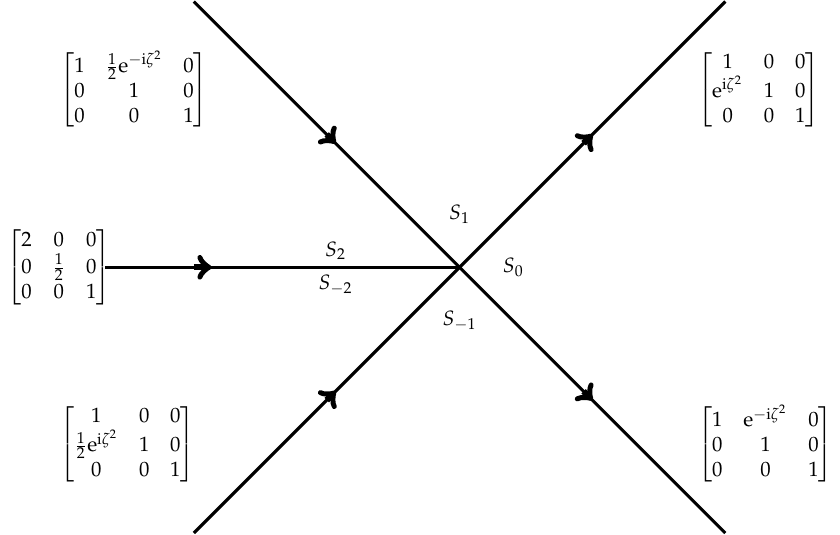}
\caption{The jump satisfied by $\mathbf{U}_a(\zeta)$ near the point $z=a(v)$ and $\mathbf{U}_b(\zeta)$ near $z=b(v)$, that is $\mathbf{U}_+(\zeta)=\mathbf{U}_-(\zeta)\mathbf{V}(\zeta)$, where $\mathbf{V}(\zeta)$ is in the five rays.}
\label{contour-1}
\end{figure}
With the definition of $\mathbf{U}_a(\zeta), \mathbf{U}_b(\zeta)$ in Eq.\eqref{eq:uaub}, we can replace the matrix $\mathbf{R}(z; \chi, v)$ with the outer parametrix $\dot{\mathbf{R}}(z; \chi, v)$ to reveal the asymptotics of $\zeta$ when $\zeta\to\infty$,
\begin{multline}\label{eq:HaHb}
\begin{bmatrix}\ii\pmb{\sigma}_2&\mathbf{0}\\
\mathbf{0}&1
\end{bmatrix}^{-1}{\rm diag}\left(\ee^{\ii \chi^{1/2}\widetilde{\vartheta}(b; v)\hat{\pmb{\sigma}}_3}, 1\right)\dot{\mathbf{R}}^{\rm out}(z; \chi, v){\rm diag}\left(\ee^{\ii \chi^{1/2}\widetilde{\vartheta}(b; v)\hat{\pmb{\sigma}}_3}, 1\right)\begin{bmatrix}\ii \pmb{\sigma}_2&\mathbf{0}\\
\mathbf{0}&1
\end{bmatrix}\\=\begin{bmatrix}\ii \pmb{\sigma}_2&\mathbf{0}\\
\mathbf{0}&1
\end{bmatrix}^{-1}{\rm diag}\left(\chi^{-\ii \frac{p}{4}\hat{\pmb{\sigma}}_3}, 1\right)\mathbf{H}^{a}(z; v){\rm diag}\left(\zeta_a^{-\ii  p\hat{\pmb{\sigma}}_3}, 1\right),
\\\mathbf{H}^a(z; v):={\rm diag}\left(\left(z-b\right)^{-\ii  p\hat{\pmb{\sigma}}_3}, 1\right){\rm diag}
\left(\left(\frac{z-a}{f_{a}(z; v)}\right)^{\ii  p\hat{\pmb{\sigma}}_3}, 1\right)\begin{bmatrix}\ii \pmb{\sigma}_2&\mathbf{0}\\
\mathbf{0}&1
\end{bmatrix},\\
{\rm diag}\left(\ee^{\ii \chi^{1/2}\widetilde{\vartheta}(b; v)\hat{\pmb{\sigma}}_3}, 1\right)\dot{\mathbf{R}}^{\rm out}(z; \chi, v){\rm diag}\left(\ee^{-\ii \chi^{1/2}\widetilde{\vartheta}(b; v)\hat{\pmb{\sigma}}_3}, 1\right)={\rm diag}\left(\chi^{\ii \frac{p}{4}\hat{\pmb{\sigma}}_3}, 1\right)\mathbf{H}^{b}(z; v){\rm diag}\left(\zeta_b^{-\ii  p\hat{\pmb{\sigma}}_3}, 1\right), \\
\mathbf{H}^{b}(z; v):={\rm diag}\left(\left(z-a\right)^{\ii  p\hat{\pmb{\sigma}}_3}, 1\right){\rm diag}\left(\left(\frac{f_b(z; v)}{z-b}\right)^{\ii  p\hat{\pmb{\sigma}}_3}, 1\right).
\end{multline}
From the definition of $f_a$ and $f_b$, we know the matrices $\mathbf{H}^a(z; v)$ and $\mathbf{H}^b(z; v)$ are analytic at $z=a(v)$ and $z=b(v)$ respectively, and the singularities appear in the last factor $\zeta^{-\ii  p\hat{\pmb{\sigma}}_3}$, which is a critical factor to the boundary condition in the following parabolic cylinder function:
\begin{rhp}\label{RHP4}(Parabolic cylinder parametrix) Find a $3\times 3$ matrix $\mathbf{U}(\zeta)$ satisfying the following properties:
\begin{itemize}
\item \textbf{Analyticity}:$\mathbf{U}(\zeta)$ is analytic in five sectors as shown in Fig.\ref{contour-1}, $S_0: |\arg (\zeta)|<\frac{1}{4}\pi, S_1: \frac{1}{4}\pi<\arg(\zeta)<\frac{3}{4}\pi, S_{-1}: -\frac{3}{4}\pi<\arg(\zeta)<-\frac{1}{4}\pi, S_2: \frac{3}{4}\pi <\arg(\zeta)< \pi, S_{-2}: -\pi<\arg(\zeta)<-\frac{3}{4}\pi. $ Their orientation are given with the arrow.
\item\textbf{Jump Condition}: $\mathbf{U}_{+}(\zeta)=\mathbf{U}_-(\zeta)\mathbf{V}(\zeta)$, where $\mathbf{V}(\zeta)$ is shown in Fig.\ref{contour-1}.
\item \textbf{Normalization}: $\mathbf{U}(\zeta){\rm diag}\left(\zeta^{\ii  p\hat{\pmb{\sigma}}_3}, 1\right)\to \mathbb{I}$ as $\zeta\to \infty$.
\end{itemize}
\end{rhp}
The unique solution of RHP \ref{RHP4} can also be given by the classical parabolic cylinder $D_a(\zeta)$. Compared with the $2\times 2$ order problem, the solution of $3\times 3$ type in this case just has a little difference, we will show how to derive the second order ordinary equation through the following theorem.
\begin{theorem}\label{thm-painleve}
Suppose $\zeta\to\infty,$ the asymptotic expression of $\mathbf{U}(\zeta){\rm diag}\left(\zeta^{\ii  p\hat{\pmb{\sigma}}_3}, 1\right)$ is
\begin{equation}
\mathbf{U}(\zeta){\rm diag}\left(\zeta^{\ii  p\hat{\pmb{\sigma}}_3}, 1\right):=\mathbb{I}+\frac{1}{2\ii\zeta}\begin{bmatrix}0&r&0\\
-s&0&0\\
0&0&0
\end{bmatrix}+O(\zeta^{-2}),
\end{equation}
then we have
\begin{equation}
r=-s^*=2^{5/4}\ee^{\frac{\ii\ln(2)^2}{2\pi}}\sqrt{\pi}\ee^{\frac{\pi}{4}\ii}\left(\Gamma\left(\ii\frac{\ln(2)}{2\pi}\right)\right)^{-1}.
\end{equation}
\end{theorem}
\begin{proof}
We first define a new matrix $\widetilde{\mathbf{U}}$ as
\begin{equation}
\begin{aligned}
\widetilde{\mathbf{U}}(\zeta):&=\mathbf{U}(\zeta), &\zeta\in S_1,\\
\widetilde{\mathbf{U}}(\zeta):&=\mathbf{U}(\zeta)\begin{bmatrix}1&\frac{1}{2}\ee^{-\ii\zeta^2}&0\\
0&1&0\\
0&0&1
\end{bmatrix}, &\zeta\in S_2,\\
\widetilde{\mathbf{U}}(\zeta):&=\mathbf{U}(\zeta)\begin{bmatrix}1&0&0\\
-\frac{1}{2}\ee^{\ii\zeta^2}&1&0\\
0&0&1
\end{bmatrix}, &\zeta\in S_{-2},\\
\widetilde{\mathbf{U}}(\zeta):&=\mathbf{U}(\zeta), &\zeta\in S_{-1},\\
\widetilde{\mathbf{U}}(\zeta):&=\mathbf{U}(\zeta)\begin{bmatrix}1&0&0\\
\ee^{\ii\zeta^2}&1&0\\
0&0&1
\end{bmatrix}, &\zeta\in S_{0}^{[2]},\\
\widetilde{\mathbf{U}}(\zeta):&=\mathbf{U}(\zeta)\begin{bmatrix}1&-\ee^{-\ii\zeta^2}&0\\
0&1&0\\
0&0&1
\end{bmatrix}, &\zeta\in S_{0}^{[-2]},\\
\end{aligned}
\end{equation}
where $S_0^{[2]}$ and $S_0^{[-2]}$ are in the upper and lower domain in $S_0$, which are symmetric to the $S_2$ and $S_{-2}$ respectively. Then the new jump matrices about $\widetilde{\mathbf{U}}(\zeta)$ transfer to the real axis, it changes into
\begin{equation}
\widetilde{\mathbf{U}}_{+}(\zeta)=\widetilde{\mathbf{U}}_{-}(\zeta)\begin{bmatrix}2&\ee^{-\ii\zeta^2}&0\\
\ee^{\ii\zeta^2}&1&0\\
0&0&1
\end{bmatrix}, \qquad \zeta\in \mathbb{R}.
\end{equation}
When $\zeta\to\infty$, $\mathbf{U}(\zeta)$ and $\widetilde{\mathbf{U}}(\zeta)$ has the same asymptotic expression, suppose the asymptotic expression is
\begin{equation}\label{eq:uasymptotic}
\widetilde{\mathbf{U}}(\zeta){\rm diag}\left(\zeta^{\ii  p\hat{\pmb{\sigma}}_3}, 1\right)=\mathbb{I}+\frac{1}{2\ii\zeta}\begin{bmatrix}0&r&0\\
-s&0&0\\0&0&0
\end{bmatrix}+O(\zeta^{-2}),\qquad \zeta\to\infty.
\end{equation}
Set $\mathbf{W}(\zeta):=\widetilde{\mathbf{U}}(\zeta){\rm diag}\left(\ee^{-\frac{\ii\zeta^2}{2}\hat{\pmb{\sigma}}_3}, 1\right), $ then we have
\begin{equation}\label{eq:w-jump}
\mathbf{W}_+(\zeta)=\mathbf{W}_{-}(\zeta)\begin{bmatrix}2&1&0\\1&1&0\\
0&0&1
\end{bmatrix},
\end{equation}
whose jump is independent with $\zeta$, thus $\mathbf{W}(\zeta)_\zeta\mathbf{W}(\zeta)^{-1}$ is an entire function in the whole $\zeta$ plane. With the ansatz when $\zeta\to \infty$, we have
\begin{equation}
\mathbf{W}(\zeta)_{\zeta}=\begin{bmatrix}-\ii\zeta&r&0\\s&\ii\zeta&0\\
0&0&0
\end{bmatrix}\mathbf{W}(\zeta),
\end{equation}
which indicates
\begin{equation}\label{eq:ODE}
\begin{aligned}
w^{\pm}_{11,\zeta\zeta}&=\left(-\zeta^2-\ii+rs\right)w^{\pm}_{11},\\
w^{\pm}_{22,\zeta\zeta}&=\left(-\zeta^2+\ii+rs\right)w^{\pm}_{22},\\
w^{\pm}_{33,\zeta}&=0.
\end{aligned}
\end{equation}
Next, we give the solutions of Eq.\eqref{eq:ODE} with the parabolic cylinder functions. Set $t=\sqrt{2}\ee^{-\frac{3}{4}\pi\ii}\zeta$, when ${\rm Im}(\zeta)>0$, the first equation in Eq.\eqref{eq:ODE} changes into
\begin{equation}
w^{+}_{11,tt}+\left(\frac{1}{2}-\frac{t^2}{4}+\frac{\ii rs}{2}\right)w_{11}=0,
\end{equation}
which is a standard parabolic cylinder function, whose solution is
\begin{equation}
w^{+}_{11}=C_1D_{a}\left(\sqrt{2}\ee^{-\frac{3}{4}\pi\ii}\zeta\right)+C_2D_a(-\sqrt{2}\ee^{-\frac{3}{4}\pi\ii}\zeta),
\end{equation}
where $a=\frac{\ii rs}{2}$. From \cite{Whittaker-book-1927,Zhou-Annals-1993}, when $\zeta\to\infty$, the asymptotic behavior of $D_a(\zeta)$ is
\begin{equation}
\begin{split}
D_a(\zeta)&=\zeta^a\ee^{-\frac{1}{4}\zeta^2}\left(1+O(\zeta^{-2})\right), \qquad |\arg \zeta|<\frac{3\pi}{4},\\
D_a(\zeta)&=\zeta^a\ee^{-\frac{1}{4}\zeta^2}\left(1+O(\zeta^{-2})\right)
-\frac{\sqrt{2\pi}}{\Gamma(-a)}\ee^{a\pi \ii}\zeta^{-a-1}\ee^{\frac{1}{4}\zeta^2}\left(1+O(\zeta^{-2})\right), \quad \frac{\pi}{4}<\arg \zeta<\frac{5\pi}{4},\\
D_a(\zeta)&=\zeta^a\ee^{-\frac{1}{4}\zeta^2}\left(1+O(\zeta^{-2})\right)
-\frac{\sqrt{2\pi}}{\Gamma(-a)}\ee^{-a\pi \ii}\zeta^{-a-1}\ee^{\frac{1}{4}\zeta^2}\left(1+O(\zeta^{-2})\right), \quad -\frac{5\pi}{4}<\arg \zeta<-\frac{\pi}{4}.
\end{split}
\end{equation}
With this asymptotic behavior, we know $a=-\ii p$ and $c_1=\sqrt{2}^{\ii p}\ee^{\frac{3\pi}{4}p}$, so $$w_{11}^+=\left(\sqrt{2}\right)^{\ii p}\ee^{\frac{3\pi}{4}p}D_{a}(\sqrt{2}\ee^{-\frac{3\pi}{4}\ii}\zeta).$$ Similarly, we can also calculate the other terms
\begin{equation}
\begin{split}
w_{22}^{+}&=\left(\sqrt{2}\right)^{-\ii p}\ee^{-\frac{\pi}{4}p}D_{-a}\left(\sqrt{2}\ee^{-\frac{\pi}{4}\ii}\zeta\right),\\
w_{11}^{-}&=\left(\sqrt{2}\right)^{\ii p}\ee^{-\frac{\pi}{4}p}D_{a}\left(\sqrt{2}\ee^{\frac{\pi}{4}\ii}\zeta\right),\\
w_{22}^-&=\left(\sqrt{2}\right)^{-\ii p}\ee^{\frac{3\pi}{4}p}D_{-a}\left(\sqrt{2}\ee^{\frac{3\pi}{4}\ii}\zeta\right),\\
w^{\pm}_{33}&=1,\qquad w^{\pm}_{13}=w^{\pm}_{31}=w^{\pm}_{23}=w^{\pm}_{32}=0.
\end{split}
\end{equation}
Next, we should determine $r$ and $s$. From the relation Eq.\eqref{eq:w-jump}, we know
\begin{equation}
w_{11}^-w_{21}^+-w_{11}^+w_{21}^-=1, \qquad w_{22}^-w_{12}^+-w_{12}^-w_{22}^+=1,
\end{equation}
which indicates
\begin{equation}
\begin{split}
r&=\left(\sqrt{2}\right)^{2\ii p}\ee^{\frac{\pi}{2}p}\left[D_{a}\left(\sqrt{2}\ee^{\frac{\pi }{4}\ii}\zeta\right)\frac{d}{d\zeta}D_{a}\left(\sqrt{2}\ee^{-\frac{3\pi}{4}\ii}\zeta\right)
-D_{a}\left(\sqrt{2}\ee^{-\frac{3\pi }{4}\ii}\zeta\right)\frac{d}{d\zeta}D_{a}\left(\sqrt{2}\ee^{\frac{\pi}{4}\ii}\zeta\right)\right]\\
s&=\left(\sqrt{2}\right)^{-2\ii p}\ee^{\frac{\pi}{2}p}\left[D_{-a}\left(\sqrt{2}\ee^{\frac{3\pi }{4}\ii}\zeta\right)\frac{d}{d\zeta}D_{-a}
\left(\sqrt{2}\ee^{-\frac{\pi}{4}\ii}\zeta\right)
-D_{-a}\left(\sqrt{2}\ee^{-\frac{\pi }{4}\ii}\zeta\right)\frac{d}{d\zeta}D_{a}
\left(\sqrt{2}\ee^{\frac{3\pi}{4}\ii}\zeta\right)\right],\\
\Rightarrow\\
r&=\left(\sqrt{2}\right)^{2\ii p}\ee^{\frac{\pi}{2}p}W\left[D_{a}\left(\sqrt{2}\ee^{\frac{\pi }{4}\ii}\zeta\right), D_{a}\left(-\sqrt{2}\ee^{\frac{\pi }{4}\ii}\zeta\right)\right],\\
s&=-\left(\sqrt{2}\right)^{-2\ii p}\ee^{\frac{\pi}{2}p}W\left[D_{-a}\left(\sqrt{2}\ee^{-\frac{\pi }{4}\ii}\zeta\right), D_{-a}\left(-\sqrt{2}\ee^{-\frac{\pi }{4}\ii}\zeta\right)\right],
\end{split}
\end{equation}
where $W[f, g]=fg'-f'g$ is the Wronskian determinant. Throughout the properties $W\left[D_a(\zeta), D_a(-\zeta)\right]=\frac{\sqrt{2\pi}}{\Gamma\left(-a\right)}, $ we have
\begin{equation}\label{eq:rs}
\begin{split}
r&=2^{5/4}\ee^{\frac{\ii\ln(2)^2}{2\pi}}\sqrt{\pi}\ee^{\frac{\pi}{4}\ii}\left(\Gamma\left(\ii\frac{\ln(2)}{2\pi}\right)\right)^{-1},\\
s&=-2^{5/4}\ee^{-\frac{\ii\ln(2)^2}{2\pi}}\sqrt{\pi}\ee^{-\frac{\pi}{4}\ii}\left(\Gamma\left(-\ii\frac{\ln(2)}{2\pi}\right)\right)^{-1}.\\
\end{split}
\end{equation}
\end{proof}
By the properties of complex $\Gamma$ function, we have $r=\sqrt{\frac{\ln(2)}{\pi}}\ee^{\frac{\ln (2)^2}{2\pi}\ii+\frac{\pi}{4}\ii-\arg(\Gamma(\ii\frac{\ln (2)}{2\pi}))\ii},$ $|r|=\sqrt{\frac{\ln(2)}{\pi}}, s=-r^*$. Based on the definition of $\mathbf{U}(\zeta)$, we can construct the inner parametrix at $z=a, b$. Set $D_{\zeta}(\delta)$ be a disk centered at $\zeta$ with the radius $\delta$, when $\delta$ is small, we can define the following inner paramerix
\begin{multline}
\dot{\mathbf{R}}^a(z; \chi, v):={\rm diag}\left(\ee^{-\ii \chi^{1/2}\widetilde{\vartheta}(a; v)\hat{\pmb{\sigma}}_3}, 1\right){\rm diag}\left(\chi^{-\ii \frac{p}{4}\hat{\pmb{\sigma}}_3}, 1\right)\mathbf{H}^{a}(z; v)\mathbf{U}(\chi^{1/4}f_a)\\\cdot\begin{bmatrix}\ii \pmb{\sigma}_2&0\\
\mathbf{0}&1
\end{bmatrix}^{-1}{\rm diag}\left(\ee^{\ii \chi^{1/2}\widetilde{\vartheta}(a; v)\hat{\pmb{\sigma}}_3}, 1\right),\quad z\in D_{a}(\delta),\\
\dot{\mathbf{R}}^b(z; \chi, v):=\ee^{-\ii \chi^{1/2}\widetilde{\vartheta}(b; v)\pmb{\sigma}}\chi^{\ii \frac{p}{4}\pmb{\sigma}}\mathbf{H}^{b}(z; v)\mathbf{U}(\chi^{1/4}f_b){\rm diag}\left(\ee^{\ii \chi^{1/2}\widetilde{\vartheta}(b; v)\hat{\pmb{\sigma}}_3}, 1\right),\quad z\in D_{b}(\delta),
\end{multline}
thus we have an estimate to the $\dot{\mathbf{R}}^a(z; \chi, v)\dot{\mathbf{R}}^{-1}(z; \chi, v)$ and $\dot{\mathbf{R}}^b(z; \chi, v)\dot{\mathbf{R}}^{-1}(z; \chi, v),$
\begin{multline}\label{eq:rarout}
\dot{\mathbf{R}}^a(z; \chi, v)\left(\dot{\mathbf{R}}^{\rm out}(z; \chi, v)\right)^{-1}={\rm diag}\left(\ee^{-\ii \chi^{1/2}\widetilde{\vartheta}(a; v)\hat{\pmb{\sigma}}_3}, 1\right){\rm diag}\left(\chi^{-\ii \frac{p}{4}\hat{\pmb{\sigma}}_3}, 1\right)\mathbf{H}^{a}(z; v)\mathbf{U}(\chi^{1/4}f_a)\\\cdot {\rm diag}\left(\zeta_a^{\ii  p\hat{\pmb{\sigma}}_3}, 1\right)\left(\mathbf{H}^{a}(z; v)\right)^{-1}{\rm diag}\left(\chi^{\ii \frac{p}{4}\hat{\pmb{\sigma}}_3}, 1\right){\rm diag}\left(\ee^{\ii \chi^{1/2}\widetilde{\vartheta}(a; v)\hat{\pmb{\sigma}}_3}, 1\right),\, \zeta\in\partial D_{a}(\delta)\\
\dot{\mathbf{R}}^b(z; \chi, v)\left(\dot{\mathbf{R}}^{\rm out}(z; \chi, v)\right)^{-1}={\rm diag}\left(\ee^{-\ii \chi^{1/2}\widetilde{\vartheta}(b; v)\hat{\pmb{\sigma}}_3}, 1\right){\rm diag}\left(\chi^{\ii \frac{p}{4}\hat{\pmb{\sigma}}_3}, 1\right)\mathbf{H}^{b}(z; v)\mathbf{U}(\chi^{1/4}f_b)\\\cdot {\rm diag}\left(\zeta_b^{\ii  p\hat{\pmb{\sigma}}_3}, 1\right)\left(\mathbf{H}^{b}(z; v)\right)^{-1}{\rm diag}\left(\chi^{-\ii \frac{p}{4}\hat{\pmb{\sigma}}_3}, 1\right){\rm diag}\left(\ee^{\ii \chi^{1/2}\widetilde{\vartheta}(b; v)\hat{\pmb{\sigma}}_3}, 1\right), \, \zeta\in\partial D_{b}(\delta).
\end{multline}
Thus the global parametrix for $\mathbf{R}(z;\chi,v)$ is given by
\begin{equation}
\dot{\mathbf{R}}(z; \chi, v):=\left\{\begin{aligned}&\dot{\mathbf{R}}^a(z; \chi, v),&\qquad z\in D_{a}(\delta)\\
&\dot{\mathbf{R}}^b(z; \chi, v),&\qquad z\in D_{b}(\delta)\\
&\dot{\mathbf{R}}^{\rm out}(z; \chi, v),&\qquad z\in \mathbb{C}\setminus\left(\mathbb{I}\cup \overline{D_{a}(\delta)}\cup \overline{D_{b}(\delta)}\right)
\end{aligned}\right.
\end{equation}
\subsection{Error Analysis}
\label{sec:error-analysis}
After constructing the global parametrix $\dot{\mathbf{R}}(z; \chi, v)$ to $\mathbf{\mathbf{R}}(z; \chi, v)$, we can analyze the error between these two matrices, which is defined as
\begin{equation}\label{eq:f-definition}
\mathbf{F}(z; \chi, v):=\mathbf{R}(z; \chi, v)\dot{\mathbf{R}}(z; \chi, v)^{-1}.
\end{equation}
Set the jump contour about $\mathbf{F}(z; \chi, v)$ to $\Sigma_F$. Due to $\dot{\mathbf{R}}(z; \chi, v)$ has three different definitions on three different domains, we should give the error analysis with this three cases. In the neighbourhood of $D_a(\delta)$ and $D_b(\delta)$, $\mathbf{R}(z; \chi, v)$ and $\dot{\mathbf{R}}(z; \chi, v)$ satisfy the same jump matrices, so $\mathbf{F}(z; \chi, v)$ has no jump in these two domains. In the outer of the domain, the jump becomes
\begin{equation}\label{eq:VF-1}
\mathbf{V}^F(z; \chi, v)=\mathbf{F}_-(z; \chi, v)^{-1}\mathbf{F}_{+}(z; \chi, v)=\dot{\mathbf{R}}(z; \chi, v)\mathbf{R}_-(z; \chi, v)\mathbf{R}_+(z; \chi, v)\dot{\mathbf{R}}(z; \chi, v)^{-1},
\end{equation}
from the jump matrices in Eq.\eqref{eq:R-jump}, we have
\begin{equation}
\|\mathbf{V}^{F}(z; \chi, v)-\mathbb{I}\|_{z\in \Sigma_{F}\setminus \left(\partial D_a\cup \partial D_b\right)}=\mathcal{O}(\ee^{-\chi^{1/2}K(v)}).
\end{equation}
Additionally, $\mathbf{F}(z; \chi, v)$ is nonanalytic in the boundary of $D_{a}(\delta)$ and $D_{b}(\delta)$, the jump condition becomes
\begin{equation}
\begin{split}
\mathbf{V}^{F}(z; \chi, v)=\dot{\mathbf{R}}^a(z; \chi, v)\dot{\mathbf{R}}^{\rm out}(z; \chi, v)^{-1},\qquad z\in\partial D_a,\\
\mathbf{V}^{F}(z; \chi, v)=\dot{\mathbf{R}}^b(z; \chi, v)\dot{\mathbf{R}}^{\rm out}(z; \chi, v)^{-1},\qquad z\in\partial D_b,
\end{split}
\end{equation}
which is the Eq.\eqref{eq:rarout}. With the asymptotic expression in Eq.\eqref{eq:uasymptotic}, we know
\begin{equation}
\|\mathbf{V}^{F}(z; \chi, v)-\mathbb{I}\|_{z\in \partial D_a\cup \partial D_b}=\mathcal{O}(\chi^{-1/4}).
\end{equation}
Lastly, we will give the asymptotic expression to $\mathbf{q}(\chi, \tau)$. The starting point is analyzing jump matrix in Eq.\eqref{eq:f-definition}. As we have given the estimation to $\mathbf{V}^{F}(z; \chi, v)-\mathbb{I}$ in the jump contour, which is useful to the later calculation. From Eq.\eqref{eq:f-definition}, we have
\begin{equation}\label{eq:f-jump}
\mathbf{F}_+(z; \chi, v)=\mathbf{F}_-(z; \chi, v)\mathbf{V}^{F}(z; \chi, v), \qquad z\in\Sigma_F.
\end{equation}
As to this matrix-type RHP, we use the general Plemelj formula to give its solution. So we should rewrite Eq.\eqref{eq:f-jump} to another form:
\begin{equation}
\mathbf{F}_+(z; \chi, v)-\mathbf{F}_-(z; \chi, v)=\mathbf{F}_-(z; \chi, v)\left(\mathbf{V}^{F}(z; \chi, v)-\mathbb{I}\right),
\end{equation}
so the solution of $\mathbf{F}(z; \chi, v)$ is given by
\begin{multline}\label{eq:fsolution}
\mathbf{F}(z; \chi, v)=\mathbb{I}+\frac{1}{2\pi \ii}\int_{\Sigma_F}\frac{\mathbf{F}_-(z'; \chi, v)\left(\mathbf{V}^{F}(z'; \chi, v)-\mathbb{I}\right)}{z'-z}dz'\\
{=}\mathbb{I}{+}\frac{1}{2\pi \ii}\int_{\partial D_a\cup \partial D_b}\frac{\mathbf{F}_-(z'; \chi, v)\left(\mathbf{V}^{F}(z'; \chi, v){-}\mathbb{I}\right)}{z'-z}dz'\\+\frac{1}{2\pi \ii}\int_{\Sigma^F\setminus\left(\partial D_a\cup \partial D_b\right)}\frac{\mathbf{F}_-(z'; \chi, v)\left(\mathbf{V}^{F}(z'; \chi, v)-\mathbb{I}\right)}{z'-z}dz'.
\end{multline}
To get the potential $\mathbf{q}(\chi, \tau)$, we should make a series expansion to $\mathbf{F}(z; \chi, v)$ when $z\to\infty$, thus Eq.\eqref{eq:fsolution} changes into
\begin{multline}\label{eq:fsolution1}
\mathbf{F}(z; \chi, v)=\mathbb{I}-\sum\limits_{l=1}^{\infty}\frac{1}{2\pi \ii}z^{-l}\int_{\partial D_a\cup \partial D_b}\mathbf{F}_-(z'; \chi, v)\left(\mathbf{V}^{F}(z'; \chi, v)-\mathbb{I}\right)\left(z'\right)^{l-1}dz'\\-\sum_{l=1}^{\infty}\frac{1}{2\pi \ii}z^{-l}\int_{\Sigma_F\setminus\left(\partial D_a\cup \partial D_b\right)}\mathbf{F}_-(z'; \chi, v)\left(\mathbf{V}^{F}(z'; \chi, v)-\mathbb{I}\right)\left(z'\right)^{l-1}dz'.
\end{multline}
By the expansion of series, we will establish the relation between the potential $\mathbf{q}(\chi, \tau)$ and the matrix $\mathbf{F}(z; \chi, v)$. Combining with several transformations, we get
\begin{equation}
\begin{split}
\mathbf{q}(\chi, \tau)&=2\lim\limits_{z\to\infty}z\chi^{-1/2}\left(\mathbf{Q}_{H}\mathbf{S}(z; \chi, v)\mathbf{Q}_{H}^{-1}\right)_{12}(z; \chi, v)\\
&=2\lim\limits_{z\to\infty}z\chi^{-1/2}
\left(\mathbf{Q}_{H}\mathbf{R}(z; \chi, v)\mathbf{Q}_{H}^{-1}\right)_{12}\\
&=2\lim\limits_{z\to\infty}z\chi^{-1/2}
\left(\mathbf{Q}_{H}\mathbf{F}(z; \chi, v)\dot{\mathbf{R}}(z; \chi, v)\mathbf{Q}_{H}^{-1}\right)_{12}.
\end{split}
\end{equation}
From Eq.\eqref{eq:fsolution1}, we know
\begin{equation}\label{eq:q1q2}
\begin{split}
q_1(\chi, \chi^{3/2}v)&=-\frac{1}{\pi\ii \chi^{1/2}}c_1^*\int_{\partial D_a\cup\partial D_b}\mathbf{V}^{F}_{12}(z'; \chi, v)dz'+O(\chi^{-1})\\
=&-\frac{1}{\pi\ii \chi^{1/2}}c_1^*\int_{\partial D_a}-\frac{\ii}{2}s\ee^{-2\ii \chi^{1/2}\widetilde{\vartheta}(a; v)}\chi^{-\frac{\ii}{2}p}(z'-b(v))^{-2\ii p}\left(\frac{z'-a(v)}{f_a(z')}\right)^{2\ii p}\frac{1}{\zeta_a(z')}dz'\\
&-\frac{1}{\pi\ii \chi^{1/2}}c_1^*\int_{\partial D_b}-\frac{\ii}{2}r\ee^{-2\ii \chi^{1/2}\widetilde{\vartheta}(b; v)}\chi^{\frac{\ii}{2}p}(z'-a(v))^{2\ii p}\left(\frac{f_b(z')}{z'-b(v)}\right)^{2\ii p}\frac{1}{\zeta_b(z')}dz'+O(\chi^{-1}).
\end{split}
\end{equation}
Through the definition of $f_a(z; v)$ and $f_b(z; v)$, we have
\begin{equation}
f_a'(a(v); v)=-\sqrt{-\widetilde{\vartheta}''(a(v); v)}, \qquad f_b'(b(v); v)=\sqrt{\widetilde{\vartheta}''(b(v); v)},
\end{equation}
then Eq.\eqref{eq:q1q2} changes into Eq.\eqref{eq:q-near}.
\section{\appendixname : The asymptotic analysis to $\widehat{\mathbf{Q}}_{f, n}(\lambda; X, T)$ and $\widehat{\mathbf{Q}}_n(z; \tau, w)$}
\label{app:no}
From the jump matrix about the $\widehat{\mathbf{Q}}_{f, n}(\lambda; X, T)$ and $\widehat{\mathbf{Q}}_{n}(z; \tau, w)$, we can see that they satisfy the same type jump condition. Thus we only give the parametrix construction to the $\widehat{\mathbf{Q}}_{f, n}(\lambda; X, T)$, the other one can be derived similarly. Based on the framework in \cite{Peter-Duke-2019}, from the jump condition Eq.\eqref{eq:jump-hatQ}, the outer parametrix $\widehat{\mathbf{Q}}^{\rm out}_{f, n}$ can be set as
\begin{equation}
\widehat{\mathbf{Q}}_{f,n}^{\rm out}(\lambda; X, T)=\mathbf{K}(\lambda; X, T){\rm diag}\left(\left(\frac{\lambda-\lambda_c}{\lambda-\lambda_d}\right)^{\ii p\hat{\pmb{\sigma}}_3}, 1\right), \quad p=\frac{\log(2)}{2\pi}.
\end{equation}
Suppose $\mathbf{K}(\lambda; X, T)$ is bounded in the neighbourhood of $\lambda=\lambda_c, \lambda_d$, and it is analytic for $\lambda\in\mathbb{C}\setminus \left(\Sigma_{\rm u}\cup \Sigma_{\rm d}\right)$. To match the jump condition on the contour $\lambda\in \Sigma_{\rm u}\cup \Sigma_{\rm d}$, then $\mathbf{K}(\lambda; X, T)$ satisfies
\begin{equation}
\begin{split}
\mathbf{K}_{+}(\lambda; X, T)&=\mathbf{K}_{-}(\lambda; X, T){\rm diag}\left(\left(\frac{\lambda-\lambda_c}{\lambda-\lambda_d}\right)^{\ii p\hat{\pmb{\sigma}}_3}, 1\right)\begin{bmatrix}0&\ee^{N\Omega_{n}}&0\\
-\ee^{-N\Omega_{n}}&0&0\\0&0&1
\end{bmatrix}\\
\cdot &{\rm diag}\left(\left(\frac{\lambda-\lambda_c}{\lambda-\lambda_d}\right)^{-\ii p\hat{\pmb{\sigma}}_3}, 1\right),\quad \lambda\in \Sigma_{\rm u}\cup \Sigma_{\rm d}.
\end{split}
\end{equation}
To solve $\mathbf{K}(\lambda; X, T)$, we introduce a new function $k(\lambda; X, T)$, which satisfies the following condition when $\lambda\in \Sigma_{\rm u}\cup \Sigma_{\rm d}$:
\begin{equation}
k_{+}(\lambda; X, T)+k_{-}(\lambda; X, T)=2\ii p\log\left(\frac{\lambda-\lambda_c}{\lambda-\lambda_d}\right)+\ii \mu(X, T), \quad \mu(X, T)=2p\int_{\lambda_c}^{\lambda_d}\frac{1}{R(s)}ds,
\end{equation}
where $R(\lambda)$ is defined in Eq.\eqref{eq:R-no}. With a generalized residue theorem, we know $k(\lambda; X, T)$ equals to
\begin{equation}
k=\ii p\log\left(\frac{\lambda-\lambda_c}{\lambda-\lambda_d}\right)+\ii p R(\lambda)\int_{\lambda_c}^{\lambda_d}\frac{ds}{R(s)(s-\lambda)}+\frac{\ii \mu(X, T)}{2}.
\end{equation}
Set $\mathbf{K}(\lambda; X, T)=\mathbf{J}(\lambda; X, T){\rm diag}\left(\ee^{-k(\lambda; X, T)\hat{\pmb{\sigma}}_3}, 1\right)$, then the jump of $\mathbf{J}(\lambda; X, T)$ changes into
\begin{equation}
\mathbf{J}_{+}(\lambda; X, T)=\mathbf{J}_{-}(\lambda; X, T)\begin{bmatrix}
0&\ee^{N\Omega_n-\ii \mu(X, T)}&0\\
-\ee^{-N\Omega_{n}+\ii \mu(X, T)}&0&0\\
0&0&1
\end{bmatrix},\quad \lambda\in \Sigma_{\rm u}\cup \Sigma_{\rm d}..
\end{equation}

By the linear algebra, the above Riemann-Hilbert problem $\mathbf{J}(\lambda; X, T)$ can be solved as
\begin{multline}
\mathbf{J}(\lambda; X, T)={\rm diag}\left(\ee^{\frac{N\Omega_n-\ii\mu}{2}\hat{\pmb{\sigma}}_3}, 1\right)\mathbf{C}{\rm diag}\left(\left(\frac{\lambda-a_n}{\lambda-a_n^*}\right)^{\hat{\pmb{\sigma}}_3/4},1\right)\mathbf{C}^{-1}{\rm diag}\left(\ee^{\frac{-N\Omega_n+\ii\mu}{2}\hat{\pmb{\sigma}}_3}, 1\right)\\\mathbf{C}:=\frac{1}{\sqrt{2}}\begin{bmatrix}1&1&0\\
\ii&-\ii&0\\
0&0&1
\end{bmatrix}.
\end{multline}
Therefore, the outer parametrix $\widehat{\mathbf{Q}}^{\rm out}_{f, n}(\lambda; X, T)$ converts into
\begin{equation}\label{eq:hatQout}
\widehat{\mathbf{Q}}^{\rm out}_{f, n}=\mathbf{J}(\lambda; X, T){\rm diag}\left(\ee^{-k(\lambda; X, T)\hat{\pmb{\sigma}}_3}, 1\right){\rm diag}\left(\left(\frac{\lambda-\lambda_c}{\lambda-\lambda_d}\right)^{\ii p\hat{\pmb{\sigma}}_3}, 1\right).
\end{equation}
Obviously, $\widehat{\mathbf{Q}}^{\rm out}_{f, n}$ satisfies the same jump condition as $\widehat{\mathbf{Q}}_{f, n}$ when $\lambda\in \Sigma_{\rm u}\cup \Sigma_{\rm d}\cup I$. But in the end points, $\lambda=\lambda_c, \lambda_d, a_n, a_n^*$, the outer parametrix has singularities, so we should give an inner parametrix on a small neighborhood of these points, which are shown in the next subsection.
\subsection{Inner parametrix construction near $\lambda_c$ and $\lambda_d$}
\lmn{During the calculation}, we know $\lambda_c$ and $\lambda_d$ are two critical points of $h(\lambda; X, T)$. Similar to the inner parametrix in Eq.\eqref{eq:con-map}, we also set a conformal mapping between $f_{\lambda_c}, f_{\lambda_d}$ and $h(\lambda; X, T)$. To simplify the calculation, we set $\widehat{h}(\lambda; X, T)=\ii h(\lambda; X, T)/2$, then the corresponding mapping is
\begin{equation}
\begin{split}
f_{\lambda_c}(\lambda; X, T)^2=2\left(\widehat{h}_{\lambda_c}(X, T)-\widehat{h}_-(\lambda; X, T)\right),\\
f_{\lambda_d}(\lambda; X, T)^2=2\left(\widehat{h}(\lambda; X, T)-\widehat{h}_{\lambda_d}(X, T)\right),
\end{split}
\end{equation}
where $\widehat{h}_{\lambda_c}, \widehat{h}_{\lambda_d}$ indicates the value at $\lambda=\lambda_c, \lambda=\lambda_d$. Set $\zeta_{\lambda_c}=N^{1/2}f_{\lambda_c}, \zeta_{\lambda_d}=N^{1/2}f_{\lambda_d}$, then we also have
\begin{equation}
f_{\lambda_c}'(\lambda_c; X, T)^2=-\widehat{h}''_{-}(\lambda_c, X, T)>0,\quad f_{\lambda_d}'(\lambda_d; X, T)^2=\widehat{h}''(\lambda_d; X, T)>0.
\end{equation}
Due to jump condition when $\lambda\in\Sigma_{\rm u}\cup \Sigma_{\rm d}$ in the neighborhood of $\lambda_c$, the piecewise matrix will be a little difference with Eq.\eqref{eq:uaub}, which changes into
\begin{equation}\label{eq:uaub-1}
\begin{aligned}
\mathbf{U}_{\lambda_c}(\zeta)&=\left\{\begin{split}&\begin{bmatrix}\ii \pmb{\sigma}_2&0\\
\mathbf{0}&1
\end{bmatrix}^{-1}{\rm diag}\left(\ee^{\ii N\widehat{h}_{\lambda_c}\hat{\pmb{\sigma}}_3}, 1\right){\rm diag}\left(\ii, -\ii, 1\right)\widehat{\mathbf{Q}}_{f,n}(\lambda; X, T){\rm diag}\left(-\ii, \ii, 1\right)\\\cdot&{\rm diag}\left(\ee^{-\ii N\widehat{h}_{\lambda_c}\hat{\pmb{\sigma}}_3}, 1\right)\begin{bmatrix}\ii\pmb{\sigma}_2&0\\
\mathbf{0}&1
\end{bmatrix}, \qquad \qquad \qquad \quad \lambda\in D_{\lambda_c,-}(\delta),\\
&\begin{bmatrix}\ii \pmb{\sigma}_2&0\\
\mathbf{0}&1
\end{bmatrix}^{-1}{\rm diag}\left(\ee^{\ii N\widehat{h}_{\lambda_c}\hat{\pmb{\sigma}}_3}, 1\right){\rm diag}\left(\ii, -\ii, 1\right)\widehat{\mathbf{Q}}_{f, n}(\lambda; X, T)\begin{bmatrix}0&\ee^{N\Omega_n}&0\\
-\ee^{-N\Omega_n}&0&0\\0&0&1
\end{bmatrix}\\\cdot &{\rm diag}\left(-\ii, \ii, 1\right){\rm diag}\left(\ee^{-\ii N\widehat{h}_{\lambda_c}\hat{\pmb{\sigma}}_3}, 1\right)\begin{bmatrix}\ii\pmb{\sigma}_2&0\\
\mathbf{0}&1
\end{bmatrix}, \quad \lambda\in D_{\lambda_c,+}(\delta)
\end{split}\right.\\
\mathbf{U}_{\lambda_d}(\zeta)&={\rm diag}\left(\ee^{\ii N\widehat{h}_{\lambda_d}\hat{\pmb{\sigma}}_3}\right)\widehat{\mathbf{Q}}_{f,n}(\lambda; X, T){\rm diag}\left(\ee^{-\ii N\widehat{h}_{\lambda_d}\hat{\pmb{\sigma}}_3}, 1\right), \qquad \qquad \quad \lambda\in D_{\lambda_d}(\delta).\\
\end{aligned}
\end{equation}
In view of this definition, we can set the inner parametrix in the neighborhood of $\lambda_c$ and $\lambda_d$ to match the original jump condition. That is
\begin{equation}
\begin{aligned}
\widehat{\mathbf{Q}}_{f, n}^{\lambda_c}:&=\left\{\begin{split}&\mathbf{K}_{-}(\lambda; X, T){\rm diag}\left(\ee^{-\ii N\widehat{h}_{\lambda_c}\hat{\pmb{\sigma}}_3}, 1\right){\rm diag}\left(N^{-\ii \frac{p}{2}\hat{\pmb{\sigma}}_3}, 1\right){\rm diag}\left(\ii, -\ii, 1\right)\mathbf{H}^{\lambda_c}(\lambda)\mathbf{U}(N^{1/2}f_{\lambda_{c}})\\
\cdot &\begin{bmatrix}\ii \pmb{\sigma}_2&0\\
\mathbf{0}&1
\end{bmatrix}^{-1}{\rm diag}\left(\ee^{\ii N\widehat{h}_{\lambda_c}\hat{\pmb{\sigma}}_3}\right){\rm diag}\left(-\ii, \ii, 1\right),\quad \lambda\in D_{\lambda_c, -}(\delta),\\
&\mathbf{K}_{-}(\lambda; X, T){\rm diag}\left(\ee^{-\ii N\widehat{h}_{\lambda_c}\hat{\pmb{\sigma}}_3}, 1\right){\rm diag}\left(N^{-\ii \frac{p}{2}\hat{\pmb{\sigma}}_3}, 1\right){\rm diag}\left(\ii, -\ii, 1\right)\mathbf{H}^{\lambda_c}(\lambda)\mathbf{U}(N^{1/2}f_{\lambda_{c}})\\
\cdot &\begin{bmatrix}\ii \pmb{\sigma}_2&0\\
\mathbf{0}&1
\end{bmatrix}^{-1}{\rm diag}\left(\ee^{\ii N\widehat{h}_{\lambda_c}\hat{\pmb{\sigma}}_3}, 1\right){\rm diag}\left(-\ii, \ii, 1\right)\begin{bmatrix}0&\ee^{N\Omega_n}&0\\-\ee^{-N\Omega_n}&0&0\\0&0&1
\end{bmatrix},\quad \lambda\in D_{\lambda_c, +}(\delta),
\end{split}\right.\\
\widehat{\mathbf{Q}}_{f, n}^{\lambda_d}:&=\mathbf{K}(\lambda; X, T){\rm diag}\left(\ee^{-\ii N\widehat{h}_{\lambda_d}\hat{\pmb{\sigma}}_3}, 1\right){\rm diag}\left(N^{\ii \frac{p}{2}\hat{\pmb{\sigma}}_3}, 1\right)\mathbf{H}^{\lambda_d}(\lambda)\mathbf{U}(N^{1/2}f_{\lambda_d}){\rm diag}\left(\ee^{\ii N\widehat{h}_{\lambda_d}\hat{\pmb{\sigma}}_3}, 1\right),\\
&\qquad \qquad \qquad \qquad \qquad \qquad \qquad \qquad \qquad \qquad \lambda\in D_{\lambda_d}(\delta),
\end{aligned}
\end{equation}
where $\mathbf{H}^{\lambda_c}(\lambda)$ and $\mathbf{H}^{\lambda_d}(\lambda)$ is replacing the parameters $a, b, f_a $ and $f_b$in Eq.\eqref{eq:HaHb} to $\lambda_c, \lambda_d, f_{\lambda_c}, f_{\lambda_d}$ respectively. From the result in the Appendix \eqref{App:large-chi}, we know the \lmn{dominant} error comes from the error $\widehat{\mathbf{Q}}^{\lambda_c}_{f, n}(\lambda; X, T)\left(\widehat{\mathbf{Q}}^{\rm out}_{f, n}(\lambda; X, T)\right)^{-1}. $ Thus we give this formula in this case
\begin{equation}\label{eq:hatQhatQi}
\begin{aligned}
\widehat{\mathbf{Q}}^{\lambda_c}_{f, n}(\lambda; X, T)\left(\widehat{\mathbf{Q}}_{f, n}^{\rm out}(\lambda; X, T)\right)^{-1}&=\mathbf{Y}^{\lambda_c}(\lambda; X, T)\mathbf{U}(\zeta_{\lambda_c}){\rm diag}\left(\zeta_{\lambda_c}^{\ii p\hat{\pmb{\sigma}}_3}, 1\right)\mathbf{Y}^{\lambda_c}(\lambda; X, T)^{-1}, \\
\widehat{\mathbf{Q}}^{\lambda_d}_{f, n}(\lambda; X, T)\left(\widehat{\mathbf{Q}}_{f, n}^{\rm out}(\lambda; X, T)\right)^{-1}&=\mathbf{Y}^{\lambda_d}(\lambda; X, T)\mathbf{U}(\zeta_{\lambda_d})\zeta_{\lambda_d}^{\ii p\pmb{\sigma}}\mathbf{Y}^{\lambda_d}(\lambda; X, T)^{-1},
\end{aligned}
\end{equation}
where
\begin{equation}
\begin{split}
\mathbf{Y}^{\lambda_c}(\lambda; X, T):&=\mathbf{K}_{-}(\lambda; X, T) {\rm diag}\left(N^{-\frac{1}{2}\ii p\hat{\pmb{\sigma}}_3}, 1\right){\rm diag}\left(\ee^{-\ii N \widehat{h}_{\lambda_c}\hat{\pmb{\sigma}}_3}, 1\right){\rm diag}\left(\ii, -\ii, 1\right)\mathbf{H}^{\lambda_c}(\lambda; X, T),\\
\mathbf{Y}^{\lambda_d}(\lambda; X, T):&=\mathbf{K}(\lambda; X, T) {\rm diag}\left(N^{\frac{1}{2}\ii p\hat{\pmb{\sigma}}_3}, 1\right){\rm diag}\left(\ee^{-\ii N \widehat{h}_{\lambda_d}\hat{\pmb{\sigma}}_3}, 1\right)\mathbf{H}^{\lambda_d}(\lambda; X, T).
\end{split}
\end{equation}
\subsection{Inner parametrix in the neighborhood of $a_n$ and $a_n^*$}
Based on the idea in \cite{BilmanM-21}, the parametrices can be constructed with the Airy function when $\lambda$ is in the neighbourhood of $a_n$ and $a_n^*$. And the leading error lies in the boundary $\partial D_{a_n}$ and $\partial D_{a_n^*}$, and they satisfy the following condition
\begin{equation}
\begin{split}
&\sup\limits_{\lambda\in\partial D_{a_n}}\|\widehat{\mathbf{Q}}^{a_n}(\lambda; X, T)\widehat{\mathbf{Q}}^{\rm out}_{f, n}(\lambda; X, T)-\mathbb{I}\|=\mathcal{O}(N^{-1}), \\
&\sup\limits_{\lambda\in\partial D_{a_n^*}}\|\widehat{\mathbf{Q}}^{a_n^*}(\lambda; X, T)\widehat{\mathbf{Q}}^{\rm out}_{f, n}(\lambda; X, T)-\mathbb{I}\|=\mathcal{O}(N^{-1}).
\end{split}
\end{equation}
Then the global parametrix $\widetilde{\mathbf{Q}}_{f,n}(\lambda; X, T)$ can be given by
\begin{equation}
\begin{split}
\widetilde{\mathbf{Q}}_{f,n}(\lambda; X, T):=\left\{\begin{split}
&\widehat{\mathbf{Q}}^{\lambda_c}_{f,n}(\lambda; X, T), \quad \lambda\in D_{\lambda_c}(\delta),\\
&\widehat{\mathbf{Q}}^{\lambda_d}_{f,n}(\lambda; X, T), \quad \lambda\in D_{\lambda_d}(\delta),\\
&\widehat{\mathbf{Q}}^{a_n}_{f,n}(\lambda; X, T), \quad \lambda\in D_{a_n}(\delta),\\
&\widehat{\mathbf{Q}}^{a_n^*}_{f,n}(\lambda; X, T), \quad \lambda\in D_{a_n^*}(\delta),\\
&\widehat{\mathbf{Q}}^{\rm out}_{f,n}(\lambda; X, T), \quad \lambda\in \mathbb{C}\setminus \left(\Sigma_{\rm u}\cup \Sigma_{\rm d}\cup I\cup \overline{D_{\lambda_c}(\delta)\cup D_{\lambda_d}(\delta)\cup D_{a_n}(\delta)\cup D_{a_n^*}(\delta)}\right).
\end{split}\right.
\end{split}
\end{equation}
\subsection{Error analysis}
To give the leading order in the non-oscillatory region, we begin to analyze the error between the global parametrix $\widetilde{\mathbf{Q}}_{f, n}^{\rm out}(\lambda; X, T)$ and the matrix $\widehat{\mathbf{Q}}_{f, n}(\lambda; X, T)$. Define the following error term
\begin{equation}
\mathbf{F}_{f, n}(\lambda; X, T):=\widehat{\mathbf{Q}}_{f, n}(\lambda; X, T)\widetilde{\mathbf{Q}}_{f, n}(\lambda; X, T)^{-1}.
\end{equation}
Similar to the analysis in Eq.\eqref{eq:VF-1}, in the outer of domain, the jump of $\mathbf{V}^{F}_{f,n}(\lambda; X, T)$ changes into
\begin{equation}
\mathbf{V}^{F}_{f, n}=\widetilde{\mathbf{Q}}_{f,n}(\lambda; X, T)\widehat{\mathbf{Q}}_{f,n, -}(\lambda; X, T)\widehat{\mathbf{Q}}_{f,n, +}(\lambda; X, T)\widetilde{\mathbf{Q}}_{f,n}(\lambda; X, T)^{-1}.
\end{equation}
Set the jump about $F_{f,n}(\lambda; X, T)$ as $\Sigma_{F_{f,n}}$. In the outer domain, there exist a positive constant $\nu$ such that
\begin{equation}
\|\mathbf{V}^{F}_{f,n}(\lambda; X, T)-\mathbb{I}\|_{\lambda\in\Sigma_{F_{f,n}}\setminus\left(\partial D_{\lambda_c}\cup \partial D_{\lambda_d}\right)}=\mathcal{O}\left(\ee^{-\nu N}\right).
\end{equation}
And in the boundary of $\partial D_{\lambda_c}$, $\partial D_{\lambda_d}, \partial D_{a_n}$ and $\partial D_{a_n^*}$, the jump condition becomes
\begin{equation}
\begin{split}
\mathbf{V}^{F}_{f,n}&=\widehat{\mathbf{Q}}^{\lambda_c,\lambda_d}_{f,n}(\lambda; X, T)\left(\widehat{\mathbf{Q}}^{\rm out}_{f,n}(\lambda; X, T)\right)^{-1},\\
\mathbf{V}^{F}_{f,n}&=\widehat{\mathbf{Q}}^{a_n,a_n^*}_{f,n}(\lambda; X, T)\left(\widehat{\mathbf{Q}}^{\rm out}_{f,n}(\lambda; X, T)\right)^{-1}.
\end{split}
\end{equation}
With the aid of idea in \cite{BilmanM-21}, when $\lambda\in\partial D_{a_n}(\delta)\cup D_{a_n^*}(\delta)$, $\mathbf{V}^{F}_{f,n}-\mathbb{I}=\mathcal{O}(N^{-1})$, and in the boundary of $D_{\lambda_c}(\delta)$ and $D_{\lambda_d}(\delta)$, we have $\mathbf{V}^{F}-\mathbb{I}=\mathcal{O}(N^{-1/2})$. Thus the leading order error appears in the boundary of $D_{\lambda_c}(\delta)$ and $D_{\lambda_d}(\delta)$, we only need to give this error, which is shown in the next subsection.
\subsection{Asymptotics in the non-oscillatory region}
For $(X, T)$ lie in non-oscillatory region, the leading order term changes into
\begin{equation}\label{eq:q-no-1}
\begin{split}
q_{i}(X, T)&=\lim\limits_{\lambda\to\infty}2c_i^*\lambda\left(\widetilde{\mathbf{N}}_{f,n}\right)_{12}\\
&=\lim\limits_{\lambda\to\infty}2c_i^*\lambda\left(\widehat{\mathbf{Q}}_{f,n}\ee^{-Ng(\lambda)}\right)\\
&=\lim\limits_{\lambda\to\infty}2c_i^*\lambda\left(\mathbf{F}_{f,n}\widetilde{\mathbf{Q}}_{f,n}\right)_{12}\\
&=\lim\limits_{\lambda\to\infty}2c_i^*\lambda\left(F_{f,n,11}\widehat{Q}_{f,n,12}^{\rm out}+F_{f,n,12}\right), \quad (i=1,2).
\end{split}
\end{equation}
From the expression of Eq.\eqref{eq:hatQout}, the first term of Eq.\eqref{eq:q-no-1} can be given as
\begin{equation}\label{eq:lead-no}
\lim\limits_{\lambda\to\infty}2c_i^*\lambda F_{f,n,11}\widehat{Q}_{f,n,12}^{\rm out}=-c_i^*{\rm Im}(a_n)\ee^{N\Omega_n-\ii\mu}.
\end{equation}
The calculation to the second term is similar to the Appendix \ref{App:large-chi}, and the algebraic decay term is
\begin{equation}
\lim\limits_{\lambda\to\infty}2c_i^*\lambda F_{f,n,12}=-\frac{1}{\pi \ii}c_i^*\int_{\partial D_{\lambda_c}\cup \partial D_{\lambda_d}}V_{f,n,12}^{F}(\lambda')d\lambda'+\mathcal{O}(N^{-1}).
\end{equation}
From the definition of Eq.\eqref{eq:hatQhatQi}, we know
\begin{equation}
\begin{split}
V^{F}_{f,n,12}=&\frac{\ii}{2N^{1/2}f_{\lambda_c}}\ee^{2k_{-}+2\ii N \widehat{h}_{\lambda_c}}N^{\ii p}\left(\lambda-\lambda_d\right)^{2\ii p}\left(\frac{\lambda-\lambda_c}{f_{\lambda_c}}\right)^{-2\ii p}J_{-,12}^2r\\
+&\frac{\ii}{2N^{1/2}f_{\lambda_c}}\ee^{-2k_{-}-2\ii N \widehat{h}_{\lambda_c}}N^{-\ii p}\left(\lambda-\lambda_d\right)^{-2\ii p}\left(\frac{\lambda-\lambda_c}{f_{\lambda_c}}\right)^{2\ii p}J_{-,11}^2s,\quad \lambda\in \partial D_{\lambda_c}(\delta),\\
V^{F}_{f,n,12}=&\frac{-\ii}{2N^{1/2}f_{\lambda_d}}\ee^{-2k-2\ii N \widehat{h}_{\lambda_d}}N^{\ii p}\left(\lambda-\lambda_c\right)^{2\ii p}\left(\frac{\lambda-\lambda_d}{f_{\lambda_d}}\right)^{-2\ii p}J_{11}^2r\\
-&\frac{\ii}{2N^{1/2}f_{\lambda_d}}\ee^{2k+2\ii N \widehat{h}_{\lambda_d}}N^{-\ii p}\left(\lambda-\lambda_c\right)^{-2\ii p}\left(\frac{\lambda-\lambda_d}{f_{\lambda_d}}\right)^{2\ii p}J_{12}^2s,\quad \lambda\in\partial D_{\lambda_d}(\delta).\\
\end{split}
\end{equation}
Thus the last result is
\begin{equation}
\begin{split}
&-\frac{1}{\pi\ii}c_i^*\int_{\partial D_{\lambda_c}\cup\partial D_{\lambda_d}}V^{F}_{f,n,12}(\lambda')d\lambda'\\
=&-\ii c_i^*\left(\frac{\ee^{2k_{-}(\lambda_c)+2\ii N \widehat{h}_{\lambda_c}}}{N^{1/2}\sqrt{-\widehat{h}''_{-}(\lambda_c)}}N^{\ii p}(\lambda_d-\lambda_c)^{2\ii p}\left(\sqrt{-\widehat{h}_-''(\lambda_c)}\right)^{2\ii p}J_{-,12}^2(\lambda_c)r\right)\\
&-\ii c_i^*\left(\frac{\ee^{-2k_{-}(\lambda_c)-2\ii N \widehat{h}_{\lambda_c}}}{N^{1/2}\sqrt{-\widehat{h}''_{-}(\lambda_c)}}N^{-\ii p}(\lambda_d-\lambda_c)^{-2\ii p}\left(\sqrt{-\widehat{h}_-''(\lambda_c)}\right)^{-2\ii p}J_{-,11}^2(\lambda_c)s\right)\\
&-\ii c_i^*\left(\frac{\ee^{-2k(\lambda_d)-2\ii N \widehat{h}_{\lambda_d}}}{N^{1/2}\sqrt{\widehat{h}''(\lambda_d)}}N^{\ii p}(\lambda_d-\lambda_c)^{2\ii p}\left(\sqrt{\widehat{h}''(\lambda_d)}\right)^{2\ii p}J_{11}^2(\lambda_d)r\right)\\
&-\ii c_i^*\left(\frac{\ee^{2k(\lambda_d)+2\ii N \widehat{h}_{\lambda_d}}}{N^{1/2}\sqrt{\widehat{h}''(\lambda_d)}}N^{-\ii p}(\lambda_d-\lambda_c)^{-2\ii p}\left(\sqrt{\widehat{h}''(\lambda_d)}\right)^{-2\ii p}J_{12}^2(\lambda_d)s\right)+\mathcal{O}(N^{-1}).
\end{split}
\end{equation}
Moreover, together the above equation and the main leading order Eq.\eqref{eq:lead-no}, we know the asymptotics in the non-oscillatory region changes into Eq.\eqref{eq:q-no}, which completes this calculation.
\section*{Acknowledgements}
The authors sincerely thank Professor Miller for his guidance and help during the visiting on the University of Michigan and valuable suggestions on the finite order solitons to improve this project.

Liming Ling is supported by the National Natural Science Foundation of China (Grant No. 11771151), the Guangzhou Science and Technology Program of China (Grant No. 201904010362), the Fundamental Research Funds for the Central Universities of China (Grant No. 2019MS110); Xiaoen Zhang is supported by the China
Postdoctoral Science Foundation (Grant No. 2020M682692).

\end{document}